\documentclass[a4paper,11pt]{article}
\usepackage{jheppub} 
\usepackage{empheq}
\usepackage{bm}
\usepackage{slashed}
\usepackage[T1]{fontenc} 
\usepackage[utf8]{inputenc}
\usepackage[all]{xy}
\usepackage{rotating}
\usepackage{amsmath,amssymb,amsfonts,amsthm} 
\usepackage{mathtools}
\usepackage{bbm}
\usepackage{dcolumn}
\usepackage{multirow}
\usepackage{mathrsfs}
\usepackage{subcaption}
\usepackage{tikz}
\usetikzlibrary{arrows,decorations.markings}
\usepackage{cleveref}
\crefname{figure}{Fig.}{Figs.}
\usepackage{graphicx,adjustbox}
\usepackage{hyperref}

\usepackage{tikz}
\usepackage{tikz-cd,quiver}
\usetikzlibrary{calc}
\usetikzlibrary{shapes.geometric,fit}
\usepackage{custom_macros}
\usepackage{stmaryrd}
\usepackage{simpler-wick}

\newcommand*\circled[1]{\tikz[baseline=(char.base)]{
  \node[shape=circle,draw,inner sep=1pt, minimum width=0.8cm ] (char) {#1};}}
\newcommand*\rectangled[1]{\tikz[baseline=(char.base)]{
  \node[shape=rectangle,draw,inner sep=1pt, minimum width=0.8cm, minimum height=0.8cm ] (char) {#1};}}

\graphicspath{{./figures_phys/}}

\def\sb{\mathsf{b}} 
\def\sc{\mathsf{c}} 

\newcommand{\Sb}{{\beta\gamma}}
\newcommand{\Ff}{\textrm{Ff}}
\newcommand{\cp}{\chi\psi}
\newcommand{\bc}{\mathsf{bc}}
\newcommand{\jet}{J_{\infty}}

\newcommand{\CVs}{\CV^{\sharp}}
\newcommand{\CVdel}{\CV_{\partial}}
\newcommand{\CVsdel}{\CV^{\sharp}_{\partial}}

\newcommand{\CVFf}{\CV_{\rm Ff}}
\newcommand{\CVbg}{\CV_{\beta \gamma}}
\newcommand{\CVbc}{\CV_{\mathsf{bc}}}
 
\newcommand{\Zhu}{\mathcal{R}} 
\newcommand{\Qzero}{|Q_0|} 
 
\newcommand{\QBRST}{\mathcal{Q}}
\newcommand{\JBRST}{\mathcal{J}}
 
\let\a=\alpha \let\b=\beta \let\g=\gamma \let\d=\delta \let\e=\epsilon 
     \let\k=\kappa
 \let\m=\mu    \let\r=\rho  \let\i=\iota
\let\s=\sigma       
\let\ps=\psi
      
    \let\F=\Phi

\let\ph=\phantom

\let\nn=\nonumber

\title{Affine $\CW$-algebras and Miura maps\\ from 3d $\mathcal N=4$ non-Abelian quiver gauge theories} 
\author[1]{Ioana Coman,}
\author[2,1,3]{Myungbo Shim,}
\author[1,4]{Masahito Yamazaki,}
\author[1]{and Yehao Zhou}

\affiliation[1]{Kavli Institute for the Physics and Mathematics of the Universe (WPI),\\University of Tokyo, Kashiwa, Chiba 277-8583, Japan}
\affiliation[2]{Yau Mathematical Sciences Center, Tsinghua University\\Beijing 100084, China}
\affiliation[3]{Department of Physics and Research Institute of Basic Science, Kyung Hee University,\\Seoul 02447, Korea}
\affiliation[4]{Trans-Scale Quantum Science Institute, 
University of Tokyo, Tokyo 113-0033, Japan}
\emailAdd{ioana.coman@ipmu.jp}
\emailAdd{mbshim@tsinghua.edu.cn}
\emailAdd{masahito.yamazaki@ipmu.jp}
\emailAdd{yehao.zhou@ipmu.jp}

\abstract{We study Vertex Operator Algebras (VOAs) obtained from 
the H-twist of 3d $\mathcal{N}=4$ linear quiver gauge theories.
We find that H-twisted VOAs can be regarded as the ``chiralization'' of the extended Higgs branch: many of the ingredients of the Higgs branch are naturally ``uplifted'' into the VOAs, while conversely the Higgs branch can be recovered as the associated variety of the VOA. We also discuss the connection of our VOA with affine $\mathcal{W}$-algebras.
For example, we construct an explicit homomorphism from an affine $\mathcal{W}$-algebra $\CW^{-n+1}(\mathfrak{gl}_n,f_{\min})$ into the H-twisted VOA for $T^{[2,1^{n-2}]}_{[1^n]}[\mathrm{SU}(n)]$ theories. Motivated by the relation with affine $\mathcal{W}$-algebras, we introduce 
a reduction procedure for the quiver diagram, and use this to give an algorithm to systematically construct novel free-field realizations for  VOAs associated with general linear quivers.}

\begin{document}

\maketitle 

\section{Introduction}

Over decades we have encountered a plethora of mathematical structures in the study of quantum field theories (QFTs), and this has facilitated a successful interplay between physics and mathematics.

Vertex operator algebras (VOAs) \cite{MR0996026,MR1651389,MR1849359} provide one of the mathematical characterizations of the chiral half of two-dimensional conformal field theories (CFTs) \cite{Moore:1988qv}. In addition, they often arise as the boundary edge modes of three-dimensional topological quantum field theories (TQFTs), as exemplified in the celebrated correspondence between the boundary current algebras and the bulk three-dimensional pure Chern-Simons theories \cite{Elitzur:1989nr,Witten:1988hf}.
Since we can obtain TQFTs from the
topological twists of the supersymmetric QFTs (SQFTs) \cite{Witten:1988ze}, it is natural to ask 
what VOAs we obtain when we consider
the boundary theories of the topological twists of three-dimensional 
SQFTs.

In this paper we study VOAs arising on the boundary of topological H-twists (also known as A-twists) of 3d $\mathcal{N}=4$ theories \cite{Costello:2018fnz} (see also \cite{Gaiotto:2016wcv}).
While there have been discussions of the boundary VOAs
in the literature (see e.g.\ \cite{Costello:2018swh,Creutzig:2021ext,Gang:2021hrd,Garner:2022vds,Garner:2022rwe,Yoshida:2023wyt}),
previous studied have either been limited to Abelian gauge theories or case-by-case analysis of non-Abelian gauge theories.
It is therefore left as an important problem to systematically analyze boundary VOAs for general non-Abelian gauge theories.

The goal of this paper is to study VOAs associated with non-Abelian quiver gauge theories. We explicitly write down detailed procedures for defining the H-twisted VOA for a general quiver gauge theory, as a BRST reduction of a product VOA, and then study the BRST cohomologies. 

One of the important ideas of this paper is that the H-twisted VOA should be regarded as a ``chiralization'' of the (extended) quiver variety, which is the (extended) Higgs branch of the theory. Many of the concepts of the geometry of the Higgs branch are naturally ``uplifted'' into the VOAs, and conversely the Higgs branch can be recovered as the ``associated variety'' of the VOA; see \cref{tab:dictionaries} for a dictionary, 
which will be explained in detail in the rest of this paper.

\begin{table}[!ht] 
    \centering
    \begin{tabular}{|c|c|}
    \hline
    Geometry of Quiver Variety & Vertex Operator Algebra \\
    \hline\hline
      $\mathrm{Hom}(V,V)$   & $X$ and $Y$\\
       $\mathrm{Hom}(W,V)$   & $I$ and $J$\\
       Cartan Matrix of $Q$ &  OPE coefficients of BRST invariant currents \\
       (Extended) moment maps & BRST invariant currents
       \\
Relaxed balance condition & Anomaly cancellation with extra Fermi multiplets
       \\
       Paths on $Q$ & BRST closed operators up to Fermion bilinears
       \\
       Path algebra on $Q$ & OPE structures of BRST cohomology
       \\
      Quiver reduction of $Q$ & Miura map
      \\
      \hline
    \end{tabular}
    \caption{Correspondences between the geometry of quiver varieties and VOAs.}
    \label{tab:dictionaries}
\end{table}

We also discuss the connection of our VOAs with affine $\mathcal{W}$-algebras.
We construct an explicit homomorphism from affine $\mathcal{W}$-algebras
into the H-twisted VOA for $T^{\sigma}_{[1^n]}[\mathrm{SU}(N)]$ theories
some for some choices of $\sigma$ and $N$ given in \cref{tab:pair of partition and W-algebra}.
Motivated by the relation with affine $\mathcal{W}$-algebras, we introduce a reduction procedure 
 implementing free-field realizations of the 
 VOAs, thereby recovering 
 many of the results on Miura transformations from the literature
 while also obtaining many new results in greater generality.
 
\begin{table}[ht!]
    \centering
    \begin{tabular}{|c|c|}
    \hline
    $(\s,N)$  & $\CW$-algebra 
    \\
    \hline\hline
   $ ([2,1^{N-2}],N)$ & $ \CW^{-N+1}(\mathfrak{gl}_{N},f_{\min})$
    \\
   $ ([2^2],4)$ & $\CW^{-3}(\mathfrak{gl}_{4},[2^2])$
    \\
   $ ([N-1,1],N)$ & $\CW^{-N+1}(\mathfrak{sl}_{N},f_{\mathrm{sub}})$
   \\
    $ ([1^N],N)$ & $V^{-N+1}(\mathfrak{sl}_N)$
    \\
      \hline
    \end{tabular}
    \caption{Correspondence between $T_{\rho}^{\sigma}[\SU(N)]$ theories associated with pairs of partitions $(\rho=[1^N],\s)$, and the $\CW$-algebras obtained from their boundary VOAs.}
    \label{tab:pair of partition and W-algebra}
\end{table}

While this paper is written primarily in physics language, many of our contents have bearings in mathematics
and require new mathematical results. For this reason, we prepared a companion paper \cite{math_draft}
written in the style of mathematics, and comment on the relations between the two manuscripts where appropriate.\footnote{A similar discussion for a particular example of the Jordan quiver has recently appeared in \cite{arakawa2023hilbert}.}

The rest of this paper is organized as follows (see \cref{fig.flow_chart} for interrelations between different sections).
In \cref{sec:3d_VOA} we begin by introducing 
H-twisted VOAs for general 3d $\mathcal{N}=4$ SQFTs.
We introduce the 3d $\mathcal{N}=4$
quiver gauge theories themselves in \cref{sec:general-quiver}, and highlight the case of linear quiver gauge theories and in particular $T_{\rho}^{\sigma}[\SU(N)]$ theories in \cref{sec:T_rho_sigma}.
In \cref{sec:VOA-quiver} we describe VOAs
associated with a general quiver gauge theories, and 
in \cref{sec:VOA_linear} we will specialize to the case of 
linear quiver gauge theories.
In \cref{sec:variety} we introduce the concept of the associated variety of a vertex algebra (VA), a variety which we conjecture to be identified with an extended Higgs branch.
In \cref{sec:[21..1]} we propose an explicit homomorphism from an affine $\mathcal{W}$-algebra into the H-twisted VOAs for 
$T^{[2, 1^{n-2}]}_{[1^n]}[\SU(n)]$ theories,
and verify our proposal for the $T^{[2,1^2]}_{[1^4]}[\SU(4)]$
theory in \cref{sec:[211]}.
In \cref{sec:[nn]} we discuss another example, specifically the 
$T^{[n,n]}_{[1^{2n}]}[\SU(2n)]$ theory, and construct a homomorphism from a rectangular affine $\mathcal{W}$-algebra into the corresponding H-twisted VOA. 
In \cref{sec:localization} we introduce the localization of quiver vertex algebras and 
make contact with our companion paper \cite{math_draft}.
One of the highlights of this paper is \cref{sec:quiver reduction},
in which we introduce
a reduction procedure for the quiver diagram, wherein we obtain a free-field realization of the VOAs for general quiver gauge theories. We will illustrate this idea with explicit examples in \cref{sec:reduction_examples}.
In \cref{sec:4d-VOA} we discuss the relations of our VOAs with those originating from 
4d $\mathcal{N}=2$ theories.
We conclude in \cref{sec:discussions}
with a summary and comments on future directions.
We also include \cref{sec:glossary} for a glossary of notations, and more appendices with technical materials.

\begin{figure}[htbp]
\centering 
\begin{tikzpicture}[scale=0.8,
every path/.style={thick, >={Stealth[scale=1.5]}},
every node/.style={scale=0.8}]
\tikzset{
   blob/.style={shape= rectangle,rounded corners=5pt,draw,align= center,thick},
   double-blob/.style={shape= rectangle,double,rounded corners=5pt,draw,align= center,thick}
 }

\node[blob] (S2) at (0,0) 
{VOAs for 3d $\mathcal{N}=4$ SQFTs \\ (\cref{sec:3d_VOA})};

\node[blob] (S3) at (6,-2) 
{Quiver gauge theories\\ (\cref{sec:general-quiver})};

\node[blob] (S4) at (6,-5) 
{Linear Quivers\\ and $T_{\rho}^{\sigma}[\SU(N)]$\\(\cref{sec:T_rho_sigma})};

\node[blob] (S5) at (0,-3) 
{VOAs associated with quivers\\ (\cref{sec:VOA-quiver})};

\node[blob] (S6) at (0,-6) 
{VOAs for linear quivers\\(\cref{sec:VOA_linear})};

\node[blob] (S7) at (-6,-3) 
{Higgs branch\\
and associated varieties\\ (\cref{sec:variety})};

\node[blob] (S8) at (0,-9) 
{VOAs for $T^{[2, 1^{n-2}]}_{[1^n]}[\SU(n)]$\\ (\cref{sec:[21..1]})};

\node[blob] (S9) at (5,-9) 
{VOAs for $T^{[n,n]}_{[1^{2n}]}[\SU(2n)]$\\ (\cref{sec:[nn]})};

\node[blob] (S10) at (-6,-6) 
{Localization of quiver VOAs\\ (\cref{sec:localization})};

\node[double-blob] (S11) at (-6,-9) 
{Quiver reductions\\ and Miura Transformations\\
(\cref{sec:quiver reduction})};

\node[blob] (S12) at (0,-13) 
{Examples of \\ quiver reductions \\(\cref{sec:reduction_examples})};

\node[blob] (S13) at (-6,-12) 
{Comparison with 4d $\mathcal{N}=2$ VOAs\\(\cref{sec:4d-VOA})};

\draw[->] (S2) -- (S5);
\draw[->] (S3) -- (S5);
\draw[->] (S2) -- (S7);
\draw[->] (S7) -- (S10);
\draw[->] (S5) -- (S10);
\draw[->] (S10) -- (S11);
\draw[->] (S11) -- (S12);
\draw[<->, dashed] (S8) -- (S12);
\draw[<->, dashed] (S9) -- (S12);
\draw[->] (S5) -- (S6);
\draw[->] (S6) -- (S8);
\draw[->] (S6) -- (S9);
\draw[->] (S3) -- (S4);
\draw[->] (S4) -- (S6);
\draw[->, dotted, thick] (S7.west) to [bend right=35] (S13.west);

\end{tikzpicture}
\caption{Approximate dependencies of the sections of this paper.}
\label{fig.flow_chart}
\end{figure}
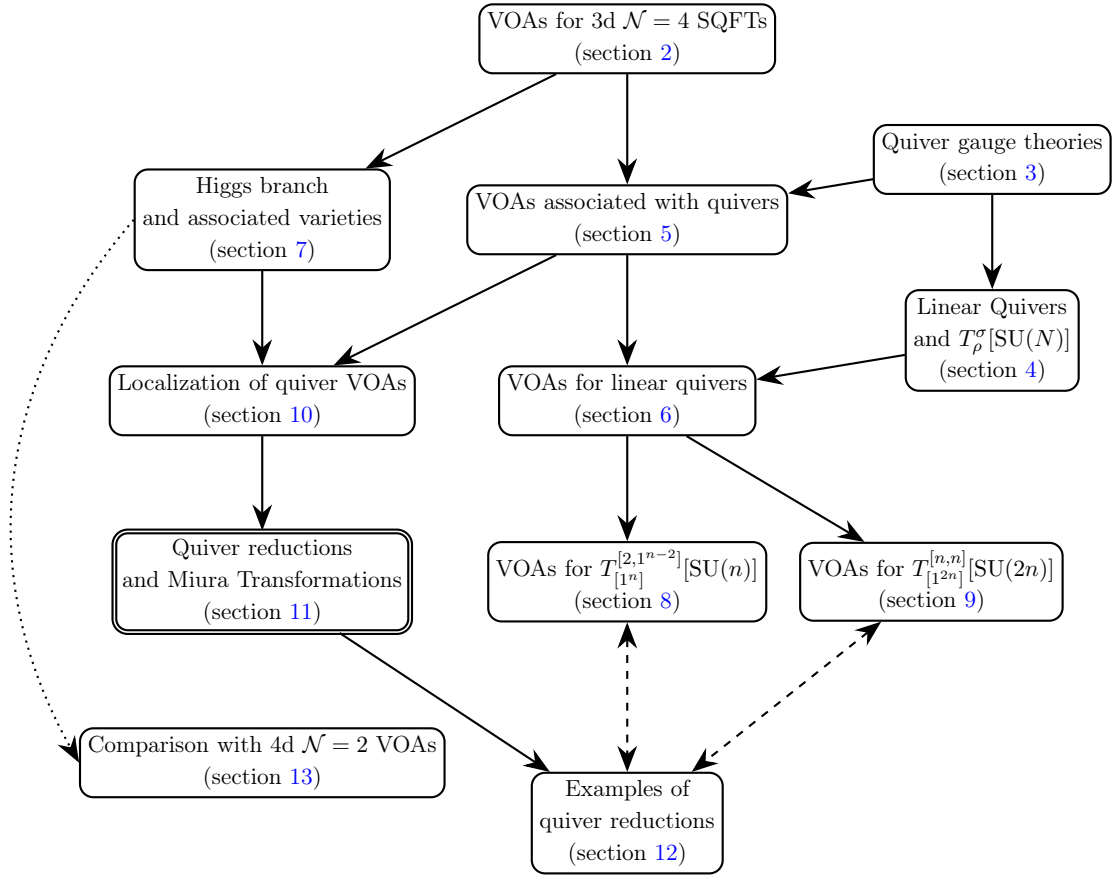

\section{VOAs from 3d \texorpdfstring{$\mathcal{N}=4$}{N=4} SQFTs}
\label{sec:3d_VOA}

\subsection{Setup}

Suppose that we are given a 3d $\mathcal{N}=4$ supersymmetric quantum field theory (SQFT) $\mathcal{T}$
with a Lagrangian description. This theory is defined by 
\begin{itemize}
    \item a gauge group $G$ ($\mathcal{N}=4$ vector multiplet),
    \item and a symplectic representation $M$ of $G$ ($\mathcal{N}=4$ hypermultiplet).
\end{itemize}

The vacuum moduli space of the theory consists of the Higgs branch $\mathcal{M}_H$,
the Coulomb branch $\mathcal{M}_C$, as well as the mixed branches.
This theory has $\SO(4)_R\simeq \SU(2)_H \times \SU(2)_C$ R-symmetry, 
where $\SU(2)_H$ ($\SU(2)_C)$ is the symmetry acting non-trivially on the Higgs branch (Coulomb branch);
the Higgs/Coulomb branches are hyperK\"{a}hler manifolds and the $\SU(2)_{H,C}$ symmetries rotate the triplet of
complex structures therein. The two branches $\mathcal{M}_H$ and $\mathcal{M}_C$, and hence their corresponding symmetries $\SU(2)_H$ and $\SU(2)_C$,
are exchanged by the celebrated 3d mirror symmetry \cite{Intriligator:1996ex}.

We consider topological twists of the theory $\mathcal{T}$. On a 3-manifold $M$, 
the $\SO(3)_M \simeq \SU(2)_M$ symmetry acting on the tangent space of $M$ is mixed either 
with the $\SU(2)_H$ or $\SU(2)_C$ symmetry,
and correspondingly we have an H-twisted TQFT or a C-twisted TQFT; these theories
are also called A-twisted and B-twisted TQFT, respectively.\footnote{The H/A-twist is the dimensional reduction of the Donaldson twist \cite{Witten:1988hf} in four dimensions.
The early references on C/B-twist include \cite{Blau:1996bx,Rozansky:1996bq}. We can also consider a holomorphic twist \cite{Closset:2012ru,Aganagic:2017tvx,Costello:2020ndc}
along a holomorphic curve, deformation of which gives rise to topological A/B-twists.}
In this paper we will for the most part concentrate on the H-twisted/A-twisted TQFTs.

Now the H-twisted VOA $\CV_H$ (or hereafter simply $\CV$), the central player of this paper, is obtained by 
considering the H-twisted TQFT on the geometry of the form $\mathbb{R}_{t \ge 0}\times \mathcal{C}$,
where $\mathcal{C}$ is a holomorphic curve\footnote{The holomorphic structure here is more generally taken as a transverse holomorphic foliation \cite{Klare:2012gn,Closset:2012ru,Aganagic:2017tvx}.} --- the VOA is defined on the holomorphic boundary $\{t=0 \}\times \mathcal{C}$.
In more detail, we choose an $\mathcal{N}=(0,4)$-preserving boundary condition consistent with the supercharge preserved under the H-twist,
with Neumann boundary conditions for both the vector and hyper multiplets \cite{Costello:2018fnz,Gaiotto:2017euk}.
In order to cancel boundary gauge anomalies, we also need to choose boundary Fermi multiplets $\textrm{Ff}[R]$ in a representation $R$ of the bulk gauge group $G$.

\subsection{VOA as BRST reduction}

It was proposed in \cite{Costello:2018fnz} (see also \cite{Garner:2022vds}) that the H-twisted VOA $\CV$ in question is determined from the BRST reduction
(to be discussed in more detail in the next subsection):
\begin{align}
&\CV = \CV[G, M, R] :=  \mathcal{V}^{\sharp}[M, R]  \Kquotient G = 
    \left\{ \mathcal{V}^{\sharp}[M, R] \otimes \CVbc[G] , \QBRST \right \} \,, 
    \label{def:BRSTcomplex}
    \\
&\mathcal{V}^{\sharp}=\mathcal{V}^{\sharp}[M, R]  :=    \CVbg[M] \otimes \CVFf[R]  \;,
\end{align}
where the VOAs  $\CVbg[M], \CVFf[R], \CVbc[G]$ are given respectively by
\begin{itemize}
\item $\CVbg[G]$ is a set of $\beta-\gamma$ bosons (symplectic bosons \cite{Goddard:1987td}) $(X_{i}, Y_{i})_{i=1, \dots, \textrm{dim}\, M}$ originating from 3d $\mathcal{N}=4$ hypermultiplets
in a representation $M$ of $G$, with non-trivial OPE given by 
\beq
X_{i}(z)Y_{j}(w)\sim\frac{\d_{i,j}}{z-w} \;.
\eeq
The stress-energy operator is 
\begin{align}
T_{\Sb}= \frac{1}{2} \sum_{i=1}^{\textrm{dim}\, M} \mathrm{Tr}(\normord{X_i \partial Y_i} - \normord{\partial X_i Y_i}) \;,
\end{align} 
and hence the conformal dimensions of $(X_i, Y_i)$ are $(1/2, 1/2)$,
and the total central charge is $c_{\beta\gamma}=-\textrm{dim}\, M$.

\item $\CVFf[R]$ is a set of free fermions $(\F_{i}, \Psi_{i})_{i=1, \dots, \textrm{dim}\, R}$ arising from the free Fermi multiplets on the boundary,  with OPE
\beq
\F_{i}(z) \Psi_{j}(w)\sim\frac{\d_{i,j}}{z-w} \;.
\eeq
We choose the stress-energy operator to be 
\begin{align}
T_{\Ff}=-\frac{1}{2}\sum_{i=1}^{\textrm{dim}\, R} \mathrm{Tr}(\normord{\F_{i}\partial \Psi_{i}}- \normord{\partial \F_{i} \Psi_{i}}) \;,
\end{align}
and hence the conformal dimensions of $(\F_i, \Psi_i)$ are $(1/2, 1/2)$,
and the total central charge is $c_{\Ff}=\textrm{dim}\, R$.

\item $\CVbc[G]$ is the $\mathsf{bc}$-ghost $(\mathsf{b}^{\alpha}, \mathsf{c}_{\alpha})_{\alpha=1, \dots, \textrm{dim}\, G}$ in the adjoint representation of $G$.
The OPEs are given by 
\beq
\mathsf{b}^{\alpha}(z)\mathsf{c}_{\beta}(w)\sim\frac{\delta^\alpha_{\beta}}{z-w} \;.
\eeq
We choose the stress-energy operator to be 
\begin{align}
    T_{\mathsf{bc}}&=-\sum_{\alpha=1}^{\textrm{dim}\, G} \mathrm{Tr}(\normord{\mathsf{b}^{\alpha} \partial \mathsf{c}_{\alpha}} ) \;,
\end{align}
and consequently conformal dimensions of $(\mathsf{b}, \mathsf{c})$ to be $(1, 0)$.
The total central charge is $c_{\bc}=-2 \,\textrm{dim}\, G$.

\end{itemize}
This trio defines a product VOA $\mathcal{V}^{\sharp}[M, R] \otimes \CVbc[G]
= \CVbg[M] \otimes \CVFf[R] \otimes \CVbc[G]$. We denote by $\QBRST$ a supercharge defined below, with respect to which the BRST reduction is performed. 

\bigskip In general there will be boundary gauge anomalies at $\{t=0\}\times\CC$ coming separately from the $\beta\gamma$-systems, the free fermions $\Ff$ and the $\bc$-ghosts. When combined, these contributions to the boundary gauge anomaly cancel within $\mathcal{V}^{\sharp}[M, R] \otimes \CVbc[G]$. This cancellation and its conditions will be discussed further in section \ref{subsection:anomaly}. Note however that in general, there is no unique way to cancel such anomalies, and the definition of the VOA depends on the specific choice of the boundary anomaly-canceling matter contents in a representation $R$; different choices of $R$
will lead to different VOAs with different characters.

\subsection{BRST reduction}\label{subsec:BRST reduction}
In this subsection we discuss the BRST reduction in \eqref{def:BRSTcomplex}. 
This is the VOA counterpart of imposing the Gauss law for the gauge symmetry.

Let us denote the Lie algebra associated to the group $G$ by $\mathfrak{g}$.
Since the 3d $\mathcal{N}=4$ SQFT $\CT$ has a gauge symmetry $G$, we have a corresponding $\mathfrak{g}$-current, which is promoted to a holomorphic current algebra $J_{\mathcal{V}^{\sharp}}^{\alpha} (z)$ on the the boundary VOA $\mathcal{V}^{\sharp}$,
satisfying the OPE
\begin{equation}
    J_{\mathcal{V}^{\sharp}}^{\alpha} \, (z)  J_{\mathcal{V}^{\sharp}}^{\beta} \, (w) 
    \sim 
    \frac{k_{\CVs} \delta_{\alpha,\beta}}{(z-w)^2} + \sum_\gamma \, i f^{\alpha\beta}{}_\gamma \, \frac{J_{\mathcal{V}^{\sharp}}^\gamma (w)}{z-w}  \;,
\end{equation}
where $f^{\alpha\beta}{}_\gamma$ are the structure constants for the Lie algebra $\mathfrak{g}$ and we have chosen a basis for $\mathfrak{g}$ such that the Killing form is diagonalized. We have adopted the standard notation in OPE that $\sim$ represents an equivalence up to non-singular terms. 
The integer $k_{\CVs}$ is the level of the current algebra. 

The $\mathsf{bc}$ ghost system $\mathcal{V}_{\mathsf{bc}}$ in itself has an affine $\mathfrak{g}$ current algebra 
\begin{equation}
    J_{\mathsf{bc}}^\alpha (z) = f^{\alpha\beta}{}_\gamma \, \sc_\beta \sb^\gamma (z) ~
\end{equation}
with level $k_{\bc}=2\mathsf{h}^\vee_\mathfrak{g}$,  
where $\mathsf{h}_{\mathfrak{g}}^\vee$ is the dual Coxeter number of the Lie algebra $\mathfrak{g}$, defined by the relation $f^{\alpha \gamma \delta} f^{\beta}{}_{\gamma \delta} = 2\mathsf{h}_{\mathfrak{g}}^\vee \delta^{\alpha \beta}$.  

Now suppose that the $\mathfrak{g}$-current algebra is critical, 
i.e.\, 
\begin{align} \label{k_k}
 k_{\CVs}=-2\mathsf{h}_{\mathfrak{g}}^\vee\;. 
\end{align}
so that we have $k_{\CVs} + k_{\bc} =0$.
(We will relate this condition to the anomaly cancellation condition below.)
We can then define the supercharge $\QBRST$ on the product VOA $\mathcal{V}^{\sharp}\otimes \mathcal{V}_{\mathsf{bc}}$ as
\begin{align}
    \label{Q_BRST}
    \QBRST  := \oint \frac{dz}{2\pi i} \, \JBRST(z) \;,
\end{align}
where the BRST current $\JBRST$ is defined by \cite{Karabali:1989dk,Beem:2013sza}
\begin{align}
    \JBRST(z) := \mathsf{c}_\alpha \left(J_{\mathcal{V}^{\sharp}}^{\alpha} + \frac{1}{2}J_{\mathsf{bc}}^{\alpha}  \right) (z) ~.
    \label{def:BRSTcurrent}
\end{align}
We can then verify the nilpotency of $\QBRST$, and use this as a supercharge in the BRST reduction of \eqref{def:BRSTcomplex}:
the boundary VOA  $\mathcal{V}$ is
given by the $\QBRST$-cohomology on the vertex operator $\mathcal{V}^{\sharp}\otimes \mathcal{V}_{\mathsf{bc}}$,
i.e.\ $\QBRST$-closed operators up to equivalence relation defined by $\QBRST$-exact terms.
In mathematical language, the BRST complex defines a semi-infinite complex $C^{\frac{\infty}{2}}(\mathcal{V}^{\sharp} \otimes \mathcal{V}^{\bullet}_{\mathsf{bc}}, \QBRST)$,
and this defines a semi-infinite cohomology \cite{MR0740035}
$H^{\frac{\infty}{2}+\bullet}(\mathcal{V}^\sharp)$, where the grading $\bullet$ is given by the ghost number; see 
\cref{sec:BRST reduction} for further explanation. 

\section{General quiver gauge theories}
\label{sec:general-quiver}

As a preparation for our discussions of VOAs, we first summarize some basics of 
quiver gauge theories themselves.
Readers who are familiar with quiver varieties can skip this section, and come back here when necessary for notations.

\subsection{Matter contents from quiver diagrams}

A quiver $Q$ is an oriented graph consisting of a finite set of nodes $Q_0$ and a finite set of arrows $Q_1$ together with two maps $\mathsf h,\mathsf t: Q_1\to Q_0$ sending an arrow to its head and tail respectively.

The quiver diagram encapsulates the matter contents in the following manner:

\begin{itemize}

\item A vertex $i\in Q_0$ of the quiver represents a 3d $\mathcal{N}=4$ vectormultiplet; the gauge group is $\U(v_i)$  if the rank at the $i^{\rm th}$ vertex is $v_i$.
The total gauge group is thus  $G=\prod_{i\in Q_0} \U(v_i)$.

\item An edge $a\in Q_1$ of the quiver represents a 3d $\mathcal{N}=4$ hypermultiplet. There are two types of hypermultiplets:
    \begin{itemize}
    \item An edge connecting two circled vertices of the quiver represents a bifundamental hypermultiplet with respect to the gauge groups on the two ends of the edge.
    
    \item An edge connecting a circular vertex to a squared vertex represents a fundamental hypermultiplet charged under the $\U(w_i)$ flavor symmetry (\cref{fig:general quiver node for anomaly}).
    \end{itemize}
\end{itemize}

 As is clear from these discussions, we will in general need to specify the dimension vectors $\mathbf{v}=\{ v_i \}$ and $\mathbf{w}=\{ w_i \}$.

\begin{figure}[ht!]
    \centering
    \begin{tikzpicture}[scale=.75, every node/.style={scale=0.8}]  
        \draw[] (-3.5,2.7) circle (.7);
        \draw[] (-6,2.7) circle (.7);
        \draw[] (-3.5,0) circle (.7);
        \draw[] (-3.5,.7) -- (-4.8,2);
        \draw[] (-3.5,-.7) -- (-3.5,-2);
        \draw[] (-3.5,.7) -- (-6,2);
        \draw[] (-3.5,.7) -- (-3.5,2);
        \draw[] (-4.15,-2) rectangle (-2.85,-3.3);
        \node at (-3.5,0) {$v_{i}$};
       \node at (-3.5,-2.65) {$w_{i}$};
        \node at (-4.8,2.65) {$\cdots$};
        \node at (-3.5,2.65) {$v_{k}$};
        \node at (-6,2.65) {$v_{j}$};  
        \node[circle,minimum size=1.5cm] (vi) at (-3.5,0){};
        \draw[] (vi) to [in=160, out=200, looseness=5](vi);
        \draw[] (vi) to [in=150, out=210, looseness=8](vi);
        \draw[dotted, thick] (-5.5,0) -- (-5.1,0);
    \end{tikzpicture}
    \caption{Local neighborhood of a general quiver around the $i^{\rm th}$ vertex}
    \label{fig:general quiver node for anomaly}
\end{figure}
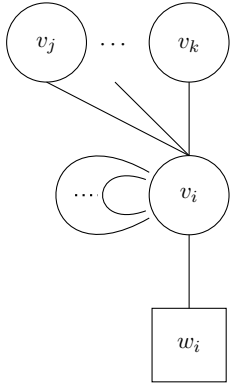

For later purposes we introduce the balancing condition \cite{Gaiotto:2008ak}
on the quiver. We will see in \cref{subsection:anomaly} that
this condition is equivalent to the vanishing of the $\SU$ gauge anomalies on the boundary.

We define the adjacency matrix $\mathsf Q$ to be
\begin{align}
    \mathsf Q_{i,j}:= \#(a\in Q_1:\mathsf h(a)=j,\mathsf t(a)=i) \;,
\end{align}
and the Cartan matrix $\mathsf C$ to be
\begin{align}
    \mathsf C_{i,j}:= 2\delta_{i,j}-\mathsf Q_{i,j}-\mathsf Q_{j,i} \;.
\end{align}
The \emph{balance condition} is then given by
\begin{equation}
   \label{general_balancing_condition}
   - \mathsf{C} \cdot \mathbf{v} + \mathbf{w}  = \mathbf{0} \;.  
\end{equation}
When this condition is satisfied the quiver is referred to be \emph{balanced}.
While not all the theories discussed in this paper are balanced, we can impose a weaker condition,
which we call the \emph{anomaly cancellation condition}
\begin{equation}
   \label{general_anomaly_cancellable_condition}
   \mathbf{d}:=- \mathsf{C} \cdot \mathbf{v} + \mathbf{w}  \in \mathbb{Z}_{\ge 0}^{Q_0} ~.
\end{equation}

\subsection{Higgs branches as quiver varieties}\label{sec:Higgs_branch}

Let us next describe the Higgs branch of the quiver gauge theories. 
In mathematical language, the Higgs branch is the
Nakajima quiver variety \cite{MR1302318,MR1604167} defined from the same quiver, whose description will be needed in the rest of this paper.

Define the space of representation of $Q$ with gauge dimension vector $\mathbf v$ and framing dimension vector $\mathbf w$ to be
\begin{align}
    \mathrm{Rep}(Q,\mathbf v,\mathbf w):=\bigoplus_{i\in Q_0}\mathrm{Hom}(W_i,V_i)\oplus \bigoplus_{a\in Q_1}\mathrm{Hom}(V_{\mathsf t(a)},V_{\mathsf h(a)}) \;,
\end{align}
where $V=\bigoplus_{i\in Q_0}V_i$ and $W=\bigoplus_{i\in Q_0}W_i$ are $Q_0$-graded vector spaces of dimensions $\dim V=\mathbf v$ and $\dim W=\mathbf w$ respectively. The space of representations $\mathrm{Rep}(\mathbf v,\mathbf w)$ naturally possesses an action of the complexified gauge group
\begin{align} \label{GL_Q}
\GL(\mathbf v):=\prod_{i\in Q_0}\GL(\mathbf v_i)\;,
\end{align}
and this action induces an Hamiltonian $\GL(\mathbf v)$-action on the cotangent bundle $T^*\mathrm{Rep}(\mathbf v,\mathbf w)$. Note that $T^*\mathrm{Rep}(\mathbf v,\mathbf w)$ is naturally isomorphic to the space of representations of the doubled quiver $\overline{Q}$, which has the same set of nodes as $Q$ while its set of arrows is $\overline{Q}_1=Q_1\sqcup Q^*_1$, where $Q^*_1$ is the orientation reverse of $Q_1$. Explicitly
\begin{multline}
    T^*\mathrm{Rep}(\mathbf v,\mathbf w)=\bigoplus_{i\in Q_0}\mathrm{Hom}(W_i,V_i)\oplus \bigoplus_{i\in Q_0}\mathrm{Hom}(V_i,W_i) \\
    \oplus\bigoplus_{a\in Q_1}\mathrm{Hom}(V_{\mathsf t(a)},V_{\mathsf h(a)})\oplus\bigoplus_{a\in Q_1}\mathrm{Hom}(V_{\mathsf h(a)},V_{\mathsf t(a)})\;.
\end{multline}
We introduce the notation for the elements in the direct summands in $T^*\mathrm{Rep}(Q,\mathbf v,\mathbf w)$ as follows
\begin{center}
\begin{tabular}{ | c | c | c | c |} 
 \hline
 $\mathrm{Hom}(V_{\mathsf t(a)},V_{\mathsf h(a)})$ & $\mathrm{Hom}(V_{\mathsf h(a)},V_{\mathsf t(a)})$ & $\mathrm{Hom}(W_i,V_i)$ & $\mathrm{Hom}(V_i,W_i)$  \\  [1ex] 
 \hline
$x_a$ & $y_a$ & $\alpha_i$ & $\beta_i$  \\ [0.5ex] 
 \hline
\end{tabular}
\label{tab:matther contents from quiver}
\end{center}
The moment map $\mu:T^*\mathrm{Rep}(Q,\mathbf v,\mathbf w)\to \mathfrak{gl}(\mathbf v)^*$ for the Hamiltonian $\GL(\mathbf v)$-action on $T^*\mathrm{Rep}(Q,\mathbf v,\mathbf w)$ is
\begin{align}\label{moment_map_general}
    \mu(x,y,\alpha,\beta)=-\sum_{a\in Q_1}[x_a,y_a]-\sum_{i\in Q_0}\alpha_i\beta_i
    \;.
\end{align}
Inside $\mathfrak{gl}(\mathbf v)^*=\bigoplus_{i\in Q_0}\mathfrak{gl}(\mathbf v_i)^*$ there is a linear subspace $\mathcal Z$ which is the span of scalar matrices 
\begin{align}\label{the center Z}
    \mathcal Z:=\bigoplus_{i\in Q_0} \mathbb C\cdot\mathrm{Id}_{V_i} \;.
\end{align}
\begin{definition}
The affine quiver variety $\CM ^0(Q,\mathbf v,\mathbf w)$ and its deformation $\widetilde{\mathcal M}^0(Q,\mathbf v,\mathbf w)$ are defined as the Hamiltonian reductions
\begin{align}\label{def_M0}
    \CM ^0(Q,\mathbf v,\mathbf w):=\mu^{-1}(0)\sslash \GL(\mathbf v)\;,
    \quad \widetilde{\mathcal M}^0(\mathbf v,\mathbf w):=\mu^{-1}(\mathcal Z)\sslash \GL(\mathbf v) \;.
\end{align}
\end{definition}

Alternatively, one can define the extended moment map 
\begin{equation}\label{extended moment map}
\begin{split}
    \widetilde\mu: T^*\mathrm{Rep}(Q,\mathbf v,\mathbf w)\times\mathcal Z&\longrightarrow \mathfrak{gl}(\mathbf v)^*\\
    \textrm{\rotatebox{90}{$\in$}} \qquad\qquad &\qquad \quad \textrm{\rotatebox{90}{$\in$}}\\
    ((x,y,\alpha,\beta),z)&\mapsto \mu(x,y,\alpha,\beta)- z \;.
\end{split}
\end{equation}
Then we obviously have
\begin{align*}
    \mu^{-1}(\mathcal Z)\cong \widetilde\mu^{-1}(0) \;.
\end{align*}
\begin{definition}
A representation $(x,y,\alpha,\beta)\in T^*\mathrm{Rep}(Q,\mathbf v,\mathbf w)$ is called stable if there is no proper $Q_0$-graded linear subspace $S_i\subset V_i$ such that 
\begin{align*}
    \alpha_i(W_i)\subset S_i\;,
    \quad 
    x_a(S_{\mathsf t(a)})\subset S_{\mathsf h(a)}\;,
    \quad 
    y_a(S_{\mathsf t(a)})\subset S_{\mathsf h(a)}\;.
\end{align*}
A representation $(x,y,\alpha,\beta)\in T^*\mathrm{Rep}(Q,\mathbf v,\mathbf w)$ is called costable if there is no proper $Q_0$-graded linear subspace $T_i\subset V_i$ such that 
\begin{align*}
    \beta_i(T_i)=0\;,
    \quad 
    x_a(T_{\mathsf t(a)})\subset T_{\mathsf h(a)}\;,
    \quad 
    y_a(T_{\mathsf t(a)})\subset T_{\mathsf h(a)}\;.
\end{align*}
\end{definition}
Denote by $T^*\mathrm{Rep}(Q,\mathbf v,\mathbf w)^{\mathrm{st}}$ (resp.\  $T^*\mathrm{Rep}(Q,\mathbf v,\mathbf w)^{\mathrm{cst}}$) the locus of stable (resp.\ costable) representations, and for every $\GL(\mathbf v)$-invariant subset $S\subset T^*\mathrm{Rep}(Q,\mathbf v,\mathbf w)$ we set
\begin{align}
    S^{\mathrm{st}}:=S\cap T^*\mathrm{Rep}(Q,\mathbf v,\mathbf w)^{\mathrm{st}}\;, 
    \quad 
    S^{\mathrm{cst}}:=S\cap T^*\mathrm{Rep}(Q,\mathbf v,\mathbf w)^{\mathrm{cst}}\;.
\end{align}
\begin{definition}
The Nakajima quiver variety ${\mathcal M}(Q,\mathbf v,\mathbf w)$ and its deformation $\widetilde{\mathcal M}(Q,\mathbf v,\mathbf w)$ are defined as the Hamiltonian reductions
\begin{align}
    {\mathcal M}(Q,\mathbf v,\mathbf w):=\mu^{-1}(0)^{\mathrm{st}}\sslash \GL(\mathbf v)\;,
    \quad 
    \widetilde{\mathcal M}(Q,\mathbf v,\mathbf w):=\mu^{-1}(\mathcal Z)^{\mathrm{st}}\sslash \GL(\mathbf v) \;.
\end{align}
\end{definition}
Note that ${\mathcal M}(Q,\mathbf v,\mathbf w)$ is a closed subscheme of $\widetilde{\mathcal M}(Q,\mathbf v,\mathbf w)$. In fact, if $\mu^{-1}(\mathcal Z)^{\mathrm{st}}$ is nonempty, then the natural projection $\mu^{-1}(\mathcal Z)^{\mathrm{st}}\to \mathcal Z$ is smooth and the action of $\GL(\mathbf v)$ on $\mu^{-1}(\mathcal Z)^{\mathrm{st}}$ is free. After passing to the quotient we get a smooth morphism 
$p: \widetilde{\mathcal M}(Q,\mathbf v,\mathbf w)\to \mathcal Z$, and ${\mathcal M}(Q,\mathbf v,\mathbf w)=p^{-1}(0)$. Since ${\mathcal M}(Q,\mathbf v,\mathbf w)$ is a good quotient, its dimension is the expected one:
\begin{align}
    \dim {\mathcal M}(Q,\mathbf v,\mathbf w)=\mathbf v\cdot(2\mathbf w-\mathsf C\mathbf v)\;.
\end{align}
Note that the right hand side is always non-negative thanks to the condition \eqref{general_anomaly_cancellable_condition}.

\section{Linear quiver gauge theories and \texorpdfstring{$T_{\r}^{\s}[\SU(N)]$}{T[SU(N)]} Theories} 
    \label{sec:T_rho_sigma}

In this section, we specialize to the case of 3d $\mathcal{N}=4$ gauge theories with 
linear quivers of type A, as shown in \cref{fig:genericAquiver}. We in particular highlight a
family of 3d $\mathcal{N}=4$ theories known as $T_{\r}^{\s}[\SU(N)]$ theories introduced in \cite{Gaiotto:2008ak}.\footnote{The theory was denoted as $T_{\r}^{\r'}[\SU(N)]$ in \cite{Gaiotto:2008ak}, where their notations and ours are related by $\s_{\rm here}=\r'_{\rm there}$ and $\r_{\rm here}=\r_{\rm there}$.} 
In the following we follow the convention of $T_{\r}^{\s}[\SU(N)]$ theories in \cite{Nishioka:2011dq}.

\begin{figure}[ht!]
    \centering
    \begin{tikzpicture}[scale=.75, every node/.style={scale=0.8}]
        \draw[] (-6,0) circle (.7);
        \draw[] (-5.3,0) -- (-4.2,0);
        \draw[] (-3.5,0) circle (.7);
        \draw[] (-2.8,0) -- (-1.75,0);
        \draw[line width=1pt,loosely dotted] (-1.75,0) -- (-.3,0);
        \draw[] (-.3,0) -- (.8,0);
        \draw[] (1.5,0) circle (.7);   
        \draw[] (2.2,0) -- (3.3,0);
        \draw[] (4,0) circle (.7);      
        \draw[] (4,-.7) -- (4,-2);
        \draw[] (1.5,-.7) -- (1.5,-2);
        \draw[] (-3.5,-.7) -- (-3.5,-2);
        \draw[] (-6,-.7) -- (-6,-2);   
        \draw[] (3.35,-2) rectangle (4.65,-3.3);
        \draw[] (0.85,-2) rectangle (2.15,-3.3);
        \draw[] (-4.15,-2) rectangle (-2.85,-3.3);
        \draw[] (-6.65,-2) rectangle (-5.35,-3.3);
        \node at (-6,0) {$v_{\Qzero}$};
        \node at (-3.5,0) {$v_{\Qzero-1}$};
        \node at (1.5,0) {$v_{2}$};
        \node at (4,0) {$v_{1}$};
         \node at (-6,-2.65) {$w_{\Qzero}$};
        \node at (-3.5,-2.65) {$w_{\Qzero-1}$};
        \node at (1.5,-2.65) {$w_{2}$};
        \node at (4,-2.65) {$w_{1}$};
    \end{tikzpicture}
    \caption{A-type linear  quiver diagram with gauge nodes $\{v_i\}_{i\in Q_0}$ and framing nodes $\{w_i\}_{i\in Q_0}$.}    
    \label{fig:genericAquiver}
\end{figure}
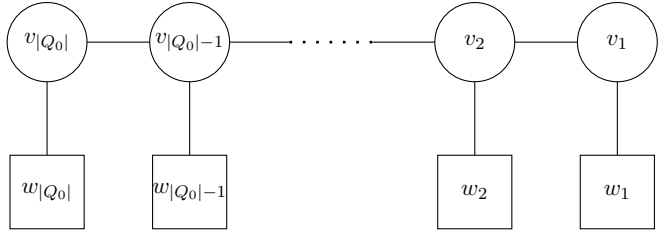

\subsection{Some notations}

In the linear quiver of 
\cref{fig:genericAquiver},
an edge of the quiver always connects  neighboring gauge node vertices.
By choosing an orientation of the arrow, the oriented adjacency matrix $\mathsf{Q}$ can be obtained as 
\begin{align}
\mathsf{Q}_{i,j}= 
    \begin{cases}
        1 & \quad  (j-i=1) \;,\\
        0 & \quad \textrm{(otherwise)}\;.
    \end{cases}
\end{align}
This means that 
$a\in Q_1$ can be traded for $i\in Q_0$ with $i=1, \dots, \Qzero-1$, where $i\in Q_0$ now represents an edge connecting vertices $i$ and $i+1$.

The Cartan matrix is given by
\begin{align}
\mathsf{C}_{i,j}= 
    \begin{cases}
        2 & \quad 
        (i=j) \;,\\
        -1 & \quad (|i-j|=1) \;,\\
        0 & \quad \textrm{(otherwise)}\;,
    \end{cases}
\end{align}
which is nothing but the standard Cartan matrix of type $A_{\Qzero}$.

The notation for the elements in the direct summands in $T^*\mathrm{Rep}(Q,\mathbf v,\mathbf w)$ is now
\begin{center}
\begin{tabular}{ | c | c | c | c |} 
 \hline
 $\mathrm{Hom}(V_{i+1},V_{i})$ & $\mathrm{Hom}(V_{i},V_{i+1})$ & $\mathrm{Hom}(W_i,V_i)$ & $\mathrm{Hom}(V_i,W_i)$  \\  [1ex] 
 \hline
$x_i$ & $y_i$ & $\alpha_i$ & $\beta_i$  \\ [0.5ex] 
 \hline
\end{tabular}
\end{center}
The moment map $\mu:T^*\mathrm{Rep}(Q,\mathbf v,\mathbf w)\to \mathfrak{gl}(v_i)^*$ for the Hamiltonian $\GL(v_i) \subset \GL(\mathbf v)$ action on $T^*\mathrm{Rep}(Q,\mathbf v,\mathbf w)$ is
\begin{align}\label{moment_map}
    \mu_i(x,y,\alpha,\beta)=-x_i y_i +y_{i-1} x_{i-1}- \alpha_i\beta_i\;,
\end{align}
where we have set $x_{\Qzero}=y_{\Qzero}=0$ for simplicity of the equation .

\subsection{Pair of partitions}

The $T_{\r}^{\s}[\SU(N)]$ theory is specified by $\rho$ and $\sigma$, each of which represents a conjugacy class of $\SU(2)$ embeddings,
\beq
    \r, \s: \SU(2)\to\SU(N) \;.
\eeq
When we denote the standard Pauli generators of $\mathfrak{su}(2)$ as 
$t_{i=1,2,3}$, we can write, up to conjugation by an element of $\SU(N)$, 
\beq
& \r(t^3)=
    \diag\left(\underbrace{\frac{n_1-1}{2},\frac{n_1-3}{2},\cdots,-\frac{n_1-1}{2}}_{n_1},\cdots,\underbrace{\frac{n_k-1}{2},\cdots,-\frac{n_k-1}{2}}_{n_k}\right) \;.
\eeq
The integers $n_i$ define a partition of $N$,
\beq
     \r=[n_1,n_2,\cdots,n_{k}]\;,
     \quad 
     n_1\geq n_{2} \geq n_{k} > 0\;,
     \quad 
     \sum_{i=1}^{k} n_{i}=N\;.
 \eeq
The same comment applies to $\s$.
 When $\r=\s=[1,\cdots,1]$ the $T^{\s}_{\r}[\SU(N)]$ theory is denoted simply as $T[\SU(N)]$. 
 
 \subsection{Partitions and quiver data}
 
The $T_{\r}^{\s}[\SU(N)]$ theory is defined by a D3-NS5-D5 boundary condition, 
where we have $N$ D3-branes on a interval, ending on the NS5-type boundary conditions defined by $\r$ 
on one end and D5-type boundary conditions defined by $\s$ on the other end, as depicted in \cref{fig:quiverexample2011} --- we have $n_i$ D3-branes ending on the $i^{\rm th}$ NS5-brane
for $\r=[n_1,n_2,\cdots,n_{k}]$, with similar pattern for the D5-branes and $\s$.\footnote{Since we have 4d $\mathcal{N}=4$ $\SU(N)$ super Yang-Mills gauge theory on $N$ D3-branes, we can regard this as a setup where 4d $\mathcal{N}=4$ $\SU(N)$ super Yang-Mills theory is placed on a interval, with D5 and NS5 boundary conditions on the two ends.}

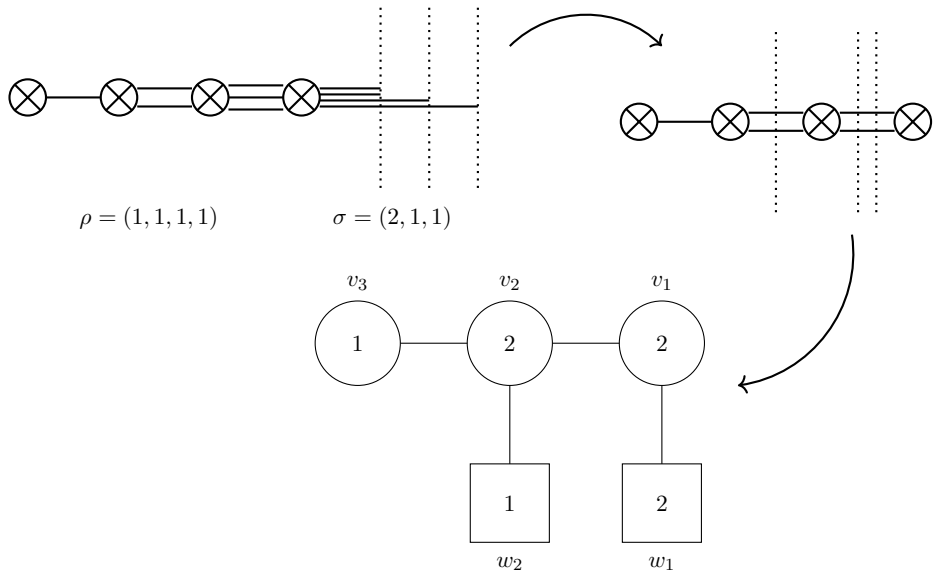
\begin{figure}[ht!]
    \centering

    \begin{tikzpicture}[decoration=snake]
    \node (leftfigure) at (0,0){
    \begin{tikzpicture}[scale=0.8, every node/.style={scale=0.8}]     
    \node [draw,circle,minimum width=0.6 cm,thick](A) at (-6,0){}; 
    \draw [thick] (A.north east) -- (A.south west)
    (A.north west) -- (A.south east);
    \node [draw,circle,minimum width=0.6 cm,thick](B) at (-4.5,0){}; 
    \draw  [thick] (B.north east) -- (B.south west)
    (B.north west) -- (B.south east);
    \node [draw,circle,minimum width=0.6 cm,thick](C) at (-3,0){}; 
    \draw  [thick](C.north east) -- (C.south west)
    (C.north west) -- (C.south east);
    \node [draw,circle,minimum width=0.6 cm,thick](D) at (-1.5,0){}; 
    \draw [thick] (D.north east) -- (D.south west)
    (D.north west) -- (D.south east);
%
    \draw[dotted,thick] (-0.2,-1.5) -- (-0.2,1.5);
    \draw[dotted,thick] (0.6,-1.5) -- (0.6,1.5);
    \draw[dotted,thick] (1.4,-1.5) -- (1.4,1.5);
%
    \draw[thick] (-5.7,0) -- (-4.8,0);
    \draw[thick] (-4.2,-0.15) -- (-3.3,-0.15);
    \draw[thick] (-4.2,0.15) -- (-3.3,0.15);
    \draw[thick] (-2.7,0) -- (-1.8,0);
    \draw[thick] (-2.7,0.2) -- (-1.8,0.2);
    \draw[thick] (-2.7,-0.2) -- (-1.8,-0.2);
    \draw[thick] (-1.2,0.15) -- (-0.2,0.15);
    \draw[thick] (-1.2,0.05) -- (-0.2,0.05);
    \draw[thick] (-1.2,-0.05) -- (0.6,-0.05);
    \draw[thick] (-1.2,-0.15) -- (1.4,-0.15);
    \node at (-4,-2) {$\rho=(1,1,1,1)$};
    \node at (0,-2) {$\sigma=(2,1,1)$};
    
    \end{tikzpicture}
    };
    
    \node (rightfigure) at (7,0)
         {
            \begin{tikzpicture}[scale=0.8, every node/.style={scale=0.8}]  
            \node [draw,circle,minimum width=0.6 cm,thick](A) at (-6,0){}; 
            \draw [thick] (A.north east) -- (A.south west)
            (A.north west) -- (A.south east);
            \node [draw,circle,minimum width=0.6 cm,thick](B) at (-4.5,0){}; 
            \draw  [thick] (B.north east) -- (B.south west)
            (B.north west) -- (B.south east);
            \node [draw,circle,minimum width=0.6 cm,thick](C) at (-3,0){}; 
            \draw  [thick](C.north east) -- (C.south west)
            (C.north west) -- (C.south east);
            \node [draw,circle,minimum width=0.6 cm,thick](D) at (-1.5,0){}; 
            \draw [thick] (D.north east) -- (D.south west)
            (D.north west) -- (D.south east);       
            \draw[dotted,thick] (-3.75,-1.5) -- (-3.75,1.5);
            \draw[dotted,thick] (-2.4,-1.5) -- (-2.4,1.5);
            \draw[dotted,thick] (-2.1,-1.5) -- (-2.1,1.5);       
            \draw[thick] (-5.7,0) -- (-4.8,0);
            \draw[thick] (-4.2,-0.15) -- (-3.3,-0.15);
            \draw[thick] (-4.2,0.15) -- (-3.3,0.15);
            \draw[thick] (-2.7,0.15) -- (-1.8,0.15);
            \draw[thick] (-2.7,-0.15) -- (-1.8,-0.15);
            
        \end{tikzpicture}
        };
         
   \node (bottomfigure) at (3.5,-4.0){
        \begin{tikzpicture}[scale=0.8, every node/.style={scale=0.8}]        
            \draw[] (-1.0,0) circle (.7); 
            \draw[] (-0.3,0) -- (0.8,0);
            \draw[] (1.5,0) circle (.7);   
            \draw[] (2.2,0) -- (3.3,0);
            \draw[] (4,0) circle (.7);           
            \draw[] (4,-.7) -- (4,-2);
            \draw[] (1.5,-.7) -- (1.5,-2);     
            \draw[] (3.35,-2) rectangle (4.65,-3.3);
            \draw[] (0.85,-2) rectangle (2.15,-3.3);
            \node at (-1.0,0) {$1$};
            \node at (1.5,0) {$2$};
            \node at (4,0) {$2$};
            \node at (-1.0,1) {$v_{3}$};
            \node at (1.5,1) {$v_{2}$};
            \node at (4,1) {$v_{1}$};
            \node at (1.5,-3.65) {$w_{2}$};
            \node at (4,-3.65) {$w_{1}$};
            \node at (1.5,-2.65) {$1$};
            \node at (4,-2.65) {$2$};
        \end{tikzpicture}
    };

    \draw[->, thick] (3.5,1) to [bend left=45] (5.5,1);
    \draw[->, thick] (8,-1.5) to [bend left=45] (6.5,-3.5);
    
    \end{tikzpicture}
    \caption{Brane realization, before and after Hanany-Witten moves, with NS5 branes (circles with crosses), D3 branes (horizontal segments) and D5 branes (dashed lines) and the corresponding quiver for the gauge theory example $T^{\sigma}_{\rho}\left[\SU(N)\right]$ with $N=4$, $\rho=(1,1,1,1)$ and $\sigma=(2,1,1)$ that appeared in \cite{Nishioka:2011dq}.}
    \label{fig:quiverexample2011}
\end{figure}

The matter contents of the gauge theory can be extracted after the Hanany-Witten transition \cite{Hanany:1996ie}.
The resulting theory is a quiver gauge theory as in \cref{fig:quiver for Trs[SU(N)] theory}, 
where the ranks $v=\left(v_1,\ldots , v_{k-1}\right)$ and $w=\left(w_1,\ldots , w_{k-1}\right)$ of the gauge and flavor symmetry groups for $T_{\r}^{\s}[\SU(N)]$ are determined by
\beq
    v_{k-1}=n_k~,\quad 2v_i=v_{i+1}+v_{i-1}+w_{i}+n_{i+1}-n_i~,\quad (i=1,\cdots,k-2)~,
\eeq
for
\beq
    &\r=[n_1,n_2,\cdots,n_{k-1},n_k]~, \\
    &\s=[\underbrace{(k-1),\cdots,(k-1)}_{w_{k-1}},\cdots,\underbrace{i,\cdots,i}_{w_i},\cdots] ~.
\eeq

\begin{figure}[ht!]
    \centering
    \begin{tikzpicture}[scale=.75, every node/.style={scale=0.8}]        
        \draw[] (-6,0) circle (.7);
        \draw[] (-5.3,0) -- (-4.2,0);
        \draw[] (-3.5,0) circle (.7);
        \draw[] (-2.8,0) -- (-1.75,0);
        \draw[line width=1pt, loosely dotted] (-1.75,0) -- (-.3,0);
        \draw[] (-.3,0) -- (.8,0);
        \draw[] (1.5,0) circle (.7);   
        \draw[] (2.2,0) -- (3.3,0);
        \draw[] (4,0) circle (.7);           
        \draw[] (4,-.7) -- (4,-2);
        \draw[] (1.5,-.7) -- (1.5,-2);
        \draw[] (-3.5,-.7) -- (-3.5,-2);
        \draw[] (-6,-.7) -- (-6,-2);       
        \draw[] (3.35,-2) rectangle (4.65,-3.3);
        \draw[] (0.85,-2) rectangle (2.15,-3.3);
        \draw[] (-4.15,-2) rectangle (-2.85,-3.3);
        \draw[] (-6.65,-2) rectangle (-5.35,-3.3);
        \node at (-6,0) {$v_{k-1}$};
        \node at (-3.5,0) {$v_{k-2}$};
        \node at (1.5,0) {$v_{2}$};
        \node at (4,0) {$v_{1}$};
         \node at (-6,-2.65) {$w_{k-1}$};
        \node at (-3.5,-2.65) {$w_{k-2}$};
        \node at (1.5,-2.65) {$w_{2}$};
        \node at (4,-2.65) {$w_{1}$};
    \end{tikzpicture}
    \caption{The quiver diagram for the $T_{\r}^{\s}[\SU(N)]$ theory.}
    \label{fig:quiver for Trs[SU(N)] theory}
\end{figure}

The 3d mirror symmetry is realized by the S-duality of type IIB string theory \cite{Hanany:1996ie}, 
which exchanges the NS5-branes with the D5-branes while preserving the D3-branes. This means that 
the mirror of the $T_{\r}^{\s}[\SU(N)]$ theory is the $T_{\s}^{\r}[\SU(N)]$ theory,
and in particular $T[\SU(N)]$ theory (whose quiver is shown in \cref{fig:quiver_TSUN}) is mirror to itself.

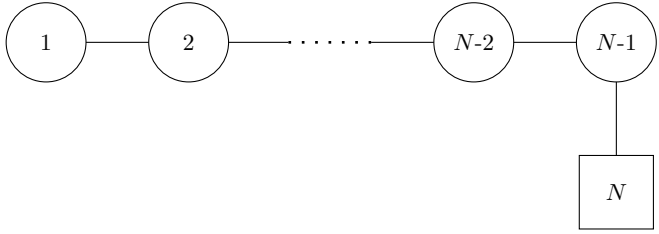
\begin{figure}[ht!]
    \centering
    \begin{tikzpicture}[scale=.75, every node/.style={scale=0.8}]        
        \draw[] (-6,0) circle (.7);
        \draw[] (-5.3,0) -- (-4.2,0);
        \draw[] (-3.5,0) circle (.7);
        \draw[] (-2.8,0) -- (-1.75,0);
        \draw[line width=1pt, loosely dotted] (-1.75,0) -- (-.3,0);
        \draw[] (-.3,0) -- (.8,0);
        \draw[] (1.5,0) circle (.7);   
        \draw[] (2.2,0) -- (3.3,0);
        \draw[] (4,0) circle (.7);           
        \draw[] (4,-.7) -- (4,-2);
        \draw[] (3.35,-2) rectangle (4.65,-3.3);
        \node at (-6,0) {$1$};
        \node at (-3.5,0) {$2$};
        \node at (1.5,0) {$N$-$2$};
        \node at (4,0) {$N$-$1$};
        \node at (4,-2.65) {$N$};
    \end{tikzpicture}
    \caption{The quiver diagram for the $T[\SU(N)]$ theory.}
    \label{fig:quiver_TSUN}
\end{figure}

\subsection{Higgs and Coulomb branches}

The $T_{\r}^{\s}[\SU(N)]$ theory has non-trivial vacuum moduli space of vacua, consisting of the Higgs branch $\mathcal{M}_H$, the Coulomb branch $\mathcal{M}_C$ and the mixed branches. The Higgs branch of $T_{\r}^{\s}[\SU(N)]$ has three different descriptions:
\begin{itemize}
    \item[(1)] Nakajima quiver variety $\mathcal M(v,w)$ associated to the quiver of \cref{fig:quiver for Trs[SU(N)] theory}, 
    \item[(2)] Slodowy variety $\overline{\mathcal S}^{\rho^t}_{\sigma}$ defined as an intersection $\overline{\mathbb O}(\rho^t)\cap \mathcal S(\sigma)$, where $\overline{\mathbb O}(\rho^t)$ is the closure of $\mathrm{SL}_N$-orbit through an nilpotent element of $\mathfrak{sl}_N$ with Jordan blocks associated to the partition $\rho^t$ (transpose of $\rho$), and $\mathcal S(\sigma)$ is the transverse slice to an nilpotent element of $\mathfrak{sl}_N$ with Jordan blocks associated to the partition $\sigma$~,
    \item[(3)] Slice $\overline{\mathcal W}^{\sigma^t}_{\rho}$ of the affine Grassmannian $\mathrm{Gr}_{\mathrm{PSL}_{|w|}}$, where $|w|=\sum_{i=1}^{k-1}w_i$, and both $\sigma^t$ and $\rho$ are regarded as coweights of $\mathrm{PSL}_{|w|}$~,
\end{itemize}
with the three varieties having been shown to be isomorphic by Mirkovi\'c-Vybornov \cite{mirkovic2007quiver}.  
By the work of Braverman-Finkelberg-Nakajima \cite{braverman2016coulomb, braverman2017ring}, 
the Coulomb branch of $T_{\s}^{\r}[\SU(N)]$ is the slice $\overline{\mathcal W}^{\s^t}_{\r}$ of the affine Grassmannian $\mathrm{Gr}_{\mathrm{PSL}_{k}}$, which is isomorphic to the Slodowy variety $\overline{\mathcal S}^{\r^t}_{\s}$ of $\mathrm{SL}_{N}$.  This is expected since the 3d mirror symmetry $T_{\s}^{\r}[\SU(N)]\leftrightarrow T_{\r}^{\s}[\SU(N)]$ swaps the Coulomb branches and Higgs branches
\begin{align}
    \overline{\mathcal S}^{\r^t}_{\s} \cong\mathcal M_{H}(T_{\r}^{\s}[\SU(N)])\cong \mathcal M_{C}(T_{\s}^{\r}[\SU(N)])\cong \overline{\mathcal W}^{\s^t}_{\r}\;.
\end{align}

The symmetry of the Higgs branch is the commutant of  
$\s(\mathrm{SU}(2))$, more precisely $\mathrm{S}(\mathrm{U}(w_1)\times\cdots\times \mathrm{U}(w_{k-1}))$. 
To see the hidden natural symmetry of the Coulomb branch, collect quiver nodes of $d_i=0$ and construct a Dynkin diagram from these (corresponding to the balanced sub-quiver), see \cite[Appendix A]{braverman2017ring}, then the Coulomb branch symmetry is the Levi subgroup of $\mathrm{PSU}(k)$ corresponding to the balanced subquiver. Equivalently, if we regard $\rho$  as a cocharacter of $\mathrm{PSU}(k)$, then the Coulomb branch symmetry is the commutant of $\rho(\mathrm{U}(1))$  in $\mathrm{PSU}(k)$.

\subsection{Higgs branch from paths on the linear quivers}\label{subsec:Higgs}

For an explicit analysis of the Higgs branch,
it is useful to have concrete parametrizations of the Higgs branch.
In the algebraic language, this is to identify an element of the coordinate ring $\mathbb{C}[\CM(Q, \mathbf{v}, \mathbf{w})]$ of the quiver variety.

The elements of the coordinate ring of the Higgs branch are given by gauge-invariant operators (mesons)\footnote{There are no baryons since the gauge groups are $\U(N)$.},
which on the quiver diagram are either (i) closed paths or (ii) open paths which begin at a framing node and 
end at another framing node. Not all of these operators are independent since 
we can use the moment map relations $\mu=\mathcal{Z}$, where $\mu$ was given in \eqref{moment_map}.

Now, Theorem 2.2.1 of \cite{2019arXiv190501810} has shown\footnote{See also Theorem 10.1.1 of \cite{2019arXiv190501810},
\cite{MR1623674} and \cite{MR2130242}. This statement holds for a general linear quiver diagram, not necessarily one of the $T_{\rho}^{\sigma}[\SU(N)]$ theories.} that the coordinate ring 
is generated by operators 
of the form (see \cref{fig:pathsquiverexample3})
\begin{equation}
 \label{path_g}
  g_{(i\to m \to j)}: = \beta_j \cdot x_j  \cdots \cdot x_{m-2} \cdot x_{m-1} \cdot y_{m-1}  \cdot  \cdots \cdot y_{i+1} \cdot y_{i} \cdot \alpha_i \;.
\end{equation}
In particular, when parameterising the Higgs branch $\CM_H$, it is only necessary to consider paths on the quiver diagram which do not loop around any of the edges.

\begin{figure}
\centering
    \begin{tikzpicture}[scale=.6, every node/.style={scale=0.8}];
        \node(-5) at (-12.5,0){};
        \node(-4) at (-10,0){$\circled{$m$}$};
        \node(-3) at (-7.5,0){};
        \node(-2) at (-5,0){};
        \node(-1) at (-2.5,0){$\circled{$i+1$}$};
        \node(0) at (0,0){$\circled{$i$}$};
        \node(1) at (2.5,0){};
        \node(2) at (5,0){};
        \node(3) at (7.5,0){$\circled{$j$}$};
        \node(4) at (10.0,0){};
        
        \node(w-4) at (-10,-2.5){$\rectangled{}$};
        \node(w-1) at (-2.5,-2.5){$\rectangled{}$};
        \node(w0) at (0,-2.5){$\rectangled{}$};
        \node(w3) at (7.5,-2.5){$\rectangled{}$};
        \draw (-5) -- (-4);
        \draw (-4) -- (-3);
        \draw[dotted] (-3) -- (-2);
        \draw (-2) -- (-1);
        \draw (-1) -- (0);
        \draw (0) -- (1);
        \draw[dotted] (1) -- (2);
        \draw (2) -- (3);
        \draw (3) -- (4);
        \draw (w-4) --(-4);
        \draw (w-1) --(-1);
        \draw (w0) --(0);
        \draw (w3) --(3);
        \draw[->,red,thick] (-0.5,-2) -- node [left] {$\alpha_{i}$}(-0.5, 0.8);
        \draw[->,red,thick] (-0.5, 0.8) -- node [below] {$y_{i}$}(-3.0, 0.8);
        \draw[dotted,red,thick] (-3.0, 0.8) -- (-8.0, 0.8);
        \draw[->,red,thick] (-8.0, 0.8) -- (-10.5, 0.8);
        \draw[->,red,thick] (-10.5, 0.8) -- (-10.5, 1.2);
        \draw[->,red,thick]  (-10.5, 1.2) -- node [above] {$x_{m}$}(-8.0, 1.2);
        \draw[dotted,red,thick] (-8.0, 1.2) -- (4.5, 1.2);
        \draw[->,red,thick] (4.5, 1.2) -- node [above] {$x_{j}$}(7.0, 1.2);
        \draw[->,red,thick] (7.0,1.2) -- node [left] {$\beta_{j}$}(7.0, -2);
    \end{tikzpicture}
    \caption{A path on the quiver diagram representing a gauge-invariat operator on the Higgs branch. In this figure we have shown the case $j<i$, but we allow for $j\ge i$ in the definition of the generator $g_{(i\to m \to j)}$.}
    \label{fig:pathsquiverexample3}
\end{figure}
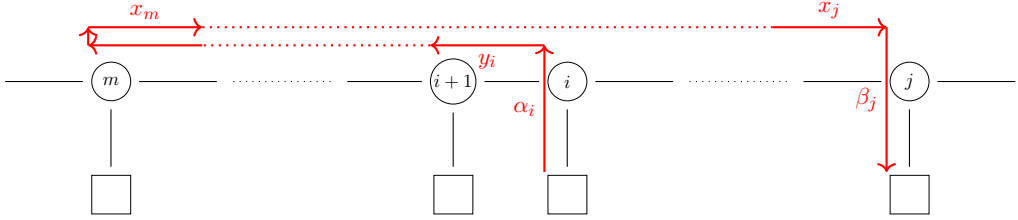

There are still dependencies between the operators of the form
\eqref{path_g} and not all of them are required to generate the Higgs branch.

\subsection{Balanced linear quivers}\label{sec:balanced_quivers}

In the following we will specialize to the case of balanced quivers.

We can verify that a linear type A quiver (as in \cref{fig:genericAquiver}) is balanced 
if and only if it can be expressed as a $T_{\r}^{\s}[\SU(N)]$ theory,
where
\begin{align}
        & N = k v_{k-1}~,\\
        &\r=[\underbrace{v_{k-1} ,v_{k-1},\cdots, v_{k-1}}_k ]~, \\
        &\s=[\underbrace{(k-1),\cdots,(k-1)}_{w_{k-1}},\cdots,\underbrace{i,\cdots,i}_{w_i},\cdots] ~.
\end{align}
This include, as special cases (with $v_{k-1}=1$),
$T_{\r}^{\s}[\SU(N)]$ theory with $k=N$ and 
\begin{subequations}
\begin{align}
        &\r=[\underbrace{1, 1,\cdots, 1}_N ]~, \\
        &\s=[\underbrace{(N-1),\cdots,(N-1)}_{w_{k-1}},\cdots,\underbrace{i,\cdots,i}_{w_i},\cdots] ~.
\end{align}
\label{rho=111}
\end{subequations}
In the latter special case, we have $T_{[1, \dots, 1]}^{\s}[\SU(N)]$ theories with a general partition $\s$, and 
$v_i$ are related to $w_i$ by
\begin{align}\label{v_i in terms of w_i}
    v_i=N-i-\sum_{j>i}(j-i)w_j \;.
\end{align}

\section{VOAs associated to quivers} \label{sec:VOA-quiver}

In the following sections we will concretely describe the VOA $\mathcal{V}$ \eqref{def:BRSTcomplex} for a quiver gauge theory, following the steps outlined in the previous section. In this section, we will actually describe a VOA $\CV(Q,\mathbf v,\mathbf w)$ which is canonically associated to a quiver $Q$ with gauge and framing dimensions $\mathbf v,\mathbf w$, in contrast to the boundary VOA of the quiver gauge theory associated to $(Q,\mathbf v,\mathbf w)$ which involves a choice of boundary fermions. The difference is that we fix a canonical choice of anomaly cancellation scheme which is determined solely by the quiver data.

We first describe the naive VA for the BRST reduction in \cref{subsection:pre-reduction}. We then analyze in \cref{subsection:anomaly}  the global $\mathfrak{g}$-symmetry currents and 
identify the boundary theories needed for anomaly cancellation. We can then carry out the BRST reduction to identify the VA $\CV$.
In \cref{sec:conformal_vector} we analyze the conformal vector, to promote the VA into a VOA.

\subsection{Vertex algebra before BRST reduction}
\label{subsection:pre-reduction}

Let us begin by writing down the VOA before the BRST reduction. Note that at this point we will set aside the boundary anomaly cancellation condition, which will be fully taken into account in the next subsection.

In our discussion of VOAs we impose $\mathcal{N}=(0,4)$-preserving Neumann boundary conditions
for $\mathcal{N}=4$ vector and hyper multiplets. This leaves the following degrees of freedom on the boundary:

\begin{itemize} 
\item a collection of ${\Sb}$-systems
    \begin{itemize}
        \item $(I_{i}, J_{i})_{i\in Q_0}$ valued in fundamental or antifundamental representations of the 
        $i\in \{1,\ldots ,\Qzero\}$ flavor nodes  and  
        \item $(X_{a}, Y_{a})_{a\in Q_1}$ valued in bifundamental representations of the gauge groups.  
    \end{itemize}
    These fields are associated to the scalar fields coming from the 3d $\mathcal{N}=4$ hypermultiplets valued in a symplectic representation of the gauge group $G$, as well as those from 3d $\mathcal{N}=4$ vectormultiplets.

\item a system $\bc$ of fermionic fields $(\mathsf{b}_{i}, \mathsf{c}_{i})_{i\in Q_0}$ from the $\bc$ ghosts needed to fix the gauge freedom. The conformal weights of $\mathsf{b}$ and $\mathsf{c}$ are chosen to be $1$ and $0$, respectively.
\end{itemize}

The convention we adopt is summarized pictorially in \cref{fig:convention2}, which presents the information locally around the $i^{\rm th}$ gauge node of the quiver.
For example $X_{a}$ is associated to the quiver edge directed from the $\mathsf{t}(a)^{\rm th}$ gauge node to the $\mathsf{h}(a)^{\rm th}$ gauge node. 

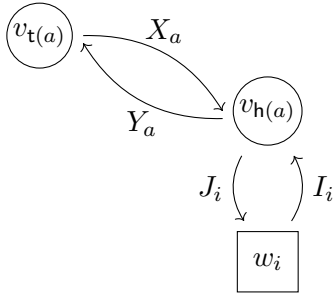
\begin{figure}[ht!]
    \centering
    \begin{tikzpicture}[]
          \node (giplus1) at (-3,-1) {$\circled{$v_{\mathsf{t}(a)}$}$};
          \node (gi) at (0,-2) {$\circled{$v_{\mathsf{h}(a)}$}$};
          \node (fi) at (0,-4) {$\rectangled{$w_i$}$};
          %
          \draw[->, bend left] (giplus1.east) to node[above]{$X_{a}$} (gi.west);
          \draw[<-, bend right] ([yshift=-0.1cm]giplus1.east) to node[below]{$Y_{a}$} ([yshift=-0.1cm]gi.west);
          \draw[->, bend right] (fi) to node[right]{$I_{i}$} (gi);
          \draw[->, bend right] (gi) to node[left]{$J_{i}$} (fi);
     \end{tikzpicture}
     \caption{Conventions for the quiver arrows and associated fields in the vertex algebra.}
     \label{fig:convention2}
\end{figure}

The boundary field content defines a vertex operator algebra $\CVbg \otimes \CVbc$ on the boundary $\partial M_3$,
where the collection of fields obeys the following nontrivial OPEs 
\begin{gather}
\begin{split}
    &  \qquad (X_{a})^{\alpha}{}_{\beta} (z) \, (Y_{b})^{\gamma}{}_{\delta} (w) \sim \frac{\delta^{\alpha}_{\delta}\delta^{\gamma}_{\beta}\delta_{a,b}}{z-w} ~, 
    \\
    & \qquad (I_{i})^{\alpha}{}_{\beta} (z) \, (J_{j})^{\gamma}{}_{\delta} (w) \sim \frac{\delta^{\alpha}_{\delta}\delta^{\gamma}_{\beta}\delta_{i,j}}{z-w} ~,
    \\
    & \qquad (\mathsf b_{i})^{\alpha}{}_{\beta} \,(z)  \, (\mathsf c_{j})^{\gamma}{}_{\delta} \,(w) \sim \frac{\delta^{\alpha}_{\delta}\delta^{\gamma}_{\beta} \delta_{i,j} }{z-w} ~, \\   
\end{split}
\end{gather}
while all other OPEs are trivial. 

\subsection{Boundary anomaly cancellation and BRST reduction}\label{subsection:anomaly}

Our boundary conditions leave unconstrained
some fermions, which contribute to gauge anomalies on the boundary.
These anomalies need to be canceled by the addition of extra degrees of freedom
on the boundary \cite{Yoshida:2014ssa,Dimofte:2017tpi,Costello:2018fnz}.

A viable prescription for quiver gauge theories is to cancel a boundary anomaly locally at each of the quiver nodes. For the supersymmetric quiver gauge theories with A-type quiver depicted in \cref{fig:genericAquiver} and gauge group of the form $G=\prod_{i\in Q_0}G_i$, with $G_i=\U(v_i)= \SU(v_i)\times \U(1)_{i, \textrm{diag}}$, one should separate the gauge anomaly for the semi-simple non-Abelian parts $\SU(v_i)$ from the Abelian components $\U(1)_{i, \textrm{diag}}$. 

\subsubsection{$\SU$ anomalies}
The contribution to the boundary gauge anomaly from a 3d $\SU(v_i)$ $\mathcal{N}=2$ gauge multiplet with Neumann boundary conditions is $v_i {\rm Tr} (F^2)$, where $F$ denotes the $\SU(v_i)$ connection field strength \cite[Section 2]{Dimofte:2017tpi}.
On the other hand, the fermions from the 3d $\mathcal{N}=2$ chiral multiplets in a fundamental or antifundamental representation of $G_i$, which correspond to edges   
\begin{figure}[ht!]
\begin{center}
 \begin{tikzpicture}
        \node at (-2,0) [circle] (GN) {};
        \node at (-1,0) [circle,draw] (GN-1) {$G_i$};
        \node at (1,0) [circle] (FN) {$w_i$};
        \node [rectangle,draw,fit=(FN),inner sep=0] (FNN) {};
        \draw (GN) -- (GN-1) -- (FNN);
    \end{tikzpicture}
\end{center}
\end{figure}
\noindent 
of the quiver diagram and which have Neumann boundary conditions, each of these contributes $-\frac{1}{2} {\rm Tr} (F^2)$ to the anomaly polynomial. 
From the $G_i$--$w_i$ part of the quiver, the number of fundamental hypermultiplets is equal to the dimension of the fundamental or antifundamental representation for the flavor symmetry. 
The total contribution at the $i^{\rm th}$ node of the quiver is therefore $-\frac{1}{2} (w_i-\sum_{j\in Q_0} \mathsf{C}_{i,j} v_j) {\rm Tr} (F^2)$ associated to the pairs of arrows on the quiver pointing away from and respectively into the gauge node (see \cref{fig:general quiver node for anomaly2}). Then in order for the $\SU(v_i)$ gauge anomaly to cancel at each node $i$, the corresponding dimension vectors $\mathbf{v}=(v_i)_{i\in Q_0}$ and $\mathbf{w}=(w_i)_{i\in Q_0}$ should satisfy the balance condition (recall \eqref{general_balancing_condition})
\begin{equation}
   w_i-\sum_{j\in Q_0} \mathsf{C}_{i,j} v_j  = 0 ~.  
\end{equation}
When this condition is satisfied the quiver is referred to as \emph{balanced} \cite{Gaiotto:2008ak}.

If instead we have the anomaly cancellation condition (recall \eqref{general_anomaly_cancellable_condition})
\begin{equation}
   \label{def_d}
   d_i:= w_i-\sum_{j\in Q_0} \mathsf{C}_{i,j} v_j   \geq 0 ~,  
\end{equation}
then it is necessary to add $d_i$ additional left-handed symplectic 3d fermions in a fundamental or antifundamental representation of the gauge group $\SU(v_i)$, which contribute $\frac{1}{2} d_i {\rm Tr} (F^2)$ to the gauge anomaly polynomial. In this case, the quiver locally resembles that depicted in \cref{fig:general quiver node for anomaly2} around the gauge node $i\in Q_0$. One could think of the framing vector space being promoted from $W_i= \mathbb C^{w_i}$ to $\mathbb C^{w_i|d_i}$. 

\begin{figure}[ht!]
    \centering
    \begin{tikzpicture}[scale=.75, every node/.style={scale=0.8}]  
        \draw[] (-3.5,2.7) circle (.7);
        \draw[] (-6,2.7) circle (.7);
        \draw[] (-3.5,0) circle (.7);
        \draw[loosely dashed] (-2.8,0) -- (-1.7,0);
        \draw[] (-3.5,.7) -- (-4.8,2);
        \draw[] (-3.5,-.7) -- (-3.5,-2);
        \draw[] (-3.5,.7) -- (-6,2);
        \draw[] (-3.5,.7) -- (-3.5,2);
        \draw[] (-4.15,-2) rectangle (-2.85,-3.3);
        \draw[loosely dashed] (-1.7,-.7) rectangle (-.3,.7);        
        \node at (-3.5,0) {$v_{i}$};
        \node at (-1,0) {$d_i$};
        \node at (-3.5,-2.65) {$w_{i}$};
        \node at (-4.8,2.65) {$\cdots$};
        \node at (-3.5,2.65) {$v_{k}$};
        \node at (-6,2.65) {$v_{j}$};       
        \node[circle,minimum size=1.5cm] (vi) at (-3.5,0){};
        \draw[] (vi) to [in=160, out=200, looseness=5](vi);
        \draw[] (vi) to [in=150, out=210, looseness=8](vi);
        \draw[dotted, thick] (-5.5,0) -- (-5.1,0);
    \end{tikzpicture}
    \caption{Local neighborhood of a general quiver around the $i^{\rm th}$ vertex. The $\SU$ anomaly cancellation requires introduction of boundary Fermi multiplets, which are represented as dotted arrows connecting to the newly-added framing node $\SU(d_i)$.}
    \label{fig:general quiver node for anomaly2}
\end{figure}
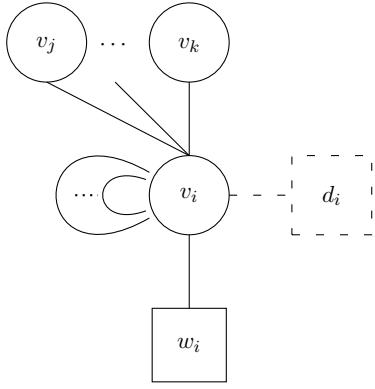

The end result of our analysis is that we need to include 
\begin{itemize}
         \item 
            a collection of boundary Fermi multiplets $\CV_{\textrm{Ff}}$ satisfying the OPE relations
            \begin{gather}\label{def:Ff}
            \begin{split}
                & \qquad (\F_{i})^{\alpha}{}_{\beta} (z) \, (\Psi_{j})^{\gamma}{}_{\delta} (w) \sim \frac{\delta^{\alpha}_{\delta}\delta^{\gamma}_{\beta}\delta_{i,j}}{z-w} ~,
            \end{split}
            \end{gather}
            (All other OPEs are non-singular.)
\end{itemize}

\paragraph{Analysis of global symmetry currents}
We can arrive at the same conclusion by 
directly analyzing the global symmetry currents.

We have explained previously that we need to identify the current algebra symmetry originating from the gauge symmetry of the 3d $\mathcal{N}=4$ theory we started with.
For a quiver gauge theory,
we have the gauge group $G=\prod_{i\in Q_0} G_i$, and 
this implies that
there exists an affine $\mathfrak{g}_i$ current $J_{\mathfrak{g}_i}(z)$ of the $\beta\gamma$ system associated to the gauge node $v_i$ for every $i\in Q_0$.
At the $i^{\rm th}$ gauge node, the generic form of which is depicted in \cref{fig:general quiver node for anomaly2}, 
there is an affine $\mathfrak{u}(v_i)$ current of the form (cf.\ \cite{Costello:2018fnz}) 

\begin{align}
\label{def:unBRSTcurrentViShifted_general}
    (J_{\mathfrak{u}(v_i)})^\alpha{}_\beta &= 
    -\sum_{\mathsf{h}(a)=i} (X_{a})^\alpha{}_\gamma (Y_{a})^{\gamma}{}_\beta
    +\sum_{\mathsf{t}(a)=i} 
    (Y_{a})^\alpha{}_\gamma(X_{a})^{\gamma}{}_\beta-(I_{i})^\alpha{}_\gamma(J_{i})^{\gamma}{}_\beta\nonumber\\ 
    & ~~~~ + (\F_{i})^\alpha{}_\gamma(\Psi_{i})^{\gamma}{}_\beta  \;.
\end{align}
In view of comparison between \eqref{def:unBRSTcurrentViShifted_general} with \eqref{moment_map_general}, we shall call $J_{\mathfrak{u}(v_i)}$ the \textit{chiral moment map} associated to the $i^{\rm th}$ node, denoted by $\mu_{\mathrm{ch},i}$.

The $\mathfrak{su}(v_i)$ part of the current \eqref{def:unBRSTcurrentViShifted_general} is given by 
\begin{align}
    (J_{ \mathfrak{su}(v_i)})^\alpha{}_\beta 
    &= (J_{ \mathfrak{u}(v_i)})^\alpha{}_\beta 
    - \frac{\delta^\alpha_\beta}{v_i} (J_{ \mathfrak{u}(v_i)})^\gamma{}_\gamma ~. 
\end{align}
These $\mathfrak{su}(v_i)$ currents are akin to the BRST 
current $J^\alpha_{\CVs}$ in \eqref{def:BRSTcomplex}, where the level of the corresponding $\bc$ ghost system current is $2 \mathsf{h}^\vee_{\mathfrak{su}(v_i)}=2v_i$. Thus the 
level of the affine $\mathfrak{su}(v_i)$ currents $J_{\mathfrak{su}(v_i)}$ must be $-2 \mathsf{h}^\vee_{\mathfrak{su}(v_i)}=-2v_i$ in order for these to be anomaly free and to define a nilponent charge in equation \eqref{def:BRSTcurrent}. We can see this is indeed the case from the OPE  
\begin{align} 
    & (J_{ \mathfrak{su}(v_i)})^\alpha{}_\beta \, (z) (J_{ \mathfrak{su}(v_i)})^\gamma{}_\delta \, (w) \nonumber\\ 
    & \qquad \qquad \qquad 
    \sim \frac{\delta^\alpha_\delta \delta^\gamma_\beta (-\sum_{j\in Q_0} \mathsf{Q}_{i,j} v_{j}-w_i + d_i)}{(z-w)^2} + \frac{\delta^\gamma_\beta (J_{ \mathfrak{su}(v_i)})^\alpha{}_\delta-\delta^\alpha_\delta (J_{ \mathfrak{su}(v_i)})^\gamma{}_\beta}{z-w} \;.
\end{align}
The level of the $\mathfrak{su}(v_i)$ current algebra is then 
$-\sum_{j\in Q_0} \mathsf{Q}_{i,j} v_{j}-w_i + d_i$,
and this coincides with $-2 v_i$ if and only if the anomaly cancellation condition \eqref{def_d} is satisfied.

\subsubsection{$\U(1)$ anomalies}

The $\mathfrak{u}(v_i)$ current \eqref{def:unBRSTcurrentViShifted_general} also has a diagonal $\mathfrak{u}(1)$ component 
\beq
    J^{(i)}_{ \mathfrak{u}(1)}
    = \frac{1}{v_i}(J_{ \mathfrak{u}(v_i)})^\gamma{}_\gamma \;, 
\label{J_U1}
\eeq
which represents its trace part and which is anomalous, as is evident from the presence of singular terms in the OPE  
\beq
    J^{(i)}_{ \mathfrak{u}(1)}(z) \, J^{(j)}_{ \mathfrak{u}(1)} (w) \sim \frac{-\mathsf{C}_{ij}}{(z-w)^2} ~,
\eeq
where the existence of the second-order pole represents the existence of the $\U(1)$ anomalies.

While there is far from a unique choice to cancel the anomalies, 
the exists a minimal choice to cancel the 
anomalies: 
\begin{itemize}
    \item 
        We add 
        $\U(1)^{Q_0}$ current algebras 
        satisfying the commutation relation
        \beq
        h_i(z) \, h_j(w) \sim \frac{\mathsf{C}_{ij}}{(z-w)^2} ~,
        \eeq
        (All other OPEs are non-singular.)
\end{itemize}

We can then modify the $\mathfrak{u}(1)$ current 
 \eqref{J_U1} to be
\beq
    J^{(i)}_{\mathfrak{u}(1)}=\frac{1}{v_i}(J_{ \mathfrak{u}(v_i)})^\gamma{}_\gamma -  h_i \;.
\eeq
We can then easily verify that this current
has a regular OPE with itself for all $i\in Q_0$.

\subsubsection{Definition of \texorpdfstring{$\CV(Q,\mathbf v,\mathbf w)$}{V(Q,v,w)}}

Let us summarize the ingredients. Starting from a quiver $Q$ with gauge and framing dimensions $\mathbf v,\mathbf w$, we associate 
\begin{itemize}
    \item $\beta\gamma$ systems $(X_a,Y_a)_{a\in Q_1}$ and $(I_i,J_i)_{i\in Q_0}$,
    \item pairs of free fermions $(\F_i,\Psi_i)_{i\in Q_0}$,
    \item Heisenberg algebra $\mathcal H=(h_i)_{i\in Q_0}$.
\end{itemize}
Let us denote the tensor product of the above trio by $\CVs(Q,\mathbf v,\mathbf w)$. We call 
\begin{align}\label{eq:chiral extended moment map}
    \widetilde\mu_{\mathrm{ch},i}:=J_{\mathfrak{u}(v_i)}-h_i\cdot\mathrm{id}
\end{align}
the \textit{chiral extended moment map} associated to the $i^{\rm th}$ node, which is a chiralization of extended moment map \eqref{extended moment map}, in the sense that  $\widetilde\mu_{\mathrm{ch}}$ gives rise to a vertex algebra map
\begin{align*}
    \widetilde\mu_{\mathrm{ch}}:V^{-\kappa_{\mathfrak g}}(\mathfrak g)\to \CV(Q,\mathbf v,\mathbf w) \;,
    \quad 
    (E_i)^\alpha{}_\beta\mapsto (\widetilde\mu_{\mathrm{ch},i})^\alpha{}_\beta \;,
\end{align*}
where $\mathfrak g=\prod_{i\in Q_0}\mathfrak{gl}(v_i)$ and $\kappa_{\mathfrak g}$ is the Killing form on $\mathfrak g$. Then we can gauge the $\widehat{\mathfrak{g}}$-symmetry, using the BRST reduction described in \cref{subsec:BRST reduction}, and define
\begin{align}\label{def:VA of quiver}
    \CV^\bullet(Q,\mathbf v,\mathbf w):=H^{\infty/2+\bullet}_{\mathrm{BRST}}(\mathfrak{g},\CVs(Q,\mathbf v,\mathbf w))\;
    \quad
    \text{and } ~ \CV(Q,\mathbf v,\mathbf w):=\CV^0(Q,\mathbf v,\mathbf w) \;,
\end{align}
from the semi-infinite cohomology $H^{\infty/2+\bullet}_{\mathrm{BRST}}(\mathfrak{g},\CVs(Q,\mathbf v,\mathbf w))$ defined by the BRST complex.

\bigskip We now make two remarks. First, the presence of the current algebra is required irrespective of the choice of the quiver. Second, when the Cartan matrix $\mathsf{C}$ has a non-trivial kernel (i.e.\ $\mathsf{C}$ is not invertible), then a linear combination of the $h$'s have a non-singular OPE with all other fields and hence decouples from the rest of the VA; in this case it more natural to decouple these fields to generate a smaller VA. The exception is the case where $\mathsf{C}$ is invertible, 
for example when the quiver is of the ADE type.

\subsection{Conformal vector}
\label{sec:conformal_vector}

Let us discuss a conformal vector of the theory so that we promote the VA into the VOA. 

Since we already know the stress-energy operator for each of the building blocks, we can immediately write down
the total stress-energy tensor to be
\begin{equation} 
    \label{def:generalTtensor_general}
    T=T_{\Sb}+T_{\Ff}+T_{h}+T_{\mathsf{bc}} 
    ~ , 
\end{equation} 
where $T_{\Sb}$, $T_{\Ff}$, $T_{h}$ and $T_{\mathsf{bc}}$ are defined by 
\begin{align}
    T_{\Sb}&:= \frac{1}{2}\sum_{a\in Q_1} \mathrm{Tr}(X_{(a)}\partial Y_{(a)}-\partial X_{(a)} Y_{(a)})
                  +\frac{1}{2}\sum_{i\in Q_0} \mathrm{Tr}(I_{i}\partial J_{i}-\partial I_{i} J_{i})~,
    \\
    T_{\Ff}&:=-\frac{1}{2}\sum_{i\in Q_0} \mathrm{Tr}(\F_{i}\partial \Psi_{i}-\partial \F_{i} \Psi_{i})~, 
    \\ 
    T_{h}&:=\frac{1}{2}\sum_{i,j\in Q_0} (\mathsf{C}^{-1})_{i,j} h_{i} h_{j} ~,
    \\
    T_{\mathsf{bc}}&:=-\sum_{i\in Q_0} \mathrm{Tr}(\mathsf b_{i}  \partial \mathsf c_{i} 
    ) ~,
\end{align}
where we assume for now that $\mathsf{C}$ is invertible. 

In order to verify that the BRST reduction defines a VOA, we need to ensure that the BRST operator is consistent with the stress-energy tensor of the VOA $\CVs \otimes \mathcal{V}_{\mathsf{bc}}$.
To see this, first note that the $\mathfrak g$-current  $J_{\CVs}^{\alpha} (z)$ is a primary field of weight 1, having OPE with the stress-energy tensor $T_{\CVs}(w)$ of $\CVs$:
\begin{align}\label{eq:JT=J}
   T_{\CVs}(z) J_{\CVs}^{\alpha}(w)\sim \frac{J_{\CVs}^{\alpha}(w)}{(z-w)^2} + \frac{\partial J_{\CVs}^{\alpha}(w)}{z-w}~.
\end{align} 
The total stress-energy tensor $T(z)=T_{\CVs}(z)+T_{\mathsf{bc}}(z)=T_{\CVs}(z)-\mathsf b^{\alpha}\partial \mathsf c_{\alpha}(z)$ of the product VOA $\mathcal \CVs\otimes \mathcal V_{\bc}$ is invariant under the BRST supercharge $\QBRST$, as can be seen by computing the OPE 
\begin{equation}
\begin{split}
     T(z) \JBRST(w) \sim \frac{c_{\alpha}J_{\CVs}^{\alpha}(w)}{(z-w)^2}+\frac{f^{\alpha\beta}{}_{\gamma}\mathsf b^{\gamma}\mathsf c_{\alpha}\mathsf c_{\beta}(w)}{2(z-w)^2} \;.
\end{split}
\end{equation}
and hence
\begin{align} 
     [\QBRST, T(w) ] =\oint_{|z-w|=1} \frac{dz}{2\pi i} \JBRST(z)\, T(w) = 0~. 
\end{align}
The stress-energy tensor \eqref{def:generalTtensor_general} is hence BRST invariant. 

There are more BRST-closed operators associated with paths of the quiver diagram.
Explicit expressions of these operators will depend on the details of the quiver, and for simplicity of the presentation we will discuss them in \cref{sec:VOA_linear},
when we specialize our discussion to linear quivers.

The central charge of the stress-energy tensor \eqref{def:generalTtensor_general} is 
\begin{equation}\label{eqn: central charge}
\begin{split}
    c_{\CVs \otimes \mathcal{V}_{\mathsf{bc}}} 
    &=\underbrace{-\sum_{i\in Q_0} \mathsf{Q}_{i,i} v_i^2-\sum_{i\neq j} \mathsf{Q}_{i,j} v_i v_{j}-\sum_{i \in Q_0} v_iw_i}_{\Sb} +\underbrace{\sum_{i \in Q_0} v_id_i}_{\Ff} 
    +  \underbrace{\Qzero}_{h} 
     \underbrace{-2\sum_{i\in Q_0} v_i^2}_{\bc}
    \\
    &=|Q_0| -\frac{1}{2}\mathbf v\cdot \mathsf C\mathbf v-3\mathbf v\cdot \mathbf v~,
\end{split}
\end{equation}
where the equality on the second line follows from substituting the anomaly cancellation condition \eqref{def_d}.

\bigskip When the Cartan matrix $\mathsf C$ is not invertible, there is a central part inside $\CV(Q,\mathbf v,\mathbf w)$ strongly generated by $\{h\in \mathbb C^{Q_0}\:|\: \mathsf Ch=0\}$. Let us denote the vertex ideal of $\CVs(Q,\mathbf v,\mathbf w)$ strongly generated by the kernel of $\mathsf C$ by $\ker(\mathsf C)$. We define $\CVs(Q,\mathbf v,\mathbf w)_{\mathrm{red}}$ to be the quotient of $\CVs(Q,\mathbf v,\mathbf w)$ by the vertex ideal $\ker(\mathsf C)$, and subsequently define $\CV(Q,\mathbf v,\mathbf w)_{\mathrm{red}}$ to be the BRST reduction of $\CVs(Q,\mathbf v,\mathbf w)_{\mathrm{red}}$, i.e.
\begin{equation}
\begin{split}
    \CV^\bullet(Q,\mathbf v,\mathbf w)_{\mathrm{red}}&:=H^{\infty/2+\bullet}_{\mathrm{BRST}}(\mathfrak{g},\CVs(Q,\mathbf v,\mathbf w)_{\mathrm{red}})\;,\\
    \CV(Q,\mathbf v,\mathbf w)_{\mathrm{red}}&:=\CV^0(Q,\mathbf v,\mathbf w)_{\mathrm{red}} \;.
\end{split}
\end{equation}
Now $\CVs(Q,\mathbf v,\mathbf w)_{\mathrm{red}}$ admits a stress-energy tensor $T=T_{\Sb}+T_{\Ff}+T_{\overline{h}}+T_{\mathsf{bc}} $, where 
\begin{align}
    T_{\overline{h}}:=\frac{1}{2}\sum_{i,j} (\overline{\mathsf{C}}^{-1})_{i,j} \overline{h}_{i} \overline{h}_{j} ~.
\end{align}
Here $\overline{h}_{i}$ are chosen from a basis of $\mathbb C^{Q_0}/\ker(\mathsf C)$, and $\overline{\mathsf{C}}$ is the induced inner form on $\mathbb C^{Q_0}/\ker(\mathsf C)$. The same argument as above shows that $T$ is BRST closed, thus it defines a stress-energy operator of $\CV(Q,\mathbf v,\mathbf w)_{\mathrm{red}}$, and its central charge can be computed similarly:
\begin{align}
    c_{\CV_{\mathrm{red}}} 
    =\mathrm{rank}(\mathsf C) -\frac{1}{2}\mathbf v\cdot \mathsf C\mathbf v-3\mathbf v\cdot \mathbf v~.
\end{align}

\section{Boundary VOAs for linear quivers gauge theories}\label{sec:VOA_linear}

Having discussed the VOAs associated with general quiver gauge theories, let us next discuss the boundary VOAs of 
linear quiver gauge theories presented in \cref{sec:T_rho_sigma}.

\subsection{Anomaly cancellation on the boundary}\label{subsection:anomaly_linear}

Since we are now merely specializing the discussion from \cref{sec:VOA-quiver}, 
all the statements therein will be applicable here.

\paragraph{$\SU$ anomalies:}

For cancellation of $\SU$ anomalies it is necessary to add $d_i$ additional left-handed symplectic 2d fermions in a fundamental or antifundamental representation of the gauge group $\SU(v_i)$, where
$\{d_i\}$ are determined by the anomaly cancellation condition \eqref{def_d}:
\begin{equation}
   \label{balancing_condition}
   d_i :=v_{i-1}+v_{i+1} -2v_i +w_i  \in \mathbb{Z}_{\ge 0}^{\Qzero} ~.  
\end{equation}
In this case, the quiver locally resembles that depicted in \cref{fig:quiver node for anomaly} around the gauge node $i\in Q_0$. 

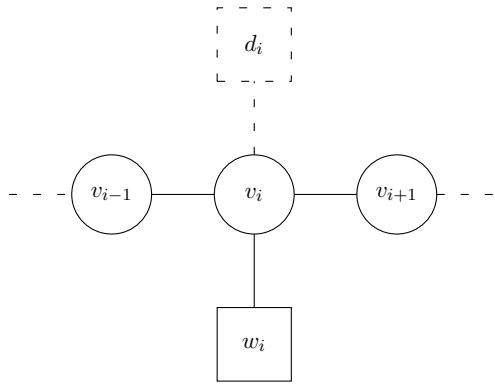
\begin{figure}[ht!]
    \centering
    \begin{tikzpicture}[scale=.75, every node/.style={scale=0.8}]  
        \draw[loosely dashed] (-7.8, 0) -- (-6.7, 0);
        \draw[] (-6,0) circle (.7);
        \draw[] (-5.3,0) -- (-4.2,0);
        \draw[] (-3.5,0) circle (.7);
        \draw[] (-2.8,0) -- (-1.7,0);
        \draw[] (-1,0) circle (.7);    
        \draw[loosely dashed] (-0.3, 0) -- (0.8, 0);
        \draw[] (-3.5,-.7) -- (-3.5,-2);
        \draw[loosely dashed] (-3.5,.7) -- (-3.5,2);
        \draw[loosely dashed] (-4.15,2) rectangle (-2.85,3.3);
        \draw[] (-4.15,-2) rectangle (-2.85,-3.3);
        \node at (-6,0) {$v_{i-1}$};
        \node at (-3.5,0) {$v_{i}$};
        \node at (-1,0) {$v_{i+1}$};
        \node at (-3.5,-2.65) {$w_{i}$};
        \node at (-3.5,2.65) {$d_{i}$};
    \end{tikzpicture}
    \caption{Local portion of an A-type quiver around the $i^{th}$ gauge node. The condition for $\SU(v_i)$ gauge anomaly cancellation is $2v_i-v_{i-1}-v_{i+1}-w_i + d_i = 0$. The dashed square node represents not a part of the quiver diagram of the bulk 3d theory, but rather boundary degrees of freedom which are introduced only to cancel the boundary gauge anomaly.}
    \label{fig:quiver node for anomaly}
\end{figure}

\paragraph{$\U(1)$ anomalies}
For linear quivers there exists a somewhat canonical way of
canceling $\U(1)$ anomalies---a collection of $\mathcal{N}=(0,2)$ Fermi multiplets with chiral free fermions $(\chi_i, \psi_i)$ and label $i=0,1, \ldots, \Qzero$ are added on the boundary to cancel the Abelian gauge anomaly $\prod_{i\in Q_0 }\U(1)_{i, \textrm{diag}}$. This generates the $\U(1)$ current algebra strongly generated by
\begin{align}\label{eq:u(1) currents from fermions}
    h_i = \normord{\chi_i \psi_i}-\normord{\chi_{i-1} \psi_{i-1}}.
\end{align}
It is straightforward to see that $h_i(z)h_j(w)\sim {\mathsf C_{i,j}}/{(z-w)^2}$, therefore the Heisenberg algebra $\{h_i\}_{i\in Q_0}$ cancels the $\U(1)$ anomaly as we discussed in \cref{subsection:anomaly}.

\subsection{VOA before BRST reduction}
\label{subsection:pre-reduction_linear}

Having discussed the anomaly cancellation, we can now write down the VA $\CVsdel \otimes \mathcal{V}_{\bc}$ before the BRST reduction:
\begin{itemize} 
    \item A collection of ${\Sb}$-systems
        \begin{itemize}
            \item $(I_{i}, J_{i})_{i\in Q_0}$ valued in fundamental or antifundamental representations of the 
            $i\in \{1,\ldots ,\Qzero\}$ flavor nodes  and  
            \item $(X_{i}, Y_{i})_{i\in \{1,\ldots ,\Qzero-1\}}$ valued in bifundamental representations of the gauge groups.  
        \end{itemize}
        These fields are associated to the scalar fields coming from the 3d $\mathcal{N}=4$ hypermultiplets valued in a symplectic representation of the gauge group $G$, which will be denoted by $\CV_{\beta\gamma}$.
        
    \item A set of free fermions $(\F_i,\Psi_i)$ for any gauge nodes with $d_i\neq 0$ for labels $i \in \{1,\ldots , \Qzero\}$, which will be denoted by $\CV_{\Ff}$.
    
    \item A set of chiral free fermions $(\chi_i, \psi_i)_{i=0,1,\ldots , \Qzero}$ introduced on the boundary to cancel the Abelian gauge anomaly, with pairs $(\chi_i, \psi_i)$ and $(\chi_{i-1}, \psi_{i-1})$ for each gauge node $i \in \{1,\ldots , \Qzero\}$. This system of free fermions will be denoted by $\CV_{\chi\psi}$.
    
    \item A $\bc$ ghost system of fermionic fields $(\mathsf{b}_{i}, \mathsf{c}_{i})_{i\in Q_0}$ which is needed to fix the gauge freedom. The conformal weights of $\mathsf{b}$ and $\mathsf{c}$ are chosen to be $1$ and $0$, respectively. The $\bc$ ghost system will be denoted $\CV_{\bc}$.
\end{itemize} 

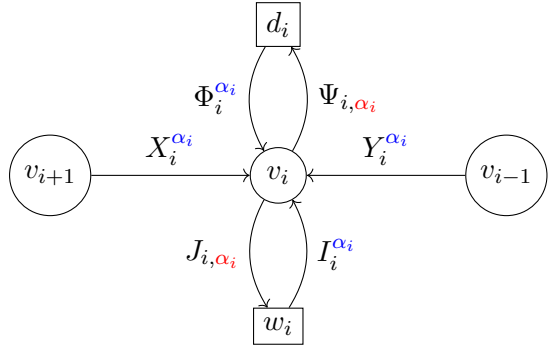
\begin{figure}[ht!]
    \centering
    \begin{tikzpicture}[vertex/.style={draw,circle,minimum size=4mm},fvertex/.style={draw,rectangle,minimum size=4mm}]
          \node[fvertex] (di) at (0,0) {$d_i$};
          \node[vertex] (giplus1) at (-3,-2) {$v_{i+1}$};
          \node[vertex] (gi) at (0,-2) {$v_i$};
          \node[vertex] (giminus1) at (3,-2) {$v_{i-1}$};
          \node[fvertex] (fi) at (0,-4) {$w_i$};
        %
          \draw[->] (giplus1) -- node[above]{$X_{i}^{\textcolor{blue}{\alpha_{i}}}$} (gi);
          \draw[->] (giminus1) -- node[above]{$Y_{i}^{\textcolor{blue}{\alpha_{i}}}$} (gi);
          \draw[->, bend right] (fi) to node[right]{$I_{i}^{\textcolor{blue}{\alpha_{i}}}$} (gi);
          \draw[->, bend right] (gi) to node[left]{$J_{i,\textcolor{red}{\alpha_{i}}}$} (fi);
          \draw[->, bend right] (di) to node[left]{$\F^{\textcolor{blue}{\alpha_{i}}}_{i}$} (gi);
          \draw[->, bend right] (gi) to node[right]{$\Psi_{i,\textcolor{red}{\alpha_{i}}}$} (di);
    \end{tikzpicture}
    \caption{Index convention at the gauge node $i$, where $\alpha_{i}$ is the gauge index for the gauge subgroup $G_{i}$. Other indices are suppressed.}
    \label{fig:convention}
\end{figure}

The convention we adopt is summarized pictorially in \cref{fig:convention}, which presents the information locally around the $i^{\rm th}$ gauge node of the quiver similarly to \cref{fig:quiver node for anomaly}. It shows how the symplectic bosons $X_{i}, Y_{i}, I_{i}$ and $J_{i}$ and the fermions $\F_{i}$ and $\Psi_{i}$ are associated to arrows on the quiver edges, together with the convention for their indices. For example $X_{i}$ is associated to the quiver edge directed from the $(i+1)^{\rm th}$ gauge node to the $i^{\rm th}$ gauge node, with the convention  
\begin{align} 
    \circled{$v_{i+1}$} \quad & \overrightarrow{ {}_{(X_{i})^{\a_{i}}{}_{\a_{i+1}}} } \quad \circled{$v_i$} \;, \nonumber
\end{align}
where $\a_{i}$ is an index for the gauge group $G_{i}$ and $\a_{i+1}$, an index for the gauge group $G_{i+1}$. 

The boundary field content defines a vertex operator algebra $\CVsdel\otimes \CV_{\bc}$ on $\partial M_3$ where
\begin{equation}
   \CVsdel := \mathcal{V}_{\Sb} \otimes \mathcal{V}_{\Ff} \otimes \mathcal{V}_{\cp}  ~. 
\end{equation}  
The collection of fields obeys the following nontrivial OPEs 
\begin{gather}
\begin{split}
    &  \qquad (X_{i})^{\alpha}{}_{\beta} (z) \, (Y_{j})^{\gamma}{}_{\delta} (w) \sim \frac{\delta^{\alpha}_{\delta}\delta^{\gamma}_{\beta}\delta_{i,j}}{z-w} ~, 
    \\
    & \qquad (I_{i})^{\alpha}{}_{\beta} (z) \, (J_{j})^{\gamma}{}_{\delta}d (w) \sim \frac{\delta^{\alpha}_{\delta}\delta^{\gamma}_{\beta}\delta_{i,j}}{z-w} ~,
    \\
    & \qquad (\F_{i})^{\alpha}{}_{\beta} (z) \, (\Psi_{j})^{\gamma}{}_{\delta} (w) \sim \frac{\delta^{\alpha}_{\delta}\delta^{\gamma}_{\beta}\delta_{i,j}}{z-w} ~,
    \\
    & \qquad (\mathsf b_{i})^{\alpha}{}_{\beta} \,(z)  \, (\mathsf c_{j})^{\gamma}{}_{\delta} \,(w) \sim \frac{\delta^{\alpha}_{\delta}\delta^{\gamma}_{\beta} \delta_{i,j} }{z-w} ~, \\
    & \qquad \chi_i (z) \, \psi_j (w) \sim \frac{\delta_{i,j}}{z-w}~,    
\end{split}
\end{gather}
while all other OPEs are trivial. 

The total stress-energy tensor is 
\begin{equation} 
    \label{def:generalTtensor}
    T=T_{\Sb}+T_{\Ff}+T_{\chi\psi}+T_{\mathsf{bc}} 
    ~ , 
\end{equation} 
where $T_{\Sb}$, $T_{\Ff}$, and $T_{\mathsf{bc}}$ are defined in the same way as in the \eqref{def:generalTtensor_general}, and $T_{\chi\psi}$ is defined by 
\begin{align}
    T_{\chi\psi}&:=-\frac{1}{2}\sum_{i=0}^{\Qzero} (\chi_{i}\partial \psi_{i}-\partial \chi_{i} \psi_{i})~, 
\end{align}
The central charge of the stress-energy tensor \eqref{def:generalTtensor} is 
\begin{equation}\label{eqn: central charge for linear quivers}
\begin{split}
    c_{ \CVs_{\partial}\otimes\CV_{\mathsf b\mathsf c}} 
    &=\underbrace{-\sum_{i=1}^{\Qzero-1} v_iv_{i+1}-\sum_{i=1}^{\Qzero} v_iw_i}_{\Sb} +\underbrace{\sum_{i=1}^{\Qzero} v_id_i}_{\Ff} 
    +  \underbrace{1+\Qzero}_{\chi \psi} 
     \underbrace{-2\sum_{i=1}^{\CI} v_i^2}_{\bc}
    \\
    &=1+\Qzero + \sum_{i=1}^{\Qzero-1} v_iv_{i+1}-4\sum_{i=1}^{\Qzero}v_i^2~,
\end{split}
\end{equation}
where the equality on the second line follows from substituting the anomaly cancellation condition \eqref{def_d}.

\subsection{BRST reduction}
\label{sec:BRST_operators}

Let us next perform a BRST reduction of the VOA $\CVsdel$. We notice that $\CVsdel$ contains the vertex subalgebra $\CVs:=\CV_{\beta\gamma}\otimes\CV_{\Ff}\otimes\mathcal H$ where $\mathcal H$ is the Heisenberg algebra strongly generated by $\{h_i\}_{i\in Q_0}$ in \eqref{eq:u(1) currents from fermions}, $\CVs$ being exactly the vertex algebra associated to the quiver $(Q,\mathbf v,\mathbf w)$ before BRST reduction defined in \cref{subsection:anomaly}. The stress-energy tensor $T$ in \eqref{def:generalTtensor} belongs to the vertex subalgebra $\CVs$. We also notice that the boson-fermion correspondence provides a decomposition of $\CVsdel$
\begin{align}
    \CVsdel=\bigoplus_{\lambda\in \mathbb Z^{\Qzero+1}}\CVs_{\partial,\lambda} ~,
\end{align}
into a lattice extension of the vertex subalgebra $\CVs_{\partial,0}$, where 
\begin{align}
    \CVs_{\partial,0}=\CVs\otimes J_{\chi\psi},\quad  J_{\chi\psi}=\sum_{i=0}^{\Qzero} \chi_i \psi_i ~.
\end{align}
Therefore we can perform BRST reduction of $\CVsdel$ using the chiral extended moment map \eqref{eq:chiral extended moment map}, and define
\begin{align}
    \CVdel^\bullet:=H^{\infty/2+\bullet}_{\mathrm{BRST}}(\mathfrak{g},\CVsdel)\;,
    \quad
    \text{and }\CVdel:=\CVdel^0 \;.
\end{align}
Then we have a decomposition of $\CVdel$
\begin{align}\label{lattice decomposition of boundary VOA}
    \CVdel=\bigoplus_{\lambda\in \mathbb Z^{\Qzero+1}}\CV_{\partial,\lambda} ~,\quad  \CV_{\partial,0}=\CV\otimes J_{\chi\psi}~,
\end{align}
into a lattice extension of the vertex subalgebra $\CV\otimes J_{\chi\psi}$, where $\CV=\CV(Q,\mathbf v,\mathbf w)$ is the vertex algebra associated to the quiver $(Q,\mathbf v,\mathbf w)$ defined in \eqref{def:VA of quiver}. The $\CV\otimes J_{\chi\psi}$-module $\CV_{\partial,\lambda}$ is
\begin{align}
    H^{\infty/2+0}_{\mathrm{BRST}}(\mathfrak{g},\CV_{\beta\gamma}\otimes\CV_{\Ff}\otimes\mathcal H_{\bar\lambda})\otimes J_{\chi\psi}||\lambda|\rangle,
\end{align}
where $\mathcal H_{\bar\lambda}$ is the module of $\mathcal H$ generated by $|\bar\lambda\rangle$ such that $h_{i,(0)}|\bar\lambda\rangle=(\lambda_i-\lambda_{i-1})|\bar\lambda\rangle$, and $||\lambda|\rangle$ generates a module of $J_{\chi\psi}$ such that $J_{\chi\psi,(0)}||\lambda|\rangle=\sum_{i=0}^{\Qzero}\lambda_i||\lambda|\rangle$.

\bigskip Let us discuss examples of bosonic operators that can be constructed for A-type quivers and which are invariant with respect to the action of BRST charges \eqref{def:BRSTcurrent} defined using the affine $\mathfrak{u}(v_i)$ currents $\widetilde\mu_{\mathrm{ch},i}$ \eqref{eq:chiral extended moment map}:

\begin{itemize}
    \item The stress energy tensor $T$ defined in equation \eqref{def:generalTtensor}~.
\end{itemize}

\begin{itemize}
    \item The $\U(1)$ current $J_{\chi\psi}$, which is constructed using  the chiral free fermions  
          \begin{equation}\label{def:U(1)chiralbilincurrent}
            J_{\chi\psi} = \sum_{i=0}^{\Qzero} \chi_i \psi_i ~,
          \end{equation}
\end{itemize}

This current is invariant with respect to the BRST currents \eqref{def:BRSTcurrent} constructed using the 
currents $\widetilde\mu_{\mathrm{ch},i}$ \eqref{eq:chiral extended moment map}, the former having only regular terms in the OPEs with the latter
\begin{equation}
   J_{\chi\psi}(z) \widetilde\mu_{\mathrm{ch},i}(w) \sim 0  ~. 
\end{equation} 

\begin{itemize}
    \item Currents associated to paths on the quiver diagram in \cref{fig:genericAquiver} which begin and end at framing nodes. For example, one can consider operators
    \begin{align}\label{G_path_def}
    \begin{split}
     G_{(i,j)} &: =   J_j X X \dots X I_i \quad (i>j) \;, \\
     G_{(i,j)} &:=   J_j Y Y \dots Y I_i \quad (i<j)\;,
    \end{split}
    \end{align}
    which are represented by a path on the quiver diagram as in \cref{fig:example path}.

\begin{figure}[htbp]
    \centering
\[\begin{tikzcd}
	{\circled{$v_{\Qzero}$}} & \cdots & {\circled{$v_3$}} && {\circled{$v_2$}} && {\circled{$v_1$}} 
	\\
	{\rectangled{$w_{\Qzero}$}} & \cdots & {\rectangled{$w_3$}} && {\rectangled{$w_2$}} && {\rectangled{$w_1$}}
	\arrow["{I^{i}_{1,a}}"', from=2-7, to=1-7]
	\arrow["{Y_{1,i}^{j}}"', from=1-7, to=1-5]
	\arrow["{Y_{2,j}^{k}}"', from=1-5, to=1-3]
	\arrow["{J_{3,k}^{b}}"', from=1-3, to=2-3]
	\arrow[dotted, no head, from=2-5, to=1-5]
	\arrow[dotted, no head, from=2-1, to=1-1]
\end{tikzcd}\]
    \caption{The operator $G_{(3,1)}=J_{3}Y_{2}Y_{1} I_{1}$}
    \label{fig:example path}
\end{figure}
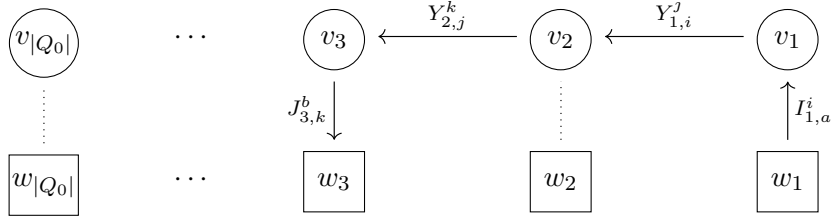
  
The operators \eqref{G_path_def} are primary fields of dimension $(|i-j|+2)/2$, and moreover are BRST exact.

\end{itemize}

To explain the BRST exactness, 
let us for example compute the OPE of 
$\widetilde\mu_{\mathrm{ch},2}$ with the operator
$G_{(3,1)}=J_{3}Y_{2}Y_{1} I_{1}$ of \cref{fig:example path}. Most of the terms in 
the definition of the current \eqref{eq:chiral extended moment map} have trivial OPEs with $G_{(3,1)}$,
with the exception of 
\begin{align}
    (\widetilde\mu_{\mathrm{ch},2})^{\alpha}{}_{\beta} \supset (X_{2})^{\alpha}{}_{\gamma} (Y_{2})^{\gamma}{}_{\beta} - (Y_{2})^{\alpha}{}_{\gamma} (X_{1})^{\gamma}{}_{\beta} \;.
\end{align}
We can compute
\begin{align}
\begin{split}
     [(X_{2})^{\alpha}{}_{\gamma} (Y_{2})^{\gamma}{}_{\beta}] (z) G_{(3,1)} (w) 
     &\sim  \wick{
     \c (X_{2})^{\alpha}{}_{\gamma} (Y_{2})^{\gamma}{}_{\beta} [J_{3} \c Y_{2} Y_{1} I_{1}](w)}\\
     &\sim 
     \frac{1}{z-w} (J_{3} Y_{2})_{\beta}(Y_{1} I_{1})^{\alpha} \;.
 \end{split}
\end{align}
Similarly, we can compute 
\begin{align}
\begin{split}
     [(Y_{1})^{\alpha}{}_{\gamma} (X_{1})^{\gamma}{}_{\beta}] (z) G_{(3,1)} (w) 
     \sim
     \frac{1}{z-w} (J_{3} Y_{2})_{\beta}(Y_{1} I_{1})^{\alpha} \;.
 \end{split}
\end{align}
and hence $(\widetilde\mu_{\mathrm{ch},2})^{\alpha}{}_{\beta} G_{(3,1)} \sim 0$. Since the $\mathsf{bc}$ ghosts have trivial OPE with $G_{(3,1)}$, we have OPE $\JBRST (z) G_{(3,1)}(w)\sim 0$ and this proves the BRST invariance of $G_{(3,1)}$. 

The proof for a more general $G_{(i,j)}$ operator \eqref{G_path_def} works in essentially the same way, since the discussion is localized in the neighborhood of the quiver diagram.

We can consider a more general path. The simplest possibility is the path in 
\cref{fig:simplests dim1}. In this case, 
the naive operator
$G_{(i,i)} \overset{?}{=} J_{i} I_{i}$
is not BRST invariant due to the 
presence of a double pole in the
OPE of $G_{(i,i)}$ and $\textrm{Tr}(I_{i} \partial J_{i} - \partial I_{i} J_{i}) \subset T$. This problem, however, can be dealt with by introducing an extra term involving fermions: 
\begin{align} \label{Gii}
    (G_{(i,i)})^a{}_b = (J_i I_i)^a{}_b + \left( \chi_i \psi_i + \chi_{i+1} \psi_{i+1} + \chi_{i+2} \psi_{i+2} + \cdots \right)\delta^a{}_b \;,
\end{align} 
    
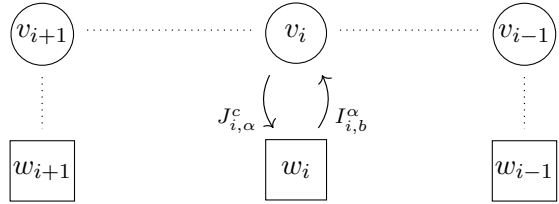
\begin{figure}[htbp]
\centering
   \[\begin{tikzcd}
	{\circled{$v_{i+1}$}} & {} & {\circled{$v_i$}} && {\circled{$v_{i-1}$}} \\
	{\rectangled{$w_{i+1}$}} && {\rectangled{$w_i$}} && {\rectangled{$w_{i-1}$}}
	\arrow["I_{i,b}^{\a}"', bend right, from=2-3, to=1-3] 
   \arrow["{J_{i,\a}^{c}}"', bend right, from=1-3, to=2-3]   
	\arrow[dotted, no head, from=1-5, to=2-5]
	\arrow[dotted, no head, from=1-1, to=2-1]
	\arrow[dotted, no head, from=1-3, to=1-5]
	\arrow[dotted, no head, from=1-1, to=1-3]
   \end{tikzcd}\]
   \caption{The bosonic generators $G_{(i,i),b}^{c}$ at framing node $w_i$.}
   \label{fig:simplests dim1}
\end{figure}

More generally, one 
considers operators which are schematically of the form
\begin{align} \label{J_Phi_I}
     G \sim   J_j \mathcal O_1 \mathcal O_2 \dots \mathcal O_m I_i \;, \quad \mathcal O_i  = X \textrm{ or }Y \;.
\end{align}
The expression \eqref{J_Phi_I} is not meant to be a precise expression in general, and as we will see we 
might need to add extra terms to make the operator BRST-invariant, similarly to those in equation \eqref{Gii}.
On the quiver diagram the operator \eqref{J_Phi_I} represents a path starting from one of the framing nodes $i$, going around the quiver, and then
ending at another framing node $j$. 

Readers will have noticed that the discussion of the BRST-closed operators in this subsection is reminiscent of the 
generators of the coordinate ring of the Higgs branch in \cref{subsec:Higgs}.
This is not a coincidence, and we will see throughout the rest of this paper the close parallel between the geometry of the Higgs branch and the its associated VOA: the VOA is a ``chiralization'' of the Higgs branch. We will in particular formalize in \cref{sec:variety} this correspondence in the language of the associated variety. 

We have to this point listed examples of BRST-invariant operators
which will survive after the BRST quotient and hence will define non-trivial operators in $\CVdel$. However, it turns out that the generators \eqref{J_Phi_I} still contain redundancies and cannot be used as a subset for a strong generator of VOA. Here a strong generator of a VOA is a set of operators generating the VOA, such that none of the generators appear in the non-singular terms of the OPEs of other generators. In other words, we do not include composite operators in the set of generators. We will come back to this issue in \cref{subsec:conjectures}. 

The connections between BRST-invariant operators and  paths on the quiver
also sheds light on the OPE between operators. Defining for example an operator $G_{(m\leftarrow j)}$ associated to a path on the quiver diagram in \cref{fig:pathsquiverexample} starting at the framing node $j$ and ending at the framing node $m$, and an operator $G_{(l \rightarrow i)}$ associated to a path starting at the framing node $l$ and ending at the framing node $i$
\begin{equation}
    G_{(m \leftarrow j)} = J_{m} Y_{l} Y_{k} Y_{j} I_{j} ~,
    \qquad
    G_{(l \rightarrow i)} = J_{i} X_{i} X_{j} X_{k} I_{l} ~,
\end{equation}
their OPE takes the form 
\begin{equation} \label{GGG}
    G_{(m \leftarrow j)}(z) G_{(l \rightarrow i)}(w) \sim \frac{ v_k G_{(m \leftarrow l)} G_{(j \rightarrow i)} (w)}{(z-w)^2} + \frac{G_{(j \rightarrow i)}\mathcal{O}_1(w)   +   G_{(m \leftarrow l)}\mathcal{O}_2 (w)}{(z-w)} ~,
\end{equation}
for some operators $\mathcal{O}_1$ and $\mathcal{O}_2$.
At the level of the paths on the quiver, this operation removes pairs of juxtaposed segments in \cref{fig:pathsquiverexample} with opposite orientation, each such pair producing a singular contribution $(z-w)^{-1}$ and operators associated to the paths formed by the remaining segments in \cref{fig:pathsquiverexample2}.  

\begin{figure}
\centering
    \begin{tikzpicture}

    \node (leftfigure) at (0,0){
        \begin{tikzpicture}[scale=.6, every node/.style={scale=0.8}];
        \node(-3) at (-7.5,0){};
        \node(-2) at (-5,0){$\circled{$m$}$};
        \node(-1) at (-2.5,0){$\circled{$l$}$};
        \node(0) at (0,0){$\circled{$k$}$};
        \node(1) at (2.5,0){$\circled{$j$}$};
        \node(2) at (5,0){$\circled{$i$}$};
        \node(3) at (7.5,0){};
        \node(w-2) at (-5,-2.5){$\rectangled{}$};
        \node(w-1) at (-2.5,-2.5){$\rectangled{}$};
        \node(w0) at (0,-2.5){$\rectangled{}$};
        \node(w1) at (2.5,-2.5){$\rectangled{}$};
        \node(w2) at (5,-2.5){$\rectangled{}$};
        \draw (-3) -- (-2);
        \draw (-2) -- (-1);
        \draw (-1) -- (0);
        \draw (0) -- (1);
        \draw (1) -- (2);
        \draw (2) -- (3);
        
        \draw (w-2) --(-2);
        \draw (w-1) --(-1);
        \draw (w0) --(0);
        \draw (w1) --(1);
        \draw (w2) --(2);

        \draw[->,red,thick] (2.0,-2) -- node [left] {$I_{i}$}(2.0, 0.8);
        \draw[->,red,thick] (2.0, 0.8) -- node [below] {$Y_{j}$}(-0.5, 0.8);
        \draw[->,red,thick] (-0.5, 0.8) -- node [below] {$Y_{k}$}(-3.0, 0.8);
        \draw[->,red,thick] (-3.0, 0.8) -- node [below] {$Y_{l}$}(-5.5, 0.8);
        \draw[<-,red,thick] (-5.5,-2) -- node [left] {$J_{m}$}(-5.5, 0.8);

        \draw[->,blue,thick] (-3.0,-2) -- node [left] {$I_{l}$}(-3.0, 1.2);
        \draw[->,blue,thick] (-3.0, 1.2) -- node [above] {$X_{k}$}(-0.5, 1.2);
        \draw[->,blue,thick] (-0.5, 1.2) -- node [above] {$X_{j}$}(2.0, 1.2);
        \draw[->,blue,thick] (2.0, 1.2) -- node [above] {$X_{i}$}(4.5, 1.2);
        \draw[->,blue,thick] (4.5,1.2) -- node [left] {$J_{i}$}(4.5, -2);
        \end{tikzpicture}
    };
    \end{tikzpicture}
    \caption{We consider the OPE of two BRST-invariant operators, one represented by a red path and another by a blue path; see the LHS of \eqref{GGG}.}
    \label{fig:pathsquiverexample}
\end{figure}
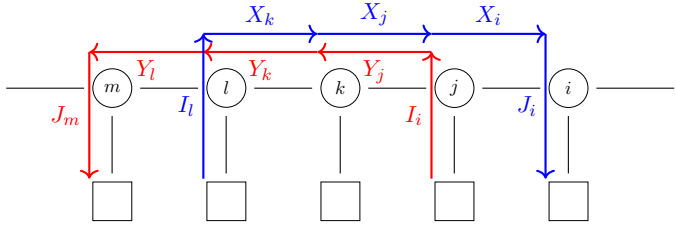

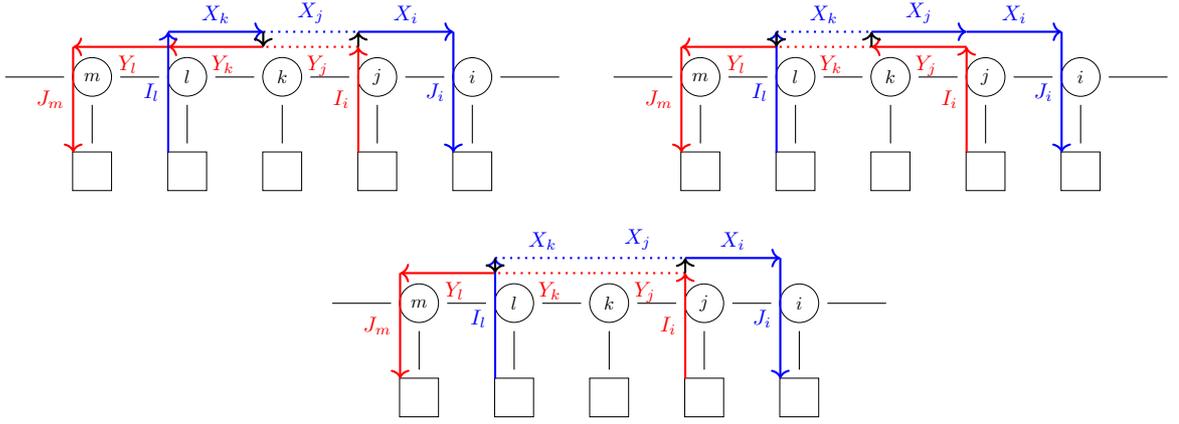
\begin{figure}
\centering
    \begin{tikzpicture}

    \node (leftfigure) at (-0.5,0){
        \begin{tikzpicture}[scale=.5, every node/.style={scale=0.8}];
        \node(-3) at (-7.5,0){};
        \node(-2) at (-5,0){$\circled{$m$}$};
        \node(-1) at (-2.5,0){$\circled{$l$}$};
        \node(0) at (0,0){$\circled{$k$}$};
        \node(1) at (2.5,0){$\circled{$j$}$};
        \node(2) at (5,0){$\circled{$i$}$};
        \node(3) at (7.5,0){};
        \node(w-2) at (-5,-2.5){$\rectangled{}$};
        \node(w-1) at (-2.5,-2.5){$\rectangled{}$};
        \node(w0) at (0,-2.5){$\rectangled{}$};
        \node(w1) at (2.5,-2.5){$\rectangled{}$};
        \node(w2) at (5,-2.5){$\rectangled{}$};
        \draw (-3) -- (-2);
        \draw (-2) -- (-1);
        \draw (-1) -- (0);
        \draw (0) -- (1);
        \draw (1) -- (2);
        \draw (2) -- (3);
        
        \draw (w-2) --(-2);
        \draw (w-1) --(-1);
        \draw (w0) --(0);
        \draw (w1) --(1);
        \draw (w2) --(2);

        \draw[->,red,thick] (2.0,-2) -- node [left] {{\small  $I_{i}$}}(2.0, 0.8);
        \draw[dotted,red,thick] (2.0, 0.8) -- node [below] {{\small ~ $Y_{j}$}}(-0.5, 0.8);
        \draw[->,red,thick] (-0.5, 0.8) -- node [below] {{\small ~ $Y_{k}$}}(-3.0, 0.8);
        \draw[->,red,thick] (-3.0, 0.8) -- node [below] {{\small ~ $Y_{l}$}}(-5.5, 0.8);
        \draw[<-,red,thick] (-5.5,-2) -- node [left] {{\small $J_{m}$}}(-5.5, 0.8);

        \draw[->,blue,thick] (-3.0,-2) -- node [left] {{\small $I_{l}$}}(-3.0, 1.2);
        \draw[->,blue,thick] (-3.0, 1.2) -- node [above] {{\small $X_{k}$}}(-0.5, 1.2);
        \draw[dotted,blue,thick] (-0.5, 1.2) -- node [above] {{\small $X_{j}$}}(2.0, 1.2);
        \draw[->,blue,thick] (2.0, 1.2) -- node [above] {{\small $X_{i}$}}(4.5, 1.2);
        \draw[->,blue,thick] (4.5,1.2) -- node [left] {{\small $J_{i}$}}(4.5, -2);

        \draw[->,thick] (-0.5, 1.2) -- (-0.5, 0.8);
        \draw[<-,thick] (2.0, 1.2) -- (2.0, 0.8);
            
        \end{tikzpicture}
    };

    \node (rightfigure) at (7.5,0){
        \begin{tikzpicture}[scale=.5, every node/.style={scale=0.8}];
        \node(-3) at (-7.5,0){};
        \node(-2) at (-5,0){$\circled{$m$}$};
        \node(-1) at (-2.5,0){$\circled{$l$}$};
        \node(0) at (0,0){$\circled{$k$}$};
        \node(1) at (2.5,0){$\circled{$j$}$};
        \node(2) at (5,0){$\circled{$i$}$};
        \node(3) at (7.5,0){};
        \node(w-2) at (-5,-2.5){$\rectangled{}$};
        \node(w-1) at (-2.5,-2.5){$\rectangled{}$};
        \node(w0) at (0,-2.5){$\rectangled{}$};
        \node(w1) at (2.5,-2.5){$\rectangled{}$};
        \node(w2) at (5,-2.5){$\rectangled{}$};
        \draw (-3) -- (-2);
        \draw (-2) -- (-1);
        \draw (-1) -- (0);
        \draw (0) -- (1);
        \draw (1) -- (2);
        \draw (2) -- (3);
        
        \draw (w-2) --(-2);
        \draw (w-1) --(-1);
        \draw (w0) --(0);
        \draw (w1) --(1);
        \draw (w2) --(2);

        \draw[->,red,thick] (2.0,-2) -- node [left] {{\small $I_{i}$}}(2.0, 0.8);
        \draw[->,red,thick] (2.0, 0.8) -- node [below] {{\small ~ $Y_{j}$}}(-0.5, 0.8);
        \draw[dotted,red,thick] (-0.5, 0.8) -- node [below] {{\small ~ $Y_{k}$}}(-3.0, 0.8);
        \draw[->,red,thick] (-3.0, 0.8) -- node [below] {{\small ~ $Y_{l}$}}(-5.5, 0.8);
        \draw[<-,red,thick] (-5.5,-2) -- node [left] {{\small $J_{m}$}}(-5.5, 0.8);

        \draw[->,blue,thick] (-3.0,-2) -- node [left] {{\small $I_{l}$}}(-3.0, 1.2);
        \draw[dotted,blue,thick] (-3.0, 1.2) -- node [above] {{\small $X_{k}$}}(-0.5, 1.2);
        \draw[->,blue,thick] (-0.5, 1.2) -- node [above] {{\small $X_{j}$}}(2.0, 1.2);
        \draw[->,blue,thick] (2.0, 1.2) -- node [above] {{\small $X_{i}$}}(4.5, 1.2);
        \draw[->,blue,thick] (4.5,1.2) -- node [left] {{\small $J_{i}$}}(4.5, -2);

        \draw[<-,thick] (-0.5, 1.2) -- (-0.5, 0.8);
        \draw[->,thick] (-3.0, 1.2) -- (-3.0, 0.8);
        \end{tikzpicture}
    };

    \node (bottomfigure) at (3.8,-3.0){
        \begin{tikzpicture}[scale=.5, every node/.style={scale=0.8}];
        \node(-3) at (-7.5,0){};
        \node(-2) at (-5,0){$\circled{$m$}$};
        \node(-1) at (-2.5,0){$\circled{$l$}$};
        \node(0) at (0,0){$\circled{$k$}$};
        \node(1) at (2.5,0){$\circled{$j$}$};
        \node(2) at (5,0){$\circled{$i$}$};
        \node(3) at (7.5,0){};
        \node(w-2) at (-5,-2.5){$\rectangled{}$};
        \node(w-1) at (-2.5,-2.5){$\rectangled{}$};
        \node(w0) at (0,-2.5){$\rectangled{}$};
        \node(w1) at (2.5,-2.5){$\rectangled{}$};
        \node(w2) at (5,-2.5){$\rectangled{}$};
        \draw (-3) -- (-2);
        \draw (-2) -- (-1);
        \draw (-1) -- (0);
        \draw (0) -- (1);
        \draw (1) -- (2);
        \draw (2) -- (3);
        
        \draw (w-2) --(-2);
        \draw (w-1) --(-1);
        \draw (w0) --(0);
        \draw (w1) --(1);
        \draw (w2) --(2);

        \draw[->,red,thick] (2.0,-2) -- node [left] {{\small $I_{i}$}}(2.0, 0.8);
        \draw[dotted,red,thick] (2.0, 0.8) -- node [below] {{\small ~ $Y_{j}$}}(-0.5, 0.8);
        \draw[dotted,red,thick] (-0.5, 0.8) -- node [below] {{\small ~ $Y_{k}$}}(-3.0, 0.8);
        \draw[->,red,thick] (-3.0, 0.8) -- node [below] {{\small ~ $Y_{l}$}}(-5.5, 0.8);
        \draw[<-,red,thick] (-5.5,-2) -- node [left] {{\small $J_{m}$}}(-5.5, 0.8);

        \draw[->,blue,thick] (-3.0,-2) -- node [left] {{\small $I_{l}$}}(-3.0, 1.2);
        \draw[dotted,blue,thick] (-3.0, 1.2) -- node [above] {{\small $X_{k}$}}(-0.5, 1.2);
        \draw[dotted,blue,thick] (-0.5, 1.2) -- node [above] {{\small $X_{j}$}}(2.0, 1.2);
        \draw[->,blue,thick] (2.0, 1.2) -- node [above] {{\small $X_{i}$}}(4.5, 1.2);
        \draw[->,blue,thick] (4.5,1.2) -- node [left] {{\small $J_{i}$}}(4.5, -2);

        \draw[<-,thick] (2.0, 1.2) -- (2.0, 0.8);
        \draw[->,thick] (-3.0, 1.2) -- (-3.0, 0.8);
        \end{tikzpicture}
    };

    \end{tikzpicture}
    \caption{The OPE of two operators represented in \cref{fig:pathsquiverexample2} generates three singular terms as in the RHS of \eqref{GGG}, each of which can be represented graphically as in this figure.
    As is evident from this example, the schematic structures OPEs between path operators are well represented by paths on the quiver.}
    \label{fig:pathsquiverexample2}
\end{figure}

\subsection{Match of central charges with \texorpdfstring{$\mathcal{W}$-algebras}{W-algebras}}

Let us comment on some numerology
before closing this section.

For the case \eqref{rho=111},
we claim that the central charge \eqref{eqn: central charge} equals the central charge of $\CW^{-N+1}(\mathfrak{gl}(N),f)$, where $f$ is a nilpotent element of $\mathfrak{gl}(N)$ of type $((N-1)^{w_{N-1}}\cdots i^{w_i}\cdots 1^{w_1})$. 

Here we use the decomposition $\CW^{-N+1}(\mathfrak{gl}(N),f)\cong \CW^{-N+1}(\mathfrak{sl}(N),f)\otimes \widehat{\mathfrak{gl}}(1)_1$, and we take the energy momentum operator of $\CW^{-N+1}(\mathfrak{gl}(N),f)$ to be the sum of energy momentum operators of $\CW^{-N+1}(\mathfrak{sl}(N),f)$ and $\widehat{\mathfrak{gl}}(1)_1$ respectively. 
The energy momentum operator of $\CW^{k}(\mathfrak{g},f)$ for a simple Lie algebra $\mathfrak{g}$ of level $k$ and a nilpotent element $f$ is given in \cite[(2.5)]{kac2003quantum}, whose central charge is \cite[(2.6)]{kac2003quantum} (see also \cite{kac2004quantum})
\begin{align}\label{eqn: central charge in Kac-Roan-Wakimoto}
    c(\mathfrak g,x,f,k)=\frac{k\dim\mathfrak{g}}{k+\mathsf{h}^\vee_{\mathfrak{g}}}-12k(x|x)-\sum_{\alpha\in S_+}(12m_{\alpha}^2-12m_{\alpha}+2)-\frac{1}{2}\dim\mathfrak{g}_{1/2}\;,
\end{align}
Here $\{x, e, f\}$ represents the $\mathfrak{sl}_2$-triple with diagonalizable $x$ (so that $[e,f]=x,[x,e]=e,[x,f]=-f$),
and $(\cdot | \cdot )$ is a non-degenerate
$\mathfrak{g}$-invariant symmetric bilinear form. The adjoint action of $x$ decomposes the Lie algebra into the weight spaces
$\mathfrak{g} = \oplus_{j\in \frac{1}{2}\mathbb{Z}} \, \mathfrak{g}_{j}$. For a root $\alpha \in \mathfrak{g}_{j}$ we define $m_{\alpha}=j$, and $S_{+}$ denotes the set of positive roots.

Although the two expressions \eqref{eqn: central charge for linear quivers} and \eqref{eqn: central charge in Kac-Roan-Wakimoto} look quite different, we verify in \cref{sec:c_W} that
\begin{align}\label{cc agree}
    c_{ \CVsdel \otimes \mathcal{V}_{\bc} } =1+ c(\mathfrak{sl}(N),x,f,-N+1)\;,
\end{align}
where $f$ is a nilpotent element of $\mathfrak{sl}(N)$ of type $((N-1)^{w_{N-1}}\cdots i^{w_i}\cdots 1^{w_1})$.

The match of the central charge discussed here is not a coincidence.
Indeed, we will see in the examples explored in later sections that there are homomorphisms from $\CW^{-N+1}(\mathfrak{gl}(N),f)$ into our H-twisted VOA $\CVdel$, such that the stress-energy operator of the former is mapped into the stress-energy operator $T$ described above.

\section{Higgs branch and associated varieties} 
\label{sec:variety}

In the previous section, we pointed out that the H-twisted VOA $\CV$ can be regarded as the chiralization of the Higgs branch. One can then ask if there is a way to extract the latter from the former. It turns out that the answer is positive, and 
we can ``de-chiralize'' the VOA by considering Zhu's $C_2$ algebra, from which we can extract the geometry of the associated variety, as we will explain in this section.

\subsection{Zhu's algebra and associated varieties}

We introduce some standard concepts in the literature on vertex algebras (see e.g.\ \cite{arakawa2017introduction}).

The canonical filtration\footnote{Also known as the Li's filtration.} $F^\bullet$ of a vertex algebra $\CV$ is defined by 
\beq\label{eq:canonical filtration}
\begin{split}
    \CV &=F^{0}\CV\supset F^{1}\CV\supset\cdots\supset F^{p}\CV\supset F^{p+1}\CV\supset\cdots, 
    \\
    F^{p}\CV &:=\mathrm{Span}\left\{\mathcal{O}^{1}_{(-n_1-1)}\cdots \mathcal{O}^{k}_{(-n_k-1)}| 0\rangle\Biggr|\, 
    \mathcal{O}^1, \dots , \mathcal{O}^k \in \CV,  \sum_{i=1}^{k}n_i\geq p\in\BZ_{\geq0}\right\}\;.  
\end{split}
\eeq
It is known that the associated graded vector space
\beq
    \mathrm{gr}_{F}\CV:=\bigoplus_{p=0}\mathrm{gr}_{F}^{p}\CV \;,
    \quad 
    \mathrm{gr}_{F}^{p}\CV:= F^{p}\CV/F^{p+1}\CV \;,
\eeq  
admits a natural vertex Poisson algebra structure \cite[Proposition 3.14]{arakawa2017introduction}.
Note that $\mathrm{gr}_{F}\CV$ is isomorphic to $\cV$ as vector spaces, while the vertex Poisson algebra structure of the former replaces the vertex operator algebra structure of the latter.
In this sense, the former can be regarded as a ``semiclassical analogue'' of the latter.

Zhu's $C_2$-algebra of $\CV$ is defined to be the zeroth graded algebra 
\beq
\Zhu({\CV}):= F^{0}\CV/F^{1}\CV \;.
\eeq 
It was proven that $\Zhu({\CV})$ is a Poisson algebra \cite[Proposition 3.15]{arakawa2017introduction}. 
We denote the image of $\mathcal{O}\in \CV$ under the surjection $\CV\twoheadrightarrow  \Zhu({\CV})$ by $\overline{\mathcal{O}}$.
In physics language, Zhu's algebra is basically generated by primary operators with respect to the global $\SL(2, \mathbb{R})$ conformal group, except there are extra subtleties when we have null states
relating primaries to descendants; one can also think of Zhu's algebra as describing the dimensional reduction of the VOA \cite{Dedushenko:2019mzv,Ashwinkumar:2023zbu}.

We say that $\CV$ is finitely strongly generated by a set of elements $\mathcal{O}^1, \mathcal{O}^2,\cdots, \mathcal{O}^n\in \CV$ if their images $\overline{\mathcal{O}^1}, \overline{\mathcal{O}^2},\cdots,\overline{\mathcal{O}^n}\in \Zhu({\CV})$ generate $\Zhu({\CV})$ as a $\mathbb C$-algebra.

We have defined Zhu's $C_2$ algebra $\Zhu({\CV})$ as a reduction of 
$\mathrm{gr}_{F}\CV$. We can conversely ask if we can recover $\mathrm{gr}_{F}\CV$ from $\CR({\CV})$. 
Physically, this is natural since all we need to do is to consider the descendant operators.
More mathematically, this can be answered by introducing the concept of the $\infty$-jet
${\jet} \CR({\CV})$, which admits a natural vertex Poisson algebra structure \cite[Theorem 3.9]{arakawa2017introduction}.
According to \cite[Theorem 3.17]{arakawa2017introduction}, there exists a natural surjective map
\beq
\jet\CR({\CV})\twoheadrightarrow \mathrm{gr}_{F}\CV \;.
\eeq
In particular, if $\CV$ is finitely strongly generated by a set of elements $\mathcal{O}^1, \mathcal{O}^2,\cdots, \mathcal{O}^n\in \CV$, then $\CV$ is spanned by elements
\begin{align*}
\mathcal{O}^{i_1}_{(-n_1)}\cdots \mathcal{O}^{i_k}_{(-n_k)}|0\rangle \;, \quad 
\textrm{ with } k\ge 0, n_i\ge 1 \;.
\end{align*}
In this sense, we do not lose much information by 
replacing the VOA by its $C_2$ algebra.

The associated variety $X_{\CV}$ is defined as the geometry corresponding to the 
reduced part of the $C_2$ algebra, i.e.\ geometry corresponding to the non-nilpotent part of the algebra \cite{MR2875849}:
\begin{align}
(\mathcal{R}({\CV}))_{\rm red} =  \mathbb{C}[X_{\CV}] \;.
\end{align}

Summarizing, given a vertex algebra $\CV$ we can extract the associated variety $X_{\CV}$ by a sequence of surjective maps
\begin{align}
    \label{V_to_M}
    \CV \twoheadrightarrow \Zhu(\CV) \twoheadrightarrow (\Zhu(\CV))_{\rm red} = \mathbb{C}[X_{\CV}] \;.
\end{align}

\subsection{Conjecture}
\label{subsec:conjectures}

We can now formulate the following conjecture:
\begin{conjecture}\label{conj:associated variety = Higgs branch}
When $\CVdel$ is the H-twisted boundary vertex algebra for a type A linear quiver gauge theory specified by $(Q, \mathbf{v}, \mathbf{w})$, then its associated variety $X_{\CVdel}$ can be identified with the affine quiver variety ${\CM}^0(Q, \mathbf{v}, \mathbf{w})$ of the theory.
\end{conjecture}

The vertex algebra is called quasi-lisse when the associated variety has finitely many symplectic leaves. The conjecture implies that the H-twisted VOA is quasi-lisse.

Our conjecture is reminiscent of, and is partly inspired by, a similar discussion for VOAs associated with 4d $\mathcal{N}=2$ theories \cite{Beem:2017ooy}. As we will discuss in \cref{sec:4d-VOA}, the precise relation between VOAs associated with 3d $\mathcal{N}=4$ theories and those with 4d $\mathcal{N}=2$ theories is subtle, and at this point we do not expect that the conjecture in \cite{Beem:2017ooy} 
immediately implies our conjecture (or vice versa).

In our discussion to this point we have obtained the Higgs branch from the VOA. We can try to go in the other direction, by ``uplifting'' Higgs branch operators to BRST-closed VOA operators. In this regard, we can formulate the following corollary to \Cref{conj:associated variety = Higgs branch}:
\begin{cor}
Let $Q$ be a balanced linear quiver. We can choose a set of elements $\{\mathcal{O}^i\}$ of $\CVdel$ such that their images under the natural map \eqref{V_to_M} generate $\mathbb{C}[X_{\CVdel}]=\mathbb C[{\mathcal{M}}^0]$.
Moreover, $\{\mathcal{O}^i\}$ can be taken as a subset of strong generators of the VOA.
\end{cor}

Notice that we have here stated ``a subset of a strong generator'' instead of ``a strong generator''.
This is because the associated variety is defined from the reduced part of Zhu's algebra,
and hence we will in general lose the information of some generators by the time we arrive at the Higgs branch. 
In general, we do expect some operators in the non-reduced part. For example, 
the stress-energy operator is BRST-closed and defines a non-trivial operator in the VOA;
it is expected to be absent after passing to associated variety (cf.\ \cite{Beem:2017ooy}).

\section{Vertex algebra for \texorpdfstring{$T^{[2,1^{n-2}]}_{[1^n]}[\mathrm{SU}(n)]$}{TSU(n)}}\label{sec:[21..1]}

In the previous sections we discussed boundary VOAs of 3d $\mathcal{N}=4$ theories
associated with linear balanced A-type quivers. In this section, we specialize further to a family of such theories
whose quivers are given in Fig.~\ref{fig:quiver for T^[2,1^{n-2}]_[1^n][SU(n)] theory}.
These theories are $T^{[2,1^{n-2}]}_{[1^n]}[\SU(n)]$ theories
in the language of section~\ref{sec:T_rho_sigma}.

\begin{figure}[ht!]
    \centering
    \begin{tikzpicture}[scale=.75, every node/.style={scale=0.8}]    
        \draw[] (-6,0) circle (.75cm);
        \draw[] (-5.25,0) -- (-4.25,0);
        \draw[] (-3.5,0) circle (.75cm);
        \draw[] (-2.75,0) -- (-1.75,0);
        \draw[line width=1pt,loosely dotted] (-1.75,0) -- (-.25,0);
        \draw[] (-.25,0) -- (.75,0);
        \draw[] (1.5,0) circle (.75);   
        \draw[] (2.25,0) -- (3.25,0);
        \draw[] (4,0) circle (.75);       
        \draw[] (4,-.75) -- (4,-2);
        \draw[] (1.5,-.75) -- (1.5,-2);       
        \draw[] (3.35,-2) rectangle (4.65,-3.3);
        \draw[] (0.85,-2) rectangle (2.15,-3.3);
        \node at (-6,0) {$1$};
        \node at (-3.5,0) {$2$};
        \node at (1.5,0) {$n-2$};
        \node at (1.5,-2.65) {$1$};
        \node at (4,0) {$n-2$};
        \node at (4,-2.65) {$n-2$};
        \node at (-6,1) {$(n-1)$};
        \node at (-3.5,1) {$(n-2)$};
        \node at (1.5,1) {$(2)$};
        \node at (4,1) {$(1)$};
    \end{tikzpicture}
    \caption{Quiver diagram for a $T^{[2,1^{n-2}]}_{[1^n]}[\mathrm{SU}(n)]$ theory. This quiver has two framing nodes (represented by boxes). The labeling of the quiver nodes is also shown as numbers in parentheses.}
    \label{fig:quiver for T^[2,1^{n-2}]_[1^n][SU(n)] theory}
\end{figure}
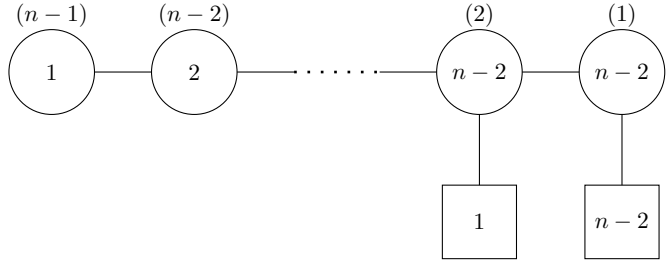

One of the important properties of the $T^{[2,1^{n-2}]}_{[1^n]}[\SU(n)]$ theories is that they have balanced quivers, as we have seen in \eqref{sec:balanced_quivers}.
Since the anomaly cancellation condition \eqref{balancing_condition} is satisfied for these examples,
the simplification for the associated H-twisted VOAs is that 
we do not need to incorporate extra Fermi multiplets to 
cancel $\SU$ anomalies. (We still need to cancel $\U(1)$ anomalies, however.)

In this section, we discuss a set of BRST-closed operators. As a byproduct we afterwards propose an 
explicit vertex algebra homomorphism into the H-twisted boundary VOA $\CV(T^{[2,1^{n-2}]}_{[1^n]}[\mathrm{SU}(n)])$ from another known VOA in the literature: the 
affine $\mathcal{W}$-algebra $\CW^{-n+1}(\mathfrak{gl}_n , f_{\rm min})$, where $f_{\rm min}$ refers to $[2, 1^{n-2}]$.

\subsection{BRST reduction}
\label{BRST-cohomology-for-A-twisted-Trhosigmasimple}

\paragraph{Ingredients:} Following the recipe, we assign $\U(1)$ Fermi multiplets at each gauge node $i$
\beq
    (\chi_i,\ps_i)_{i=1,\cdots,n-1}\;, &&\textrm{for a gauge vertex } i~,
\eeq
symplectic bosons corresponding to scalar fields from bifundamental $\CN=4$ hypermultiplets
\beq
    (X_i,Y_i)_{i=1,2,\cdots,n-2}\;, &&\textrm{for a gauge vertex } {i}~,
\eeq
and symplectic bosons for scalar fields from fundamental $\CN=4$ hypermultiplets in flavor
\beq
    (I_i,J_i)_{i=1,2}\;, &&\textrm{for gauge-flavor nodes } {i}~,
\eeq
where $(X_i,I_i)$ and $(Y_i,J_i)$ are incoming to and outgoing from the $i$-th vertex.

\paragraph{BRST currents:} We can take the BRST currents defined in equation \eqref{eq:chiral extended moment map}. 

\paragraph{BRST cohomology}

We choose the bosonic generators for the BRST cohomology to be:\footnote{While not needed for the rest of the paper, let us comment on fermionic generators. It turns out that some fermionic generators are also represented by a path on the quiver diagram. 
This representation reproduces the examples in \cite{Costello:2018fnz}.
We can generalize fermionic BRST-closed operators charged under $\U(1)_k$ as in \cite{Costello:2018fnz} as
\beq
\begin{split}
    \cO_{0}&=\ps_{0}\e^{ij}\e_{ab}I^{a}_{1,i}I^{b}_{1,j},\\
    \cO_{1}&=\ps_{1}\e^{ab}J_{2,a}J_{2,b},\\
    \cO_{2}&=\ps_{2}\e^{ab}Y_{2,a}J_{2,b},\\
    \cO_{3}&=\ps_{3}X_{2}^{a}J_{2,a},
    \label{eq:fermionic BRST}
\end{split}
\eeq
Operators in \eqref{eq:fermionic BRST} are charged under the Coulomb branch symmetry.}
\begin{itemize}
    \item[(1)] Stress-energy operator $T=T_{\beta\gamma}+T_{\psi\chi}+T_{\bc}$.
    
    \item[(2)] $\U(1)$ current: $J_F=\sum_{i=0}^{n-1}\chi_{i}\ps_{i} $~.

    \item[(3)] Extension of Kac-Moody currents: $(G^{(1,1)})^i_j, (G^{(1,2)})_j, (G^{(2,1)})^i$, where $1\le i,j\le n-2$
    \begin{gather}
    \begin{split}\label{eqdef:genericBRSTbgen}
        &(G^{(1,1)})^{i}_{\ph{j}j}=(J_{1})^{i}_{\ph{k}k}(I_1)^{k}_{\ph{j}j}+\chi_{0}\ps_{0}\d^{i}_{\ph{a}j}~,\\
        &(G^{(1,2)})^1_{\ph{j}j}=(J_{1})^1_{\ph{k}k}(X_{1})^k_{\ph{i}i}(I_{2})^i_{\ph{j}j}~, \\
        &(G^{(2,1)})^i_{\ph{1}1}=(J_{2})^i_{\ph{j}j}(Y_{1})^j_{\ph{k}k}(I_{1})^k_{\ph{1}1} ~,
    \end{split}
    \end{gather}
     The currents listed in equations \eqref{eqdef:genericBRSTbgen} are associated to particularly simple paths as in \cref{fig:path for KM} on the quiver diagram depicted in Fig.~\ref{fig:quiver for T^[2,1^{n-2}]_[1^n][SU(n)] theory}.

We claim that it is sufficient to consider
only these three types of operators, out of many operators of the form \eqref{J_Phi_I}. 
This is proven in \cref{sec:reduction}.

\end{itemize}

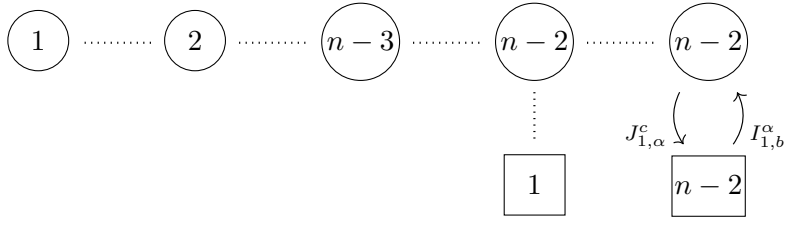
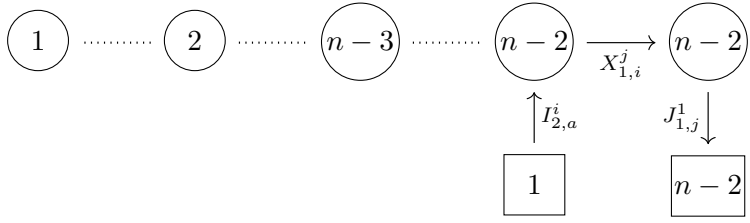
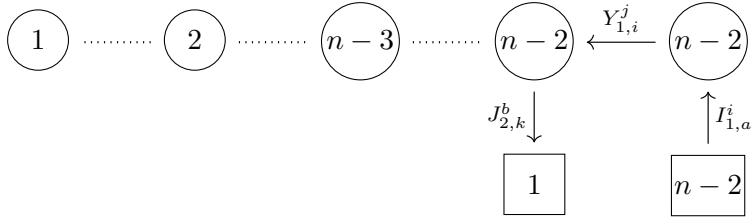
\begin{figure}[ht!]
\centering
\begin{subfigure}[a]{\textwidth}
\centering
\[\begin{tikzcd}
{\circled{$1$}}&{\circled{$2$}}& {\circled{$n-3$}} & {\circled{$n-2$}} & {\circled{$n-2$}} 
\\
{}&{}&{} & {\rectangled{$1$}} & {\rectangled{$n-2$}}
    \arrow["I_{1,b}^{\a}"', bend right, from=2-5, to=1-5] 
    \arrow["{J_{1,\a}^{c}}"', bend right, from=1-5, to=2-5] 
    \arrow[dotted, thick, no head, from=1-4, to=1-5]
    \arrow[dotted, thick, no head, from=1-3, to=1-4]
    \arrow[dotted, thick, no head, from=1-4, to=2-4]
    \arrow[dotted, thick, no head, from=1-2, to=1-1]
    \arrow[dotted, thick, no head, from=1-2, to=1-3]
\end{tikzcd}\]
\caption{The path for $G^{(1,1)}$}
\end{subfigure}

\begin{subfigure}[a]{\textwidth}
\centering
   \[\begin{tikzcd}
   {\circled{$1$}}&{\circled{$2$}}& {\circled{$n-3$}} & {\circled{$n-2$}} & {\circled{$n-2$}}
   \\
   {}&{}&{} & {\rectangled{$1$}} & {\rectangled{$n-2$}}
    \arrow["{I^{i}_{2,a}}"', from=2-4, to=1-4]
    \arrow["{X_{1,i}^{j}}"', from=1-4, to=1-5]
     \arrow["{J_{1,j}^{1}}"', from=1-5, to=2-5]
    \arrow[dotted, thick, no head, from=1-2, to=1-1]
    \arrow[dotted, thick, no head, from=1-2, to=1-3]
    \arrow[dotted, thick, no head, from=1-3, to=1-4]
   \end{tikzcd}\]
\caption{The path of $G^{(1,2)}$}
\end{subfigure}

\begin{subfigure}[a]{\textwidth}
\centering
    \[\begin{tikzcd}
        {\circled{$1$}}&{\circled{$2$}}& {\circled{$n-3$}} & {\circled{$n-2$}} & {\circled{$n-2$}}
        \\
        {}&{}&{} & {\rectangled{$1$}} & {\rectangled{$n-2$}}
        	\arrow["{I^{i}_{1,a}}"', from=2-5, to=1-5]
        	\arrow["{Y_{1,i}^{j}}"', from=1-5, to=1-4]
        	\arrow["{J_{2,k}^{b}}"', from=1-4, to=2-4]
            \arrow[dotted, thick, no head, from=1-2, to=1-1]
           \arrow[dotted, thick, no head, from=1-2, to=1-3]
           \arrow[dotted, thick, no head, from=1-3, to=1-4]
    \end{tikzcd}\]
    \caption{The path of $G^{(2,1)}$}
 \end{subfigure}
 
    \caption{Simple paths on the quiver diagram from Fig.~\ref{fig:quiver for T^[2,1^{n-2}]_[1^n][SU(n)] theory}, for BRST closed operators.}
    \label{fig:path for KM}
\end{figure}

\paragraph{Moment Map} 

\beq\label{eqdef:genericmomentmaprelations}
(\widetilde{\m}_{\mathrm{ch},i})^{a}_{\ph{a}b}=X^{a}_{i}Y_{i,b}-Y^{a}_{i-1}X_{i-1,b}+\d^{a}_{\ph{a}b}(\chi_{i-1}\ps_{i-1}-\chi_{i} \ps_{i})+I^{a}_{i}J_{i,b}\d_{(i<3)} \;.
\eeq

\subsection{Map from \texorpdfstring{$\CW$}{W}-algebra to boundary VOA}

We now claim there exists an injection as a VOA from 
an affine $W$-algebra $\CW^{-n+1}(\mathfrak{gl}_{n},f_{\mathrm{min}})$
into the H-twisted VOA for the $T^{[2,1^{n-2}]}_{[1^n]}[\mathrm{SU}(n)]$ theory, where $f_{\rm min}$ refers to $[2, 1^{n-2}]$.

The strong generators and their OPEs for $\CW^{-n+1}(\mathfrak{gl}_{n},[2, 1^{n-2}])$ were obtained by Ueda \cite{ueda2022example}.\footnote{We have $[2, 1^{n-2}]$, whereas \cite{ueda2022example} has $[2^n, 1^{m-n}]$. So we need to send $n_{\rm there}\to 1, m_{\rm there}\to n-1$, and hence $m_{\rm there}-n_{\rm there}\to n-2$, in the result from \cite{ueda2022example}. We take the level to be $k=-n+1$, which sets the parameter $\alpha_1, \alpha_2$ in \cite{ueda2022example} as $\alpha_1=-n+2, \alpha_2=0$.}
The strong generators are given by
\begin{align}
\begin{split}
    \left\{  W_{i,j}^{(1)} | 1 \le i \le n-2 , 1\le j \le n-1 \textrm{ or } i=j= n-1 \right \} 
    \cup \left\{ W_{i=n-1,j}^{(2)} | 1 \le j \le n-1 \right\} \;,
\end{split}
\end{align}
and their OPEs are given by:
\begin{align}\label{eq: W([2,1^{n-2}]) OPE}
    (W^{(1)}_{p,q})_{(0)} W_{i,j}^{(2)} &= -\delta_{p,j} W^{(2)}_{i,q} + \delta_{i,q} \delta(p>2) W^{(2)}_{p,j} +\delta(p\leq 2, q>2) (W^{(1)}_{p,j})_{(-1)} W^{(1)}_{i,q} \nonumber\\ 
    & ~~~~ - \delta_{i,q}\delta(p\leq 2)  \sum_{w\leq 2} (W^{(1)}_{w,j})_{(-1)} W^{(1)}_{p,w} \;,\nonumber\\
    (W^{(1)}_{p,q})_{(1)} W_{i,j}^{(2)} &= \delta_{p,j} \alpha_1 \delta(q>2) W^{(1)}_{i,q} - \delta_{i,q} \delta(p\leq 2) \alpha_1 W^{(1)}_{p,j} \nonumber \\
    &~ + \delta_{p,q} \delta(j>2) W^{(1)}_{i,j} - \delta_{i,j} \delta(q\leq 2) W^{(1)}_{p,q} + \delta_{p,j} \delta_{q,i} \sum_{w\leq 2} W^{(1)}_{w,w}\;, \nonumber \\
    (W^{(1)}_{p,q})_{(2)} W_{i,j}^{(2)} &= -\delta_{q,i} \delta_{p,j} \alpha_1^2- \delta_{p,q} \delta_{i,j} \delta(p,q\leq 2)\alpha_1 \;, \nonumber \\
    (W^{(1)}_{p,q})_{(s)} W_{i,j}^{(2)} &=0~ \quad \forall s>2 ~.
\end{align}
Here $(\mathcal{O}_{1})_{(n)} (\mathcal{O}_{2})$ for $n=0, 1, \dots$ are defined from the OPE of the operators as
\begin{align}
\mathcal{O}_1(z) \mathcal{O}_2(w) \sim \sum_{n} \frac{1}{(z-w)^{n+1}} ((\mathcal{O}_{1})_{(n)} (\mathcal{O}_2)) (w)\;,
\end{align}
with $\alpha_1:=-n+2$.

We propose the following map from 
an affine $W$-algebra $\CW^{-n+1}(\mathfrak{gl}_{n},f_{\mathrm{min}})$
into the bosonic part of the H-twisted VOA for the $T^{[2,1^{n-2}]}_{[1^n]}[\mathrm{SU}(n)]$ theory:
\beq
&W^{(1)}_{i,j}=(G^{(1,1)})^i_j,\; 1\le i,j\le n-2\;,\\
&W^{(1)}_{i,n-1}=(G^{(1,2)})^i,\; 1\le i\le n-2\;,\\
&W^{(1)}_{n-1,n-1}=J_F-\sum_{i=1}^{n-2}(G^{(1,1)})^i_i\;,\\
&W^{(2)}_{n-1,l}=(G^{(2,1)})_l,\; 1\le l\le n-2\;,\\
&W^{(2)}_{n-1,n-1}=:\chi_{(0)}\ps_{(0)}W^{(1)}_{n-1,n-1}:+::X^c_{1,b}Y^a_{1,c}::J^1_{2,a}I^b_{2,1}:: + \Tr{(\mathsf b_1\partial \mathsf c_1)}
\nn
\\
&-(:{\partial J}_{1,a}^{i} I_{1,i}^{a}:-:{\partial X}^b_{1,a}Y^a_{1,b}:-(n-2):J^1_{1,a}{\partial I}^a_{2,1}:+\partial \chi_{0}\ps_{0}-\chi_{0}\partial \ps_{0})  \;.
\label{eq: VOA map}
\eeq

In principle it is a matter of computations to
verify that our proposed map is a vertex algebra homomorphism. Some OPEs are 
easier to verify. For example, 
we can check OPEs between $W^{(1)}$'s as follows:
\begin{align}\label{W1W1}
W^{(1)}_{p,q}(z)W^{(1)}_{i,j}(w)&\sim  \frac{(2-n)\delta_{p,j}\delta_{q,i}+\delta_{p,q}\delta_{i,j}(1+\delta(p,i>n-2))}{(z-w)^2}+\frac{\delta_{q,i}W^{(1)}_{p,j}(w)-\delta_{p,j}W^{(1)}_{i,q}(w)}{z-w} \;.
\end{align}
We can also check  some of the OPEs by
remembering the discussions in terms of paths in \cref{sec:BRST_operators}. In general, however, the OPEs are complicated and require repeated use of the moment map relations. We will discuss the special case of $n=4$ in detail in the next section, where we check explicitly all the relations of the $\mathcal{W}$-algebra (and hence prove that we have a vertex algebra homomorphism).

Later we will show that the map \eqref{eq: VOA map} is indeed a vertex algebra map for all $n\ge 3$, and in fact it is injective, see \cref{sec:quiverred211111}. But there is a caveat: we need to modify the target of the map \eqref{eq: VOA map}. Namely what we actually find is a vertex algebra map from $\CW^{-n+1}(\mathfrak{gl}_n,f_{\min})$ to a certain vertex algebra denoted by $\mathsf D^{\mathrm{ch}}(\widetilde\CM)$, which can be regarded as a ``chiralization'' of the extended Higgs branch $\widetilde\CM$. $\mathsf D^{\mathrm{ch}}(\widetilde\CM)$ is closely related to the H-twisted boundary vertex algebra of $T^{[2,1^{n-2}]}_{[1^n]}[\mathrm{SU}(n)]$ theory, in the sense that there exists a vertex algebra map $\CV_{\partial,0}\to \mathsf D^{\mathrm{ch}}(\widetilde\CM)$, where $\CV_{\partial,0}$ is the degree zero component in the decomposition \eqref{lattice decomposition of boundary VOA}. When $n=4$, we can show that such map is injective, thus providing another way to show that \eqref{eq: VOA map} is a vertex algebra map.

Before proceeding to the computation of OPE, we note that the map \eqref{eq: VOA map} is a vertex algebra homomorphism (instead of a VOA homomorphism) and have not specified the stress-energy operator (i.e. conformal vector). It will be shown in the later section that our stress-energy operator $T$ in \eqref{def:generalTtensor} is in the image of $\CW^{-n+1}(\mathfrak{gl}_n,f_{\min})$, moreover it equals to $T_{\mathrm{KRW}}+\frac{1}{2n}J_F^2$, where $T_{\mathrm{KRW}}$ is the stress-energy operator of $\CW^{-n+1}(\mathfrak{sl}_n,f_{\min})$ defined by Kac-Roan-Wakimoto in \cite[(2.5)]{kac2003quantum}. See \cref{sec:quiverred211111} for details.

\subsection{Check of OPE Relations for \texorpdfstring{$n=4$}{n=4}}\label{sec:[211]}

In this section we work out in more detail the $T_{[1^4]}^{[2,1^2]}[\SU(4)]$ theory,
which is the $n=4$ case of the $T^{[2,1^{n-2}]}_{[1^n]}[\mathrm{SU}(n)]$ theory,
whose quiver diagram was given previously in Fig.~\ref{fig:quiverexample2011}. 

The BRST invariant currents, fermionic operators, moment map equations, and the map between generators from the $\CW$-algebra are obtained from the formulae in section \ref{BRST-cohomology-for-A-twisted-Trhosigmasimple} by taking $n=4$.

\paragraph{BRST invariant currents}
 Extension of Ka$\check{c}$-Moody currents: $G^{(1,1)}$ $G^{(1,2)}$ $G^{(2,1)}$
\beq \label{three_Gs}
(G^{(1,1)})^{i}_{\ph{j}j}=J_{1}^{i}I_{1,j}+\chi_{0}\ps_{0}\d^{i}_{\ph{a}j} \;,
&&
G^{(1,2)}=J_{1}X_{1}I_{2} \;,
&&
G^{(2,1)}=J_{2}Y_{1}I_{1} \;,
\eeq

\paragraph{Moment map equation}

\begin{equation}
(\widetilde{\m}_{\mathrm{ch},i})^{a}_{\ph{a}b}=X^{a}_{i}Y_{i,b}-Y^{a}_{i-1}X_{i-1,b}+\d^{a}_{\ph{a}b}(\chi_{i-1}\ps_{i-1}-\chi_{i} \ps_{i})+I^{a}_{i}J_{i,b} \delta_{i<3} \;.
\end{equation}
In the discussion below, we will use short-hand notation $\m_i$ for the extended chiral moment map $\widetilde{\m}_{\mathrm{ch},i}$.

Using the map \eqref{eq: VOA map}, we will confirm in the following that the OPE of the strong generators follows the relations provided in \cite{ueda2022example}. 

\paragraph{\texorpdfstring{$W^{(1)} W^{(1)}$}{WW} OPE}
We checked their OPEs as in \eqref{W1W1} (with $n=4$), 
and we can also check that we have no non-trivial OPEs  involving $W^{(1)}_{3,(q,j)\leq2}$.

\paragraph{\texorpdfstring{$W^{(1)}_{1,3} W^{(2)}_{3,2}$}{WW} OPE}
Using Mathematica\footnote{The computations are facilitated with the help of the Mathematica package \textsf{OPEdefs.m} \cite{Thielemans:1994er}.}, we compute that
\begin{equation}\label{eqn: W1[1,3] W2[3,2] OPE}
\begin{split}
    &(W^{(1)}_{1,3})_{(0)} W^{(2)}_{3,2}+\sum_{w\leq 2} (W^{(1)}_{w,2})_{(-1)} W^{(1)}_{1,w} - (W^{(1)}_{1,2})_{(-1)} W^{(1)}_{3,3}=\\
    &::J^1_{1,b}I^a_{1,2}::I^b_{1,k}J^k_{1,a}::-:W^{(1)}_{1,2}W^{(1)}_{3,3}:+2W^{(1)}_{1,2}\chi_{0}\psi_{0}+W^{(1)}_{1,2}:J^1_{2,a}I^a_{2,1}:\\
    &+::J^1_{1,b}I^a_{1,2}::X^b_{1,c}Y^c_{1,a}:: \;.
\end{split}
\end{equation}
The moment map equation implies that 
$$
::J^1_{1,b}I^a_{1,2}::I^b_{1,k}J^k_{1,a}::+::J^1_{1,b}I^a_{1,2}::X^b_{1,c}Y^c_{1,a}::+W^{(1)}_{1,2}(\chi_{0}\psi_{0}-\chi_{1}\psi_{1})=0\;,
$$
In detail, with the explicit form of $W^{(1)}_{1,2}$,
\beq
    &::J^1_{1,b}I^a_{1,2}:[(\m_{1})^{b}_{a}-\d^{b}_{a}(\chi_{0}\ps_{0}-\chi_{1}\ps_{1})]+:J_{1,a}^{1}I^{a}_{1,2}:(\chi_{0}\psi_{0}-\chi_{1}\psi_{1})\nn
    \\
    &=::J^1_{1,b}I^a_{1,2}:(\m_{1})^{b}_{a}:=0 \;,
\eeq
where in the last step we have used the fact that $:J^1_{1,b}I^a_{1,2}:$ transforms under the $\m_1$ as an adjoint representation, and the discussion in \cref{subsec:Some BRST exact elements} implies that $::J^1_{1,b}I^a_{1,2}:(\m_{1})^{b}_{a}:$ is BRST exact. Thus the right-hand-side of \eqref{eqn: W1[1,3] W2[3,2] OPE} equals to 
\begin{align}
    W^{(1)}_{1,2}(\chi_{0}\psi_{0}+\chi_{1}\psi_{1})-:W^{(1)}_{1,2}W^{(1)}_{3,3}:+W^{(1)}_{1,2}:J^1_{2,a}I^a_{2,1}: \;,
\end{align}
and using the moment map again, we find
\begin{align}
    \underbrace{\chi_{0}\psi_{0}+\chi_{1}\psi_{1}-W^{(1)}_{3,3}+:J^1_{2,a}I^a_{2,1}:}_{\Tr{(\m_1+\m_2+\m_3)}}=0 \;,
\end{align}
therefore $(W^{(1)}_{1,3})_{(0)} W^{(2)}_{3,2}$ agrees with \cite{ueda2022example}. On the other hand, we
verified
\begin{align}
    (W^{(1)}_{1,3})_{(1)} W^{(2)}_{3,2}=2W^{(1)}_{1,2}\;, \quad (W^{(1)}_{1,3})_{(n)} W^{(2)}_{3,2}=0,\;\forall n>1\;.
\end{align}
Hence $W^{(1)}_{1,3} W^{(2)}_{3,2}$ OPE are checked in all orders.

\paragraph{\texorpdfstring{$W^{(1)}_{1,3} W^{(2)}_{3,1}$}{WW} OPE}
\bi
\item[1] $(W^{(1)}_{1,3})_{(2)} W^{(2)}_{3,1}$ OPE is easily checked. 
\item[2] $(W^{(1)}_{1,3})_{(1)} W^{(2)}_{3,1}$ OPE is checked with the combinations of the moment map equations listed below.

\item[] The moment map for $(W^{(1)}_{1,3})_{(1)} W^{(2)}_{3,1}$ : $-3\Tr(\m_{1})-2\Tr(\m_{2})-2\Tr(\m_{3})$
\beq
(W^{(1)}_{1,3})_{(1)} W^{(2)}_{3,1}
=-2W^{(1)}_{3,3} +3 W^{(1)}_{1,1}+ W^{(1)}_{2,2}-3\Tr(\m_{1})-2\Tr(\m_{2})-2\Tr(\m_{3})
\;.
\eeq
\item[3] $(W^{(1)}_{1,3})_{(0)} W^{(2)}_{3,1}$ OPE is checked by
\begin{small}
\beq
&(W^{(1)}_{1,3})_{(0)} W^{(2)}_{3,1}+\sum_{w\leq 2} (W^{(1)}_{w,1})_{(-1)} W^{(1)}_{1,w} - (W^{(1)}_{1,1})_{(-1)} W^{(1)}_{3,3}+W^{(2)}_{3,3}
\label{eqn: W1[1,3] W2[3,1] OPE}
\nn
\\
&=\underbrace{W^{(2)}_{3,3} - :W^{(1)}_{1,1}W^{(1)}_{3,3}:+:{\partial J}_{1,a}^{i} I_{1,i}^{a}:-:{\partial X}^b_{1,a}Y^a_{1,b}:-2:J^1_{2,a}{\partial I}^a_{2,1}:+\partial \chi_{0}\ps_{0}-\chi_{0}\partial \ps_{0}}_{W^{(2)}_{3,3}-\left(W^{(2)}_{3,3}-\mathrm{Tr}(\mathsf b\partial c)\right)=\mathrm{Tr}(\mathsf b\partial \mathsf c)}
\nn
\\
&+\underbrace{:J^1_{1,b}I^a_{1,1}:(\chi_{0}\ps_{0}+\chi_{1}\ps_{1})+:J^1_{1,a}I^a_{1,1}::J^1_{2,b}I^b_{2,1}:-:J^1_{1,b}I^a_{1,1}W^{(1)}_{3,3}:}_{:J^1_{1,a}I^a_{1,1}:\Tr{(\m_{1}+\m_{2}+\m_{3})}=0}
\nn
\\
&+\underbrace{:J^1_{1,a}I^a_{1,1}:(\chi_{0}\ps_{0}-\chi_{1}\ps_{1})+::J^1_{1,b}I^a_{1,1}::X^b_{1,c}Y^c_{1,a}::+::J^1_{1,b}I^a_{1,1}::I^b_{1,k}J^k_{1,a}::}_{:J^1_{1,b}I^a_{1,1}:(\m_{1})^{b}_{a}=-\mathrm{Tr}(\mathsf b\partial \mathsf c)}
\nn
\\
&=0 \;,
\eeq
where in the last step we have used the fact that $:J^1_{1,b}I^a_{1,1}:$ transforms under the $\m_1$ as an adjoint representation with anomalous level $-1$, and discussion in \cref{subsec:Some BRST exact elements} implies that $::J^1_{1,b}I^a_{1,2}:(\m_{1})^{b}_{a}:$ is BRST equivalent to $-\mathrm{Tr}(\mathsf b\partial \mathsf c)$.
\end{small}

\ei

\paragraph{\texorpdfstring{$W^{(1)}_{2,3} W^{(2)}_{3,2}$}{WW} OPE}

The computation is essentially the same as the $W^{(1)}_{1,3} W^{(2)}_{3,1}$ OPE, and we do not repeat here.

\paragraph{\texorpdfstring{$W^{(1)}_{1,3} W^{(2)}_{3,3}$}{WW} and \texorpdfstring{$W^{(1)}_{2,3} W^{(2)}_{3,3}$}{WW} OPE}

\beq
    & (W^{(1)}_{1,3})_{(0)} W^{(2)}_{3,3}+ :W^{(1)}_{1,3} (W^{(1)}_{1,1} - W^{(1)}_{3,3}):+ :W^{(1)}_{2,3}W^{(1)}_{1,2} :
    \nn
    \\
    &=:W^{(1)}_{1,3}\Tr{(\m_{1}+\m_{2}+\m_{3})}:-:J^{1}_{1,b}X^{a}_{1,c}I^{c}_{2,1}:(\m_{1})^{b}_{a}=0 \;.
    \eeq
    \beq
    & (W^{(1)}_{2,3})_{(0)} W^{(2)}_{3,3}+ :W^{(1)}_{2,3} (W^{(1)}_{2,2} - W^{(1)}_{3,3}):+ :W^{(1)}_{1,3}W^{(1)}_{2,1} :
    \nn
    \\
    &=:W^{(1)}_{2,3}\Tr{(\m_{1}+\m_{2}+\m_{3})}:-:J^{2}_{1,b}X^{a}_{1,c}I^{c}_{2,1}:(\m_{1})^{b}_{a}=0 \;.
\eeq

\section{Vertex algebra for \texorpdfstring
{$T^{[n,n]}_{[1^{2n}]}[\SU(2n)]$}
{T[n,n][1{2n}][\SU(2n)]}}
\label{sec:[nn]}

Let us next discuss one more example of a class of  
balanced quivers, those for the $T^{[n,n]}_{[1^{2n}]}[\SU(2n)]$ theories. The corresponding quiver is presented in \cref{fig:T-shaped}. 

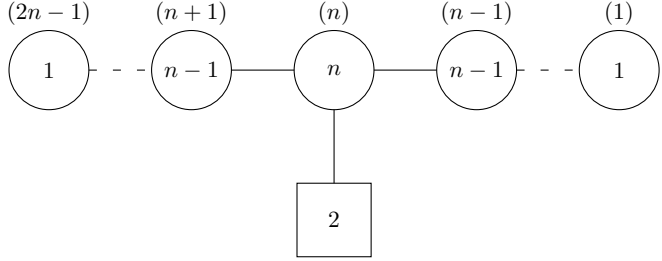
\begin{figure}[ht!]
    \centering
    \begin{tikzpicture}[scale=.75, every node/.style={scale=0.8}]  
        \draw[] (-8.5,0) circle (.7);
        \draw[loosely dashed] (-7.8, 0) -- (-6.7, 0);
        \draw[] (-6,0) circle (.7);
        \draw[] (-5.3,0) -- (-4.2,0);
        \draw[] (-3.5,0) circle (.7);
        \draw[] (-2.8,0) -- (-1.7,0);
        \draw[] (-1,0) circle (.7);    
        \draw[loosely dashed] (-0.3, 0) -- (0.8, 0);
        \draw[] (1.5,0) circle (.7);    
        \draw[] (-3.5,-.7) -- (-3.5,-2); 
        \draw[] (-4.15,-2) rectangle (-2.85,-3.3);
        \node at (-8.5,0) {$1$};
        \node at (-6,0) {$n-1$};
        \node at (-3.5,0) {$n$};
        \node at (-1,0) {$n-1$};
        \node at (1.5,0) {$1$};
        \node at (-3.5,-2.65) {$2$};
        \node at (-8.5,1) {$(2n-1)$};
        \node at (-6,1) {$(n+1)$};
        \node at (-3.5,1) {$(n)$};
        \node at (-1,1) {$(n-1)$};
        \node at (1.5,1) {$(1)$};  
    \end{tikzpicture}
    \caption{A quiver diagram for $T^{[n,n]}_{[1^{2n}]}[\SU(2n)]$}
    \label{fig:T-shaped}
\end{figure}

\subsection{General \texorpdfstring{$n$}{n}}
\paragraph{Ingredients:} Following the recipe, we assign $\U(1)$ Fermi multiplets at each gauge node $i$
\beq
    (\chi_i,\ps_i)_{i}\;, \qquad {i=1,\cdots,2n-1}~,
\eeq
Symplectic bosons:
\beq
    (X_i,Y_i)_{i=1,2,\cdots,2n-2} \;, &&(I_i,J_i)_{i=n} \;,
\eeq
where $(X_i,I_i)_i$ and $(Y_i,J_i)_{i}$ are incoming to and outgoing from the $i$-th vertex.

The moment maps are 
\beq
    ({\m}_{i})^{\a}_{\ph{\a}{\b}}&=X^{\a}_{i}Y_{i,\b}-Y^{\a}_{i-1}X_{i-1,\b}+\d^{\a}_{\ph{\a}{\b}}(\chi_{i-1}\ps_{i-1}-\chi_{i}\ps_{i})+I^{\a}_{a}J_{\b}^{a}\d_{i,n} \;, 
\eeq
where $i=1,\cdots,2n-2$, and $a=1,2$.

The BRST cohomology consists of $T$, $\{(G^{(i)})^a_{\ph{a}b}\}$, and $J_F=\sum_i \chi_{i}\ps_{i}$ as before. The extension of Kac-Moody currents is represented by $n$ independent paths on the quiver in \cref{fig:T-shaped}. All of the bosonic operators corresponding to paths on the quiver are written as below,
\beq
\begin{split}
J^{a}\underbrace{X_{n}\cdots X_{n+k_1-1}}_{k_1}\underbrace{Y_{n+k_1-1}\cdots Y_{n-k_2}}_{k_1+k_2}\underbrace{X_{n-k_2}\cdots X_{n-1}}_{k_2}I_{b} \;,
\\
J^{a}\underbrace{Y_{n-1}\cdots Y_{n-k_1}}_{k_1}\underbrace{X_{n-k_1}\cdots X_{n+k_2-1}}_{k_1+k_2}\underbrace{Y_{n+k_2-1}\cdots Y_{n}}_{k_2}I_{b}\;,
\end{split}
\eeq
where $k_i=0,\cdots,n-1$, and we omit all dummy indices that are contracted. One may choose for example the linearly independent bosonic operators to correspond to paths on the quiver that begin at the framing node, turn left at the rank-$n$ gauge node and loop around the rank-$(n+k)$ gauge node for some $0\leq k \leq 2n-1$, before returning to the framing node. All of the other bosonic operators, which correspond to the remaining possible paths, are related to these via the moment map relations. The corresponding BRST closed operators are $(G^{(1)})^a_{b}$ and $(G^{(m+1)})^a_{b}$ of the form
\beq
    (G^{(1)})^a_{b}&=J^{a}I_{b}+\frac{1}{2}\d^{a}_{b}\left(\sum_{k=0}^{n-1}\chi_{k}\ps_{k}-\sum_{k=n}^{2n-1}\chi_{k}\ps_{k}\right),
    \\
    (G^{(m)})^a_{b}&=J^{a}\underbrace{X_{n}\cdots X_{n+m-2}}_{n\ge m\geq 2}\underbrace{Y_{n+m-2}\cdots Y_{n}}_{n\ge m\geq 1}I_{b}\, + (\textrm{corr.}) \;,
\eeq
for $2\leq m\leq n$ and where (\textrm{corr.}) represent correction terms involving $(G^{(m')})^a_{b}$ with $m'<m$ and fermion bilinears. 

\subsection{\texorpdfstring
{$n=2$}
{n=2}}
\label{sec:[22]}

\subsubsection{BRST reduction}

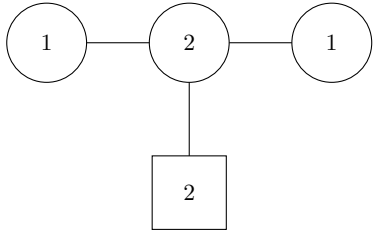
\begin{figure}[ht!]
    \centering
    \begin{tikzpicture}[scale=.75, every node/.style={scale=0.8}]        
        \draw[] (-6,0) circle (.7);
        \draw[] (-5.3,0) -- (-4.2,0);
        \draw[] (-3.5,0) circle (.7);
        \draw[] (-2.8,0) -- (-1.7,0);
        \draw[] (-1,0) circle (.7);         
        \draw[] (-3.5,-.7) -- (-3.5,-2);
        \draw[] (-4.15,-2) rectangle (-2.85,-3.3);
        \node at (-6,0) {$1$};
        \node at (-3.5,0) {$2$};
        \node at (-1,0) {$1$};
        \node at (-3.5,-2.65) {$2$};
    \end{tikzpicture}
    \caption{A quiver for a rectangular $\CW$-algebra}
    \label{fig:rectangular}
\end{figure}

Following the recipe in section \ref{subsection:pre-reduction}, we have the ingredients for boundary VOA of the H-twisted $T^{[2,2]}_{[1,1,1,1]}[\SU(4)]$ theory:
\beq
X_{1,\a},Y_{1}^{\a},X_{2}^\a,Y_{2,\a},I^{\a}_{a},J_{\a}^{a},\quad  (\a,a=1,2),
&&
\chi_{i},\ps_{i} \;,\quad (i=0,1,2,3) \;.
\eeq
with the gauge and the flavor indices $\a$ and $a$ respectively. 

\paragraph{Moment map relations}
\beq
{\m}_{1}&=X_{1,\a}Y^{\a}_{1}+(\chi_{0}\ps_{0}-\chi_{1}\ps_{1}) \;,
\\
({\m}_{2})^{\a}_{\ph{\a}{\b}}&=X^{\a}_{2}Y_{2,\b}-Y^{\a}_{1}X_{1,\b}+\d^{\a}_{\ph{\a}{\b}}(\chi_{1}\ps_{1}-\chi_{2}\ps_{2})+I^{\a}_{a}J_{\b}^{a} \;,
\\
{\m}_{3}&=-Y_{2,\a}X^{\a}_{2}+(\chi_{2}\ps_{2}-\chi_{3}\ps_{3}) \;.
\eeq
We can write down the following set of bosonic operators which are classically gauge invariant
\beq
\left\{J_{\a}^{a}I^{\a}_{b},\, J_{\a}^{a}Y^{\a}_{1}X_{1,\b}I^{\b}_{b},\, J_{\a}^{a}X^{\a}_{2}Y_{2,\b}I^{\b}_{b},\,
J_{\a}^{a}X^{\a}_{2}Y_{\b}Y^{\b}_{1}X_{1,\g}I^{\g}_{b}\right\} \;.
\eeq
Following the strategy for the general $n$, one may choose the bosonic operators corresponding to paths on the quiver that begin at the framing node and turn left at the rank-2 gauge node before returning to the starting point, namely the operators $J_{\a}^{a}I^{\a}_{b}$ and $J_{\a}^{a}X^{\a}_{2}Y_{2,\b}I^{\b}_{b}$.

\paragraph{BRST cohomology}

The bosonic generators for the BRST cohomology are:
\begin{itemize}
    \item[(1)] Stress-energy operator $T=T_{\beta\gamma}+T_{\psi\chi}+T_{bc}$.
    \item[(2)] $\U(1)$ current: $J_F=\sum_{i=0}^{3}\chi_{i}\ps_{i} $.
    \item[(3)] Extension of Kac-Moody currents: $(G^{(1)})^a_{\ph{a}b}, (G^{(2)})^a_{\ph{a}b}$. 

\end{itemize}
From all the possible paths, the extension of Kac-Moody currents are chosen as
\beq
(G^{(1)})^a_{\ph{a}b}&=J_{\a}^{a}I^{\a}_{b}+\frac{1}{2}\d^{a}_{b}(\chi_{0}\ps_{0}+\chi_{1}\ps_{1}-\chi_{2}\ps_{2}-\chi_{3}\ps_{3}) \;,
\\
(G^{(2)})^{a}_{\ph{a}b}&=-(J_{\a}^{a}X^{\a}_{2})_{(-1)}Y_{2,\b}I^{\b}_{b}+\rho (G^{(1)})^a_{\ph{a}b}-(\rho^2+\partial\rho)\delta^a_b \;,
\eeq
where $\rho=(3b_3-b_2-b_1-b_0)/4$. The free bosons $b_0,\cdots,b_3$ are introduced such that $b_i(z)b_j(w)\sim {\delta_{ij}}/{(z-w)^2}$, and $b_i$ is related to the free fermions by $b_i=\chi_i\psi_i$. 

We will come back to this example in
section \ref{[2,2]-revisited}. We will see, in particular, that
the two generators $G^{(1)}$ and $G^{(2)}$, after a suitable combination with the current $J_F$, will be mapped the strong generators of the rectangular $\mathcal{W}$-algebra 
$\CW^{-3}(\mathfrak{gl}_4,[2^2])$.

\bigskip

\noindent We remark now that the VOA of the H-twisted $T^{[2,2]}_{[1,1,1,1]}[\SU(4)]$ theory can also be obtained by the quiver reduction of the  $T_{[1,1,1,1]}^{[2,1,1]}[\SU(4)]$ theory, up to gauging of a $\U(1)$ framing, by specializing the discussion in section \ref{sec:quiverred211111}. 

\section{Localization of quiver vertex algebra}\label{sec:localization}

In \cref{sec:[21..1]} and \cref{sec:[nn]}, we have discussed the homomorphism from affine W-algebras into our VOAs. The practical computations, however,
 become rather involved as we consider more complicated quivers, and it is clearly desirable to have a systematic procedure for analyzing the VOAs associated with general linear quivers.

Fortunately, there indeed exists such a systematic procedure, which we will discuss in detail in the coming three sections. One of the motivations comes from the by-now-familiar motto that the VOA is a chiralization of the Higgs branch. On the geometry side, while the geometry of the Higgs branch is complicated in itself, one can always choose a local patch to simplify the geometry dramatically.
There is a vertex algebra counterpart of this statement,
in the form of the localizations of the quiver vertex algebras.

In this section, we sketch the $\hbar$-adic vertex algebra and its sheaf version, define the $\hbar$-adic BRST cohomology and its sheaf version, and state the main result (vanishing theorem and injectivity theorem), which will be proven in the companion paper \cite{math_draft}.

The quiver vertex algebra $\CV(Q,\mathbf v,\mathbf w)$ can be thought as a vertex algebra analog of the extended quiver variety $\widetilde\CM^0(Q,\mathbf v,\mathbf w)$ \eqref{def_M0}, which is defined as the affine quotient
\begin{align*}
    \widetilde\CM^0(Q,\mathbf v,\mathbf w):=\widetilde\mu^{-1}(0)\sslash \mathrm{GL}(\mathbf v) \;,
\end{align*}
where $\mathrm{GL}(\mathbf v)=\prod_{i\in Q_0}\mathrm{GL}(\mathbf v_i)$ is the complexified gauge group \eqref{GL_Q}, and $\widetilde\mu:T^*\mathrm{Rep}(Q,\mathbf v,\mathbf w)\times \mathbb C^{Q_0}\to \mathfrak{gl}(\mathbf v)^*$ is the extended moment map: 
\begin{align*}
    \widetilde\mu((X,Y,I,J)\in T^*\mathrm{Rep}(Q,\mathbf v,\mathbf w), h\in \mathbb C^{Q_0})=\left(\sum_{a:j\to i}X_aY_a-\sum_{b:i\to j}Y_bX_b+I_iJ_i+h_i\right)_{i\in Q_0}.
\end{align*}
Suppose that $U\subset T^*\mathrm{Rep}(Q,\mathbf v,\mathbf w)\times \mathbb C^{Q_0}$ is a $\mathrm{GL}(\mathbf v)$-invariant open subset, then there is a natural map between varieties
\begin{align*}
    (\widetilde\mu^{-1}(0)\cap U)\sslash \mathrm{GL}(\mathbf v)\to \widetilde\CM^0(Q,\mathbf v,\mathbf w) \;.
\end{align*}
A distinguished choice of such $U$ is the stable locus, which consists of those $(X,Y,I,J,z)$ such that the actions of $X,Y$ on the image of $I$ generate all gauge node vector spaces. Denote by $\widetilde\mu^{-1}(0)^{\mathrm{st}}$ the intersection between $\widetilde\mu^{-1}(0)$ and the stable locus, then the extended Nakajima quiver variety $\widetilde\CM(Q,\mathbf v,\mathbf w)$ is defined as the quotient 
\begin{align*}
    \widetilde\CM(Q,\mathbf v,\mathbf w):=\widetilde\mu^{-1}(0)^{\mathrm{st}} \sslash \mathrm{GL}(\mathbf v) \;,
\end{align*}
and it has a natural projection $\pi: \widetilde\CM(Q,\mathbf v,\mathbf w)\to \widetilde\CM^0(Q,\mathbf v,\mathbf w)$. By the geometric invariant theory, $\pi$ is a projective morphism. When the moment map $\mu$ is 
flat\footnote{In fact, the moment map $\mu$ is flat if and only if $\pi$ induces an isomorphism between global sections of the structure sheaf.}, the projection $\pi$ induces an isomorphism between global sections of structure sheaf, i.e.
\begin{align*}
    \mathbb C[\widetilde\CM^0(Q,\mathbf v,\mathbf w)]\cong \Gamma(\widetilde\CM(Q,\mathbf v,\mathbf w), \mathscr O_{\widetilde\CM(Q,\mathbf v,\mathbf w)}) \;.
\end{align*}
The variety $\widetilde\CM(Q,\mathbf v,\mathbf w)$ is smooth, and the quotient map $q:\widetilde\mu^{-1}(0)^{\mathrm{st}}\to \widetilde\CM(Q,\mathbf v,\mathbf w)$ is a principal $\mathrm{GL}(\mathbf v)$-bundle, whereas $\widetilde\CM^0(Q,\mathbf v,\mathbf w)$ is singular in general and the action of $\mathrm{GL}(\mathbf v)$ on $\widetilde\mu^{-1}(0)$ is always non-free, this makes $\widetilde\CM(Q,\mathbf v,\mathbf w)$ a better geometric object to study than $\widetilde\CM^0(Q,\mathbf v,\mathbf w)$. 

For simplicity we will abbreviate $\widetilde\CM(Q,\mathbf v,\mathbf w)$ (resp.\ $\widetilde\CM^0(Q,\mathbf v,\mathbf w)$) to $\widetilde\CM$ (resp.\ $\widetilde\CM^0$).

\bigskip Such a construction has a vertex algebra analog. To properly state it, we need to make sense of sheafification of vertex algebra. In the classical world, i.e.\ in a commutative algebra, this is captured by the notion of localization. For quantized operators, we need to be careful about in what sense we localize. For example the Weyl algebra $D(\mathbb C^2):=\mathbb C\langle x,y\rangle/([x,y]-1)$ can be localized just for the variable $x$, or just for the variable $y$. Namely, we can consider $f(x)g(x,y)$ where $f(x)$ is a rational function of $x$ and $g(x,y)$ is a polynomial in $x,y$, then the commutator $[f_1(x)g_1(x,y), f_2(x)g_2(x,y)]$ is well-defined, in this sense the Weyl algebra can be localized with respect to $x$. Similarly, it can be localized with respect to $y$. If we attempt to localize with respect to both $x$ and $y$, however, then we immediately run into a problem: the commutator $[x^{-1},y^{-1}]$ does not make sense. To cure this issue, one can introduce an $\hbar$-adic version of the Weyl algebra $D(\mathbb C^2)_\hbar:=\mathbb C[\![\hbar]\!]\langle x,y\rangle/([x,y]-\hbar)$, where we treat $\hbar$ as a formal variable, then it makes sense to define the $\star$-multiplication as
\begin{align*}
F(x,y)\star G(x,y):=
m\circ e^{\frac{\hbar}{2} \partial_x\wedge\partial_y}(F(x,y)\otimes G(x,y)) \;,
\end{align*}
where $F,G$ are arbitrary rational functions in $x,y$, and $m$ is the usual multiplication for rational functions. The RHS is treated as a power series in $\hbar$. The $\star$ product allows us to define a sheaf $\mathscr D_{\mathbb C^2,\hbar}$ of non-commutative algebras on $\mathbb C^2$ such that $D(\mathbb C^2)_\hbar=\Gamma(\mathbb C,\mathscr D_{\mathbb C^2,\hbar})$.

\bigskip The lesson we learn from the example of Weyl algebra is that localization might involve infinite sums, which can be made sense of by allowing formal power series with the aid of formal variables. For vertex algebras, a notion of $\hbar$-adic vertex algebra was introduced in \cite{li2004vertex}. The formal definition is as follows. An $\hbar$-adic vertex algebra is a tuple $(V,|0\rangle,\partial,Y(-,z))$ such that $V$ is a flat $\mathbb C[\![\hbar]\!]$-module complete in $\hbar$ topology, and $Y(-,z):V\to \mathrm{End}_{\mathbb C[\![\hbar]\!]}(V)[\![z^{\pm}]\!]$ is a $\mathbb C[\![\hbar]\!]$-linear map such that
\begin{enumerate}
    \item $a_{(n)}$ is continuous with respect to the $\hbar$-adic topology, for all $a\in V, n\in \mathbb Z$.
    \item $(V/\hbar^NV,|0\rangle,\partial,Y(-,z))$ is a vertex algebra for all $N\in \mathbb Z_{\ge 1}$.
\end{enumerate}
An example of $\hbar$-adic vertex algebra is a $\beta\gamma$-system:
\begin{align*}
    \beta(z)\gamma(w)\sim \frac{\hbar}{z-w}\;.
\end{align*}
In this example, the underlying vector space (vacuum module) is 
\begin{align*}
    D^{\mathrm{ch}}(\mathbb C^2)_{\hbar}:=\mathbb C[\![\hbar]\!][\beta_{(-n)},\gamma_{(-n)}\:|\:n\in \mathbb Z_{\ge 1}]|0\rangle \;.
\end{align*}
Note that $D^{\mathrm{ch}}(\mathbb C^2)/(\hbar)$ is a commutative vertex algebra, in fact a vertex Poisson algebra, which is isomorphic to the coordinate ring on the jet bundle ${\jet}\mathbb C^2$ endowed with the natural vertex Poisson algebra structure induced by the Poisson structure on $\mathbb C^2=T^*\mathbb C$. 

We define a sheaf of $\hbar$-adic vertex algebra $\mathscr D^{\mathrm{ch}}_{\mathbb C^2,\hbar}$ on $\mathbb C^2$ by 
\begin{align*}
    \mathscr D^{\mathrm{ch}}_{\mathbb C^2,\hbar}(U_f):=\mathbb C[{\jet}U_f][\![\hbar]\!] \;
\end{align*}
where $f$ is a polynomial in $\beta,\gamma$, and $U_f\subset \mathbb C^2$ is the open subset of nonzero locus of $f$. The $\hbar$-adic vertex algebra structure on the the vector space $\mathbb C[{\jet}U_f][\![\hbar]\!]$ is defined using Wick's contraction formula:
\begin{align*}
    A(z)B(w)\sim :e^{\frac{\hbar}{z-w}\partial_\beta\wedge\partial_\gamma}(A(z)\otimes B(w)): \;
\end{align*}
where $A,B\in \mathbb C[\partial^n\beta,\partial^n\gamma\:|\:n\in \mathbb Z_{\ge 0}][f^{-1}]$, and $::$ is the normal-ordered product. For example let $f=\beta\gamma$, then take $A=\beta^{-1}$ and $B=\gamma^{-1}$, the OPE reads:
\begin{align*}
    \beta^{-1}(z)\gamma^{-1}(w)\sim \sum_{n=1}^{\infty}n!\hbar^n\frac{:\beta^{-n-1}(z)\gamma^{-n-1}(w):}{(z-w)^n} \;.
\end{align*}

The above example naturally generalizes to multiple pairs of $\beta\gamma$'s, i.e. we have $$D^{\mathrm{ch}}(T^*\mathbb C^m)_{\hbar}=D^{\mathrm{ch}}(T^*\mathbb C)_{\hbar}^{\otimes m}$$ for all $m\in \mathbb Z_{\ge 1}$. The choice of coordinate system on $\mathbb C^m$ does not affect the $\hbar$-adic vertex algebra structure on $D^{\mathrm{ch}}(T^*\mathbb C^m)_{\hbar}$, so we will simply write it as $D^{\mathrm{ch}}(T^*V)_{\hbar}$, where $V$ is an $m$-dimensional vector space. We can define the localization of $D^{\mathrm{ch}}(T^*V)_{\hbar}$ using Wick's contraction formula, and therefore get a sheaf of $\hbar$-adic vertex algebras on $T^*V$, denoted by $\mathscr D^{\mathrm{ch}}_{T^*V,\hbar}$; $\mathscr D^{\mathrm{ch}}_{T^*V,\hbar}$ is called the sheaf of chiral differential operators (CDO) on $T^*V$.

We can also sheafify the more general $\beta\gamma\mathsf b\mathsf c$ system, i.e. the $\hbar$-adic vertex superalgebra $D^{\mathrm{ch}}(T^*(V\oplus\Pi S))_{\hbar}$, where $V$ and $S$ are vectors spaces and $\Pi S$ is the odd vector space associated to $S$. We only localize the bosonic variables, i.e.\ only the $\beta\gamma$ part is allowed to take the inverse. Then we get a sheaf of $\hbar$-adic vertex superalgebras $\mathscr D^{\mathrm{ch}}_{T^*(V\oplus\Pi S),\hbar}$ on the variety $T^*V$.

Another example is the $\hbar$-adic Heisenberg algebra and its localization. Let $W=\bigoplus_{i=1}^m\mathbb C\cdot a^i$ be a vector space, endowed with a symmetric inner product $(\cdot,\cdot)$, then we define $\hbar$-adic vertex algebra 
\begin{align*}
    \mathcal{H}(W,(\cdot,\cdot))_\hbar:=\mathbb C[\![\hbar]\!][a^1_{(-n)},\cdots,a^m_{(-n)}\:|\:n\in \mathbb Z_{\ge 1}]|0\rangle \;,
\end{align*}
whose $\hbar$-adic vertex algebra structure is given by the OPE
\begin{align*}
    a^i(z)a^j(w)\sim \frac{\hbar^2(a^i,a^j)}{(z-w)^2} \;.
\end{align*}
Since the Heisenberg algebra also has the Wick contraction formula, we can define its localization similarly to the localization of $\beta\gamma$ system. Namely for a polynomial $f\in \mathbb C[a^1,\cdots,a^m]$, and $A,B\in \mathbb C[\partial^n a^1,\cdots,\partial^na^m\:|\:n\in \mathbb Z_{\ge 0}][f^{-1}]$, the OPE reads
\begin{align*}
    A(z)B(w)\sim :e^{\frac{\hbar^2(a^i,a^j)}{(z-w)^2}\partial_{a^i}\otimes\partial_{a^j}}(A(z)\otimes B(w)): \;
\end{align*}
where $::$ is the normal-ordered product. For example let $f=a^1a^2$, then take $A=a^1$ and $B=a^2$, the OPE reads:
\begin{align*}
    (a^1)^{-1}(z)(a^2)^{-1}(w)\sim \sum_{n=1}^{\infty}n!\hbar^{2n}(a^1,a^2)^n\frac{:(a^1)^{-n-1}(z)(a^2)^{-n-1}(w):}{(z-w)^{2n}} \;.
\end{align*}
We define a sheaf $\mathscr H^{(\cdot,\cdot)}_{W,\hbar}$ by
\begin{align*}
    \mathscr H^{(\cdot,\cdot)}_{W,\hbar}(U_f):=\mathbb C[J_\infty U_f][\![\hbar]\!] \;,
\end{align*}
endowed with the aforementioned $\hbar$-adic vertex algebra structure.

\bigskip Now coming back to our quiver setting, we define a sheaf of $\hbar$-adic vertex superalgebras on the prequotient $\mathfrak R:=T^*\mathrm{Rep}(Q,\mathbf v,\mathbf w)\times \mathbb C^{Q_0}$, denoted by $\mathscr D^{\mathrm{ch}}_{\mathfrak R,\hbar}$, by the following
\begin{align}
    \mathscr D^{\mathrm{ch}}_{\mathfrak R,\hbar}:=\mathscr D^{\mathrm{ch}}_{T^*\mathrm{Rep}(Q,\mathbf v,\mathbf w|\mathbf d),\hbar}\widehat{\boxtimes} \mathscr H^{(\cdot,\cdot)_Q}_{\mathbb C^{Q_0},\hbar} \;,
\end{align}
where $\mathbf d=\mathbf w-\mathsf C\mathbf v$, and $(\cdot,\cdot)_{Q}$ is the Cartan inner product on $\mathbb C^{Q_0}$:
\begin{align}
    (h_i,h_j)_Q:=2\delta_{ij}-\#(a:i\to j\in Q_1)-\#(a:j\to i\in Q_1)\;.
\end{align}
$\mathscr D^{\mathrm{ch}}_{\mathfrak R,\hbar}$ is the sheafified version of $\CVs(Q,\mathbf v,\mathbf w)$, namely the $\hbar$-adic completion $\CVs(Q,\mathbf v,\mathbf w)_{\hbar}$ is isomorphic to the global section $\Gamma(\mathfrak R,\mathscr D^{\mathrm{ch}}_{\mathfrak R,\hbar})$.

\subsection{Sheaf of BRST cohomologies}

In this subsection we construct the sheaf of BRST cohomologies of $\mathscr D^{\mathrm{ch}}_{\mathfrak R,\hbar}$ on $\widetilde\CM$, following \cite{arakawa2023hilbert} and \cite{kuwabara2017vertex}. 

\bigskip First of all, let us define $\hbar$-adic BRST reduction. Let $\mathfrak g$ be a reductive Lie algebra and $\kappa$ be a $\mathfrak g$-invariant symmetric inner form on $\mathfrak g$, then the $\hbar$-adic current algebra $V^\kappa(\mathfrak g)_\hbar$ is the $\mathbb C[\![\hbar]\!]$-module $$\mathbb C[J_\infty\mathfrak g][\![\hbar]\!]:=\mathbb C[\![\hbar]\!][E^\alpha_{(-n)}\:|\:n\in \mathbb Z_{\ge 1}, E^\alpha\in\text{a basis of }\mathfrak g]|0\rangle$$ endowed with OPE
\begin{align*}
    E^{\alpha}(z)E^\beta(w)\sim \frac{\hbar^2\kappa^{\alpha,\beta}}{(z-w)^2}+\frac{\hbar f^{\alpha\beta}_\gamma E^\gamma(w)}{z-w} \;,
\end{align*}
where $\kappa^{\alpha,\beta}=\kappa(E^\alpha,E^\beta)$ is the matrix element of the inner form $\kappa$, and $f^{\alpha\beta}{}_\gamma$ is the structure constant of $\mathfrak g$ (so that $f^{\alpha\beta}{}_\gamma E^\gamma=[E^\alpha,E^\beta]$).

Denote by $\kappa_{\mathfrak g}$ the Killing form on $\mathfrak g$. Suppose that $V$ is an $\hbar$-adic vertex superalgebra such that there exists an $\hbar$-adic vertex superalgebra map
\begin{align*}
    J_V: V^{-\kappa_{\mathfrak g}}(\mathfrak g)_\hbar\to V \;,
\end{align*}
then we define $\hbar$-adic vertex superalgebra 
\begin{align}
    \widetilde C(\mathfrak g, V)_\hbar:=V\otimes \CV_{\mathsf b\mathsf c}(\mathfrak g)_\hbar \;,
\end{align}
where $\CV_{\mathsf b\mathsf c}(\mathfrak g)_\hbar$ is the $\hbar$-adic $\mathsf b\mathsf c$ system of $\mathfrak g$, i.e. chiral differential operators $D^{\mathrm{ch}}(\Pi T^*\mathfrak g)_\hbar$ of the odd symplectic vector space $\Pi T^*\mathfrak g$. Note that $\CV_{\mathsf b\mathsf c}(\mathfrak g)_\hbar$ admits an $\hbar$-adic vertex superalgebra map
\begin{align*}
    J_{\mathsf b\mathsf c}: V^{\kappa_{\mathfrak g}}(\mathfrak g)_\hbar\to \CV_{\mathsf b\mathsf c}(\mathfrak g)_\hbar\;, 
    \quad 
    J_{\mathsf b\mathsf c}(E^\alpha)=f^{\alpha\beta}{}_{\gamma}\mathsf b^\gamma\mathsf c_{\beta}\;.
\end{align*}
$\CV_{\mathsf b\mathsf c}(\mathfrak g)_\hbar$ is graded by ghost numbers, i.e. $\deg \mathsf c=+1,\deg\mathsf b=-1$, so the image of $J$ is contained in the degree zero piece $\CV^0_{\mathsf b\mathsf c}(\mathfrak g)_\hbar$. Consider the ghost number $1$ element 
\begin{align}
    \JBRST:=\mathsf c_\alpha\left(J_V(E^\alpha)+\frac{1}{2}J_{\mathsf b\mathsf c}(E^\alpha)\right)\in  \widetilde C^1(\mathfrak g, V)_\hbar \;.
\end{align}
The operator $\mathcal{Q}_{\hbar}:=\frac{1}{\hbar}\JBRST_{(0)}$ acts on $\widetilde C(\mathfrak g, V)_\hbar$, and it squares to zero $\mathcal{Q}_{\hbar}^2=0$, therefore $(\widetilde C(\mathfrak g, V)_\hbar, \mathcal{Q}_{\hbar})$ is a chain complex of $\hbar$-adic vertex superalgebra, which is graded by ghost numbers. We call it the $\hbar$-adic BRST complex. 

Note that the linear map 
\begin{align}
    E^\alpha\mapsto J(E^\alpha):=\mathcal{Q}_{\hbar}(\mathsf b^\alpha)=J_V(E^\alpha)+J_{\mathsf b\mathsf c}(E^\alpha)
\end{align}
is an $\hbar$-adic vertex superalgebra chain complex map 
$$
J: (V^0(\mathfrak g)_\hbar,0)\to (\widetilde C(\mathfrak g, V)_\hbar, \mathcal{Q}_{\hbar}) \;,
$$ 
where the LHS is endowed with trivial differential. In particular the Lie algebra $\mathfrak g$ acts on the complex $(\widetilde C(\mathfrak g, V)_\hbar, \mathcal{Q}_{\hbar})$ by $E^\alpha\mapsto \frac{1}{\hbar}J(E^\alpha)_{(0)}$. 

The relative $\hbar$-adic BRST complex is defined to be the subspace 
\begin{align}
    C(\mathfrak g, V)_\hbar:=\{v\in \widetilde C(\mathfrak g, V)_\hbar\:|\: J(E^\alpha)_{(0)}v=\mathsf b^\alpha_{(0)}v=0,\forall \alpha\} \;,
\end{align}
which is a sub-complex because $[J(E^\alpha)_{(0)},\mathcal{Q}_{\hbar}]=0$ and $[\mathsf b^\alpha_{(0)},\mathcal{Q}_{\hbar}]=J(E^\alpha)_{(0)}$. We define the relative $\hbar$-adic BRST cohomology to be
\begin{align}
    H^{\infty/2+\bullet}_{\hbar\mathrm{BRST}}(\mathfrak g,V):=H^{\bullet}(C(\mathfrak g, V)_\hbar, \mathcal{Q}_{\hbar})\;.
\end{align}
This is a $\mathbb Z$-graded $\hbar$-adic vertex superalgebra.

\bigskip Now we come back to our quiver setting. According to our previous discussion in \cref{subsection:anomaly}, the chiral extended moment map $\widetilde\mu_{\mathrm{ch}}$ is an $\hbar$-adic vertex superalgebra map from the current algebra $V^{-\kappa_{\mathfrak{gl}(\mathbf v)}}(\mathfrak{gl}(\mathbf v))_\hbar$ to $\CVs(Q,\mathbf v,\mathbf w)_\hbar$. We define a presheaf of $\mathbb Z$-graded $\hbar$-adic vertex superalgebras on $\widetilde\CM$ by assigning any open subset $U\subset \widetilde\CM$ to the relative $\hbar$-adic BRST cohomology:
\begin{align}
    U\mapsto H^{\infty/2+\bullet}_{\hbar\mathrm{BRST}}(\mathfrak{gl}(\mathbf v),\mathscr D^{\mathrm{ch}}_{\mathfrak R,\hbar}(\widetilde U))\;.
\end{align}
Here $\widetilde U$ is an open subset of $\mathfrak R$ such that $\widetilde U\cap \widetilde\mu^{-1}(0)=q^{-1}(U)$ where $q^{-1}(U)$ is the preimage of $U$ under the quotient map $q:\widetilde\mu^{-1}(0)^{\mathrm{st}}\to \widetilde\CM$. Note that 
\begin{itemize}
    \item $\widetilde U$ always exists, and the largest such choice is $\widetilde U_{\mathrm{max}}=\mathfrak R\setminus(\widetilde\mu^{-1}(0)\setminus q^{-1}(U))$,
    \item if $U$ is affine then $\widetilde U$ can be chosen such that it is also affine,
    \item $H^{\infty/2+\bullet}_{\hbar\mathrm{BRST}}(\mathfrak{gl}(\mathbf v),\mathscr D^{\mathrm{ch}}_{\mathfrak R,\hbar}(\widetilde U))$ does not depend on the choice of $\widetilde U$, because a generalization of \cite[Lemma 4.13]{arakawa2023hilbert} holds, i.e. if $D\subset \mathfrak R$ is an open subset such that $D\cap \mu^{-1}(0)=\emptyset$ then $H^{\infty/2+\bullet}_{\hbar\mathrm{BRST}}(\mathfrak{gl}(\mathbf v),\mathscr D^{\mathrm{ch}}_{\mathfrak R,\hbar}(D))=0$, see the companion paper \cite{math_draft} for more details.
\end{itemize}
The sheaf of $\mathbb Z$-graded $\hbar$-adic vertex superalgebras on $\widetilde\CM$ associated to the presheaf $U\mapsto H^{\infty/2+\bullet}_{\hbar\mathrm{BRST}}(\mathfrak{gl}(\mathbf v),\mathscr D^{\mathrm{ch}}_{\mathfrak R,\hbar}(\widetilde U))$ is denoted by $\mathscr H^{\infty/2+\bullet}_{\hbar\mathrm{BRST}}(\mathfrak{gl}(\mathbf v),\mathscr D^{\mathrm{ch}}_{\mathfrak R,\hbar})$, and denote its cohomological degree zero piece by 
\begin{align}
    \mathscr D^{\mathrm{ch}}_{\widetilde\CM,\hbar}:= \mathscr H^{\infty/2+0}_{\hbar\mathrm{BRST}}(\mathfrak{gl}(\mathbf v),\mathscr D^{\mathrm{ch}}_{\mathfrak R,\hbar})\;.
\end{align}
We note that if $U\subset \widetilde\CM$ is an open affine subset, then 
\begin{align}\label{eq:sheaf on affine}
    \mathscr D^{\mathrm{ch}}_{\widetilde\CM,\hbar}(U)=H^{\infty/2+0}_{\hbar\mathrm{BRST}}(\mathfrak{gl}(\mathbf v),\mathscr D^{\mathrm{ch}}_{\mathfrak R,\hbar}(\widetilde U))\;.
\end{align}
The argument in \cite[Section 4]{arakawa2023hilbert} naturally generalizes to our setting, and we have the following.
\begin{thm}
The BRST cohomology sheaf $\mathscr H^{\infty/2+n}_{\hbar\mathrm{BRST}}(\mathfrak{gl}(\mathbf v),\mathscr D^{\mathrm{ch}}_{\mathfrak R,\hbar})$ vanishes for $n\neq 0$. Moreover, there is an isomorphism between sheaves of vertex Poisson superalgebras
\begin{align}
    \mathscr D^{\mathrm{ch}}_{\widetilde\CM,\hbar}/\hbar\mathscr D^{\mathrm{ch}}_{\widetilde\CM,\hbar}\cong \mathscr O_{J_\infty\widetilde\CM(Q,\mathbf v,\mathbf w|\mathbf d)}\;,
\end{align}
where $\mathbf d:=\mathbf w-\mathsf C\mathbf v$ and $\widetilde\CM(Q,\mathbf v,\mathbf w|\mathbf d)$ is the extended quiver super variety.
\end{thm}

\subsection{Vertex superalgebra of global sections}\label{subsec:global section}

Let $\mathbb S$ be an one-dimensional torus that acts on $\mathfrak R$ by assigning the following weights to coordinates
\begin{align}\label{eq:S-wts of XYIJh}
    \mathrm{wt}(X_{(-n)})=\mathrm{wt}(Y_{(-n)})=\mathrm{wt}(I_{(-n)})=\mathrm{wt}(J_{(-n)})=n\;,
    \quad 
    \mathrm{wt}(h_{(-n)})=n+1 \;.
\end{align}
We extend the $\mathbb S$-action to $\widetilde C(\mathfrak{gl}(\mathbf v),\CVs(Q,\mathbf v,\mathbf w))$ by assigning the following weights
\begin{align}\label{eq:S-wts of bc and hbar}
    \mathrm{wt}(\mathsf b_{(-n)})=n+1\;,\quad\mathrm{wt}(\mathsf c_{(-n)})=n-1\;,\quad\mathrm{wt}(\hbar)=1\;,
\end{align}
then $\mathcal{Q}_{\hbar}$ has zero weight, so $(\widetilde C(\mathfrak{gl}(\mathbf v),\CVs(Q,\mathbf v,\mathbf w)),\mathcal{Q}_{\hbar})$ is an $\mathbb S$-equivariant complex. Moreover, the relative BRST complex $(C(\mathfrak{gl}(\mathbf v),\CVs(Q,\mathbf v,\mathbf w)),\mathcal{Q}_{\hbar})$ is an $\mathbb S$-equivariant sub-complex. 

Since $\widetilde\CM$ has an $\mathbb S$-equivariant ample line (tensor product of tautological line bundles), there exists an $\mathbb S$-equivariant closed embedding $\widetilde\CM\hookrightarrow\widetilde\CM^0\times\mathbb P(L)$ where $L$ is an $\mathbb S$-module. So we can find an open cover $\widetilde\CM=\bigcup_{i=1}^nU_i$ such that each $U_i$ is affine and $\mathbb S$-invariant. Set $\widetilde U_{i,\mathrm{max}}=\mathfrak R\setminus(\widetilde\mu^{-1}(0)\setminus q^{-1}(U_i))$, then $\widetilde U_{i,\mathrm{max}}$ is the largest among those open subsets of $\mathfrak R$ whose intersection with $\widetilde\mu^{-1}(0)$ equals to $q^{-1}(U)$. It follows that $\widetilde U_{i,\mathrm{max}}$ is also $\mathbb S$-invariant, therefore $\mathscr D^{\mathrm{ch}}_{\widetilde\CM,\hbar}(U_i)=H^{\infty/2+0}_{\hbar\mathrm{BRST}}(\mathfrak{gl}(\mathbf v),\mathscr D^{\mathrm{ch}}_{\mathfrak R,\hbar}(\widetilde U_{i,\mathrm{max}}))$ is a complete $\mathbb S$-module. Similarly $\mathscr D^{\mathrm{ch}}_{\widetilde\CM,\hbar}(U_i\cap U_j)$ is a complete $\mathbb S$-module. Then it follows that $\mathscr D^{\mathrm{ch}}_{\widetilde\CM,\hbar}(\widetilde\CM)=\ker\left(\bigoplus_i\mathscr D^{\mathrm{ch}}_{\widetilde\CM,\hbar}(U_i)\to \bigoplus_{i,j}\mathscr D^{\mathrm{ch}}_{\widetilde\CM,\hbar}(U_i\cap U_j)\right)$ is a complete $\mathbb S$-module. 

Since the weights of $\mathscr D^{\mathrm{ch}}_{\widetilde\CM,\hbar}(\widetilde\CM)$ are bounded below\footnote{This can be shown as follows. Since $\hbar^p\mathscr D^{\mathrm{ch}}_{\widetilde\CM,\hbar}/\hbar^{p+1}\mathscr D^{\mathrm{ch}}_{\widetilde\CM,\hbar}\cong \hbar^p\mathscr O_{J_\infty\widetilde\CM(Q,\mathbf v,\mathbf w|\mathbf d)}$, it is enough to show that weights of $\mathscr O_{J_\infty\widetilde\CM(Q,\mathbf v,\mathbf w|\mathbf d)}(\widetilde\CM)$ are bounded below. Since the action of $\mathbb S$ on $\widetilde\CM^0$ is repelling with unique fixed point $0$, it follows that the fixed point set $\widetilde\CM^{\mathbb S}$ is contained in the central fiber $\pi^{-1}(0)$ of the projection $\pi:\widetilde\CM\to \widetilde\CM^0$. $\pi^{-1}(0)$ is proper, so there must be a connected component $F\subset \widetilde\CM^{\mathbb S}$ such that normal weights are all repelling. Let $U\subset \widetilde\CM$ be the repelling locus of $F$, then $\mathscr O_{J_\infty\widetilde\CM(Q,\mathbf v,\mathbf w|\mathbf d)}(\widetilde\CM)\subset \mathscr O_{J_\infty\widetilde\CM(Q,\mathbf v,\mathbf w|\mathbf d)}(U)$. The limit $x\mapsto \lim_{t\to \infty}\sigma(t)\cdot x$ gives a map $\lim: U\to F$, which realizes $U$ as an $\mathbb S$-equivariant affine bundle over $F$. It follows that there is a filtration on $\mathscr O_{J_\infty\widetilde\CM(Q,\mathbf v,\mathbf w|\mathbf d)}(U)$ such that its associated graded is an $\mathbb S$-equivariant subspace of $\Gamma(F,S^\bullet(\mathcal N^\vee_{F/J_\infty\widetilde\CM(Q,\mathbf v,\mathbf w|\mathbf d)}))$, where $\mathcal N_{F/J_\infty\widetilde\CM(Q,\mathbf v,\mathbf w|\mathbf d)}$ is the normal bundle of $F$ in $J_\infty\widetilde\CM(Q,\mathbf v,\mathbf w|\mathbf d)$. Since there could be only finitely many normal directions of $F$ which have positive weights, and moreover all such weight spaces are odd by our choice of $F$, it follows that weights in $\Gamma(F,S^\bullet(\mathcal N^\vee_{F/J_\infty\widetilde\CM(Q,\mathbf v,\mathbf w|\mathbf d)}))$ are bounded below.}, we have
\begin{align*}
    \mathscr D^{\mathrm{ch}}_{\widetilde\CM,\hbar}(\widetilde\CM)={\prod_{m\ge -M}}'\mathscr D^{\mathrm{ch}}_{\widetilde\CM,\hbar}(\widetilde\CM)_m \;,
\end{align*}
where $\mathscr D^{\mathrm{ch}}_{\widetilde\CM,\hbar}(\widetilde\CM)_m$ is the weight $m$ subspace \footnote{The argument in the previous footnote shows that if $(Q,\mathbf v,\mathbf w)$ is balanced, then $\mathscr D^{\mathrm{ch}}_{\widetilde\CM,\hbar}(\widetilde\CM)$ is non-negatively graded, and $\mathscr D^{\mathrm{ch}}_{\widetilde\CM,\hbar}(\widetilde\CM)_0=\mathbb C|0\rangle$.}, and ${\prod}'$ means the subspace of the product space consisting of elements $(x_i)_{i=-M}^{\infty}$ such that there exists $N$ and $x_n\in \hbar^{n-N}\mathscr D^{\mathrm{ch}}_{\widetilde\CM,\hbar}(\widetilde\CM)_N$ for all $n>N$. Consider the subspace of $\mathbb S$-finite elements
\begin{align}
    \mathscr D^{\mathrm{ch}}_{\widetilde\CM,\hbar}(\widetilde\CM)_{\mathbb S\text{-fin}}:=\bigoplus_{m\ge 0}\mathscr D^{\mathrm{ch}}_{\widetilde\CM,\hbar}(\widetilde\CM)_m \;,
\end{align}
which is a $\mathbb C[\hbar]$-module. Moreover the OPE preserves the $\mathbb S$-weight, and every OPE must depend on $\hbar$ in a polynomial way since $\mathrm{wt}(\hbar)=1$ and weights of $\mathscr D^{\mathrm{ch}}_{\widetilde\CM,\hbar}(\widetilde\CM)$ are bounded below. Therefore $\mathscr D^{\mathrm{ch}}_{\widetilde\CM,\hbar}(\widetilde\CM)_{\mathbb S\text{-fin}}$ is a vertex superalgebra over the base ring $\mathbb C[\hbar]$.

\bigskip Denote by $\mathsf D^{\mathrm{ch}}(\widetilde\CM)$ the $\hbar=1$ specialization of the $\mathbb S$-finite part: 
\begin{align}
\mathsf D^{\mathrm{ch}}(\widetilde\CM):=\mathscr D^{\mathrm{ch}}_{\widetilde\CM,\hbar}(\widetilde\CM)_{\mathbb S\text{-fin}}/(\hbar-1)\;.
\end{align}
We claim that there exists a natural vertex superalgebra map
\begin{align}\label{quiver VOA to global section}
    \CV(Q,\mathbf v,\mathbf w)\to \mathsf D^{\mathrm{ch}}(\widetilde\CM)\;.
\end{align}
In fact, we consider the $\mathbb S$-finite part of the relative $\hbar$-adic BRST complex:
\begin{align*}
    C(\mathfrak{gl}(\mathbf v),\CVs(Q,\mathbf v,\mathbf w))[\hbar]:=C(\mathfrak{gl}(\mathbf v),\CVs(Q,\mathbf v,\mathbf w))_{\hbar,\mathbb S-\mathrm{fin}}\;,
\end{align*}
then $\mathcal{Q}_{\hbar}$ has $\mathbb S$-weight zero and thus stabilizes $C(\mathfrak{gl}(\mathbf v),\CVs(Q,\mathbf v,\mathbf w))[\hbar]$, which makes it a cochain complex of $\mathbb C[\hbar]$-vertex superalgebras. We also note that the $\hbar=1$ specialization gives the ordinary relative BRST complex:
\begin{align*}
    (C(\mathfrak{gl}(\mathbf v),\CVs(Q,\mathbf v,\mathbf w))[\hbar],\mathcal{Q}_{\hbar})/(\hbar-1)\cong (C(\mathfrak{gl}(\mathbf v),\CVs(Q,\mathbf v,\mathbf w)),\mathcal{Q}) \;.
\end{align*}
$C(\mathfrak{gl}(\mathbf v),\CVs(Q,\mathbf v,\mathbf w))[\hbar]$ is sub-complex of $C(\mathfrak{gl}(\mathbf v),\CVs(Q,\mathbf v,\mathbf w))_\hbar$, and the latter admits a natural chain complex map to $C(\mathfrak{gl}(\mathbf v),\mathscr D^{\mathrm{ch}}_{\mathfrak R,\hbar}(\widetilde U))_\hbar$ for every open subset $U\subset \widetilde\CM$, whence we obtain a $\mathbb C[\hbar]$-module map $H^0(C(\mathfrak{gl}(\mathbf v),\CVs(Q,\mathbf v,\mathbf w))[\hbar],\mathcal{Q}_{\hbar})\to \mathscr D^{\mathrm{ch}}_{\widetilde\CM,\hbar}(\widetilde\CM)$, whose image is contained in the $\mathbb S$-finite part. This map is compatible with taking OPE, thus we have a map between $\mathbb C[\hbar]$-vertex superalgebras $H^0(C(\mathfrak{gl}(\mathbf v),\CVs(Q,\mathbf v,\mathbf w))[\hbar],\mathcal{Q}_{\hbar})\to \mathscr D^{\mathrm{ch}}_{\widetilde\CM,\hbar}(\widetilde\CM)_{\mathbb S\text{-fin}}$. Specializing to $\hbar=1$ gives the desired vertex superalgebra map \eqref{quiver VOA to global section}.

In the companion paper \cite{math_draft}, we will prove the following result.
\begin{thm}\label{thm:vanishing and embedding}
Let $Q$ be a quiver and choose the gauge and framing dimensions $\mathbf v,\mathbf w$ such that $\mathbf w-\mathsf C\mathbf v\in \mathbb Z^{Q_0}_{\ge 0}$. Assume moreover that one of the following conditions holds:
\begin{itemize}
    \item $Q$ is totally negative, i.e. $\mathsf C_{ij}<0$ for every pair $i,j\in Q_0$,
    \item $Q$ is a Dynkin quiver, and $\sum_{i\in Q_0}(\mathbf v_i-1)\le 2$,
    \item $Q$ is the $A_1$ quiver,
    \item $Q$ is the Jordan quiver, and $\mathbf w_i\ge \mathbf v_i$,
\end{itemize}
then we have
\begin{enumerate}
    \item $H^{\infty/2+n}_{\mathrm{BRST}}(\mathfrak{gl}(\mathbf v),\CVs(Q,\mathbf v,\mathbf w))$ vanishes for $n<0$,
    \item the map $\CV(Q,\mathbf v,\mathbf w)\to \mathsf D^{\mathrm{ch}}(\widetilde\CM)$ in \eqref{quiver VOA to global section} is injective.
\end{enumerate}
\end{thm}

\subsection{Associated variety}\label{subsec:Associated variety}

Recall that we have a canonical filtration \eqref{eq:canonical filtration} $F^\bullet V$ for every vertex algebra $V$. This construction naturally extends to a sheaf $\mathscr D$ of $\hbar$-adic vertex algebras. Namely, we define the subsheaf $F^p\mathscr D$ to be the sheafification of the sub-presheaf
\begin{align*}
    U\mapsto F^p(\mathscr D(U)) \;,
\end{align*}
where $F^p(\mathscr D(U))$ is defined the same way as \eqref{eq:canonical filtration}. This allows us to define the $C_2$ sheaf
\begin{align}
    \CR({\mathscr D}):=F^0\mathscr D/F^1\mathscr D\;.
\end{align}
Now let us go back to the quiver vertex algebra case $\mathscr D=\mathscr D^{\mathrm{ch}}_{\widetilde\CM,\hbar}$. Then we have
\begin{align*}
    \CR(\mathscr D^{\mathrm{ch}}_{\widetilde\CM,\hbar})/\hbar \CR(\mathscr D^{\mathrm{ch}}_{\widetilde\CM,\hbar})\cong \mathscr D^{\mathrm{ch}}_{\widetilde\CM,\hbar}/(\hbar \mathscr D^{\mathrm{ch}}_{\widetilde\CM,\hbar}+F^2\mathscr D^{\mathrm{ch}}_{\widetilde\CM,\hbar}) \cong  \CR(\mathscr O_{J_\infty\widetilde\CM(Q,\mathbf v,\mathbf w|\mathbf d)})\cong  \mathscr O_{\widetilde\CM(Q,\mathbf v,\mathbf w|\mathbf d)}\;,
\end{align*}
which induces an isomorphism 
\begin{align}
    \CR(\mathscr D^{\mathrm{ch}}_{\widetilde\CM,\hbar})\cong \mathscr O_{\widetilde\CM(Q,\mathbf v,\mathbf w|\mathbf d)}[\![\hbar]\!]\;.
\end{align}
We have an obvious inclusion 
\begin{align*}
    F^1(\mathscr D^{\mathrm{ch}}_{\widetilde\CM,\hbar}(\widetilde\CM))\subset F^1(\mathscr D^{\mathrm{ch}}_{\widetilde\CM,\hbar})(\widetilde\CM)\;,
\end{align*}
thus we get the following natural maps between $\hbar$-adic Poisson algebras:
\begin{align*}
    \CR(\mathscr D^{\mathrm{ch}}_{\widetilde\CM,\hbar}(\widetilde\CM))\twoheadrightarrow \mathscr D^{\mathrm{ch}}_{\widetilde\CM,\hbar}(\widetilde\CM)/F^1(\mathscr D^{\mathrm{ch}}_{\widetilde\CM,\hbar})(\widetilde\CM)\hookrightarrow \CR(\mathscr D^{\mathrm{ch}}_{\widetilde\CM,\hbar})(\widetilde\CM)\cong \mathscr O_{\widetilde\CM(Q,\mathbf v,\mathbf w|\mathbf d)}(\widetilde\CM)[\![\hbar]\!] \;,
\end{align*}
where the first map is surjective and the second map is injective. Restricting to $\mathbb S$-finite part and specializing to $\hbar=1$, we get a map between Poisson algebras:
\begin{align}\label{eq:map from C_2 alg to function ring of quiver variety}
    \CR(\mathsf D^{\mathrm{ch}}(\widetilde\CM))\longrightarrow \mathscr O_{\widetilde\CM(Q,\mathbf v,\mathbf w|\mathbf d)}(\widetilde\CM) \;.
\end{align}
\begin{conjecture}
The map \eqref{eq:map from C_2 alg to function ring of quiver variety} is an isomorphism.
\end{conjecture}

\begin{prop}\label{prop:injectivity criterion}
Suppose that $\Gamma$ is a vertex algebra map from a bosonic vertex algebra $\CW$ to $\mathsf D^{\mathrm{ch}}(\widetilde\CM)$. Assume moreover that
\begin{enumerate}
    \item $\CW$ has a PBW basis\footnote{By definition, $\CW$ has a PBW basis if and only if $\mathrm{gr}_F\CW=\jet\CR(\CW)$ and $\mathcal R(\CW)$ is a finitely generated polynomial ring.},
    \item after reducing to $C_2$ algebra, the induced map $\CR(\CW)\to \mathscr O_{\widetilde\CM}(\widetilde\CM)$ is injective.
\end{enumerate}
Then $\Gamma$ is injective.
\end{prop}

\begin{proof}
By the definition, $\mathrm{gr}_F\CW=\jet\CR(\CW)$ and $\mathcal R(\CW)$ is a finitely generated polynomial ring. In particular the associated variety $X(\CW)=\mathrm{Spec}\mathcal R(\CW)$ is an affine space. It follows from the injectivity of $\CR(\CW)\to \mathscr O_{\widetilde\CM}(\widetilde\CM)$ that the induced morphism $f:\widetilde\CM\to X(\CW)$ is dominant. By generic smoothness, there exist open subsets $U\subset \widetilde\CM, U'\subset X(\CW)$ such that the $f|_U:U\to U'$ is smooth, thus $J_\infty f|_{J_\infty U}: J_\infty U\to J_\infty U' $ is a flat morphism. According to \cite[Lemma 4.3.7]{liu2002algebraic}, $J_\infty f|_{J_\infty U}$ is a dominant morphism. $\jet\CR(\CW)$ is an integral domain because $\CR(\CW)$ is a polynomial ring, then it follows that the induced map between jets $\jet\CR(\CW)\to \mathbb C[{\jet}\widetilde\CM]$ is injective. We note that $\jet\CR(\CW)\to \mathbb C[{\jet}\widetilde\CM]$ factors through $\mathrm{gr}_F\mathsf D^{\mathrm{ch}}(\widetilde\CM)$, therefore the induced map between associated graded algebra $\mathrm{gr}_F(\Gamma):\mathrm{gr}_F\CW\to \mathrm{gr}_F\mathsf D^{\mathrm{ch}}(\widetilde\CM)$ is injective. This implies that $\Gamma$ is injective.
\end{proof}

A typical situation is that $\Gamma:\CW\to \mathsf D^{\mathrm{ch}}(\widetilde\CM)$ maps the strong generators of $\CW$ to path elements in $\mathscr O_{\widetilde\CM}(\widetilde\CM)$, after reducing to the $C_2$ algebra. For example in the $T[\SU(n)]$ quiver in \cref{fig:T[SU(n)]}
\begin{figure}
\centering
    \begin{tikzpicture}[scale=0.9, every node/.style={scale=1.0}]
    \node(LL) at (-5,0) {$\circled{$1$}$};
    \node(L) at (-2.5,0) {$\cdots$};
    \node(M) at (0,0) {$\circled{$n$-$2$}$};
    \node(R) at (2.5,0) {$\circled{$n$-$1$}$};
    \node(B) at (2.5,-2.5) {$\rectangled{$n$}$};
    \draw[->] ([yshift=0.1cm]LL.east) to node[above]{$X_{n-2}$} ([yshift=0.1cm]L.west);
    \draw[->] ([yshift=-0.1cm]L.west) to node[below]{$Y_{n-2}$} ([yshift=-0.1cm]LL.east);
    \draw[->] ([yshift=0.1cm]L.east) to node[above]{$X_2$} ([yshift=0.1cm]M.west);
    \draw[->] ([yshift=-0.1cm]M.west) to node[below]{$Y_2$} ([yshift=-0.1cm]L.east);
    \draw[->] ([yshift=0.1cm]M.east) to node[above]{$X_1$} ([yshift=0.1cm]R.west);
    \draw[->] ([yshift=-0.1cm]R.west) to node[below]{$Y_1$} ([yshift=-0.1cm]M.east);
    \draw[->] ([xshift=0.1cm]R.south) to node[right]{$J$} ([xshift=0.1cm]B.north);
    \draw[->] ([xshift=-0.1cm]B.north) to node[left]{$I$} ([xshift=-0.1cm]R.south);
    \node(h3) at (-5,0.7) {$h_{n-1}$};
    \node(h2) at (0,0.7) {$h_2$};
    \node(h1) at (2.5,0.7) {$h_1$};
    \end{tikzpicture}
    \caption{The $T[\SU(n)]$ quiver.}
    \label{fig:T[SU(n)]}
\end{figure}
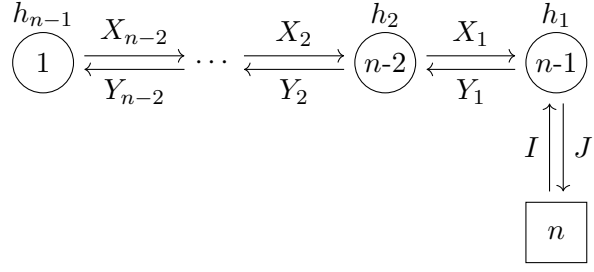
we have a vertex algebra map $\Gamma: V^{-n+1}(\mathfrak{sl}_n)\to \CV(Q)$ such that
\begin{align}\label{eq:mebdding of V(sl_n)}
\Gamma(E^i_j)=J^i_kI^k_j-\frac{\delta^i_j}{n}J^l_kI^k_l \;.
\end{align}
It is well-known that for any level $k\in \mathbb C$, the current algebra $V^{k}(\mathfrak{g})$ has a PBW basis, and its $C_2$ algebra is $\mathbb C[\mathfrak g^*]$ \cite[Example 4.9]{arakawa2017introduction}. In our case, the map $\varphi: \widetilde\CM\to \mathfrak{sl}_n^*$ induced from $\Gamma$ is given by 
\begin{align}
    \varphi(x,y,\alpha,\beta)=\beta\alpha-\frac{\mathrm{Tr}(\beta\alpha)}{n}\mathrm{id} \;.
\end{align}
Consider the map $p:\widetilde\CM\to \mathrm{End}(\mathbb C^n)$ such that $p(x,y,\alpha,\beta)=\beta\alpha$, then it is easy to see that the image of $p$ is the determinantal variety $\{A\in \mathrm{End}(\mathbb C^n)\:|\:\det(A)=0\}$. Since $\varphi=\mathfrak{Z}\circ p$, where $\mathfrak{Z}:\mathrm{End}(\mathbb C^n)\to \mathfrak{sl}_n^*$ maps a matrix $A$ to its traceless part $A-\frac{\mathrm{Tr}(A)}{n}\mathrm{id}$, we see that $\varphi$ is surjective because every traceless matrix $B$ equals to $\mathfrak Z(B-\lambda\cdot\mathrm{id})$ where $\lambda$ is an eigenvalue of $B$. It follows that $\CR(V^{-n+1}(\mathfrak{sl}_n))\to \mathbb C[\widetilde\CM]$ is injective, whence $\Gamma:V^{-n+1}(\mathfrak{sl}_n)\to \mathsf D^{\mathrm{ch}}(\widetilde\CM)$ is injective by Proposition \ref{prop:injectivity criterion}. 

\section{Quiver reduction}\label{sec:quiver reduction}

In this section, we develop a tool to study quiver vertex superalgebras inductively by reducing a quiver to another quiver with fewer nodes. 
This will enable us to obtain free-field realizations for quiver VOAs associated with a very general class of quiver gauge theories.

We begin with an observation on the geometry side. Let $Q$ be a quiver, and choose gauge dimension vector $\mathbf v$ and framing dimension vector $\mathbf w$ such that $\mathbf w-\mathsf C\mathbf v\in \mathbb Z^{Q_0}_{\ge 0}$. Assume moreover that there exists a node $i\in Q_0$ such that $\mathbf w_i= \mathbf v_i$, and consider the open locus $\mathcal U\subset \widetilde\CM(Q,\mathbf v,\mathbf w)$ where the framing map $I_i:\mathbb C^{\mathbf w_i}\to V_i$ is an isomorphism. On the locus $\mathcal U$ the isomorphism $\mathbb C^{\mathbf w_i}\cong V_i$ effectively removes the gauge node, and replaces it by a framing vector space of dimension $\mathbf w_i$, and the quiver data are updated accordingly: 
\begin{align*}
X_a\mapsto 
\begin{cases}
X_a, & \text{if }a\text{ is not connected to }i\\
I_i^{-1}X_a,& \text{if }a: j\to i\text{ where }j\neq i\\
X_aI_i,& \text{if }a: i\to j\text{ where }j\neq i\\
I_i^{-1}X_aI_i,& \text{if }a: i\to i
\end{cases},\quad 
Y_a\mapsto 
\begin{cases}
Y_a, & \text{if }a\text{ is not connected to }i\\
I_i^{-1}Y_a,& \text{if }a: i\to j\text{ where }j\neq i\\
Y_aI_i,& \text{if }a: j\to i\text{ where }j\neq i\\
I_i^{-1}Y_aI_i,& \text{if }a: i\to i
\end{cases}.
\end{align*}
The extended moment map equation at the node $i$ determines the $J_i$ by
\begin{align*}
    J_i=-I_i^{-1}z_i+\sum_{a:i\to j} I_i^{-1}Y_aX_a-\sum_{b:j\to i} I_i^{-1}X_b Y_b\;,
\end{align*}
so $J_i$ can be removed from the quiver data. Therefore we have an isomorphism
\begin{align*}
    \mathcal U_i\cong \widetilde\CM(Q',\mathbf v',\mathbf w')\times T^*\mathrm{End}(\mathbb C^{\mathbf v_i})^{\oplus \mathsf g_i}\times \mathbb C\;,
\end{align*}
where $Q'$ is the sub-quiver of $Q$ obtained by deleting the node $i$, and $\mathsf g_i$ is the number of edge loops at the node $i$ in the quiver $Q$, and
\begin{align*}
    \mathbf v'_j=\mathbf v_j\;,\quad \mathbf w'_j=\mathbf w_j-\mathsf C_{ij}\mathbf v_i\;.
\end{align*}
It is easy to see that
\begin{align*}
    \mathbf w'_j-\sum_{l\in Q'_0}\mathsf C'_{j,l}\mathbf v'_l= \mathbf w_j-\sum_{l\in Q_0}\mathsf C_{j,l}\mathbf v_l, \quad\forall j\in Q'_0\;,
\end{align*}
in particular $\mathbf w'-\mathsf C'\mathbf v'\in \mathbb Z^{Q'_0}_{\ge 0}$, where $\mathsf C'$ is the Cartan matrix for the new quiver $Q'$. 

\bigskip We can slightly generalize the above argument to the situation when there exists a node $i\in Q_0$ such that $\mathbf w_i\ge \mathbf v_i$. In this case we split the framing vector space into $\mathbb C^{\mathbf w_i}=\mathbb C^{\mathbf v_i}\oplus\mathbb C^{\mathbf w_i-\mathbf v_i}$, and consider the open locus $\mathcal U_i\subset \widetilde\CM(Q,\mathbf v,\mathbf w)$ where the framing map $I_i:\mathbb C^{\mathbf v_i}\to V_i$ is an isomorphism. Then we have a natural isomorphism
\begin{align*}
    \mathcal U_i\cong \widetilde\CM(Q',\mathbf v',\mathbf w')\times T^*\mathrm{End}(\mathbb C^{\mathbf v_i})^{\oplus \mathsf g_i}\times T^*\mathrm{Hom}(\mathbb C^{\mathbf v_i},\mathbb C^{\mathbf w_i-\mathbf v_i})\times \mathbb C\;,
\end{align*}
where $Q'$ is the sub-quiver of $Q$ obtained by deleting the node $i$, and $\mathsf g_i$ is the number of edge loops at the node $i$ in the quiver $Q$, and
\begin{align*}
    \mathbf v'_j=\mathbf v_j\;,\quad \mathbf w'_j=\mathbf w_j-\mathsf C_{ij}\mathbf v_i\;.
\end{align*}
In this case we also have $\mathbf w'_j-\sum_{l\in Q'_0}\mathsf C'_{j,l}\mathbf v'_l= \mathbf w_j-\sum_{l\in Q_0}\mathsf C_{j,l}\mathbf v_l, \;\forall j\in Q'_0$, in particular $\mathbf w'-\mathsf C'\mathbf v'\in \mathbb Z^{Q'_0}_{\ge 0}$. 

\subsection{The building block}\label{subsec:building block}

We start the vertex algebra analog of the aforementioned geometry observation by looking into the basic building block in the above construction. Namely we pick a node $i\in Q_0$ for which the reduction will be performed and then re-arrange the neighboring nodes of $i$ into the LHS of \cref{fig:reduction building block}. 

We will show in the following that the relative BRST reduction of the gauge node on a certain localization of the LHS will produce the RHS of \cref{fig:reduction building block}.

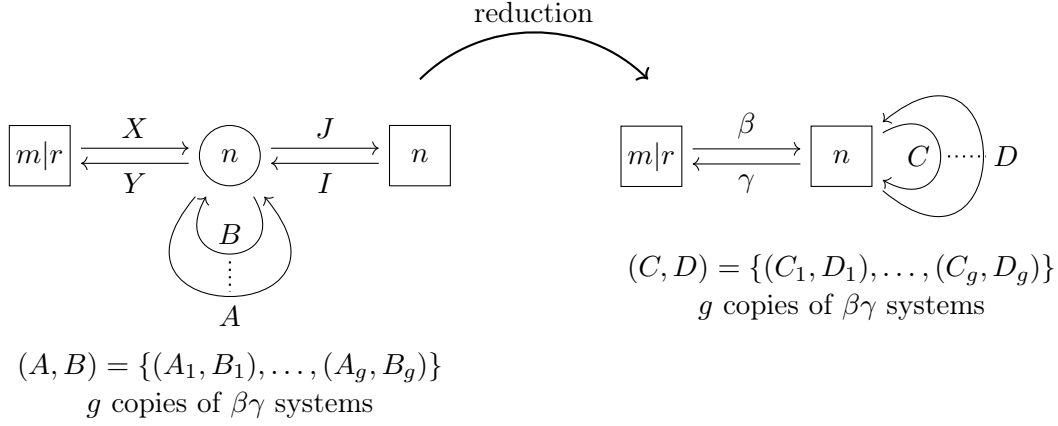
\begin{figure}

    \begin{tikzpicture}

    \node (leftfigure) at (0,0){
    \begin{tikzpicture}
    \node(L) at (-2.5,0) {$\rectangled{$m|r$}$};
    \node(M) at (0,0) {$\circled{$n$}$};
    \node(R) at (2.5,0) {$\rectangled{$n$}$};
    \draw[->] ([yshift=0.1cm]L.east) to node[above]{$X$} ([yshift=0.1cm]M.west);
    \draw[->] ([yshift=-0.1cm]M.west) to node[below]{$Y$} ([yshift=-0.1cm]L.east);
    \draw[->] ([yshift=0.1cm]M.east) to node[above]{$J$} ([yshift=0.1cm]R.west);
    \draw[->] ([yshift=-0.1cm]R.west) to node[below]{$I$} ([yshift=-0.1cm]M.east);
    \draw[thick,dotted] (0,-1.4) -- (0,-1.8);
    \draw[<-] (M) to[in=300, out=240, looseness=4.8] node[above]{$B$}(M);
    \draw[->] (M) to[in=310, out=230, looseness=6.5] node[below]{$A$} (M);
    \node(bottomtext) at (0,-2.8) {$(A, B) = \{(A_1, B_1), \dots, (A_g, B_g)\}$};
    \node(bottomtext-2) at (0,-3.3) {$g$ copies of $\beta \gamma$ systems};
    \end{tikzpicture}
    };

    \draw[->,thick] (2.5,2.5) to [bend left=45] node[above]{reduction} (5.5, 2.5);
     
    \node (rightfigure) at (8,1.3){
    
    \begin{tikzpicture}
    \node(L) at (-2.5,0) {$\rectangled{$m|r$}$};
    \node(M) at (0,0) {$\rectangled{$n$}$};
    \draw[->] ([yshift=0.1cm]L.east) to node[above]{$\beta$} ([yshift=0.1cm]M.west);
    \draw[->] ([yshift=-0.1cm]M.west) to node[below]{$\gamma$} ([yshift=-0.1cm]L.east);
    \draw[thick,dotted] (1.4,0) -- (1.9,0);
    \draw[<-] (M) to[in=390, out=330, looseness=4.8] node[left]{$C$}(M);
    \draw[->] (M) to[in=400, out=320, looseness=6.5] node[right]{$D$} (M);
    \node(bottomtext) at (0,-1.5) {$(C, D) = \{(C_1, D_1), \dots, (C_g, D_g)\}$};
    \node(bottomtext-2) at (0,-2.0) {$g$ copies of $\beta \gamma$ systems};
    \end{tikzpicture}
    };
    
    \end{tikzpicture}
    \caption{A reduction procedure for the quiver, explained locally around the middle vertex. Here we have $m-r=(1-2g)n$.}
    \label{fig:reduction building block}
\end{figure}

Denote the $\beta\gamma\mathsf b\mathsf c$-Heisenberg system on the LHS (before BRST reduction) by $D^{\mathrm{ch}}(\mathfrak R)$. The nontrivial OPEs are 
\begin{alignat*}{2}
    &X^b_j(z)Y^i_a(w) \sim \frac{\delta^b_a\delta^i_j}{z-w}\;,  \quad 
    &&I^b_\nu(z)J^\mu_a(w)\sim \frac{\delta^b_a\delta^\mu_\nu}{z-w}\;,\\
    \quad 
    &A^a_{s,b}(z)B^c_{t,d}(w)\sim \frac{\delta^a_d\delta^c_b\delta_{s,t}}{z-w}\;, \quad
    &&h(z)h(w)\sim \frac{2-2\mathsf g}{(z-w)^2}\;.
\end{alignat*}
Here $(X,Y)$ contain both $\beta\gamma$ systems and $\mathsf b\mathsf c$ systems, and they are distinguished by the OPE
\begin{align*}
    Y^i_a(z)X^b_j(w)\sim (-1)^{1+|i|}\frac{\delta^b_a\delta^i_j}{z-w}\;
\end{align*}
where $|i|$ is the parity of $i$, which equals to $0$ if $1\le i\le m$, and it equals to $1$ if $m<i\le m+r$. The chiral extended moment map reads:
\begin{align}
(\widetilde\mu_{\mathrm{ch}})^a_b=-Y^i_bX^a_i-I^a_\rho J^\rho_b-A^a_c B^c_b+B^a_dA^d_b-h\delta^a_b\;.
\end{align}
Let us localize the vertex algebra $D^{\mathrm{ch}}(\mathfrak R)$ by allowing $I$ to be invertible, i.e. we add $\det(I)^{-1}$ to the vertex algebra\footnote{$\det(I)$ is defined as the normal order product $\epsilon_{a_1a_2\cdots a_n}:I^{a_1}_1I^{a_2}_2\cdots I^{a_n}_n:$.}, and denote the resulting vertex algebra by $D^{\mathrm{ch}}(\mathfrak R)[I^{-1}]$. Then 
\begin{align*}
    D^{\mathrm{ch}}(\mathfrak R)[I^{-1}]\cong \mathbb C[\det(I)^{-1}_{(-1)},X_{(-n)},Y_{(-n)},I_{(-n)},J_{(-n)},A_{(-n)},B_{(-n)},h_{(-n)}\:|\:n\in \mathbb Z_{\ge 1}]|0\rangle \;.
\end{align*}
Define
\begin{align}
    \beta:=I^{-1}X\;, \quad \gamma:= YI\;,\quad C:=I^{-1}AI\;,\quad D:=I^{-1}BI\;,
\end{align}
where we contract indices when multiplying matrix-valued operators. We also define an element
\begin{align}
    \mathring{\alpha}:=h+\frac{2-2\mathsf g}{(1-2\mathsf g)n}\left(\mathrm{Tr}(YX)-\mathrm{Tr}(\gamma\beta)\right)\;.
\end{align}
We can also rewrite $\mathring{\alpha}$ as
\begin{align}
    \mathring{\alpha}=h+(2-2\mathsf g)\mathrm{Tr}(I^{-1}\partial I)\;.
\end{align}
Thus $\mathring{\alpha}$ is a Heisenberg field of level $2-2\mathsf g$. When $I$ is invertible, we can use $(\beta,\gamma,C,D,\mathring{\alpha})$ together with $\widetilde\mu_{\mathrm{ch}}$ to represent $(X,Y,I,J,A,B,h)$, and vice versa. Therefore we have an isomorphism
\begin{multline}
    D^{\mathrm{ch}}(\mathfrak R)[I^{-1}]\cong \mathbb C[\beta_{(-n)},\gamma_{(-n)},C_{(-n)},D_{(-n)},\mathring{\alpha}_{(-n)}\:|\:n\in \mathbb Z_{\ge 1}]\otimes\\
    \otimes\mathbb C[\det(I)^{-1}_{(-1)},I_{(-n)},\widetilde\mu_{\mathrm{ch},(-n)}\:|\:n\in \mathbb Z_{\ge 1}]|0\rangle\;.
\end{multline}
We notice that $\mathbb C[\beta_{(-n)},\gamma_{(-n)},C_{(-n)},D_{(-n)},\mathring{\alpha}_{(-n)}\:|\:n\in \mathbb Z_{\ge 1}]|0\rangle$ is exactly the $\beta\gamma\mathsf b\mathsf c$-Heisenberg system on the RHS of \cref{fig:reduction building block}, and we denote it by $D^{\mathrm{ch}}(\mathfrak R')$. We note that $D^{\mathrm{ch}}(\mathfrak R')$ commutes with $\mathbb C[\det(I)^{-1}_{(-1)},I_{(-n)},\widetilde\mu_{\mathrm{ch},(-n)}\:|\:n\in \mathbb Z_{\ge 1}]|0\rangle$, and the classical limit of the latter is the coordinate ring of $J_\infty T^*\mathrm{GL}_n$. Therefore we have
\begin{align}\label{eq:BRST cohomology of building block}
    H^{\infty/2+i}_{\mathrm{BRST}}(\mathfrak{gl}_n,D^{\mathrm{ch}}(\mathfrak R)[I^{-1}])=\begin{cases}
        D^{\mathrm{ch}}(\mathfrak R')\;,& i=0\;,\\
        0\;, & i\neq 0\;.
    \end{cases}
\end{align}
The natural embedding $D^{\mathrm{ch}}(\mathfrak R)\hookrightarrow D^{\mathrm{ch}}(\mathfrak R)[I^{-1}]$ of vertex superalgebras induces a natural map between relative BRST complexes: $(C(\mathfrak{gl}_n,D^{\mathrm{ch}}(\mathfrak R)),\mathcal{Q})\to(C(\mathfrak{gl}_n,D^{\mathrm{ch}}(\mathfrak R)[I^{-1}]),\mathcal{Q})$, and taking degree zero cohomology we get a vertex superalgebra map
\begin{align*}
    \nu: H^{\infty/2+0}_{\mathrm{BRST}}(\mathfrak{gl}_n,D^{\mathrm{ch}}(\mathfrak R))\to D^{\mathrm{ch}}(\mathfrak R')\;.
\end{align*}
Let $f(x_1,y_1;\cdots;x_{\mathsf g},y_{\mathsf g})$ be a monomial in $2\mathsf g$ variables $x_1,y_1,\cdots,x_{\mathsf g},y_{\mathsf g}$ such that each pair $(x_i,y_i)$ can not appear in $f$ simultaneously. For example $x_1x_2y_1$ is not allowed. Then it easy to see that $Yf(A_1,B_1;\cdots;A_{\mathsf g},B_{\mathsf g})I$, $Jf(A_1,B_1;\cdots;A_{\mathsf g},B_{\mathsf g})X$, $\mathrm{Tr}(f(A_1,B_1;\cdots;A_{\mathsf g},B_{\mathsf g}))$, and $\mathrm{Tr}(A_iB_i)$ are BRST closed, so they represent cohomology classes in $H^{\infty/2+0}_{\mathrm{BRST}}(\mathfrak{gl}_n,D^{\mathrm{ch}}(\mathfrak R))$. Then we have the following.

\begin{lem}
\begin{align}\label{eq:YfI}
    \nu(Yf(A_1,B_1;\cdots;A_{\mathsf g},B_{\mathsf g})I)&=\gamma f(C_1,D_1;\cdots;C_{\mathsf g},D_{\mathsf g})\;,
    \\
    \label{eq:JfX}
    \nu(Jf(A_1,B_1;\cdots;A_{\mathsf g},B_{\mathsf g})X)&=M f(C_1,D_1;\cdots;C_{\mathsf g},D_{\mathsf g})\beta\;,
    \\
    \label{eq:tr(f)}
    \nu(\mathrm{Tr}(f(A_1,B_1;\cdots;A_{\mathsf g},B_{\mathsf g})))&=\mathrm{Tr}(f(C_1,D_1;\cdots;C_{\mathsf g},D_{\mathsf g}))\;,
     \\
    \label{eq:tr(AB)}
    \nu(\mathrm{Tr}(A_iB_i))&=\mathrm{Tr}(C_iD_i)\;,
\end{align}
where $M^a_b:=-\gamma^i_b\beta^a_i-C^a_cD^c_b+D^a_dC^d_b-\mathring{\alpha}\delta^a_b$, and multiplications are understood as normal ordered products of matrix-valued operators. Furthermore, we have the following map of BRST-closed elements 
\begin{align}\label{eq:YfX}
    &\nu\left(\left(Y^i_bX^a_j-\frac{\delta^i_j}{(1-2\mathsf g)n}Y^k_bX^a_k\right)f(A_1,B_1;\cdots;A_{\mathsf g},B_{\mathsf g})^b_a\right) \nonumber\\
    &\qquad \qquad =\left(\gamma^i_b\beta^a_j-\frac{\delta^i_j}{(1-2\mathsf g)n}\gamma^k_b\beta^a_k\right)f(C_1,D_1;\cdots;C_{\mathsf g},D_{\mathsf g})^b_a\;, \\
     \label{eq:JI}
    &\nu\left(JI-\frac{\mathrm{Tr}(JI)}{n}\mathrm{id} \right)
    =-\beta \gamma+\frac{\mathrm{Tr}(\beta\gamma)}{n}\mathrm{id} -CD+DC\;,
    \\
    \label{eq:tr(XY)-tr(IJ)}
    &\nu\left(\frac{1}{1-2\mathsf g}\mathrm{Tr}(YX)-\mathrm{Tr}(JI)\right)=n\mathring{\alpha}+\frac{2-2\mathsf g}{1-2\mathsf g}\mathrm{Tr}(\beta\gamma)\;.
\end{align}
\end{lem}

\begin{proof}
The equations \eqref{eq:YfI}, \eqref{eq:tr(f)}, \eqref{eq:tr(AB)}, \eqref{eq:YfX} are straightforward to derive. Then equation \eqref{eq:tr(XY)-tr(IJ)} follows from the trace of the moment map $\mathrm{Tr}(\widetilde\mu_{\mathrm{ch}})$, which is zero in the BRST cohomology. To derive \eqref{eq:JI}, let us consider the operator
\begin{align*}
\mathcal O^{\rho}_\sigma:=-(\widetilde\mu_{\mathrm{ch}})^a_{b,(-1)}(I^{-1})^{\rho}_{a,(-1)} I^b_\sigma\;.
\end{align*}
Since for every pair of fixed $\rho,\sigma$ the operator $(I^{-1})^{\rho}_{a,(-1)} I^b_\sigma$ transforms in the adjoint representation with respect to the action of $\widetilde\mu_{\mathrm{ch}}$, then according to \cref{subsec:Some BRST exact elements}, $\mathcal O^{\rho}_\sigma$ vanishes in the BRST cohomology. On the other hand we can expand $\mathcal O^{\rho}_\sigma$ as follows:
\begin{align*}
\mathcal O^{\rho}_\sigma&=(I^{-1})^\rho_a Y^i_b X^a_iI^b_\sigma+(I^{-1}ABI-I^{-1}BAI)^{\rho}_\sigma+(JI)^{\rho}_\sigma-I^a_{\tau,(-2)}J^{\tau}_{b,(0)}(I^{-1})^{\rho}_{a,(-1)} I^b_\sigma+h\delta^{\rho}_\sigma\\
&=\gamma^i_\sigma\beta^\rho_i+(CD-DC+JI)^{\rho}_\sigma+h\delta^{\rho}_\sigma+(1-2\mathsf g)\mathrm{Tr}(I^{-1}\partial I)\delta^{\rho}_\sigma \;.
\end{align*}
Using the equation $\mathcal O-\frac{\mathrm{Tr}(\mathcal O)}{n}\mathrm{id}\equiv 0$ in the BRST cohomology, we find \eqref{eq:JI}. We note that \eqref{eq:JI} implies that 
\begin{align*}
    M=\nu\left(JI-\frac{1}{(1-2\mathsf g)n}\mathrm{Tr}(YX)\mathrm{id}\right)+\frac{1}{(1-2\mathsf g)n}\mathrm{Tr}(\gamma\beta)=JI-\mathrm{Tr}(I^{-1}\partial I)\mathrm{id}\;.
\end{align*}
Thus the RHS of \eqref{eq:JfX} equals to
\begin{align*}
    (JI-\mathrm{Tr}(I^{-1}\partial I)\mathrm{id})_{(-1)}I^{-1}f(A_1,B_1;\cdots;A_{\mathsf g},B_{\mathsf g})X\;,
\end{align*}
which equals to LHS of \eqref{eq:JfX}.
\end{proof}

The stress-energy operator $T_{\mathrm{total}}$ is the sum of stress-energy operators for the $\beta\gamma\mathsf b\mathsf c$ systems formed by $(X,Y,I,J,A,B)$, the Heisenberg field $h$, and the $\mathsf b\mathsf c$ ghost system. Namely
\begin{align*}
    T_{\mathrm{total}}&=\frac{1}{2}\mathrm{Tr}(\partial Y X-Y\partial X)+\frac{1}{2}\mathrm{Tr}(\partial J I-J\partial I)+\frac{1}{2}\mathrm{Tr}(\partial B A-B\partial A)\\
    &\quad +\frac{1}{4-4\mathsf g}h^2-\mathrm{Tr}(\mathsf b\partial\mathsf c)\;.
\end{align*}
We note that when $\mathsf g=1$, the above formula does not make sense, and we need to remove the Heisenberg field $h$ from the formula, since $h$ is commutative. 

From our previous discussion in  \cref{sec:BRST_operators}, the total stress-energy operator $T_{\mathrm{total}}$ is BRST closed, thus it represents an element in $H^{\infty/2+0}_{\mathrm{BRST}}(\mathfrak{gl}_n,D^{\mathrm{ch}}(\mathfrak R))$. 
\begin{lem}
Under the reduction map $\nu$, 
\begin{align}\label{eq:stress-energy operator under reduction}
\nu(T_{\mathrm{total}})
=-\mathrm{Tr}(\gamma\partial \beta)+\frac{1}{2}\mathrm{Tr}(\partial D C-C\partial D)+\begin{cases}
    \frac{1}{4-4\mathsf g}\mathring{\alpha}^2-\frac{n}{2}\partial\mathring{\alpha}, & \mathsf g\neq 1\\
    0, & \mathsf g=1
\end{cases}\;.
\end{align}
\end{lem}

\begin{proof}
To prove \eqref{eq:stress-energy operator under reduction}, we notice that $\beta,\gamma,C,D$ are primary fields with respect to $T_{\mathrm{total}}$ with dimensions $0,1,1/2,1/2$ respectively; and we have OPE
\begin{align*}
    T_{\mathrm{total}}(z)\mathring{\alpha}(w)\sim \frac{(2-2\mathsf g)n}{(z-w)^3}+\frac{\mathring{\alpha}}{(z-w)^2}+\frac{\partial\mathring{\alpha}}{z-w}.
\end{align*}
This implies that $\nu(T_{\mathrm{total}})-(T_{\beta\gamma}+T_{CD}+T_{\mathring{\alpha}})$ commutes with $\beta,\gamma,C,D,\mathring{\alpha}$. Since the center of $D^{\mathrm{ch}}(\mathfrak R')$ is spanned by the vacuum, we must have $\nu(T_{\mathrm{total}})=T_{\beta\gamma}+T_{CD}+T_{\mathring{\alpha}}+\lambda|0\rangle$ for some $\lambda\in \mathbb C$. On the other hand, under the grading on $C(\mathfrak{gl}_n,D^{\mathrm{ch}}(\mathfrak R))$ induced by $L_{0,\mathrm{total}}=\oint zT_{\mathrm{total}}(z)dz$, both $T_{\mathrm{total}}$ and $T_{\beta\gamma}+T_{CD}+T_{\mathring{\alpha}}$ are of degree $2$, whereas $\deg |0\rangle=0$, thus $\lambda=0$ and \eqref{eq:stress-energy operator under reduction} is proven.
\end{proof}

We note that the special case when $m=n=1,r=\mathsf g=0$, the quiver reduction leads to the Wakimoto realization of $V^{-1}(\mathfrak{sl}_2)$. Indeed, the assignment of strong generators 
\begin{align}
    e\mapsto YI\;,\quad  h\mapsto YX-JI\;,\quad f\mapsto JX\;,
\end{align}
gives a vertex algebra map $\sigma: V^{-1}(\mathfrak{sl}_2)\to H^{\infty/2+0}_{\mathrm{BRST}}(\mathfrak{gl}_n,D^{\mathrm{ch}}(\mathfrak R))$. Let $\mathcal{H}^k$ be the Heisenberg algebra associated to the one dimensional vector space $\mathbb C$ with inner form $(x,y)=kxy$, then $D^{\mathrm{ch}}(\mathfrak R')\cong D^{\mathrm{ch}}(T^*\mathbb C)\otimes \mathcal{H}^2$. The composition $\nu\circ\sigma: V^{-1}(\mathfrak{sl}_2)\to D^{\mathrm{ch}}(\mathfrak R')=D^{\mathrm{ch}}(T^*\mathbb C)\otimes \mathcal{H}^2$ maps the strong generators as follows
\begin{align}
    \nu\circ\sigma(e)=\gamma\;,\quad \nu\circ\sigma(h)=2\gamma\beta+\mathring{\alpha}\;,\quad \nu\circ\sigma(f)=-\beta^2\gamma+\partial\beta-\mathring{\alpha}\beta\;.
\end{align}
$\nu\circ\sigma$ agrees with the Wakimoto realization of $V^{-1}(\mathfrak{sl}_2)$.

\subsection{Reduction for vertex (super)algebras associated to quivers}\label{subsec:reduction for VOA}

Consider a quiver $Q$ with gauge dimension $\mathbf v$ and framing dimension $\mathbf w$ such that $\mathbf d:=\mathbf w-\mathsf C\mathbf v\in \mathbb Z^{Q_0}_{\ge 0}$. Suppose that there is a node $i\in Q_0$ such that $\mathbf w_i\ge \mathbf v_i$, then we decompose the framing vector space into $\mathbb C^{\mathbf w_i}=\mathbb C^{\mathbf v_i}\oplus\mathbb C^{\mathbf w_i-\mathbf v_i}$. Without loss of generality, we assume that $i$ is a sink, i.e. there is no arrow in $Q_1$ going out from $i$, then we recognize the neighborhood of the node $i$ as the situation in the previous subsection, see \cref{fig:reduction algorithm}.

\begin{figure}
\begin{center}
    \begin{tikzpicture}

    \node (leftfigure) at (0,0){
    \begin{tikzpicture}
    \node(C) at (0,0) {$\circled{$v_i$}$};
    \node(1) at (-2.16,-1.25) {$\circled{$v_1$}$};
    \node(j) at (-2.16,1.25) {$\circled{$v_j$}$};
    \node(t) at (0,2.5) {$\circled{$v_t$}$};
    \node(w) at (0,-2.5) {$\rectangled{$v_i$}$};
    \node(w2) at (2.16,-1.25) {$\rectangled{$w_i-v_i | d_i$}$};
    \draw[->] ([yshift=0.1cm]1.east) to node[above]{$X_1$} ([yshift=0.1cm]C.south west);
    \draw[->] ([yshift=-0.1cm]C.south west) to node[below]{$Y_1$} ([yshift=-0.1cm]1.east);
    \draw[->] ([yshift=0.1cm]j.east) to node[above]{$X_j$} ([yshift=0.1cm]C.north west);
    \draw[->] ([yshift=-0.1cm]C.north west) to node[below]{$Y_j$} ([yshift=-0.1cm]j.east);
    \draw[->] ([xshift=0.1cm]t.south) to node[right]{$X_t$} ([xshift=0.1cm]C.north);
    \draw[->] ([xshift=-0.1cm]C.north) to node[left]{$Y_t$} ([xshift=-0.1cm]t.south);
    \draw[->] ([xshift=-0.1cm]w.north) to node[left]{$I$} ([xshift=-0.1cm]C.south);
    \draw[->] ([xshift=0.1cm]C.south) to node[right]{$J$} ([xshift=0.1cm]w.north);
    \draw[->] ([yshift=-0.1cm]w2.west) to node[above,yshift=0.3cm]{$\tilde{J}$} ([yshift=-0.1cm]C.south east);
    \draw[->] ([yshift=0.1cm]C.south east) to node[below,yshift=-0.3cm]{$\tilde{I}$} ([yshift=0.1cm]w2.west);
    \draw[dotted,thick] (1.2, 0.6) -- (1.4, 0.7);
    \draw[bend left,dotted,thick] (1) to (j);
    \draw[bend left,dotted,thick] (j) to (t);
    \draw[<-] (M) to[in=10, out=50, looseness=5] node[left,yshift=-0.2cm]{$A_1$}(M);
    \draw[<-] (M) to[in=60, out=0, looseness=8] node[right]{$B_g$} (M);
    \node(h1) at (-2.85,-1.2) {$h_1$};
    \node(hj) at (-2.85,1.25) {$h_j$};
    \node(ht) at (0,3.2) {$h_t$};
    \node(h) at (-0.65,0.1) {$h$};
    \end{tikzpicture}
    };
    
    \draw[->,thick] (1.5,2.5) to [bend left=45] node[above]{reduction} (4.5, 2.5);
     
    \node (rightfigure) at (7,0){
    
    \begin{tikzpicture}
    \node(C) at (0,0) {$\rectangled{$v_i$}$};
    \node(1) at (-2.16,-1.25) {$\circled{$v_1$}$};
    \node(j) at (-2.16,1.25) {$\circled{$v_j$}$};
    \node(t) at (0,2.5) {$\circled{$v_t$}$};
    \node(w) at (0,-2.5) {$\rectangled{$w_i-v_i | d_i$}$};
    \draw[->] ([yshift=0.1cm]1.east) to node[above]{$\beta_1$} ([yshift=0.1cm]C.south west);
    \draw[->] ([yshift=-0.1cm]C.south west) to node[below]{$\gamma_1$} ([yshift=-0.1cm]1.east);
    \draw[->] ([yshift=0.1cm]j.east) to node[above]{$\beta_j$} ([yshift=0.1cm]C.north west);
    \draw[->] ([yshift=-0.1cm]C.north west) to node[below]{$\gamma_j$} ([yshift=-0.1cm]j.east);
    \draw[->] ([xshift=0.1cm]t.south) to node[right]{$\beta_t$} ([xshift=0.1cm]C.north);
    \draw[->] ([xshift=-0.1cm]C.north) to node[left]{$\gamma_t$} ([xshift=-0.1cm]t.south);
    \draw[->] ([xshift=-0.1cm]w.north) to node[left]{$\tilde{\beta}$} ([xshift=-0.1cm]C.south);
    \draw[->] ([xshift=0.1cm]C.south) to node[right]{$\tilde{\gamma}$} ([xshift=0.1cm]w.north);
    \draw[dotted,thick] (1.2, 0.6) -- (1.4, 0.7);
    \draw[bend left,dotted,thick] (1) to (j);
    \draw[bend left,dotted,thick] (j) to (t);
    \draw[<-] (M) to[in=10, out=50, looseness=5] node[left,yshift=-0.2cm]{$C_1$}(M);
    \draw[<-] (M) to[in=60, out=0, looseness=8] node[right]{$D_g$} (M);
    \node(h1) at (-2.85,-1.2) {$h_1'$};
    \node(hj) at (-2.85,1.25) {$h_j'$};
    \node(ht) at (0,3.2) {$h_t'$};
    \node(h) at (-0.65,0.1) {$\alpha$};
    \end{tikzpicture}
    };
    
    \end{tikzpicture}
    \caption{A general reduction procedure for a quiver diagram.}
    \label{fig:reduction algorithm}
\end{center}
\end{figure}
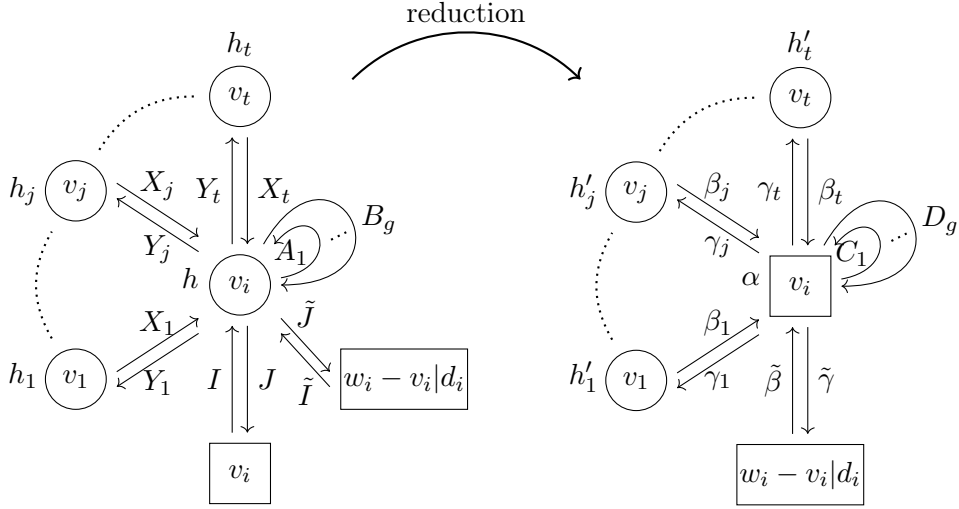

The new quiver $Q'$ on the RHS of \cref{fig:reduction algorithm} is the sub-quiver of $Q$ obtained by deleting the node $i$, and $\mathbf v'_j=\mathbf v_j$, $\mathbf w'_j=\mathbf w_j-\mathsf C_{ij}\mathbf v_i$. 

The $\beta\gamma\mathsf b\mathsf c$ system on the LHS (resp.\ RHS) of \cref{fig:reduction algorithm} is denoted by $\CVs(Q,\mathbf v,\mathbf w)$ (resp.\  $\CVs(Q',\mathbf v',\mathbf w')$). Moreover, if we add the inverse of $I$ into $\CVs(Q,\mathbf v,\mathbf w)$, then according to \eqref{eq:BRST cohomology of building block} we have
\begin{multline}\label{eq:reduction at node i}
    H^{\infty/2+0}_{\mathrm{BRST}}(\mathfrak{gl}(\mathbf v_i),\CVs(Q,\mathbf v,\mathbf w)[I^{-1}])=\left(\CVs(Q',\mathbf v',\mathbf w')\rtimes \mathcal{H}^{\frac{\mathbf w_i-\mathbf d_i}{\mathbf v_i}}\right)\otimes\\
    \otimes D^{\mathrm{ch}}(T^*\mathrm{End}(\mathbb C^{\mathbf v_i})^{\oplus\mathsf g_i})\otimes D^{\mathrm{ch}}(T^*\mathrm{Hom}(\mathbb C^{\mathbf w_i-\mathbf v_i|\mathbf d_i},\mathbb C^{\mathbf v_i}))\;,
\end{multline}
and $ H^{\infty/2+n}_{\mathrm{BRST}}(\mathfrak{gl}(\mathbf v_i),\CVs(Q,\mathbf v,\mathbf w)[I^{-1}])=0$ for $n\neq 0$. The fields on the RHS of \eqref{eq:reduction at node i} are given by
\begin{equation}\label{eq:fields under the reduction}
\begin{split}
    &\CVs(Q',\mathbf v',\mathbf w')\ni\Big\{\beta_j=I^{-1}X_j, \; \gamma_j= Y_jI ,\;h_j'=h_j-\frac{1}{\mathbf v_j}\mathrm{Tr}(Y_jX_j)+\frac{1}{\mathbf v_j}\mathrm{Tr}(\gamma_j\beta_j)\Big\}\;,\\
    &\mathcal{H}^{\frac{\mathbf w_i-\mathbf d_i}{\mathbf v_i}}\ni \Big\{\alpha=h_i+(2-2\mathsf g_i)\mathrm{Tr}(I^{-1}\partial I)+\frac{1}{\mathbf v_i}\sum_{j\to i\in Q_1}\mathrm{Tr}(\gamma_j\beta_j)\Big\}\;,\\
    &D^{\mathrm{ch}}(T^*\mathrm{End}(\mathbb C^{\mathbf v_i})^{\oplus\mathsf g_i})\ni\{C=I^{-1}AI,\; D=I^{-1}BI \}\;,\\
    &D^{\mathrm{ch}}(T^*\mathrm{Hom}(\mathbb C^{\mathbf w_i-\mathbf v_i|\mathbf d_i},\mathbb C^{\mathbf v_i}))\ni \{\widetilde\beta=I^{-1}\widetilde I, \; 
    \widetilde\gamma=\widetilde{J}I \}\;.
\end{split}
\end{equation}
The semi-direct product in \eqref{eq:reduction at node i} means that the Heisenberg algebra $\mathcal{H}^{\frac{\mathbf w_i-\mathbf d_i}{\mathbf v_i}}$ acts on $\CVs(Q',\mathbf v',\mathbf w')$. The action is given by the following OPEs
\begin{align*}
    \alpha(z)h'_j(w)\sim \frac{\mathsf C_{i,j}}{(z-w)^2}\;,\quad \alpha(z)\beta_j(w)\sim \frac{-\mathbf v_i^{-1}\beta_j(w)}{z-w}\;,\quad \alpha(z)\gamma_j(w)\sim \frac{\mathbf v_i^{-1}\gamma_j(w)}{z-w}\;.
\end{align*}
From now on, we denote $D^{\mathrm{ch}}(T^*\mathrm{End}(\mathbb C^{\mathbf v_i})^{\oplus\mathsf g_i})\otimes D^{\mathrm{ch}}(T^*\mathrm{Hom}(\mathbb C^{\mathbf w_i-\mathbf v_i|\mathbf d_i},\mathbb C^{\mathbf v_i}))$ by $D^{\mathrm{ch}}(\mathcal E_i)$.

The current algebra $V^{-2\mathbf v_j}(\mathfrak{gl}(\mathbf v_j))$ corresponding to the node $j\neq i$ is automatically BRST-closed with respect to the $\mathfrak{gl}(\mathbf v_i)$-cohomology, so there is a natural vertex algebra map 
\begin{align*}
    \widetilde\mu'_{\mathrm{ch},j}&:V^{-2\mathbf v_j}(\mathfrak{gl}(\mathbf v_j))\to \CVs(Q',\mathbf v',\mathbf w')\;,\\
    \widetilde\mu'_{\mathrm{ch},j}(E^a_b)&=(Y_j)^a_c(X_j)^c_b-h_j\delta^a_b+\sum\text{terms which are not neighboring to node }i\\
    &=(\gamma_j)^a_c(\beta_j)^c_b-h'_j\delta^a_b+\sum\text{terms which are not neighboring to node }i \;.
\end{align*}
So we see that $\widetilde\mu'_{\mathrm{ch},j}$ agrees with the extended chiral moment map corresponding to the node $j$ in the new quiver $Q'$. Moreover, $\widetilde\mu'_{\mathrm{ch},j}$ commutes with the action of $\mathcal{H}^{\frac{\mathbf w_i-\mathbf d_i}{\mathbf v_i}}$, in fact our choice of the Heisenberg field $\alpha$ instead of $\mathring{\alpha}$ is precisely of this reason. It follows that there is a canonical isomorphism
\begin{multline*}
H^{\infty/2+\bullet}_{\mathrm{BRST}}\left(\bigoplus_{j\neq i}\mathfrak{gl}(\mathbf v_j),H^{\infty/2+0}_{\mathrm{BRST}}(\mathfrak{gl}(\mathbf v_i),\CVs(Q,\mathbf v,\mathbf w)[I^{-1}])\right)\cong \\
    \left(H^{\infty/2+\bullet}_{\mathrm{BRST}}(\mathfrak{gl}(\mathbf v'),\CVs(Q',\mathbf v',\mathbf w'))\rtimes \mathcal{H}^{\frac{\mathbf w_i-\mathbf d_i}{\mathbf v_i}}\right)\otimes D^{\mathrm{ch}}(\mathcal E_i)\;.
\end{multline*}
Consider the relative BRST complex $(C(\mathfrak{gl}(\mathbf v), \CVs(Q,\mathbf v,\mathbf w)[I^{-1}]),\mathcal{Q})$, and we decompose the BRST differential 
\begin{align*}
\mathcal{Q}=\mathcal{Q}^i+\mathcal{Q}^{\hat{i}}
\end{align*}
into the differential $\mathcal{Q}^i$ corresponding to node $i$ and the differential $\mathcal{Q}^{\hat{i}}=\sum_{j\neq i}\mathcal{Q}^j$ corresponding to nodes other than $i$. $\mathcal{Q}^i$ anti-commutes with $\mathcal{Q}^{\hat{i}}$, then we have a spectral sequence with second page
\begin{align*}
    E^{p,q}_2=H^q(H^p(C(\mathfrak{gl}(\mathbf v), \CVs(Q,\mathbf v,\mathbf w)[I^{-1}]),\mathcal{Q}^i),\mathcal{Q}^{\hat{i}})\;.
\end{align*}
According to our previous discussion in \eqref{eq:BRST cohomology of building block}, we have 
\begin{align*}
    H^0(C(\mathfrak{gl}(\mathbf v), \CVs(Q,\mathbf v,\mathbf w)[I^{-1}]),\mathcal{Q}^i)\cong
    \left(C(\mathfrak{gl}(\mathbf v'),\CVs(Q',\mathbf v',\mathbf w'))\rtimes \mathcal{H}^{\frac{\mathbf w_i-\mathbf d_i}{\mathbf v_i}}\right)
    \otimes D^{\mathrm{ch}}(\mathcal E_i)\;,
\end{align*}
and $H^p(C(\mathfrak{gl}(\mathbf v), \CVs(Q,\mathbf v,\mathbf w)[I^{-1}]),\mathcal{Q}^i)$ vanishes for $p\neq 0$. It follows that $E^{p,q}$ degenerates at the second page, and it converges to $H^{p+q}(C(\mathfrak{gl}(\mathbf v), \CVs(Q,\mathbf v,\mathbf w)[I^{-1}]),\mathcal{Q})$. Thus we have an isomorphism
\begin{multline}
    H^{\infty/2+\bullet}_{\mathrm{BRST}}(\mathfrak{gl}(\mathbf v),\CVs(Q,\mathbf v,\mathbf w)[I^{-1}])\cong \\
    \left(H^{\infty/2+\bullet}_{\mathrm{BRST}}(\mathfrak{gl}(\mathbf v'),\CVs(Q',\mathbf v',\mathbf w'))\rtimes \mathcal{H}^{\frac{\mathbf w_i-\mathbf d_i}{\mathbf v_i}}\right)
    \otimes D^{\mathrm{ch}}(\mathcal E_i)\;.
\end{multline}
The natural vertex superalgebra map $\CVs(Q,\mathbf v,\mathbf w)\to \CVs(Q,\mathbf v,\mathbf w)[I^{-1}]$ induces a vertex superalgebra map between zeroth BRST cohomologies:
\begin{align}\label{eq:reduction map}
    \nu_i: \CV(Q,\mathbf v,\mathbf w)\to \CV(Q',\mathbf v',\mathbf w')\rtimes \mathcal{H}^{\frac{\mathbf w_i-\mathbf d_i}{\mathbf v_i}}
    \otimes D^{\mathrm{ch}}(\mathcal E_i)\;.
\end{align}
We call $\nu_i$ the quiver reduction map at the node $i$.

\bigskip The above construction naturally generalizes to sheaves of $\hbar$-adic vertex superalgebras. Let $\mathcal U_i\subset\widetilde\CM(Q,\mathbf v,\mathbf w)$ be the open subset on which the map $I$ in the LHS of \cref{fig:reduction algorithm} is an isomorphism, then we have an isomorphism
\begin{align}\label{eq:geometry reduction}
    \mathcal U_i\cong \widetilde\CM(Q',\mathbf v',\mathbf w')\times \mathbb C\times \mathcal E_i\;,
\end{align}
where $\mathcal{E}_i=T^*\mathrm{End}(\mathbb C^{\mathbf v_i})^{\oplus \mathsf g_i}\times T^*\mathrm{Hom}(\mathbb C^{\mathbf v_i},\mathbb C^{\mathbf w_i-\mathbf v_i|\mathbf d_i})$.  
\begin{lem}\label{lem:reduction for sheaves}
The restriction of $\mathscr D^{\mathrm{ch}}_{\widetilde\CM(Q,\mathbf v,\mathbf w),\hbar}$ to $\mathcal U_i$ is naturally isomorphic to
\begin{align*}
    \mathscr D^{\mathrm{ch}}_{\widetilde\CM(Q',\mathbf v',\mathbf w'),\hbar}\rtimes \mathscr H^{\frac{\mathbf w_i-\mathbf d_i}{\mathbf v_i}}_{\mathbb C,\hbar}\boxtimes \mathscr D^{\mathrm{ch}}_{\mathcal E_i,\hbar}\;.
\end{align*}
\end{lem}

\begin{proof}
Denote by $\mathcal R_i$ the open subset of $T^*\mathrm{Rep}(Q,\mathbf v,\mathbf w)\times\mathbb C^{Q_0}$ on which the map $I$ in the LHS of \cref{fig:reduction algorithm} is an isomorphism. $\mathcal R_i$ is affine because it is the nonzero set of a single function $\det(I)$. The action of $\mathrm{GL}(\mathbf v_i)$ on $\mathcal R_i$ is free and there is an isomorphism
\begin{align*}
    \mathcal R_i\cap \widetilde\mu^{-1}_i(0)\sslash \mathrm{GL}(\mathbf v_i)\cong T^*\mathrm{Rep}(Q',\mathbf v',\mathbf w')\times\mathbb C^{Q'_0}\times \mathbb C\times\mathcal E_i\;,
\end{align*}
where $\widetilde\mu_i$ is the extended moment map associated to the node $i$. Denote by $q_i: \mathcal R_i\cap \widetilde\mu^{-1}_i(0)\to T^*\mathrm{Rep}(Q',\mathbf v',\mathbf w')\times\mathbb C^{Q'_0}\times \mathbb C\times\mathcal E_i$ the quotient map.

Denote by $\widetilde\mu':T^*\mathrm{Rep}(Q',\mathbf v',\mathbf w')\times\mathbb C^{Q'_0}\to \mathfrak{gl}(\mathbf v')^*$ the extended moment map for $Q'$, and denote by $q': \widetilde\mu'^{-1}(0)^{\mathrm{st}}\times \mathbb C\times\mathcal E_i\to  \widetilde\CM(Q',\mathbf v',\mathbf w')\times \mathbb C\times \mathcal E_i$ the quotient map. 

Now we take an arbitrary affine open subset $U\subset \mathcal U_i$ and choose an affine open subset $\widetilde U'\subset T^*\mathrm{Rep}(Q',\mathbf v',\mathbf w')\times\mathbb C^{Q'_0}\times \mathbb C\times\mathcal E_i$ such that $\widetilde U'\cap (\widetilde\mu^{-1}(0)\times \mathbb C\times\mathcal E_i)=q'^{-1}(U)$. Then we choose an affine open subset $\widetilde U\subset \mathcal R_i$ such that $\widetilde U\cap \widetilde\mu^{-1}_i(0)=q_i^{-1}(\widetilde U')$.

According to \eqref{eq:sheaf on affine}, we have 
\begin{align*}
    \mathscr D^{\mathrm{ch}}_{\widetilde\CM(Q,\mathbf v,\mathbf w),\hbar}(U)=H^{\infty/2+0}_{\hbar\mathrm{BRST}}(\mathfrak{gl}(\mathbf v),\mathscr D^{\mathrm{ch}}_{\mathfrak R,\hbar}(\widetilde U))\;.
\end{align*}
Consider the decomposition $\mathcal{Q}_{\hbar}=\mathcal{Q}^{i}_{\hbar}+\mathcal{Q}^{\hat{i}}_{\hbar}$ into differential $\mathcal{Q}^{i}_{\hbar}$ corresponding to the node $i$ and differential $\mathcal{Q}^{\hat{i}}_{\hbar}=\sum_{j\neq i}\mathcal{Q}^{j}_{\hbar}$ corresponding to nodes other than $i$. $\mathcal{Q}^{i}_{\hbar}$ anti-commutes with $\mathcal{Q}^{\hat{i}}_{\hbar}$, then we have a spectral sequence with second page
\begin{align*}
    E^{p,q}_2=H^q(H^p(C(\mathfrak{gl}(\mathbf v), \mathscr D^{\mathrm{ch}}_{\mathfrak R,\hbar}(\widetilde U))_\hbar,\mathcal{Q}^{i}_{\hbar}),\mathcal{Q}^{\hat{i}}_{\hbar})\;.
\end{align*}
According to our previous discussion in \eqref{eq:BRST cohomology of building block}, we have 
\begin{align*}
    H^0(C(\mathfrak{gl}(\mathbf v), \mathscr D^{\mathrm{ch}}_{\mathfrak R,\hbar}(\widetilde U))_\hbar,\mathcal{Q}^{i}_{\hbar})\cong
    C\left(\mathfrak{gl}(\mathbf v'),\left(\mathscr D^{\mathrm{ch}}_{\mathfrak R',\hbar}\rtimes\mathscr H^{\frac{\mathbf w_i-\mathbf d_i}{\mathbf v_i}}_{\mathbb C,\hbar}\boxtimes \mathscr D^{\mathrm{ch}}_{\mathcal E_i,\hbar}\right)(\widetilde U')\right)_\hbar\;,
\end{align*}
and $H^p(C(\mathfrak{gl}(\mathbf v), \mathscr D^{\mathrm{ch}}_{\mathfrak R,\hbar}(\widetilde U))_\hbar,\mathcal{Q}^{i}_{\hbar})$ vanishes for $p\neq 0$. It follows that the spectral sequence $E^{p,q}$ degenerates at the second page, and it converges to 
\begin{align*}
    H^{p+q}(C(\mathfrak{gl}(\mathbf v), \mathscr D^{\mathrm{ch}}_{\mathfrak R,\hbar}(\widetilde U))_\hbar,\mathcal{Q}_{\hbar})=H^{\infty/2+p+q}_{\hbar\mathrm{BRST}}(\mathfrak{gl}(\mathbf v),\mathscr D^{\mathrm{ch}}_{\mathfrak R,\hbar}(\widetilde U))\;.
\end{align*}
On the other hand, 
\begin{align*} E^{0,0}_2&=H^{\infty/2+0}_{\hbar\mathrm{BRST}}\left(\mathfrak{gl}(\mathbf v'),\left(\mathscr D^{\mathrm{ch}}_{\mathfrak R',\hbar}\rtimes\mathscr H^{\frac{\mathbf w_i-\mathbf d_i}{\mathbf v_i}}_{\mathbb C,\hbar}\boxtimes \mathscr D^{\mathrm{ch}}_{\mathcal E_i,\hbar}\right)(\widetilde U')\right)\\
&=\mathscr D^{\mathrm{ch}}_{\widetilde\CM(Q',\mathbf v',\mathbf w'),\hbar}\rtimes \mathscr H^{\frac{\mathbf w_i-\mathbf d_i}{\mathbf v_i}}_{\mathbb C,\hbar}\boxtimes \mathscr D^{\mathrm{ch}}_{\mathcal E_i,\hbar}(U) \;,
\end{align*}
and $E^{p,q}_2$ vanishes for $(p,q)\neq (0,0)$. Thus we have isomorphism between $\hbar$-adic vertex superalgebras:
\begin{align*}
    \mathscr D^{\mathrm{ch}}_{\widetilde\CM(Q,\mathbf v,\mathbf w),\hbar}(U)\cong \mathscr D^{\mathrm{ch}}_{\widetilde\CM(Q',\mathbf v',\mathbf w'),\hbar}\rtimes \mathscr H^{\frac{\mathbf w_i-\mathbf d_i}{\mathbf v_i}}_{\mathbb C,\hbar}\boxtimes \mathscr D^{\mathrm{ch}}_{\mathcal E_i,\hbar}(U) \;.
\end{align*}
By construction such isomorphism is compatible with restriction to smaller affine open subsets, thus we get an isomorphism between sheaves of $\hbar$-adic vertex superalgebras:
\begin{align*}
    \mathscr D^{\mathrm{ch}}_{\widetilde\CM(Q,\mathbf v,\mathbf w),\hbar}\big|_{\mathcal U_i}\cong \mathscr D^{\mathrm{ch}}_{\widetilde\CM(Q',\mathbf v',\mathbf w'),\hbar}\rtimes \mathscr H^{\frac{\mathbf w_i-\mathbf d_i}{\mathbf v_i}}_{\mathbb C,\hbar}\boxtimes \mathscr D^{\mathrm{ch}}_{\mathcal E_i,\hbar} \;.
\end{align*}
This finishes the proof of the lemma.
\end{proof}

The restriction to $\mathcal U_i$ gives rise to an injective map between $\hbar$-adic vertex superalgebras:
\begin{align}\label{eq:restriction to U_i}
    \mathscr D^{\mathrm{ch}}_{\widetilde\CM(Q,\mathbf v,\mathbf w),\hbar}(\widetilde\CM(Q,\mathbf v,\mathbf w))\hookrightarrow \mathscr D^{\mathrm{ch}}_{\widetilde\CM(Q,\mathbf v,\mathbf w),\hbar}(\mathcal U_i)\;,
\end{align}
and the RHS is isomorphic to $\mathscr D^{\mathrm{ch}}_{\widetilde\CM(Q',\mathbf v',\mathbf w'),\hbar}(\widetilde\CM(Q',\mathbf v',\mathbf w'))\rtimes \mathcal{H}^{\frac{\mathbf w_i-\mathbf d_i}{\mathbf v_i}}_\hbar
\otimes D^{\mathrm{ch}}(\mathcal E_i)_\hbar$ by the Lemma \ref{lem:reduction for sheaves}. 

We note that the $\mathbb S$-action on the sheaf $\mathscr D^{\mathrm{ch}}_{\widetilde\CM(Q',\mathbf v',\mathbf w'),\hbar}\rtimes \mathscr H^{\frac{\mathbf w_i-\mathbf d_i}{\mathbf v_i}}_{\mathbb C,\hbar}\boxtimes \mathscr D^{\mathrm{ch}}_{\mathcal E_i,\hbar}$ is given by the weights assignment
\begin{align}\label{eq:S-wts of XYIJh'}
\mathrm{wt}(X_{(-n)})=\mathrm{wt}(Y_{(-n)})=\mathrm{wt}(I_{s,(-n)})=\mathrm{wt}(J_{s,(-n)})=n\;,
\quad \mathrm{wt}(h'_{(-n)})=n+1\;,
\end{align}
for those nodes $s\in Q'_0$ which are not connected to the node $i$ in the original quiver $Q$, and 
\begin{align}
\mathrm{wt}(\widetilde I_{(-n)})=\mathrm{wt}(\beta_{j,(-n)})=n-1\;,\quad \mathrm{wt}(\widetilde J_{(-n)})=\mathrm{wt}(\gamma_{j,(-n)})=n+1\;,
\end{align}
for those nodes $j\in Q'_0$ which are connected to the node $i$ in the original quiver $Q$, and 
\begin{align}
\mathrm{wt}(\alpha_{(-n)})=n+1\;,\quad\mathrm{wt}(C_{(-n)})=\mathrm{wt}(D_{(-n)})=n\;,
\end{align}
and the weights assignment \eqref{eq:S-wts of bc and hbar} to the $\mathsf b\mathsf c$ ghosts and $\hbar$. Although the $\mathbb S$-action on $\mathscr D^{\mathrm{ch}}_{\widetilde\CM(Q',\mathbf v',\mathbf w'),\hbar}(\widetilde\CM(Q',\mathbf v',\mathbf w'))$ is not the same as the one defined in section \ref{subsec:global section}, the next lemma asserts that these two actions have the same finite-weights part.

\begin{lem}\label{lem:S-fin = S'-fin}
Let $\mathbb S'$ be a one-dimensional torus, and define $\mathbb S'$ action on $\mathscr D^{\mathrm{ch}}_{\widetilde\CM(Q',\mathbf v',\mathbf w'),\hbar}$ by weights assignment \eqref{eq:S-wts of XYIJh'} and $\mathrm{wt}(\beta_{j,(-n)})=\mathrm{wt}(\gamma_{j,(-n)})=n$. Then we have
\begin{align}
    \mathscr D^{\mathrm{ch}}_{\widetilde\CM(Q',\mathbf v',\mathbf w'),\hbar}(\widetilde\CM(Q',\mathbf v',\mathbf w'))_{\mathbb S\mathrm{-fin}}=\mathscr D^{\mathrm{ch}}_{\widetilde\CM(Q',\mathbf v',\mathbf w'),\hbar}(\widetilde\CM(Q',\mathbf v',\mathbf w'))_{\mathbb S'\mathrm{-fin}}\;.
\end{align}
\end{lem}

\begin{proof}
Let $\mathscr D^{\mathrm{ch}}_{\widetilde\CM(Q',\mathbf v',\mathbf w'),\hbar}(\widetilde\CM(Q',\mathbf v',\mathbf w'))_m$ be the $\mathbb S$-degree $m$ eigenspace. Since the $\mathbb S$-action commutes with the $\mathbb S'$-action, we have
\begin{align*}
    \mathscr D^{\mathrm{ch}}_{\widetilde\CM(Q',\mathbf v',\mathbf w'),\hbar}(\widetilde\CM(Q',\mathbf v',\mathbf w'))_m={\prod_{n\ge -M}}' \mathscr D^{\mathrm{ch}}_{\widetilde\CM(Q',\mathbf v',\mathbf w'),\hbar}(\widetilde\CM(Q',\mathbf v',\mathbf w'))_{m,n} \;,
\end{align*}
where the subscript $m,n$ denotes $\mathbb S$-degree $m$ $\mathbb S'$-degree $n$ eigenspace. The restricted product ${\prod}'$ means the subspace of the product space consisting of elements $(x_i)_{i=-M}^{\infty}$ such that there exists $N$ and $x_n\in \hbar^{n-N}\mathscr D^{\mathrm{ch}}_{\widetilde\CM(Q',\mathbf v',\mathbf w'),\hbar}(\widetilde\CM(Q',\mathbf v',\mathbf w'))_{N+m-n,N}$ for all $n>N$. Since $\mathbb S$-degrees are bounded from below, $x_n$ vanishes for sufficiently large $n$, thus we have
\begin{align*}
    \mathscr D^{\mathrm{ch}}_{\widetilde\CM(Q',\mathbf v',\mathbf w'),\hbar}(\widetilde\CM(Q',\mathbf v',\mathbf w'))_m=\bigoplus_{n\ge -M} \mathscr D^{\mathrm{ch}}_{\widetilde\CM(Q',\mathbf v',\mathbf w'),\hbar}(\widetilde\CM(Q',\mathbf v',\mathbf w'))_{m,n}\;.
\end{align*}
It follows that $\mathscr D^{\mathrm{ch}}_{\widetilde\CM(Q',\mathbf v',\mathbf w'),\hbar}(\widetilde\CM(Q',\mathbf v',\mathbf w'))_{\mathbb S\mathrm{-fin}}\subseteq\mathscr D^{\mathrm{ch}}_{\widetilde\CM(Q',\mathbf v',\mathbf w'),\hbar}(\widetilde\CM(Q',\mathbf v',\mathbf w'))_{\mathbb S'\mathrm{-fin}}$. Exchange the role of $\mathbb S$ and $\mathbb S'$ in the above argument, and we have the reverse inclusion, this finishes the proof.
\end{proof}

By the definition,
\begin{align*}
    \mathsf D^{\mathrm{ch}}(\widetilde\CM(Q',\mathbf v',\mathbf w'))=\mathscr D^{\mathrm{ch}}_{\widetilde\CM(Q',\mathbf v',\mathbf w'),\hbar}(\widetilde\CM(Q',\mathbf v',\mathbf w'))_{\mathbb S'\mathrm{-fin}}/(\hbar-1)\;.
\end{align*}
According to the Lemma \ref{lem:S-fin = S'-fin}, we can also take the $\mathbb S$-finite part in the above definition. We take the $\mathbb S$-finite parts of the two sides of the inclusion \eqref{eq:restriction to U_i} and set $\hbar=1$, then we get an embedding of vertex superalgebras:
\begin{align}\label{eq:reduction map of global sections}
    \bm{\nu}_i: \mathsf D^{\mathrm{ch}}(\widetilde\CM(Q,\mathbf v,\mathbf w))\hookrightarrow \mathsf D^{\mathrm{ch}}(\widetilde\CM(Q',\mathbf v',\mathbf w'))\rtimes \mathcal{H}^{\frac{\mathbf w_i-\mathbf d_i}{\mathbf v_i}}
\otimes D^{\mathrm{ch}}(\mathcal E_i)\;\;.
\end{align}
\eqref{eq:reduction map} and \eqref{eq:reduction map of global sections} are compatible, namely the following diagram commutes:
\begin{equation}
\begin{tikzcd}
\CV(Q,\mathbf v,\mathbf w)\ar[r,"\nu_i"]\ar[d,"\iota"] & \CV(Q',\mathbf v',\mathbf w')\rtimes \mathcal{H}^{\frac{\mathbf w_i-\mathbf d_i}{\mathbf v_i}}
\otimes D^{\mathrm{ch}}(\mathcal E_i)\ar[d,"\iota'"]\\
\mathsf D^{\mathrm{ch}}(\widetilde\CM(Q,\mathbf v,\mathbf w))\ar[r,"\bm{\nu}_i", hook] & \mathsf D^{\mathrm{ch}}(\widetilde\CM(Q',\mathbf v',\mathbf w'))\rtimes \mathcal{H}^{\frac{\mathbf w_i-\mathbf d_i}{\mathbf v_i}}
\otimes D^{\mathrm{ch}}(\mathcal E_i)
\end{tikzcd}
\end{equation}
where the vertical arrows are the maps in \eqref{quiver VOA to global section}.

\begin{remark}\label{rmk:alpha bar}
The Heisenberg field $\alpha$ in \eqref{eq:fields under the reduction} is chosen such that it commutes with the extended chiral moment maps for the new quiver $Q'$. The drawback of such a choice is that $\alpha$ acts on $\CVs(Q',\mathbf v',\mathbf w')$ nontrivially, so we get a semi-direct product instead of a direct product between $\CVs(Q',\mathbf v',\mathbf w')$ and $\mathcal{H}^{\frac{\mathbf w_i-\mathbf d_i}{\mathbf v_i}}$. However, if $Q'$ satisfies certain conditions which will be specified below, then we can deform $\alpha$ using fields in $\CVs(Q',\mathbf v',\mathbf w')$ such that the deformation is still a Heisenberg field and it acts on $\CVs(Q',\mathbf v',\mathbf w')$ trivially.  

Let $\mathsf C'$ be the Cartan matrix of $Q'$. Suppose that there exists $\zeta\in \mathbb C^{Q'_0}$ such that
\begin{align*}
    (\mathsf C'\zeta)_j=\mathsf C_{i,j}\;, \quad\forall j\in Q'_0\;,
\end{align*}
then we define
\begin{align}\label{eq:modified Heisenberg 1}
    \bar\alpha:=\alpha-\frac{1}{\mathbf v_i}\sum_{j\to i\in Q_1}\mathrm{Tr}(\gamma_j\beta_j)-\sum_{j\in Q'_0} \zeta_jh'_j\;.
\end{align}
Now it is straightforward to see that OPEs between $\bar\alpha$ and strong generators of $\CVs(Q',\mathbf v',\mathbf w')$ are regular, thus $\bar\alpha$ acts on the latter trivially. The level of $\bar\alpha$ equals to $2-2\mathsf g_i-\sum_{j\in Q'_0}\mathsf C_{i,j}\zeta_j$. 

For example, if $Q'$ is a disjoint union of Dynkin quivers, then $\mathsf C'$ is invertible, thus the existence (and uniqueness) of $\zeta$ is guaranteed, namely we take $\zeta_j=\sum_{k\in Q'_0}(\mathsf C'^{-1})_{j,k}\mathsf C_{i,k}$. In this case the level of $\bar\alpha$ is ${\det(\mathsf C)}/{\det(\mathsf C')}$.
\end{remark}

\subsection{Generalization}
In our previous discussion of quiver reductions, we needed the condition that $\mathbf w_i\ge \mathbf v_i$. In this section, we generalize it to the following situation. 

Let $Q$ be a quiver and let $i\in Q_0$ be a node. Without loss of generality let us assume that $i$ is a sink. Choose gauge and framing dimensions $\mathbf v,\mathbf w$ such that $\mathbf w-\mathsf C\mathbf v\in \mathbb Z^{Q_0}_{\ge 0}$. Assume moreover that there exists a set of nodes $\mathcal N\subset Q_0\setminus \{i\}$ such that $\mathsf C_{ij}\neq 0,\forall j\in \mathcal N$, and $\sum_{j\in \mathcal N}\mathbf v_j\le\mathbf v_i$, and $\mathbf w_i+\sum_{j\in \mathcal N}\mathbf v_j\ge\mathbf v_i$. Let us split the framing vector space into
\begin{align*}
    \mathbb C^{\mathbf w_i}=\mathbb C^{\mathbf w_{i,1}}\oplus\mathbb C^{\mathbf w_{i,2}}\;,
    \quad
    \text{where }\quad\mathbf w_{i,1}=\mathbf w_i-\mathbf v_i+\sum_{j\in \mathcal N}\mathbf v_j \;,
\end{align*}
and fix a choice of $a_j\in$ the set of arrows from $j$ to $i$, then consider the open locus $\mathcal U\subset \widetilde\CM(Q,\mathbf v,\mathbf w)$ on which the linear map
\begin{align*}
    \mathcal I:=(I|_{\mathbb C^{\mathbf w_{i,2}}},(X_{a_j})_{j\in \mathcal N}):\; \mathbb C^{\mathbf w_{i,2}}\oplus \bigoplus_{j\in \mathcal N} V_j\longrightarrow V_i\;\text{ is an isomorphism}\;.
\end{align*}
An argument similar to the beginning of this section shows that 
\begin{align}
    \mathcal U\cong \widetilde\CM(Q',\mathbf v',\mathbf w')\times \mathbb C\times T^*\mathrm{End}(\mathbb C^{\mathbf w_{i,2}})^{\oplus \mathsf g_i}\times T^*\mathrm{Hom}(\mathbb C^{\mathbf w_{i,1}},\mathbb C^{\mathbf w_{i,2}})\;,
\end{align}
where $Q'$ is the quiver whose set of nodes is $Q'_0=Q_0\setminus \{i\}$, and the Cartan matrix $\mathsf C'$ for $Q'$ is
\begin{align}
    \mathsf C'_{j,l}=\begin{cases}
        \mathsf C_{j,l}\;, & j\notin \mathcal N, l\notin \mathcal N\;,\\
        \mathsf C_{j,l}+\mathsf C_{j,i}\;, & j\notin \mathcal N, l\in \mathcal N,\;\\
        \mathsf C_{j,l}+\mathsf C_{i,l}\;, & j\in \mathcal N, l\notin \mathcal N\;,\\
        \mathsf C_{j,l}+\mathsf C_{i,l}+\mathsf C_{j,i}+\mathsf C_{i,i}\;, & j,l\in \mathcal N \;.
    \end{cases}
\end{align}
The gauge and framing dimension vectors are determined by
\begin{align}
    \mathbf v'_j=\mathbf v_j\;,\quad \mathbf w'_j=\begin{cases}
        \mathbf w_j-\mathsf C_{j,i}\mathbf w_{i,2},& j\notin \mathcal N\;,\\
        \mathbf w_j+\mathbf w_i-\mathsf C_{j,i}\mathbf w_{i,2}-\mathsf C_{i,i}\mathbf w_{i,2}, & j\in \mathcal N\;.
    \end{cases}
\end{align}
It is easy to see that
\begin{align*}
    (\mathbf w'-\mathsf C'\mathbf v')_j=\begin{cases}
        (\mathbf w-\mathsf C\mathbf v)_j\;, & j\notin \mathcal N\;,\\
        (\mathbf w-\mathsf C\mathbf v)_j+(\mathbf w-\mathsf C\mathbf v)_i\;, & j\in \mathcal N\;.
    \end{cases}
\end{align*}
In particular $\mathbf w'-\mathsf C'\mathbf v'\in \mathbb Z^{Q'_0}_{\ge 0}$. 

\bigskip Using the analysis in the section \ref{subsec:building block}, we obtain the following generalization of previous results in the section \ref{subsec:reduction for VOA}.
\begin{itemize}
    \item[(1)] There exists a natural map of vertex superalgebras:
    \begin{align}\label{eq: reduction map_generalization}
        \nu: \CV(Q,\mathbf v,\mathbf w)\to \CV(Q',\mathbf v',\mathbf w')\rtimes \mathcal H^{\lambda}\otimes D^{\mathrm{ch}}(\mathcal E)\;.
    \end{align}
    where $\mathcal H^{\lambda}$ is the Heisenberg algebra of certain level $\lambda$, and $D^{\mathrm{ch}}(\mathcal E)$ is the $\beta\gamma\mathsf b\mathsf c$ system on $\mathcal E=T^*\mathrm{End}(\mathbb C^{\mathbf w_{i,2}})^{\oplus \mathsf g_i}\times T^*\mathrm{Hom}(\mathbb C^{\mathbf w_{i,1}},\mathbb C^{\mathbf w_{i,2}})$. We also have generalization of \eqref{eq:fields under the reduction}, for simplicity we only list the Heisenberg fields as follows
\begin{equation}\label{eq:general h'_j}
h_j'=\begin{cases}
h_j+\mathsf C_{i,j}\mathrm{Tr}(\mathcal I^{-1}\partial \mathcal I)\;, & j\notin \mathcal N\;,\\
h_j+h_i+(\mathsf C_{i,j}+\mathsf C_{i,i})\mathrm{Tr}(\mathcal I^{-1}\partial \mathcal I)\;, & j\in \mathcal N\;,
\end{cases}
\end{equation}
\begin{equation}\label{eq:general alpha}
\begin{split}
\alpha=h_i+&\mathsf C_{i,i}\mathrm{Tr}(\mathcal I^{-1}\partial \mathcal I)+\text{linear comb. of }\mathrm{Tr}(\beta_j\gamma_j),\:\mathrm{Tr}(I'_jJ'_j),\:\mathrm{Tr}(X'_aY'_a).
\end{split}
\end{equation}

\item[(2)]  There exists an embedding of vertex superalgebras:
\begin{align}\label{eq:reduction map of global sections_generalization}
    \bm{\nu}: \mathsf D^{\mathrm{ch}}(\widetilde\CM(Q,\mathbf v,\mathbf w))\hookrightarrow \mathsf D^{\mathrm{ch}}(\widetilde\CM(Q',\mathbf v',\mathbf w'))\rtimes \mathcal{H}^{\lambda}\otimes D^{\mathrm{ch}}(\mathcal E)\;.
\end{align}
\eqref{eq: reduction map_generalization} and \eqref{eq:reduction map of global sections_generalization} are compatible, namely the following diagram commutes:
\begin{equation}
\begin{tikzcd}
\CV(Q,\mathbf v,\mathbf w)\ar[r,"\nu"]\ar[d,"\iota"] & \CV(Q',\mathbf v',\mathbf w')\rtimes \mathcal{H}^{\lambda}
\otimes D^{\mathrm{ch}}(\mathcal E)\ar[d,"\iota'"]\\
\mathsf D^{\mathrm{ch}}(\widetilde\CM(Q,\mathbf v,\mathbf w))\ar[r,"\bm{\nu}", hook] & \mathsf D^{\mathrm{ch}}(\widetilde\CM(Q',\mathbf v',\mathbf w'))\rtimes \mathcal{H}^{\lambda}
\otimes D^{\mathrm{ch}}(\mathcal E)
\end{tikzcd}
\end{equation}
where the vertical arrows are the maps in \eqref{quiver VOA to global section}.
\end{itemize}

\begin{remark}\label{rmk:alpha bar_general}
The Heisenberg field $\alpha$ in \eqref{eq:general alpha} is chosen such that it commutes with the extended chiral moment maps for the new quiver $Q'$. The drawback of such choice is that $\alpha$ acts on $\CV(Q',\mathbf v',\mathbf w')$ nontrivially, so we get a semi-direct product instead of a direct product between $\CV(Q',\mathbf v',\mathbf w')$ and $\mathcal{H}^{\lambda}$. This drawback can be cured in the same way as in Remark \ref{rmk:alpha bar}.

Let $\mathsf C'$ be the Cartan matrix of $Q'$. Suppose that there exists $\zeta\in \mathbb C^{Q'_0}$ such that
\begin{align*}
    \forall j\in Q'_0,\quad(\mathsf C'\zeta)_j=\begin{cases}
        \mathsf C_{i,j}\;, & j\notin \mathcal N\;,\\
        \mathsf C_{i,j}+\mathsf C_{i,i}\;, & j\in \mathcal N\;,
    \end{cases}
\end{align*}
then we define
\begin{align}\label{eq:modified Heisenberg 2}
    \bar\alpha:=h_i+\mathsf C_{i,i}\mathrm{Tr}(\mathcal I^{-1}\partial \mathcal I)-\sum_{j\in Q'_0} \zeta_jh'_j \;.
\end{align}
Now it is straightforward to see that $\bar\alpha$ commutes with $\CVs(Q',\mathbf v',\mathbf w')$, so we get a direct product between $\CV(Q',\mathbf v',\mathbf w')$ and a Heisenberg algebra. The level of $\bar\alpha$ equals to 
$$
2-2\mathsf g_i-(2-2\mathsf g_i)\sum_{j\in \mathcal N}\zeta_j-\sum_{j\in Q'_0}\mathsf C_{i,j}\zeta_j\;.
$$ 
For example, if $Q'$ is a disjoint union of Dynkin quivers, then $\mathsf C'$ is invertible, thus the existence (and uniqueness) of $\zeta$ is guaranteed, namely we take 
$$
\zeta_j=\sum_{k\in Q'_0}(\mathsf C'^{-1})_{j,k}\mathsf C_{i,k}+(2-2\mathsf g_i)\sum_{k\in \mathcal N}(\mathsf C'^{-1})_{j,k}\;.
$$ 
In this case the level of $\bar\alpha$ is $\det(\mathsf C)/\det(\mathsf C')$.
\end{remark}

\bigskip The following lemma tells us that every good linear quiver contains a node for which the quiver reduction can be performed.
\begin{lem}\label{lem:good linear quiver always reduces}
Let $Q$ be a finite type A quiver whose Cartan matrix is denoted by $\mathsf C$, and choose gauge dimension $\mathbf v$ and framing dimension $\mathbf w$ such that $\mathbf w-\mathsf C\mathbf v\in \mathbb Z^{Q_0}_{\ge 0}$, then there exists a node $i\in Q_0$ such that one of the following condition holds
\begin{enumerate}
    \item $\mathbf w_i\ge \mathbf v_i$, 
    \item $\mathbf w_i<\mathbf v_i$ and there exists $j\in \{i+1,i-1\}$ such that $\mathbf v_j\le \mathbf v_i$ and $\mathbf w_i+\mathbf v_j\ge\mathbf v_i$.
\end{enumerate}
\end{lem}

\begin{proof}
Suppose that the condition 1 fails for all nodes, i.e. $\mathbf w_i<\mathbf v_i$ for all $i\in Q_0$, then we claim that the condition 2 holds for at least one node. Let us take $i\in Q_0$ with maximal $\mathbf v_i$, i.e. $\mathbf v_i\ge \mathbf v_k$ for all $k\in Q_0$. Choose $j\in \{i-1,i+1\}$, then $\mathbf v_j\le \mathbf v_i$ by our choice of $i$, and we must have $\mathbf w_i+\mathbf v_j\ge \mathbf v_i$, otherwise we would have $\mathbf w_i+\mathbf v_{i-1}+\mathbf v_{i+1}< 2\mathbf v_i$ which contradicts with the assumption that $\mathbf w-\mathsf C\mathbf v\in \mathbb Z^{Q_0}_{\ge 0}$. This finishes the proof.
\end{proof}

The Lemma \ref{lem:good linear quiver always reduces} implies that for a good linear quiver $Q$, we can always perform a quiver reduction to reduce it to a quiver $Q'$ with one less nodes. Moreover, our quiver reduction algorithm asserts that $Q'$ is either a good linear quiver or a disjoint union of two good linear quivers, so we can repeat the quiver reduction algorithm for $Q'$, using Lemma \ref{lem:good linear quiver always reduces} again. Starting from $Q$, we recursively reduce it, and we end up with a tensor product of $\beta\gamma\mathsf b\mathsf c$ system and some Heisenberg algebra. We summarize it in the following proposition.

\begin{thm}\label{thm:free-field realization}
Let $Q$ be a finite type A quiver whose Cartan matrix is denoted by $\mathsf C$, and choose gauge dimension $\mathbf v$ and framing dimension $\mathbf w$ such that $\mathbf w-\mathsf C\mathbf v\in \mathbb Z^{Q_0}_{\ge 0}$, then there exists an embedding of vertex superalgebras:
\begin{align}\label{eq:free-field realization}
    \Theta: \mathsf D^{\mathrm{ch}}(\widetilde\CM(Q,\mathbf v,\mathbf w))\hookrightarrow D^{\mathrm{ch}}(T^*\mathbb C^{n|m})\otimes\bigotimes_{i\in Q_0}\mathcal H^{\lambda_i}\;,
\end{align}
where $n=\mathbf v\cdot(\mathbf w-\frac{1}{2}\mathsf C\mathbf v)$, and $m=\mathbf v\cdot\mathbf d=\mathbf v\cdot(\mathbf w-\mathsf C\mathbf v)$.
\end{thm}

In other words, we get a free-field realization of $\mathsf D^{\mathrm{ch}}(\widetilde\CM(Q,\mathbf v,\mathbf w))$ via recursively applying quiver reductions. 

\section{Examples of quiver reductions}\label{sec:reduction_examples}
\subsection{Example: \texorpdfstring{$T[\SU(n)]$}{T} quiver} \label{T[SU(n)]-revisited}

We revisit the $T[\SU(n)]$ quiver in this section. According to our previous discussions in \cref{subsec:Associated variety}, there exists a vertex algebra embedding $V^{-n+1}(\mathfrak{sl}_{n})\hookrightarrow \mathsf D^{\mathrm{ch}}(\widetilde\CM(Q))$ determined by the map on strong generators \eqref{eq:mebdding of V(sl_n)}. Composing with the embedding \eqref{eq:free-field realization} we get a free-filed realization $V^{-n+1}(\mathfrak{sl}_{n})\hookrightarrow D^{\mathrm{ch}}(T^*\mathbb C^{n(n-1)/2})\otimes\bigotimes_{i=1}^{n-1}\mathcal H^{\lambda_i}$. We will show in this section that this is exactly the Wakimoto realization of $V^{-n+1}(\mathfrak{sl}_{n})$ \cite{feigin1990affine,bouwknegt1990quantum}.

It turns out that it will be more convenient to work with $\mathfrak{gl}_{n}$ instead of $\mathfrak{sl}_{n}$. Accordingly we add one Heisenberg field $J_B$ of level $n$ to $\mathsf D^{\mathrm{ch}}(\widetilde\CM(Q))$, or equivalently we take $n$ copies of independent level $1$ Heisenberg algebras $b_0,\cdots,b_{n-1}$, and set $h_i:=b_i-b_{i-1}$ and $J_B=b_0+\cdots+b_{n-1}$. Then there exists a vertex subalgebra $V^{-n+1}(\mathfrak{gl}_n)\subset \mathsf D^{\mathrm{ch}}(\widetilde\CM(Q))\otimes J_B$ which is strongly generated by \footnote{Here our convention for $V^\lambda(\mathfrak{gl}_n)$ is that $E^a_b(z)E^c_d(w)\sim \frac{\lambda\delta^a_d\delta^c_b+\delta^a_b\delta^c_d}{(z-w)^2}+$lower order term.}
\begin{align*}
    E^i_j:=J^i_kI^k_j-b_0\delta^i_j\;, \quad 1\le i,j\le n\;.
\end{align*}
The quiver reduction is depicted as follows.
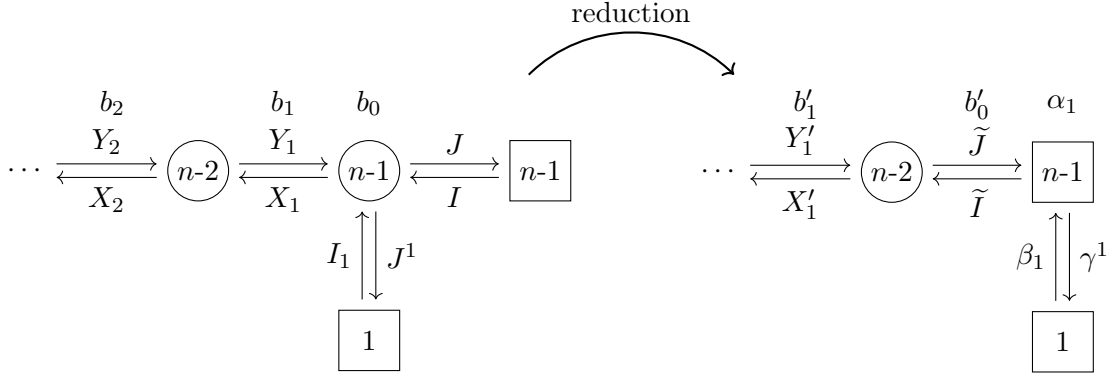
\begin{figure}
\centering
    \begin{tikzpicture}[scale=0.9, every node/.style={scale=1.0}]
    \node (leftfigure) at (-3,0){
    \begin{tikzpicture}[scale=0.9]
    \node(L) at (-2.5,0) {$\cdots$};
    \node(M) at (0,0) {$\circled{$n$-$2$}$};
    \node(R) at (2.5,0) {$\circled{$n$-$1$}$};
    \node(B) at (2.5,-2.5) {$\rectangled{$1$}$};
    \node(F) at (5,0) {$\rectangled{$n$-$1$}$};
    \draw[->] ([yshift=0.1cm]L.east) to node[above]{$Y_2$} ([yshift=0.1cm]M.west);
    \draw[->] ([yshift=-0.1cm]M.west) to node[below]{$X_2$} ([yshift=-0.1cm]L.east);
    \draw[->] ([yshift=0.1cm]M.east) to node[above]{$Y_1$} ([yshift=0.1cm]R.west);
    \draw[->] ([yshift=-0.1cm]R.west) to node[below]{$X_1$} ([yshift=-0.1cm]M.east);
    \draw[->] ([yshift=0.1cm]R.east) to node[above]{$J$} ([yshift=0.1cm]F.west);
    \draw[->] ([yshift=-0.1cm]F.west) to node[below]{$I$} ([yshift=-0.1cm]R.east);
    \draw[->] ([xshift=0.1cm]R.south) to node[right]{$J^1$} ([xshift=0.1cm]B.north);
    \draw[->] ([xshift=-0.1cm]B.north) to node[left]{$I_1$} ([xshift=-0.1cm]R.south);
    \node(b2) at (-1.25,1) {$b_2$};
    \node(b1) at (1.25,1) {$b_1$};
    \node(b0) at (2.5,1) {$b_0$};
    \end{tikzpicture}
    };

    \draw[->,thick] (0.5,2.3) to [bend left=45] node[above]{reduction} (3.5, 2.3);
     
    \node (rightfigure) at (6,0){
    \begin{tikzpicture}[scale=0.9]
    \node(L) at (-2.5,0) {$\cdots$};
    \node(M) at (0,0) {$\circled{$n$-$2$}$};
    \node(R) at (2.5,0) {$\rectangled{$n$-$1$}$};
    \node(B) at (2.5,-2.5) {$\rectangled{$1$}$};
    \draw[->] ([yshift=0.1cm]L.east) to node[above]{$Y'_1$} ([yshift=0.1cm]M.west);
    \draw[->] ([yshift=-0.1cm]M.west) to node[below]{$X'_1$} ([yshift=-0.1cm]L.east);
    \draw[->] ([yshift=0.1cm]M.east) to node[above]{$\widetilde J$} ([yshift=0.1cm]R.west);
    \draw[->] ([yshift=-0.1cm]R.west) to node[below]{$\widetilde I$} ([yshift=-0.1cm]M.east);
    \draw[->] ([xshift=0.1cm]R.south) to node[right]{$\gamma ^1$} ([xshift=0.1cm]B.north);
    \draw[->] ([xshift=-0.1cm]B.north) to node[left]{$\beta_1$} ([xshift=-0.1cm]R.south);
    \node(b1) at (-1.25,1) {$b'_1$};
    \node(b0) at (1.25,1) {$b'_0$};
    \node(alpha1) at (2.5,1) {$\alpha_1$};
    \end{tikzpicture}
    };
    \end{tikzpicture}
    \caption{Reduction procedure of the quiver, as applied to the rightmost node of the $T[\SU(n)]$ quiver.}
    \label{fig:reduction for T[SU(n)]}
\end{figure}
According to our reduction algorithm, we have
\begin{align}
    \widetilde I=X_1 I\;, \quad \widetilde J=I^{-1} Y_1\;,\quad \beta_1=I^{-1}I_1\;,\quad \gamma^1=J^1 I\;,
\end{align}
and 
\begin{align}
   \alpha_1=b_0-\mathrm{Tr}(I^{-1}\partial I)\;,\quad
   b'_i=\begin{cases}
        b_{i+1}\;, & i>0\;,\\
        b_1+\mathrm{Tr}(I^{-1}\partial I)\;, & i=0\;.
    \end{cases}
\end{align}
The quiver reduction gives a vertex algebra embedding
\begin{align*}
        \bm{\nu}:\mathsf D^{\mathrm{ch}}(\widetilde\CM(Q))\otimes J_B\hookrightarrow \mathsf D^{\mathrm{ch}}(\widetilde\CM(Q'))\otimes J'_B\otimes D^{\mathrm{ch}}(T^*\mathbb C^{n-1})\otimes \alpha_1\;,
\end{align*}
where $J'_B=\sum_{i=0}^{n-2}b'_i$ is a level $n-1$ Heisenberg algebra, $D^{\mathrm{ch}}(T^*\mathbb C^{n-1})$ is a $\beta\gamma$ system which is strongly generated by $\beta^i_1,\gamma^1_i$ for $i=2,\cdots,n$, and $\alpha_1$ is a level $1$ Heisenberg algebra. There is a vertex subalgebra $V^{-n+2}(\mathfrak{gl}_{n-1})\subset \mathsf D^{\mathrm{ch}}(\widetilde\CM(Q'))\otimes J'_B$, strongly generated by
\begin{align}
    \widetilde E^i_j:=\widetilde J^i_k\widetilde I^k_j-b_0'\delta^i_j\;, \quad 2\le i,j\le n\;.
\end{align}
Now we compute the image of $E^i_j$ after reduction:
\begin{equation}\label{reduction map for T[SU(n)]}
\begin{split}
\bm{\nu}(E^1_1)&=\sum_{i=2}^{n}\gamma^1_i\beta^i_1-\alpha_1 \;,\\
\bm{\nu}(E^1_i)&=\gamma^1_i\;, \quad 1< i\le n\;,\\
\bm{\nu}(E^i_1)&=\sum_{k=2}^n\widetilde E^i_{k}\beta^k_1+\alpha_1\beta^i_1-\sum_{l=2}^n(\beta^i_1\gamma^1_l)_{(-1)}\beta^l_1\;, \quad 1< i\le n\;,\\
\bm{\nu}(E^i_j)&=\widetilde E^i_{j}-\beta^i_1\gamma^1_j\;, \quad 1< i,j\le n\;.
\end{split}
\end{equation}
In particular, we find that the reduction $\nu$ maps $V^{-n+1}(\mathfrak{gl}_{n})$ into $V^{-n+2}(\mathfrak{gl}_{n-1})\otimes D^{\mathrm{ch}}(T^*\mathbb C^{n-1})\otimes \alpha_1$. This enables us to use the induction on $n$, and end up with getting the following explicit form of a free-field realization of $V^{-n+1}(\mathfrak{gl}_n)$.

\begin{prop}\label{prop:2}
There exists an embedding of vertex algebras
\begin{align*}
    V^{-n+1}(\mathfrak{gl}_n)\hookrightarrow D^{\mathrm{ch}}(T^*\mathbb C^{n(n-1)/2})\otimes(\mathcal H^{1})^{\otimes n},
\end{align*}
such that
\begin{equation}\label{eq:Wakimoto for gl(n+1)}
\begin{split}
E^i_j\mapsto & \gamma^i_j-\sum_{l=1}^{i-1}\beta^i_l\gamma^l_j\;,\quad 1\le i <j\le n\;,\\
E^i_i\mapsto &\sum_{k=i+1}^n\gamma^i_k\beta^k_i-\sum_{l=1}^{i-1}\beta^i_l\gamma^l_i-\alpha_i\;,\quad 1\le i \le n\;,\\
E^i_{i-1}\mapsto &\sum_{k=i+1}^n\gamma^i_k\beta^k_i\beta^i_{i-1}-\sum_{k=i}^n\gamma^{i-1}_k\beta^k_{i-1}\beta^i_{i-1}+(\alpha_{i-1}-\alpha_i)\beta^i_{i-1}\\
&+\sum_{k=i+1}^n\gamma^i_k\beta^k_{i-1}-\sum_{l=1}^{i-2}\beta^i_l\gamma^l_{i-1}+(n-i)\partial\beta^i_{i-1}\;.
\end{split}
\end{equation}
Here $\beta^j_i,\gamma^i_j$, $1\le i<j\le n$ strongly generates $D^{\mathrm{ch}}(T^*\mathbb C^{n(n-1)/2})$ with nontrivial OPEs $\beta^i_j(z)\gamma^k_l(w)\sim\frac{\delta^i_l\delta^k_j}{z-w}$, and $\alpha_1,\cdots,\alpha_n$ are $n$ copies of independent level $1$ Heisenberg algebras.
\end{prop}
If we define 
\begin{align*}
    E_i:=E^{i}_{i+1}\;,\quad H_i:=E^{i}_{i}-E^{i+1}_{i+1}\;,\quad F_i:=E^{i+1}_{i}\;, \quad 1\le i\le n-1\;,
\end{align*}
and in addition set
\begin{align*}
    \gamma^{i,j}:=\beta^{j}_{i}\;,\quad& \beta^{i,j}:=\gamma^{i}_{j}\;,\quad 1\le i<j\le n\;,\\
    \sqrt{-1}\partial\phi_k&:=\alpha_{k}\;, \quad 1\le k\le n \;,
\end{align*}
then \eqref{eq:Wakimoto for gl(n+1)} gives rise to the explicit formula of the Wakimoto realization for $V^{-n+1}(\mathfrak{sl}_{n})$ in \cite[Proposition 3.2]{bouwknegt1990quantum}. 

\bigskip The stress-energy operators $T=T_{\beta\gamma}+T_b+T_{\bc}$ is BRST closed, where $T_{\beta\gamma}$ and $T_{\bc}$ are defined in the same way as in the \eqref{def:generalTtensor_general}, and $T_{b}$ is defined by 
\begin{align}
    T_{b}:=\frac{1}{2}\sum_{i=0}^{n-1} b_i^2~.
\end{align}
$T$ gives rises to a stress-energy operator of $\mathsf D^{\mathrm{ch}}(\widetilde\CM(Q))\otimes J_B$. We claim that $T$ is inside the vertex subalgebra $V^{-n+1}(\mathfrak{gl}_{n})$, in fact $T$ agrees with the Sugawara operator:
\begin{align}\label{T equals Sugawara}
    T=\frac{1}{2}\mathrm{Tr}(E^2)~.
\end{align}
To see this, we compute the image of the both sides of \eqref{T equals Sugawara} under the reduction map $\bm{\nu}$. Using \eqref{reduction map for T[SU(n)]}, we have
\begin{align*}
    \frac{1}{2}\bm{\nu}(\mathrm{Tr}(E^2))=\frac{1}{2}\mathrm{Tr}(\widetilde E^2)+\frac{1}{2}\partial \mathrm{Tr}(\widetilde E)+\frac{1}{2}\alpha_1^2+\frac{n-1}{2}\partial\alpha_1-\mathrm{Tr}(\partial \beta_1\cdot\gamma^1)~.
\end{align*}
Using \eqref{eq:stress-energy operator under reduction}, we have
\begin{align*}
    \bm{\nu}(T)=\widetilde T+\frac{1}{2}\partial \mathrm{Tr}(\widetilde E)+\frac{1}{2}\alpha_1^2+\frac{n-1}{2}\partial\alpha_1-\mathrm{Tr}(\partial \beta_1\cdot\gamma^1)~,
\end{align*}
where $\widetilde T$ is the stress energy operator for $\mathsf D^{\mathrm{ch}}(\widetilde\CM(Q'))\otimes J'_B$. When $n=2$, then it follows from the above computation that $\bm{\nu}(T-\frac{1}{2}\mathrm{Tr}(E^2))=0$. Since $\bm{\nu}$ is injective, we have $T=\frac{1}{2}\mathrm{Tr}(E^2)$ in this case. For general $n$, we use induction on $n$ and obtain $\bm{\nu}(T-\frac{1}{2}\mathrm{Tr}(E^2))=0$, thus we have $T=\frac{1}{2}\mathrm{Tr}(E^2)$ in general.

\subsection{Example: \texorpdfstring{$T_{[1^4]}^{[2^2]}[\SU(4)]$}{T} quiver} \label{[2,2]-revisited}

Now we turn back to the example of the $T_{[1^4]}^{[2^2]}[\SU(4)]$ quiver introduced in section \ref{sec:[nn]}, and apply the quiver reduction with respect to the middle gauge node, see Fig.~\ref{fig:gl(4), (2,2)}.

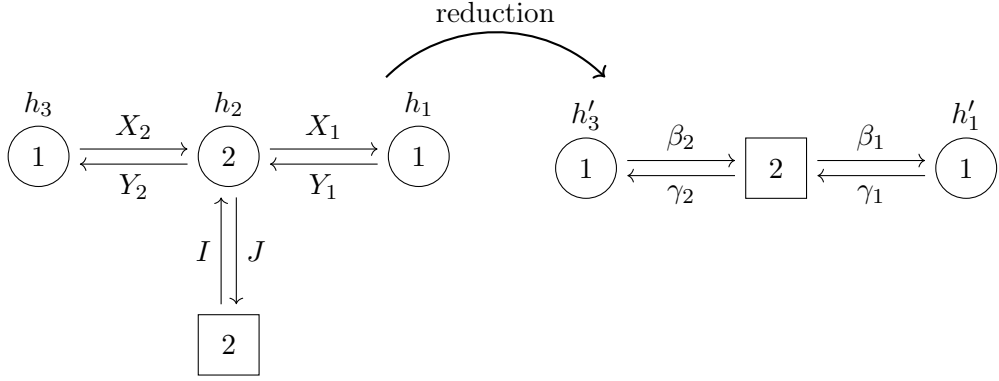
\begin{figure}
\centering
    \begin{tikzpicture}[scale=0.9, every node/.style={scale=1.0}]

    \node (leftfigure) at (0,0){
    \begin{tikzpicture}
    \node(L) at (-2.5,0) {$\circled{$1$}$};
    \node(M) at (0,0) {$\circled{$2$}$};
    \node(R) at (2.5,0) {$\circled{$1$}$};
    \node(B) at (0,-2.5) {$\rectangled{$2$}$};
    \draw[->] ([yshift=0.1cm]L.east) to node[above]{$X_2$} ([yshift=0.1cm]M.west);
    \draw[->] ([yshift=-0.1cm]M.west) to node[below]{$Y_2$} ([yshift=-0.1cm]L.east);
    \draw[->] ([yshift=0.1cm]M.east) to node[above]{$X_1$} ([yshift=0.1cm]R.west);
    \draw[->] ([yshift=-0.1cm]R.west) to node[below]{$Y_1$} ([yshift=-0.1cm]M.east);
    \draw[->] ([xshift=0.1cm]M.south) to node[right]{$J$} ([xshift=0.1cm]B.north);
    \draw[->] ([xshift=-0.1cm]B.north) to node[left]{$I$} ([xshift=-0.1cm]M.south);
    \node(h1) at (-2.5,0.7) {$h_3$};
    \node(h2) at (0,0.7) {$h_2$};
    \node(h3) at (2.5,0.7) {$h_1$};
    \end{tikzpicture}
    };

    \draw[->,thick] (2.3,2.3) to [bend left=45] node[above]{reduction} (5.5, 2.3);
     
    \node (rightfigure) at (8,1.2){
    \begin{tikzpicture}
    \node(L) at (-2.5,0) {$\circled{$1$}$};
    \node(M) at (0,0) {$\rectangled{$2$}$};
    \node(R) at (2.5,0) {$\circled{$1$}$};
    \draw[->] ([yshift=0.1cm]L.east) to node[above]{$\beta_2$} ([yshift=0.1cm]M.west);
    \draw[->] ([yshift=-0.1cm]M.west) to node[below]{$\gamma_2$} ([yshift=-0.1cm]L.east);
    \draw[->] ([yshift=0.1cm]M.east) to node[above]{$\beta_1$} ([yshift=0.1cm]R.west);
    \draw[->] ([yshift=-0.1cm]R.west) to node[below]{$\gamma_1$} ([yshift=-0.1cm]M.east);
    \node(h1) at (-2.5,0.7) {$h_3'$};
    \node(h3) at (2.5,0.7) {$h_1'$};
    \end{tikzpicture}
    };
    \end{tikzpicture}
    \caption{Reduction procedure of the quiver, as applied to the middle node of the $T_{[1^4]}^{[2^2]}[\SU(4)]$ quiver.}
    \label{fig:gl(4), (2,2)}
\end{figure}

According to our reduction algorithm, we have
\begin{align}
\beta_1=X_1I\;, \quad \beta_2=I^{-1}X_2\;,\quad \gamma_1=I^{-1}Y_1\;,\quad\gamma_2=Y_2I\;,
\end{align}
and
\begin{equation}
\begin{split}
h'_1&=h_1+Y_1X_1-\gamma_1\beta_1=h_1-\mathrm{Tr}(I^{-1}\partial I)\;,\\
h'_3&=h_3-X_2Y_2+\beta_2\gamma_2=h_3-\mathrm{Tr}(I^{-1}\partial I)\;,\\
\alpha&=h_2+2\mathrm{Tr}(I^{-1}\partial I)\;,\\
\bar\alpha&=\alpha+\frac{h_1'+h_3'}{2}\;.
\end{split}
\end{equation}
After reduction, we get a tensor product of two copies of $T[\SU(2)]$ quiver VOA together with one copy of the Heisenberg field $\bar\alpha$ of level $1$. We note that there is a vertex subalgebra $V^{-1}(\mathfrak{sl}_2)\otimes V^{-1}(\mathfrak{sl}_2)\subset \CV(Q')$ strongly generated by matrix-valued operators
\begin{align}
J^{[1]}=\gamma_1\beta_1+\frac{h'_1}{2}\mathrm{id}\;,\quad J^{[2]}=-\beta_2\gamma_2+\frac{h'_3}{2}\mathrm{id}\;.
\end{align}
To match with our setting in the linear quiver VOA, we introduce free bosons $b_0,\cdots,b_3$ such that $b_i(z)b_j(w)\sim\frac{\delta_{ij}}{(z-w)^2}$ and set $h_i=b_i-b_{i-1}$. We note that $b_i$ is related to the free fermions by $b_i=\chi_i\psi_i$. The field $J_B=\sum_{i=0}^3b_i$ commutes with $h_1,h_2,h_3$. Then we can re-arrange the fields $J^{[1]},J^{[2]},\bar\alpha,J_B$ into the following
\begin{align}
    E^{[1]}=J^{[1]}-\left(\frac{1}{4}J_B+\frac{1}{2}\bar\alpha\right)\mathrm{id}\;,\quad E^{[2]}=J^{[2]}-\left(\frac{1}{4}J_B-\frac{1}{2}\bar\alpha\right)\mathrm{id}\;,
\end{align}
which strongly generates a vertex subalgebra $V^{-1}(\mathfrak{gl}_2)\otimes V^{-1}(\mathfrak{gl}_2)\subset \CV(Q')\otimes\bar\alpha\otimes J_B$.

\bigskip Now consider the following $2\times 2$ matrix-valued operators before reduction:
\begin{equation}
\begin{split}
G^{(1)}&=JI+\frac{h_1+2h_2+h_3}{2}\mathrm{id}\;,\\
G^{(2)}&=-(JX_2)_{(-1)}(Y_2I)+\rho G^{(1)}-(\rho^2+\partial\rho)\cdot\mathrm{id}\;, 
\end{split}
\end{equation}
where $\rho=\frac{3h_3+2h_2+h_1}{4}$.
We note that $G^{(1)},G^{(2)}$ are BRST closed. Let us compute their images after reduction:
\begin{align*}
    \nu(G^{(1)})=-\beta_2\gamma_2+\gamma_1\beta_1+\frac{h_1'+h_3'}{2}\mathrm{id}=J^{[1]}+J^{[2]}\;,
\end{align*}
and
\begin{align*}
\nu(G^{(2)})=((\beta_2\gamma_2-\gamma_1\beta_1)_{(-1)}\beta_2)_{(-1)}\gamma_2+(\alpha\beta_2)_{(-1)}\gamma_2+\nu(\rho)\nu( G^{(1)})-(\nu(\rho)^2+\partial\nu(\rho))\cdot\mathrm{id}\;,
\end{align*}
where $\nu(\rho)=\frac{3h'_3+2\alpha+h'_1}{4}$. We can expand the first two terms in the above expression of $\nu(G^{(2)})$ as follows:
\begin{align*}
&((\beta_2\gamma_2-\gamma_1\beta_1)_{(-1)}\beta_2)_{(-1)}\gamma_2+(\alpha\beta_2)_{(-1)}\gamma_2\\
=&-\gamma_1\beta_1\beta_2\gamma_2-\partial(\gamma_1\beta_1)+\alpha\beta_2\gamma_2+\partial\alpha\cdot\mathrm{id}+((\beta_2\gamma_2)_{(-1)}\beta_2)_{(-1)}\gamma_2\\
=&-\gamma_1\beta_1\beta_2\gamma_2-\partial(\gamma_1\beta_1)+\alpha\beta_2\gamma_2+\partial\alpha\cdot\mathrm{id}+(\gamma_2\beta_2)_{(-1)}(\beta_2\gamma_2)+\partial(\gamma_2\beta_2)\cdot\mathrm{id}\\
=&-\gamma_1\beta_1\beta_2\gamma_2-\partial(\gamma_1\beta_1)+(h'_3+\alpha)\beta_2\gamma_2+\partial(h'_3+\alpha)\cdot\mathrm{id} \;,
\end{align*}
where in the last line we have imposed the moment map relation $\gamma_2\beta_2\equiv h'_3$. Plug the above expression back to $\nu(G^{(2)})$ and we get
\begin{align*}
\nu(G^{(2)})=&-\gamma_1\beta_1\beta_2\gamma_2-\partial(\gamma_1\beta_1)+(h'_3+\alpha)\beta_2\gamma_2+\nu(\rho)\nu( G^{(1)})\\
&-(\nu(\rho)^2+\partial(\nu(\rho)-h'_3-\alpha))\cdot\mathrm{id}\\
=&\left(J^{[1]}-\frac{\bar\alpha}{2}\mathrm{id}\right)\left(J^{[2]}+\frac{\bar\alpha}{2}\mathrm{id}\right)-\partial\left(J^{[1]}-\frac{\bar\alpha}{2}\mathrm{id}\right)\;.
\end{align*}
Let us re-organize the fields $G^{(1)},G^{(2)},J_B$ in the following way:
\begin{equation}
\begin{split}
W^{(1)}&=G^{(1)}-\frac{J_B}{2}\mathrm{id}\;,\\
W^{(2)}&=G^{(2)}-\frac{1}{4}J_BG^{(1)}+\left(\frac{J_B^2}{16}+\frac{\partial J_B}{4}\right)\cdot\mathrm{id}\;,
\end{split}
\end{equation}
then we get
\begin{align}\label{eq:W in terms of E}
    \nu(W^{(1)})=E^{[1]}+E^{[2]}\;,
    \quad \nu(W^{(2)})=E^{[1]}E^{[2]}-\partial E^{[1]}\;.
\end{align}
The equation \eqref{eq:W in terms of E} can be equivalently presented as the matrix-valued Miura operator:
\begin{align}
    \partial_z^2+\partial_z\cdot \nu(W^{(1)})+\nu(W^{(2)})=(\partial_z+E^{[1]})(\partial_z+E^{[2]}).
\end{align}
Therefore according to \cite{arakawa2017explicit}, $\nu(W^{(1)}),\nu(W^{(2)})$ strongly generate the rectangular $\mathcal{W}$-algebra $\CW^{-3}(\mathfrak{gl}_4,[2^2])$ as a vertex subalgebra of $V^{-1}(\mathfrak{gl}_2)\otimes V^{-1}(\mathfrak{gl}_2)$. 

\bigskip We note that the $T_{[1^4]}^{[2^2]}[\SU(4)]$ quiver satisfies the technical assumption in Theorem \ref{thm:vanishing and embedding}, then it follows from the theorem that the map $\iota:\CV(Q)\to \mathsf D^{\mathrm{ch}}(\widetilde\CM(Q))$ is injective, whence $W^{(1)},W^{(2)}$ strongly generate the rectangular $\mathcal{W}$-algebra $\CW^{-3}(\mathfrak{gl}_4,[2^2])$ as a vertex subalgebra of $\CV(Q)\otimes J_B$. We also note that $\iota':\CV(Q')\to \mathsf D^{\mathrm{ch}}(\widetilde\CM(Q'))$ is injective because the $T[\SU(2)]$ quiver also satisfies the technical assumption in Theorem \ref{thm:vanishing and embedding}. We summarize our conclusion in the following commutative diagram
\begin{equation}
\begin{tikzcd}
\CW^{-3}(\mathfrak{gl}_4,[2^2])\ar[rr,"\text{Miura map}",hook]\ar[d,hook] & & V^{-1}(\mathfrak{gl}_2)\otimes V^{-1}(\mathfrak{gl}_2) \ar[d,hook] \\
\CV(Q)\otimes J_B \ar[rr,"\text{reduction}",hook]\ar[d,"\iota",hook] & & \CV(Q')\otimes\bar\alpha \otimes J_B
\ar[d,"\iota'",hook]\\
\mathsf D^{\mathrm{ch}}(\widetilde\CM(Q))\otimes J_B\ar[rr,"\text{reduction}", hook] & & \mathsf D^{\mathrm{ch}}(\widetilde\CM(Q'))\otimes\bar\alpha \otimes J_B
\end{tikzcd}
\end{equation}
where all arrows are embeddings of vertex algebras.

\bigskip The stress-energy operators $T=T_{\beta\gamma}+T_h+T_{\bc}+T_{J_B}$ is BRST closed, where $T_{\beta\gamma}$, $T_h$, and $T_{\bc}$ are defined in the same way as in the \eqref{def:generalTtensor_general}, and $T_{J_B}=\frac{1}{8}J_B^2$. $T$ gives rise to a stress-energy operator of $\CV(Q)\otimes J_B$. We claim that $T$ is inside the vertex subalgebra $\CW^{-3}(\mathfrak{gl}_4,[2^2])$, and it can be written in terms of $W^{(1)},W^{(2)}$ as
\begin{align}\label{eq:stress-energy operators of [2,2] quiver}
    T=\frac{1}{2}\mathrm{Tr}((W^{(1)})^2)-\mathrm{Tr}(W^{(2)})-\frac{1}{2}\mathrm{Tr}(\partial W^{(1)})~.
\end{align}
To see this, we compute the image of $T$ under the reduction map $\nu$. By \eqref{eq:stress-energy operator under reduction} we have
\begin{align*}
    \nu(T)&=T_{Q'}+\frac{1}{2}\partial\mathrm{Tr}(\gamma_1\beta_1)-\frac{1}{2}\partial\mathrm{Tr}(\beta_2\gamma_2)-\partial\alpha+\frac{1}{8}J_B^2+\frac{1}{2}\bar\alpha^2\\
    &=T_{Q'}-\partial\bar\alpha+\frac{1}{8}J_B^2+\frac{1}{2}\bar\alpha^2~,
\end{align*}
where $T_{Q'}$ is the stress energy operator \eqref{def:generalTtensor_general} for the quiver $Q'$. Using \eqref{T equals Sugawara} and the injectivity of $\iota':\CV(Q')\to \mathsf D^{\mathrm{ch}}(\widetilde\CM(Q'))$, we see that $T_{Q'}$ equals to the Sugawara operator of $V^{-1}(\mathfrak{sl}_2)\otimes V^{-1}(\mathfrak{sl}_2)$, i.e.
\begin{align*}
    T_{Q'}=\frac{1}{2}\mathrm{Tr}((J^{[1]})^2)+\frac{1}{2}\mathrm{Tr}((J^{[2]})^2)~.
\end{align*}
Thus we find
\begin{align*}
    \nu(T)&=\frac{1}{2}\mathrm{Tr}((E^{[1]})^2)+\frac{1}{2}\mathrm{Tr}((E^{[2]})^2)+\frac{1}{2}\mathrm{Tr}(\partial E^{[1]})-\frac{1}{2}\mathrm{Tr}(\partial E^{[2]})\\
    &=\nu\left(\frac{1}{2}\mathrm{Tr}((W^{(1)})^2)-\mathrm{Tr}(W^{(2)})-\frac{1}{2}\mathrm{Tr}(\partial W^{(1)})\right)~.
\end{align*}
By the injectivity of $\nu$, we obtain \eqref{eq:stress-energy operators of [2,2] quiver}. We make two remarks:
\begin{enumerate}
    \item $T$ in \eqref{eq:stress-energy operators of [2,2] quiver} is exactly the stress-energy operator of $\CW^{-3}(\mathfrak{gl}_4,[2^2])$ in \cite[(3.19)]{eberhardt2019matrix};
    \item $T$ equals to $T_{\mathrm{KRW}}+\frac{1}{8}J_B^2$ where $T_{\mathrm{KRW}}$ is Kac-Roan-Wakimoto's stress-energy operator of $\CW^{-3}(\mathfrak{sl}_4,[2^2])$ in \cite[(2.5)]{kac2003quantum}, because $T$ is the unique stress-energy operator of $\CW^{-3}(\mathfrak{gl}_4,[2^2])$ such that $W^{(1)}$ is primary of weight $1$ and $W^{(2)}$ has conformal weight $2$ with respect to $T$ \cite[Section 3.4]{eberhardt2019matrix}, and $T_{\mathrm{KRW}}+\frac{1}{8}J_B^2$ also satisfies these properties.
\end{enumerate}
\bigskip We make a final remark that if we perform the following automorphism on $\CV(Q)\otimes J_B$:
\begin{align*}
    (X^a_b,Y^c_d)\mapsto (Y^b_a,-X^d_c)\;,\quad (I^a_i,J^j_b)\mapsto(J^i_a,-I^b_j)\;,\quad h_i\mapsto -h_i\;,\quad J_B\mapsto -J_B\;,
\end{align*}
then $W^{(1)}$ and $W^{(2)}$ transforms accordingly:
\begin{align*}
    (W^{(1)})^i_j\mapsto -(W^{(1)})^j_i\;,\quad (W^{(2)})^i_j\mapsto (W^{(2)})^j_i-(\partial W^{(1)})^j_i\;.
\end{align*}
The above automorphism on the rectangular $\mathcal{W}$-algebra $\CW^{-3}(\mathfrak{gl}_4,[2^2])$ is equivalent to the one induced from the automorphism of Miura operators
\begin{align*}
(\partial_z+E^{[1]})^i_l(\partial_z+E^{[2]})^l_j\mapsto (\partial_z-E^{[2]})^l_i(\partial_z-E^{[1]})^j_l\;.
\end{align*}

\subsection{Example: \texorpdfstring{$T_{[1^n]}^{[2,1^{n-2}]}[\SU(n)]$}{T} quiver}\label{sec:quiverred211111}

In this subsection we apply the quiver reduction to the rightmost node of the $T_{[1^{n}]}^{[2, 1^{n-2}]}[\SU(n)]$ quiver introduced in section \ref{sec:[21..1]}, see the Fig.~\ref{fig:gl(n), (21..1)}.

\begin{figure}
    \begin{tikzpicture}

    \node (leftfigure) at (0,0){
    \begin{tikzpicture}[scale=0.8, every node/.style={scale=0.8}]
    \node(L) at (-2.5,0) {$\circled{$1$}$};
    \node(M) at (0,0) {$\circled{$2$}$};
    \node(R) at (2.5,0) {$\circled{$n-2$}$};
    \node(RR) at (5,0) {$\circled{$n-2$}$};
    \node(B) at (2.5,-2.5) {$\rectangled{$1$}$};
    \node(BB) at (5,-2.5) {$\rectangled{$n-2$}$};
    \draw[->] ([yshift=0.1cm]L.east) to node[above]{$X_{n-1}$} ([yshift=0.1cm]M.west);
    \draw[->] ([yshift=-0.1cm]M.west) to node[below]{$Y_{n-1}$} ([yshift=-0.1cm]L.east);
    \draw[->] ([yshift=0.1cm]R.east) to node[above]{$X_1$} ([yshift=0.1cm]RR.west);
    \draw[->] ([yshift=-0.1cm]RR.west) to node[below]{$Y_1$} ([yshift=-0.1cm]R.east);
    \draw[dotted,thick] ([yshift=-0.1cm]M.east) to ([yshift=-0.1cm]R.west);
    \draw[dotted,thick] ([yshift=0.1cm]M.east) to ([yshift=0.1cm]R.west);
    \draw[->] ([xshift=0.1cm]RR.south) to node[right]{$J_1$} ([xshift=0.1cm]BB.north);
    \draw[->] ([xshift=-0.1cm]BB.north) to node[left]{$I_1$} 
    ([xshift=-0.1cm]RR.south);
    \draw[->] ([xshift=0.1cm]R.south) to node[right]{$J_2$} ([xshift=0.1cm]B.north);
    \draw[->] ([xshift=-0.1cm]B.north) to node[left]{$I_2$} 
    ([xshift=-0.1cm]R.south);
    \node(hn) at (-2.5,0.7) {$h_{n-1}$};
    \node(hnm) at (0,0.7) {$h_{n-2}$};
    \node(h2) at (2.5,0.7) {$h_2$};
    \node(h1) at (5,0.7) {$h_1$};
    \end{tikzpicture}
    };

    \draw[->,thick] (2.3,1.8) to [bend left=45] node[above]{reduction} (5.5, 1.8);
     
    \node (rightfigure) at (8,0.1){
    \begin{tikzpicture}[scale=0.8, every node/.style={scale=0.8}]
    \node(L) at (-2.5,0) {$\circled{$1$}$};
    \node(M) at (0,0) {$\circled{$2$}$};
    \node(R) at (2.5,0) {$\circled{$n-2$}$};
    \node(RR) at (5,0) {$\rectangled{$n-2$}$};
    \node(B) at (2.5,-2.5) {$\rectangled{$1$}$};
    
    \draw[->] ([yshift=0.1cm]L.east) to node[above]{$X_{n-1}$} ([yshift=0.1cm]M.west);
    \draw[->] ([yshift=-0.1cm]M.west) to node[below]{$Y_{n-1}$} ([yshift=-0.1cm]L.east);
    \draw[->] ([yshift=0.1cm]R.east) to node[above]{$\beta$} ([yshift=0.1cm]RR.west);
    \draw[->] ([yshift=-0.1cm]RR.west) to node[below]{$\gamma$} ([yshift=-0.1cm]R.east);
    \draw[dotted,thick] ([yshift=-0.1cm]M.east) to ([yshift=-0.1cm]R.west);
    \draw[dotted,thick] ([yshift=0.1cm]M.east) to ([yshift=0.1cm]R.west);
    \draw[->] ([xshift=0.1cm]R.south) to node[right]{$J_2$} ([xshift=0.1cm]B.north);
    \draw[->] ([xshift=-0.1cm]B.north) to node[left]{$I_2$} 
    ([xshift=-0.1cm]R.south);
    \node(hn) at (-2.5,0.7) {$h_{n-1}'$};
    \node(hnm) at (0,0.7) {$h_{n-2}'$};
    \node(h2) at (2.5,0.7) {$h_2'$};
    \node(h1) at (5,0.7) {$\alpha$};
    \end{tikzpicture}
    };
    \end{tikzpicture}
    \caption{Reduction procedure of the quiver, as applied to the second node of the $T_{[1^n]}^{[2,1^{n-2}]}[\SU(n)]$ quiver.}
    \label{fig:gl(n), (21..1)}
\end{figure}
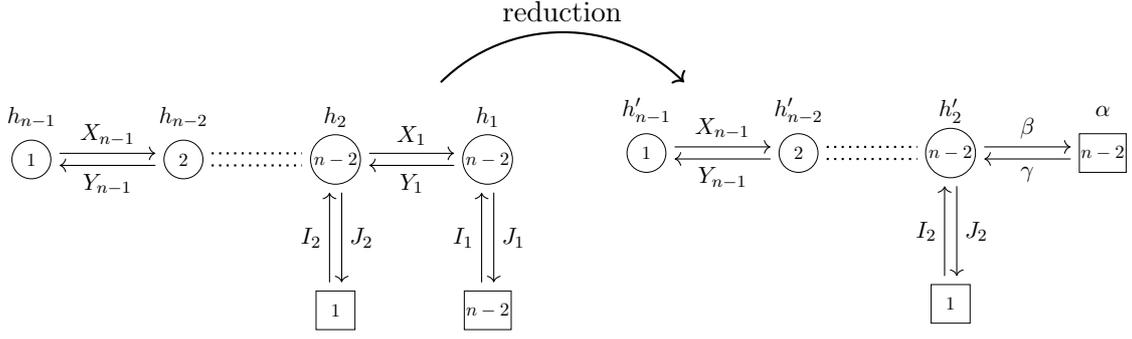

According to our reduction algorithm, we have
\begin{align}
\beta=I^{-1}X_1\;, \quad \gamma=Y_1I\;,
\end{align}
and
\begin{align}
\alpha=h_1+2\mathrm{Tr}(I^{-1}\partial I)\;,
\quad h'_2=h_2+\frac{h_1-\alpha}{2}=h_2-\mathrm{Tr}(I^{-1}\partial I),\quad h'_j=h_j\;, 
\: j>2 \;.
\end{align}
The modified Heisenberg field \eqref{eq:modified Heisenberg 2} reads:
\begin{align}
    \bar\alpha=\alpha+\sum_{j=1}^{n-2}\frac{jh'_{n-j}}{n-1}\;.
\end{align}
 After reduction, we get a tensor product of $T[\SU(n-1)]$ quiver VOA together with one copy of Heisenberg field $\bar\alpha$ of level $\frac{n}{n-1}$. We note that there is a vertex subalgebra $V^{-n+2}(\mathfrak{sl}_{n-1})\subset \CV(Q')$ strongly generated by operators
\begin{align}
J^i_j=\begin{cases}
-(\beta\gamma)^i_j+\sum_{\ell=1}^{n-2}\frac{\ell h'_{n-\ell}}{n-1}\delta^i_j\;,& 1\le i,j\le n-2\;,\\
-(J_2\gamma)_j\;, & i=n-1, 1\le j\le n-2\;,\\
(\beta I_2)^i\;, & 1\le i\le n-2, j=n-1\;,\\
J_2I_2+\sum_{\ell=1}^{n-2}\frac{\ell h'_{n-\ell}}{n-1}\;, & i=j=n-1\;.
\end{cases}
\end{align}
To match with our setting in the linear quiver VOA, we introduce free bosons $b_0,\cdots,b_{n-1}$ such that $b_i(z)b_j(w)\sim\frac{\delta_{ij}}{(z-w)^2}$ and set $h_i=b_i-b_{i-1}$. We note that $b_i$ is related to the free fermions by $b_i=\chi_i\psi_i$. The field $J_B=\sum_{i=0}^{n-1}b_i$ commutes with $h_1,\cdots,h_{n-1}$. We re-arrange the fields $J^i_j,\bar\alpha,J_B$ into the following
\begin{align}
    E^i_j=J^i_j-\frac{\bar\alpha+J_B}{n}\delta^i_j\;,\quad \widetilde\alpha=\frac{n-1}{n}\bar\alpha-\frac{1}{n}J_B\;,
\end{align}
then $E$ commutes with $\widetilde\alpha$ and they generate a vertex subalgebra $V^{-n+2}(\mathfrak{gl}_{n-1})\otimes \mathcal{H}^1\subset \CV(Q')\otimes\bar\alpha\otimes J_B$.

\bigskip Now consider the following operators before reduction: 
\begin{equation}\label{eq:G generators}
\begin{split}
(G^{(1,1)})^i_j&=-(J_1I_1)^i_j+b_0\delta^i_j,\;\\
(G^{(1,2)})^i&=-(J_1X_1I_2)^i\;,\\
(G^{(2,1)})_i&=(J_2Y_1I_1)_i\;,\\
G^{(2,2)}_1&=J_B-\mathrm{Tr}(G^{(1,1)})\;,\\
G^{(2,2)}_2&=g^{(2,2)}-S+\frac{b_0+b_1}{4}(b_0+b_1-2J_2I_2)-(n-2)\mathrm{Tr}(I_2\partial J_2)\;,
\end{split}
\end{equation}
where $1\le i,j\le n-2$. $g^{(2,2)}$ and $S$ are the following fields:
\begin{equation}
\begin{split}
&g^{(2,2)}:=(I_2J_2)^a_{b}\left((Y_1X_1)^b_a+\frac{h_1}{2}\delta^a_b\right)\;,\\
&S:=\mathrm{Tr}(X_1\partial Y_1)-\mathrm{Tr}(J_1\partial I_1)+\frac{h_1^2}{4}-\mathrm{Tr}(\mathsf b_1\partial \mathsf c_1)\;,
\end{split}
\end{equation}
where $\mathsf b_1,\mathsf c_1$ are the $\mathsf b\mathsf c$ ghosts of the rightmost node. We note that $S$ is a stress-energy operator corresponding to the neighboring fields of the rightmost gauge node. It is easy to see that $G^{(1,1)}$, $G^{(1,2)}$, $G^{(2,1)}$, and $G^{(2,2)}_1$ are BRST closed. It is true, but not obvious, that $G^{(2,2)}_2$ is also BRST closed. Nevertheless, it is easy to see that $G^{(2,2)}_2$ is BRST closed with respect to the rightmost gauge node. We will show that $G^{(2,2)}_2$ is also BRST closed with respect to all gauge nodes using the quiver reduction.

Let us compute their images after reduction.
\begin{equation}
\begin{split}
\nu((G^{(1,1)})^i_j)&=-\nu\left((J_1I_1)^i_j+\sum_{\ell=1}^{n-1}\frac{\ell h_{n-\ell}}{n}\delta^i_j\right)+\frac{J_B}{n}\delta^i_j\\
&=(\beta\gamma)^i_j-\left(\sum_{\ell=1}^{n-2}\frac{\ell h'_{n-\ell}}{n}-\frac{\alpha+J_B}{n}\right)\delta^i_j\\
&=-E^i_j\;.\\
\nu((G^{(2,1)})_i)&=(J_2\gamma)_i=-E^{n-1}_i\;, \\
\nu((G^{(1,2)})^i)&=(\beta\gamma)^i_{j,(-1)}(\beta I_2)^j+\alpha(\beta I_2)^i=\widetilde\alpha E^i_{n-1}-\sum_{\ell=1}^{n-2}E^i_{\ell,(-1)} E^\ell_{n-1}.\; \\
\nu(G^{(2,2)}_1)&=J_B+\sum_{i=1}^{n-2}E^i_i=-E^{n-1}_{n-1}-\widetilde\alpha\;.
\end{split}
\end{equation}
The hard one is $\nu(G^{(2,2)}_2)$, and let us compute $\nu(g^{(2,2)})$ first:
\begin{align*}
\nu(g^{(2,2)})&=(I_2J_2)^a_{b}\left((\gamma\beta)^b_a+\frac{\alpha}{2}\delta^a_b\right)=J_2\gamma\beta I_2+\frac{\alpha}{2}J_2I_2\\
&=(J_2\gamma)_{i,(-1)}(\beta I_2)^i+\frac{\alpha}{2}J_2I_2+(n-2)\mathrm{Tr}(I_2\partial J_2) +\mathrm{Tr}(\beta\partial\gamma)\\
&=-\sum_{\ell=1}^{n-2}E^{n-1}_{\ell,(-1)}E^\ell_{n-1}+\left(\widetilde\alpha+\frac{b_0+b_1}{2}\right)J_2I_2+(n-2)\mathrm{Tr}(I_2\partial J_2) +\mathrm{Tr}(\beta\partial\gamma)\\
&=\widetilde\alpha E^{n-1}_{n-1}-\sum_{\ell=1}^{n-2}E^{n-1}_{\ell,(-1)}E^\ell_{n-1}+\frac{\alpha^2}{4}+\mathrm{Tr}(\beta\partial\gamma)\\
&~~ +\frac{b_0+b_1}{2}I_2J_2+(n-2)\mathrm{Tr}(I_2\partial J_2)-\frac{(b_0+b_1)^2}{4}\;.
\end{align*}
According to \eqref{eq:stress-energy operator under reduction}, we have 
\begin{align*}
    \nu(S)=\frac{\alpha^2}{4}+\mathrm{Tr}(\beta\partial\gamma)\;.
\end{align*}
Together we get
\begin{align}
    \nu(G^{(2,2)}_2)=\widetilde\alpha E^{n-1}_{n-1}-\sum_{\ell=1}^{n-2}E^{n-1}_{\ell,(-1)}E^\ell_{n-1}\;.
\end{align}
In particular, we see that $\nu(G^{(2,2)}_2)$ is BRST closed. This implies that $G^{(2,2)}_2$ becomes BRST closed after localizing $\det(I_1)$. Since localization induces an injective map between relative BRST complexes, hence $G^{(2,2)}_2$ is BRST closed.

According to \cite[Appendix A]{ueda2022example}, the BRST cohomology classes $\nu(G^{(1,1)})$, $\nu(G^{(1,2)})$, $\nu(G^{(2,1)})$, $\nu(G^{(2,2)}_1)$, $\nu(G^{(2,2)}_2)$ strongly generate a vertex subalgebra $$\CW^{-n+1}(\mathfrak{gl}_n,f_{\min})\subset V^{-n+2}(\mathfrak{gl}_{n-1})\otimes \mathcal{H}^1\subset \CV(Q')\otimes\bar\alpha\otimes J_B.$$ The map is given by
\begin{equation}
\begin{split}
\nu((G^{(1,1)})^i_j)\mapsto W^{(1)}_{j,i},\quad &\nu((G^{(2,1)})_j)\mapsto W^{(1)}_{j,n-1},\quad \nu(G^{(2,2)}_1)\mapsto W^{(1)}_{n-1,n-1},\\
\nu((G^{(1,2)})^i)&\mapsto W^{(2)}_{n-1,i},\quad 
\nu(G^{(2,2)}_2)\mapsto W^{(2)}_{n-1,n-1}\;.
\end{split}
\end{equation}
Since the quiver reduction map \eqref{eq:reduction map of global sections} for the vertex algebra of the global section is injective, we see that the BRST cohomology classes $G^{(1,1)}$, $G^{(1,2)}$, $G^{(2,1)}$, $G^{(2,2)}_1$, $G^{(2,2)}_2$ strongly generate a vertex subalgebra in $\mathsf D^{\mathrm{ch}}(\widetilde\CM(Q))\otimes J_B$ which is isomorphic to $\CW^{-n+1}(\mathfrak{gl}_n,f_{\min})$.
We summarize it in the following commutative diagram:
\begin{equation}
\begin{tikzcd}
\CW^{-n+1}(\mathfrak{gl}_n,f_{\min})\ar[rr,"\text{Miura map}",hook]\ar[d,hook] & & V^{-n+2}(\mathfrak{gl}_{n-1})\otimes \mathcal{H}^1 \ar[d,hook] \\
\mathsf D^{\mathrm{ch}}(\widetilde\CM(Q))\otimes J_B\ar[rr,"\text{reduction}", hook] & & \mathsf D^{\mathrm{ch}}(\widetilde\CM(Q'))\otimes\bar\alpha \otimes J_B
\end{tikzcd}
\end{equation}
We also note that the $T_{[1^4]}^{[2,1^2]}[\SU(4)]$ quiver satisfies the technical assumption in the Theorem \ref{thm:vanishing and embedding}, then it follows from the theorem that the map $\iota:\CV(Q)\to \mathsf D^{\mathrm{ch}}(\widetilde\CM(Q))$ is injective, whence $W^{(1)},W^{(2)}$ strongly generate the minimal nilpotent $\mathcal{W}$-algebra $\CW^{-3}(\mathfrak{gl}_4,[2,1^2])$ as a vertex subalgebra of $\CV(Q)\otimes J_B$. Thus we also have the commutative diagram
\begin{equation}
\begin{tikzcd}
\CW^{-3}(\mathfrak{gl}_4,[2,1^2])\ar[rr,"\text{Miura map}",hook]\ar[d,hook] & & V^{-2}(\mathfrak{gl}_{3})\otimes \mathcal{H}^1 \ar[d,hook] \\
\CV(Q)\otimes J_B \ar[rr,"\text{reduction}",hook] & & \CV(Q')\otimes\bar\alpha \otimes J_B
\end{tikzcd}
\end{equation}

\bigskip The stress-energy operator $T=T_{\beta\gamma}+T_h+T_{\bc}+T_{J_B}$ is BRST closed, where $T_{\beta\gamma}$, $T_h$, and $T_{\bc}$ are defined in the same way as in the \eqref{def:generalTtensor_general}, and $T_{J_B}=\frac{1}{2n}J_B^2$. $T$ gives rise to a stress-energy operator of $\mathsf D^{\mathrm{ch}}(\widetilde\CM(Q))\otimes J_B$. We claim that $T$ is inside the vertex subalgebra $\CW^{-n+1}(\mathfrak{gl}_n,f_{\min})$, and it can be written in terms of $W^{(1)},W^{(2)}$ as
\begin{align}
    T=-W^{(2)}_{n-1,n-1}+\frac{1}{2}(W^{(1)}_{n-1,n-1})^2+\frac{n-2}{2}\partial W^{(1)}_{n-1,n-1}+\frac{1}{2}\sum_{i,j=1}^{n-2}W^{(1)}_{i,j}W^{(1)}_{j,i}~.
\end{align}
To see this, we compute the image of $T$ under the reduction map $\nu$. By \eqref{eq:stress-energy operator under reduction} we have
\begin{align*}
    \nu(T)&=T_{Q'}-\frac{1}{2}\partial\mathrm{Tr}(\beta\gamma)-\frac{n-2}{2}\partial\alpha+\frac{1}{2(n-1)}(J_B+\widetilde\alpha)^2+\frac{1}{2}{\widetilde\alpha}^2~,
\end{align*}
where $T_{Q'}$ is the stress energy operator \eqref{def:generalTtensor_general} for the quiver $Q'$. Using \eqref{T equals Sugawara}, we have
\begin{align*}
\nu(T)&=\frac{1}{2}\sum_{i,j=1}^{n-1}E^i_jE^j_i+\frac{1}{2}\sum_{i=1}^{n-2}\partial E^i_i-\frac{n-2}{2}\partial\widetilde\alpha+\frac{1}{2}{\widetilde\alpha}^2\\
&=-W^{(2)}_{n-1,n-1}+\frac{1}{2}(W^{(1)}_{n-1,n-1})^2+\frac{n-2}{2}\partial W^{(1)}_{n-1,n-1}+\frac{1}{2}\sum_{i,j=1}^{n-2}W^{(1)}_{i,j}W^{(1)}_{j,i}~.
\end{align*}
Therefore the claim is proven. We note that $T$ agrees with $T_{\mathrm{KRW}}+\frac{1}{2n}J_B^2$, where $T_{\mathrm{KRW}}$ is Kac-Roan-Wakimoto's stress-energy operator of $\CW^{-n+1}(\mathfrak{sl}_n,f_{\mathrm{min}})$ in \cite[(2.5)]{kac2003quantum}, because there is a unique\footnote{Uniqueness follows from the fact that $W^{(1)}$ and $W^{(2)}_{n-1,i}$ for $1\le i\le n-2$ generate the whole $\CW^{-n+1}(\mathfrak{gl}_n,f_{\mathrm{min}})$, cf. \eqref{eq: W([2,1^{n-2}]) OPE}. In fact, if $T_1$ and $T_2$ are two stress-energy operators such that $W^{(1)}$ and $W^{(2)}_{n-1,i}$ for $1\le i\le n-2$ are primary and their weights under the $T_1$ and $T_2$ are equal, then $T_1-T_2$ commutes with $W^{(1)}$ and $W^{(2)}_{n-1,i}$ for $1\le i\le n-2$, whence it commutes with $\CW^{-n+1}(\mathfrak{gl}_n,f_{\mathrm{min}})$. According to Lemma \ref{lem:criterion for two Virasoro coincide_1}, we must have $T_1=T_2$.} stress-energy operator of $\CW^{-n+1}(\mathfrak{gl}_n,f_{\mathrm{min}})$ such that $W^{(1)}_{i,j}$ for $1\le i,j\le n-2$ and $W^{(1)}_{n-1,n-1}$ are primary of weight $1$, and $W^{(1)}_{i,n-1}$ and $W^{(2)}_{n-1,i}$ for $1\le i\le n-2$ are primary of weight $3/2$, with respect to such stress-energy operator. Both $T$ and $T_{\mathrm{KRW}}+\frac{1}{2n}J_B^2$ satisfy such properties, thus
\begin{align}
    T=T_{\mathrm{KRW}}+\frac{1}{2n}J_B^2~.
\end{align}

\bigskip In section \ref{sec:[22]}, we mentioned that $T^{[2,2]}_{[1^4]}[\SU(4)]$ can be obtained from the quiver reduction of $T^{[2,1,1]}_{[1^4]}[\SU(4)]$. As in Fig.~\ref{fig:gl(4), (211) and (22)}, one can notice $T^{[2,1,1]}_{[1^4]}[\SU(4)]$ after the quiver reduction provides a slightly different outcome compared to the quiver reduction of $T^{[2,2]}_{[1^4]}[\SU(4)]$. More precisely,  $V^{-1}(\mathfrak{gl}_2)\otimes D^{\mathrm{ch}}(T^*\mathbb C^{2}) \otimes (\cH^{1})^{\otimes2}$ and $V^{-1}(\mathfrak{gl}_2)\otimes V^{-1}(\mathfrak{gl}_2)$, and the former contains the latter. Further application of Miura map, following Proposition \ref{prop:2}, yield $D^{\mathrm{ch}}(T^*\mathbb C^{3}) \otimes (\cH^{1})^{\otimes4}$ and $D^{\mathrm{ch}}(T^*\mathbb C^{2}) \otimes (\cH^{1})^{\otimes4}$. The gauging of $\U(1)$ reduces the complex dimension of the associated super vector space $D^{\mathrm{ch}}(T^*\BC^{3})\to D^{\mathrm{ch}}(T^*\mathbb C^{2})$.

\bigskip We make a final remark that our choice of generators \eqref{eq:G generators} is related to the previous ones in \eqref{eq: VOA map} by the automorphism $(X^a_b,Y^c_d)\mapsto (Y^b_a,-X^d_c)$, $(I^a_i,J^j_b)\mapsto(J^i_a,-I^b_j)$.

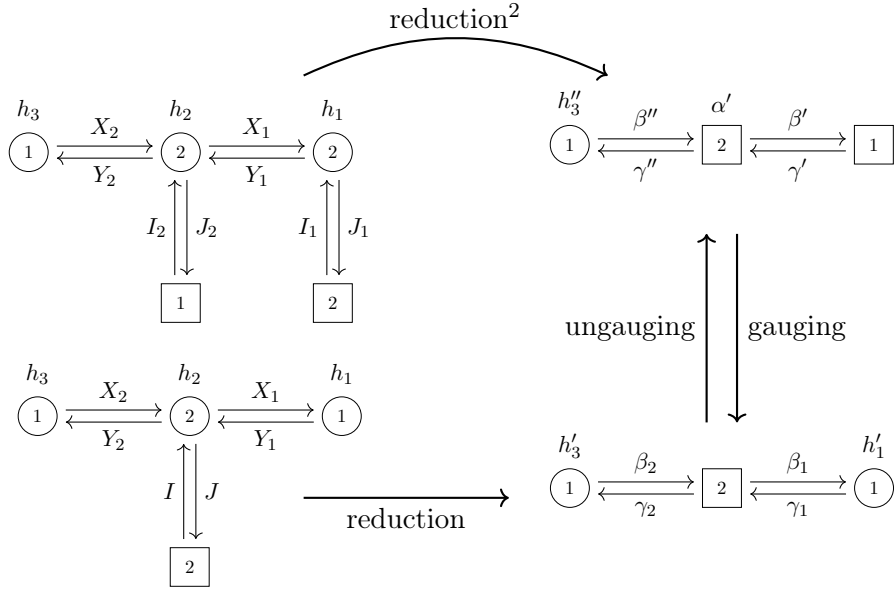
\begin{figure}[ht!]
\centering
    \begin{tikzpicture}
    \node (leftfigure) at (0,0){
    \begin{tikzpicture}[scale=0.8, every node/.style={scale=0.8}]
    \node(M) at (0,0) {$\circled{$1$}$};
    \node(R) at (2.5,0) {$\circled{$2$}$};
    \node(RR) at (5,0) {$\circled{$2$}$};
    \node(B) at (2.5,-2.5) {$\rectangled{$1$}$};
    \node(BB) at (5,-2.5) {$\rectangled{$2$}$};
    \draw[->] ([yshift=0.1cm]M.east) to node[above]{$X_{2}$} ([yshift=0.1cm]R.west);
    \draw[->] ([yshift=-0.1cm]R.west) to node[below]{$Y_{2}$} ([yshift=-0.1cm]M.east);
    \draw[->] ([yshift=0.1cm]R.east) to node[above]{$X_1$} ([yshift=0.1cm]RR.west);
    \draw[->] ([yshift=-0.1cm]RR.west) to node[below]{$Y_1$} ([yshift=-0.1cm]R.east);

    \draw[->] ([xshift=0.1cm]RR.south) to node[right]{$J_1$} ([xshift=0.1cm]BB.north);
    \draw[->] ([xshift=-0.1cm]BB.north) to node[left]{$I_1$} 
    ([xshift=-0.1cm]RR.south);
    \draw[->] ([xshift=0.1cm]R.south) to node[right]{$J_2$} ([xshift=0.1cm]B.north);
    \draw[->] ([xshift=-0.1cm]B.north) to node[left]{$I_2$} 
    ([xshift=-0.1cm]R.south);
    \node(hnm) at (0,0.7) {$h_{3}$};
    \node(h2) at (2.5,0.7) {$h_2$};
    \node(h1) at (5,0.7) {$h_1$};
    \end{tikzpicture}
    };

    \draw[->,thick] (1.5,1.8) to [bend left=25] node[above]{reduction$^2$} (5.5, 1.8);

    \node (rightfigure) at (7,1){
    \begin{tikzpicture}[scale=0.8, every node/.style={scale=0.8}]
    \node(M) at (0,0) {$\circled{$1$}$};
    \node(R) at (2.5,0) {$\rectangled{$2$}$};
    \node(RR) at (5,0) {$\rectangled{$1$}$};
    
    \draw[->] ([yshift=0.1cm]M.east) to node[above]{$\beta''$} ([yshift=0.1cm]R.west);
    \draw[->] ([yshift=-0.1cm]R.west) to node[below]{$\gamma''$} ([yshift=-0.1cm]M.east);
    \draw[->] ([yshift=0.1cm]R.east) to node[above]{$\beta'$} ([yshift=0.1cm]RR.west);
    \draw[->] ([yshift=-0.1cm]RR.west) to node[below]{$\gamma'$} ([yshift=-0.1cm]R.east);

    ([xshift=-0.1cm]R.south);
    \node(hnm) at (0,0.7) {$h_{3}''$};
    \node(h2) at (2.5,0.7) {$\a'$};
    \node(h1) at (5,0.7) {$$};
    \end{tikzpicture}
    };
    \draw[->,thick] (7.2,-0.3) to  node[right]{gauging} (7.2, -2.8);
    \draw[->,thick] (6.8, -2.8) to node[left]{ungauging} (6.8,-0.3) ;

    \node (downrightfigure) at (7,-3.5){
    \begin{tikzpicture}[scale=0.8, every node/.style={scale=0.8}]
     \node(L) at (-2.5,0) {$\circled{$1$}$};
    \node(M) at (0,0) {$\rectangled{$2$}$};
    \node(R) at (2.5,0) {$\circled{$1$}$};
    \draw[->] ([yshift=0.1cm]L.east) to node[above]{$\beta_2$} ([yshift=0.1cm]M.west);
    \draw[->] ([yshift=-0.1cm]M.west) to node[below]{$\gamma_2$} ([yshift=-0.1cm]L.east);
    \draw[->] ([yshift=0.1cm]M.east) to node[above]{$\beta_1$} ([yshift=0.1cm]R.west);
    \draw[->] ([yshift=-0.1cm]R.west) to node[below]{$\gamma_1$} ([yshift=-0.1cm]M.east);
    \node(h1) at (-2.5,0.7) {$h_3'$};
    \node(h3) at (2.5,0.7) {$h_1'$};
    \end{tikzpicture}
    };
    \draw[->,thick] (1.5,-3.8) to node[below]{reduction} (4.2, -3.8);
    
    \node (downleftfigure) at (0,-3.5){
    \begin{tikzpicture}[scale=0.8, every node/.style={scale=0.8}]
    \node(L) at (-2.5,0) {$\circled{$1$}$};
    \node(M) at (0,0) {$\circled{$2$}$};
    \node(R) at (2.5,0) {$\circled{$1$}$};
    \node(B) at (0,-2.5) {$\rectangled{$2$}$};
    \draw[->] ([yshift=0.1cm]L.east) to node[above]{$X_2$} ([yshift=0.1cm]M.west);
    \draw[->] ([yshift=-0.1cm]M.west) to node[below]{$Y_2$} ([yshift=-0.1cm]L.east);
    \draw[->] ([yshift=0.1cm]M.east) to node[above]{$X_1$} ([yshift=0.1cm]R.west);
    \draw[->] ([yshift=-0.1cm]R.west) to node[below]{$Y_1$} ([yshift=-0.1cm]M.east);
    \draw[->] ([xshift=0.1cm]M.south) to node[right]{$J$} ([xshift=0.1cm]B.north);
    \draw[->] ([xshift=-0.1cm]B.north) to node[left]{$I$} ([xshift=-0.1cm]M.south);
    \node(h1) at (-2.5,0.7) {$h_3$};
    \node(h2) at (0,0.7) {$h_2$};
    \node(h3) at (2.5,0.7) {$h_1$};
    \end{tikzpicture}
    };    
    \end{tikzpicture}
    \caption{Subtle difference of quiver reduction of $T_{[1^4]}^{[2,1^{2}]}[\SU(4)]$ quiver and $T_{[1^4]}^{[2,2]}[\SU(4)]$ quiver.}
    \label{fig:gl(4), (211) and (22)}
\end{figure}

\subsection{Example: \texorpdfstring{$T_{[1^n]}^{[n-1,1]}[\SU(n)]$}{T} quiver}\label{sec:quiverred subreg}

In this subsection we recap the $T_{[1^n]}^{[n-1,1]}[\SU(n)]$ quiver which was previously studied by Kuwabara \cite{kuwabara2017vertex}. According to the result in \cite[10.1]{kuwabara2017vertex}, $\mathsf D^{\mathrm{ch}}(\widetilde\CM(Q))\otimes J_B$ contains a vertex subalgebra which is isomorphic to $\CW^{-n+1}(\mathfrak{gl}_n,f_{\mathrm{sub}})$ and is generated by 
\begin{align}\label{generators of subregular W-algebra}
    E=J_{n-1}X_{n-2}\cdots X_1 I_1,\quad F=J_1 Y_1\cdots Y_{n-2} I_{n-1}, \quad J_B=\sum_{i=0}^{n-1}b_i,\quad \mathring H=J_1I_1-b_0.
\end{align}
Here $f_{\mathrm{sub}}$ denotes a sub-regular nilpotent element in $\mathfrak{gl}_n$. We note that the $T_{[1^n]}^{[n-1,1]}[\SU(n)]$ quiver satisfies the technical assumption in the Theorem \ref{thm:vanishing and embedding}, then it follows from the theorem that the map $\iota:\CV(Q)\to \mathsf D^{\mathrm{ch}}(\widetilde\CM(Q))$ is injective. Since the strong generators \eqref{generators of subregular W-algebra} are inside the image of $\iota$, so we see that $\CW^{-n+1}(\mathfrak{gl}_n,f_{\mathrm{sub}})$ is a vertex subalgebra of $\CV(Q)\otimes J_B$.

Let us re-derive Kuwabara's result using the quiver reduction. The quiver reduction at the first node induces an injective vertex algebra map
\begin{align*}
    \nu: \CV(Q)\otimes J_B\hookrightarrow \CV(Q')\otimes J'_B\otimes b'_0,
\end{align*}
where 
\begin{align}
    J'_B=\sum_{i=1}^{n-1}b'_i,\quad b'_j=\begin{cases}
        b_0-I_1^{-1}\partial I_1, & j=0\\
        b_1+I_1^{-1}\partial I_1, & j=1\\
        b_j, & j>1
    \end{cases}.
\end{align}
After reduction, we write down the analog of operators $E,F,\mathring H$ in $\CV(Q')\otimes J'_B$:
\begin{align*}
     E'=J'_{n-1}X'_{n-2}\cdots X'_2 I'_2,\quad F'=J'_2 Y'_2\cdots Y'_{n-2} I'_{n-1}, \quad \mathring H'=J'_2I'_2-b'_1.
\end{align*}
The quiver reduction algorithm gives 
\begin{align}\label{reduction map for subregular}
    \nu(E)=E',\quad \nu(F)=\mathring H'F'+b'_0F',\quad \nu(\mathring H)=\mathring H',\quad \nu(J_B)=J'_B+b'_0.
\end{align}
Comparing with \cite[7.3]{genra2020screening}, we claim that the following map
\begin{align}\label{match of path operators with sub-regular W generators}
    E_{\mathrm{there}}=E_{\mathrm{here}},\quad F_{\mathrm{there}}=F_{\mathrm{here}},\quad Z=J_B,\quad H=\mathring H-\frac{1}{n}J_B,
\end{align}
gives rise to an isomorphism between the vertex subalgebra generated by \eqref{generators of subregular W-algebra} and $\CW^{-n+1}(\mathfrak{gl}_n,f_{\mathrm{sub}})$. When $n=2$, the map $\nu$ gives the Wakimoto realization of $V^{-1}(\mathfrak{gl}_2)$, in this case the sub-regular nilpotent element is zero, so we get $V^{-1}(\mathfrak{gl}_2)=\CW^{-1}(\mathfrak{gl}_2,f_{\mathrm{sub}})$, and \eqref{match of path operators with sub-regular W generators} gives the isomorphism. When $n>2$, we use the induction on $n$, and notice that \eqref{reduction map for subregular} agrees with the formula \cite[(7.4)-(7.7)]{genra2020screening} of the parabolic induction map 
\begin{align*}
    \Delta: \CW^{-n+1}(\mathfrak{gl}_n,f_{\mathrm{sub}})\hookrightarrow \CW^{-n+2}(\mathfrak{gl}_{n-1},f_{\mathrm{sub}})\otimes \mathcal H^1,
\end{align*}
by taking $N_1=n-1$, $N_2=1$, and identifying $W_1=b'_0$. This proves our claim. The above argument also shows that the quiver reduction map is compatible with parabolic induction, i.e. the following diagram commutes:
\begin{equation}
\begin{tikzcd}
\CW^{-n+1}(\mathfrak{gl}_n,f_{\mathrm{sub}})\ar[rr,"\Delta",hook]\ar[d,hook] & & \CW^{-n+2}(\mathfrak{gl}_{n-1},f_{\mathrm{sub}})\otimes \mathcal H^1 \ar[d,hook] \\
\CV(Q)\otimes J_B \ar[rr,"\text{reduction}", hook] & & \CV(Q')\otimes J'_B\otimes b'_0
\end{tikzcd}
\end{equation}

\bigskip $T=T_{\beta\gamma}+T_h+T_{\bc}+T_{J_B}$ is BRST closed, where $T_{\beta\gamma}$, $T_h$, and $T_{\bc}$ are defined in the same way as in the \eqref{def:generalTtensor_general}, and $T_{J_B}=\frac{1}{2n}J_B^2$.

It is shown in \cite[Theorem 6.9]{genra2017screening} that when $k+n\neq 0$, $\CW^k(\mathfrak{sl}_n,f_{\mathrm{sub}})$ is isomorphic to Feigin-Semikhatov's $\CW^{(2)}_n$ algebra of level $k$, which is defined in \cite{feigin2004wn}. Note that in our setting, $\CW^{-n+1}(\mathfrak{sl}_n,f_{\mathrm{sub}})$ is generated by $E$, $F$, and $H=\mathring H-\frac{1}{n}J_B$. It is shown in \cite{feigin2004wn} that when $k+n\neq 0$, $\CW^{(2)}_n$ has a stress-energy operators $T_{\mathrm{FS}}$ such that $E$, $F$, and $H$ are primary of weights $n/2$, $n/2$, and $1$, with respect to $T_{\mathrm{FS}}$. In fact, $T_{\mathrm{FS}}$ is the unique stress-energy operator which has such property\footnote{To see the uniqueness, assume that $T'$ is another stress-energy operator such that $E$, $F$, and $H$ are primary of weights $n/2$, $n/2$, and $1$, with respect to $T'$, then $T_{\mathrm{FS}}-T'$ commutes with $E$, $F$, and $H$, therefore it commutes with the whole $\CW^{(2)}_n$ algebra. According to Lemma \ref{lem:criterion for two Virasoro coincide_1}, we must have $T_{\mathrm{FS}}=T'$.}. Similarly $T_{\mathrm{FS}}+\frac{1}{2n}J_B^2$ is the unique stress-energy operator of $\CW^{-n+1}(\mathfrak{sl}_n,f_{\mathrm{sub}})$ such that $E$, $F$, $Z$ and $\mathring H$ are primary of weights $n/2$, $n/2$, $1$ and $1$ with respect to it. It is easy to see that $T$ also satisfy the aforementioned properties, thus we have
\begin{align}\label{eq:compare with T_FS}
    T=T_{\mathrm{FS}}+\frac{1}{2n}J_B^2~.
\end{align}
\begin{lem}\label{lem:T_FS=T_KRW}
For all $k$ such that $k+n\neq 0$, Kac-Roan-Wakimoto's stress-energy operator $T_{\mathrm{KRW}}$ of $\CW^{k}(\mathfrak{sl}_n,f_{\mathrm{sub}})$ in \cite[(2.5)]{kac2003quantum} is identified with Feigin-Semikhatov's stress-energy operator $T_{\mathrm{FS}}$ along the isomorphism $\CW^{k}(\mathfrak{sl}_n,f_{\mathrm{sub}})\cong \CW^{(2)}_n$ in \cite[Theorem 6.9]{genra2017screening}.
\end{lem}

\begin{proof}
When $n=2$ this is obvious because both $T_{\mathrm{KRW}}$ and $T_{\mathrm{FS}}$ are equal to the Sugawara operator in $V^k(\mathfrak{sl}_2)$. So let us assume $n\ge 3$. We recall that the grading induced by $T_{\mathrm{KRW}}$ on $\CW^{k}(\mathfrak{sl}_n,f_{\mathrm{sub}})$ is the one induced from the grading of $x$ from an $\mathfrak{sl}_2$ triple $(e,x,f_{\mathrm{sub}})$. Namely, let us choose
\begin{align*}
    f_{\mathrm{sub}}=e^3_2+\cdots+e^n_{n-1},\quad x=\frac{n-2}{2}e^2_2+\frac{n-4}{2}e^3_3+\cdots+\frac{-n+2}{2}e^n_n~,
\end{align*}
where $e^i_j$ is the $ij$-elementary matrix, then we get decomposition $\mathfrak{sl}_n=\bigoplus_{j\in \frac{1}{2}\mathbb Z}\mathfrak{g}_j$ where $\mathfrak{g}_j$ is the subspace of $\mathfrak{sl}_n$ such that $\mathrm{ad}(x)$ acts with eigenvalue $j$. $T_{\mathrm{KRW}}$ is constructed such that the field $a(z)$ for $a\in \mathfrak{g}_j$ has conformal weight $1-j$, cf. the remark after \cite[(2.3)]{kac2004quantum}. Now according to \cite[7.3]{genra2020screening}, we have
\begin{align*}
    E(z)=e^1_2(z),\quad F(z)=:(\hat{\partial}+(e^1_1-e^n_n)(z))(\hat{\partial}+(e^1_1-e^{n-1}_{n-1})(z))\cdots(\hat{\partial}+(e^1_1-e^3_3)(z))e^2_1(z):
\end{align*}
where $\hat{\partial}=(k+n-1)\partial_z$. Under the conformal grading of $T_{\mathrm{KRW}}$, the conformal weight of $E$ is $1+\frac{n-2}{2}=n/2$, and the conformal weight of $F$ is $(n-2)+1-\frac{n-2}{2}=n/2$. On the other hand, both $E$ and $F$ are primary fields of conformal weight $n/2$ with respect to $T_{\mathrm{FS}}$. Since $E$ and $F$ generates $\CW^{(2)}_n$, we conclude that $T_{\mathrm{KRW}}$ and $T_{\mathrm{FS}}$ induces the same conformal grading on $\CW^{k}(\mathfrak{sl}_n,f_{\mathrm{sub}})$ when $k+n\neq 0$.

According to \cite[Theorem 4.1]{kac2004quantum}, the degree one part of $\CW^{k}(\mathfrak{sl}_n,f_{\mathrm{sub}})$ is one dimensional, which is spanned by $H$. Using \cite[Theorem 2.1(b)]{kac2004quantum}, we see that $H$ is a primary field with respect to $T_{\mathrm{KRW}}$; according to \cite{feigin2004wn}, $H$ is also a primary field with respect to $T_{\mathrm{FS}}$. Note that $H(z)H(w)\sim \frac{\ell_n(k)}{(z-w)^2}$, $\ell_n(k)=\frac{n-1}{n}k+n-2$. Then Lemma \ref{lem:criterion for two Virasoro coincide_2} implies that $T_{\mathrm{KRW}}=T_{\mathrm{FS}}$ when $\ell_n(k)\neq 0$. When $\ell_n(k)= 0$, we notice that all OPE coefficients depend on $k$ algebraically, thus the fact that $T_{\mathrm{KRW}}-T_{\mathrm{FS}}$ is in the center of $\CW^{k}(\mathfrak{sl}_n,f_{\mathrm{sub}})$ for generic $k$ implies that $T_{\mathrm{KRW}}-T_{\mathrm{FS}}$ is in the center of $\CW^{k}(\mathfrak{sl}_n,f_{\mathrm{sub}})$ for all $k$ such that $k\neq -n$, in particular this holds when $\ell_n(k)= 0$. Thus $T_{\mathrm{KRW}}=T_{\mathrm{FS}}$ by Lemma \ref{lem:criterion for two Virasoro coincide_1}.
\end{proof}

It follows from Lemma \ref{lem:T_FS=T_KRW} and \eqref{eq:compare with T_FS} that
\begin{align}\label{eq:compare with T_KRW}
    T=T_{\mathrm{KRW}}+\frac{1}{2n}J_B^2~.
\end{align}

\subsection{Expectation in the case of \texorpdfstring{$T^{\sigma}_{[1^n]}[\SU(n)]$}{TSU(n)} quiver}

Motivated by the examples studied in the previous sections, we make the following conjecture.

\begin{conjecture}\label{conj:W-alg embedding}
Let $(Q,\mathbf v,\mathbf w)$ be a $T^{\sigma}_{[1^n]}[\SU(n)]$ quiver, then there exists an embedding of vertex algebras:
\begin{align}\label{eq:W-alg embedding}
    \Gamma:\CW^{-n+1}(\mathfrak{gl}_n,\sigma)\hookrightarrow \mathsf D^{\mathrm{ch}}(\widetilde\CM(Q,\mathbf v,\mathbf w))\otimes \mathcal H^n\;,
\end{align}
where $\mathcal H^n$ is a Heisenberg algebra of level $n$, denoted by $J_B$. Moreover $\Gamma$ maps the stress-energy operators $T=T_{\mathrm{KRW}}+T_{\mathcal H}$ to the one defined on the quiver side \eqref{def:generalTtensor}, where $T_{\mathrm{KRW}}$ is Kac-Roan-Wakimoto's stress-energy operator for $\CW^{-n+1}(\mathfrak{sl}_n,\sigma)$ in \cite[(2.5)]{kac2003quantum}, and $T_{\mathcal H}=\frac{1}{2n}J_B^2$.
\end{conjecture}

The Nakajima quiver variety $\CM(Q,\mathbf v,\mathbf w)$ is an example of conical symplectic resolution, and $\widetilde\CM(Q,\mathbf v,\mathbf w)$ is a family of Poisson deformation of $\CM(Q,\mathbf v,\mathbf w)$. It is shown by Namikawa in \cite{namikawa2011poisson,namikawa2010poisson} that for every conical symplectic resolution $\pi: X\to Y$, one can associate
\begin{enumerate}
    \item a family of universal Poisson deformation $\mathcal X\to \mathbb B_X=H^2(X,\mathbb C)$,
    \item a family of universal Poisson deformation $\mathcal Y\to \mathbb B_Y=H^2(X,\mathbb C)/{\mathbb W}$,
    \item a map between universal Poisson deformations: $\bm{\pi}:\mathcal X\to \mathcal Y$ which is compatible with quotient map $\mathbb B_X\to \mathbb B_Y$, i.e. the diagram commutes
    \begin{equation}\label{cd:universal Poisson deformation}
    \begin{tikzcd}
    \mathcal{X} \ar[r] \ar[d,"\bm{\pi}" '] & \mathbb B_X \ar[d,"/\mathbb W"]\\
     \mathcal Y \ar[r] & \mathbb B_Y
    \end{tikzcd}
    \end{equation}
\end{enumerate}
Here $\mathbb W$ is a finite group with a linear action on $\mathbb B_X$. $\mathbb W$ is called the Namikawa-Weyl group. It is known that $\mathbb C[\mathcal Y]=\mathbb C[\mathcal X]^{\mathbb W}$. Denoting by $\mathbb S$ the one dimensional torus with the conical action on $\pi: X\to Y$, then the $\mathbb S$-action extends naturally to every space in the diagram \eqref{cd:universal Poisson deformation} such that all maps are equivariant. We note that $\mathbb S$ acts on $\mathbb B_X$ with weight $2$.

Moreover, it is also known from the work of Losev that the Namikawa-Weyl group acts on the universal quantization as well \cite{losev2016deformations}. Namely there is a sheaf of $\hbar$-adic associative $\mathbb C[\mathbb B_X]$-algebras $\mathscr D_{\mathcal X,\hbar}$ on $\mathcal X$ such that $\mathscr D_{\mathcal X,\hbar}/\hbar\mathscr D_{\mathcal X,\hbar}\cong \mathscr O_{\mathcal X}$, together with an action $\mathbb W$ on $\mathscr D_{\mathcal X,\hbar}$ which satisfies the following properties:
\begin{enumerate}
    \item The induced $\mathbb W$-action on $\mathscr O_{\mathcal X}$ agrees with the one constructed by Namikawa.
    \item The $\mathbb S$-action naturally extends to $\mathscr D_{\mathcal X,\hbar}$ such that $\deg\hbar=1$. Set 
    $$
    \mathsf D(\mathcal X):=\mathscr D_{\mathcal X,\hbar}(\mathcal X)_{\mathbb S\text{-fin}}/(\hbar-1)\;,
    $$ 
    then every quantization of $Y$ is isomorphic to $\mathsf D(\mathcal X)_{\lambda}:=\mathsf D(\mathcal X)/(\mathfrak m_{\lambda}\mathsf D(\mathcal X))$ for some $\lambda\in \mathbb B_X$, where $\mathfrak m_{\lambda}$ is the maximal ideal of $\mathbb C[\mathbb B_X]$ corresponding to $\lambda$. 
    \item $\mathsf D(\mathcal X)_\lambda\cong \mathsf D(\mathcal X)_{\lambda'}$ if and only if $\lambda'\in\mathbb W\lambda$.
\end{enumerate}

The Namikawa-Weyl group $\mathbb W$ of the Nakajima quiver variety $\CM$ corresponding to the $T^{\sigma}_{[1^n]}[\SU(n)]$ quiver is $\mathfrak S_n$ \cite{mcgerty2019springer}. There is a linear map $\mathbb C^{Q_0}\to H^2(\CM,\mathbb C)$ by the universal property. Precisely, this map is given by sending the base vector of $i^{\rm th}$ node to $c_1(\mathscr V_i)$, where $\mathscr V_i$ is the $i^{\rm th}$ tautological bundle on $\CM$. It is not hard to see that this map is an isomorphism \footnote{By the Kirwan surjectivity \cite{mcgerty2018kirwan}, this map is surjective. On the other hand, in view of the proof of \cite[Lemma 5.1]{mcgerty2019springer}, there is only one codimension $2$ symplectic leave in $\CM^0$, and its transverse slice is an $A_{n-1}$ singularity whose preimage in $\CM$ is the minimal resolution of $A_{n-1}$ singularity, such that $c_1(\mathscr V_i)[\mathbb P^1_j]=\delta_{i,j}$, where $\mathbb P^1_j$ is the $j^{\rm th}$ copy of $\mathbb P^1$ in the central fiber of resolution of $A_{n-1}$ singularity. This implies that $\mathbb C^{Q_0}\to H^2(\CM,\mathbb C)$ is injective. Therefore it is an isomorphism.}. It follows that $\widetilde\CM$ is the universal Poisson deformation of $\CM$. According to \cite[Theorem 5.3.3]{losev2012isomorphisms}, there is an embedding of finite $\mathcal{W}$-algebra $\CW(\mathfrak{sl}_n,\sigma)$ into $\mathsf D(\widetilde\CM)$, moreover the image coincides with $\mathsf D(\widetilde\CM)^{\mathbb W}$.

\begin{conjecture}\label{conj:Namikawa-Weyl group action}
Let $(Q,\mathbf v,\mathbf w)$ be a $T^{\sigma}_{[1^n]}[\SU(n)]$ quiver, then the Namikawa-Weyl group $\mathbb W=\mathfrak S_n$ acts on $\mathsf D^{\mathrm{ch}}(\widetilde\CM(Q,\mathbf v,\mathbf w))$ by a vertex algebra automorphism, such that the orbifold subalgebra $\mathsf D^{\mathrm{ch}}(\widetilde\CM(Q,\mathbf v,\mathbf w))^{\mathbb W}\otimes \mathcal H^n$ is the image of the embedding $\Gamma$ in  Conjecture \ref{conj:W-alg embedding}.
\end{conjecture}

Finally, the embedding $\Gamma$ should be compatible with Miura maps of affine $\mathcal{W}$-algebras in the following sense. For a partition $\sigma=[l^{w_{l}},\cdots,1^{w_1}]$, one can associate a sequence $q_1\ge \cdots\ge q_l>0$ as follows
\begin{align*}
    q_i=\sum_{j=i}^lw_j\;.
\end{align*}
We note that $n=q_1+\cdots+q_l$. There exists an embedding of vertex algebras called \textit{Miura map} \cite{kac2004quantum}:
\begin{align*}
    \Xi:\CW^k(\mathfrak{gl}_n,\sigma)\hookrightarrow \bigotimes_{i=1}^l V^{k_i}(\mathfrak{gl}(q_i))\;,\quad\text{ where }k_i=k+n-q_i \;.
\end{align*}
There is a Wakimoto realization for each component $$\Omega_i:V^{k_i}(\mathfrak{gl}(q_i))\hookrightarrow D^{\mathrm{ch}}(T^*\mathbb C^{q_i(q_i-1)/2})\otimes \mathcal H(\mathbb C^{q_i},\kappa_i),$$
where $\kappa_i$ is the inner form on $\mathbb C^{q_i}=\oplus_{s=1}^{q_i}\mathbb C\cdot a_s$ such that $\kappa_i(a_s,a_t)=(k_i+q_i)\delta_{s,t}$. Composing the Miura map with Wakimoto realizations, we get an embedding 
\begin{align}\label{eq:compose Miura with Wakimoto}
\Omega\circ\Xi: \CW^k(\mathfrak{gl}_n,\sigma)\hookrightarrow D^{\mathrm{ch}}\left(T^*\mathbb C^{\sum_{i=1}^l q_i(q_i-1)/2}\right)\otimes \mathcal H(\mathbb C^{n},\kappa)\;,
\end{align}
where $\kappa$ is the inner form on $\mathbb C^{n}=\oplus_{s=1}^{n}\mathbb C\cdot a_s$ such that $\kappa(a_s,a_t)=(k+n)\delta_{s,t}$. We note that 
\begin{align}
    \sum_{i=1}^l \frac{q_i(q_i-1)}{2}=\frac{1}{2}\dim\CM(Q,\mathbf v,\mathbf w)=\mathbf v\cdot(\mathbf w-\frac{1}{2}\mathsf C\mathbf v)\;.
\end{align}
In our case, where $k=-n+1$, so the Heisenberg factor in \eqref{eq:compose Miura with Wakimoto} is the tensor product of $n$ copies of level $1$ Heisenberg algebras $\mathcal H^1$, so we can rewrite \eqref{eq:compose Miura with Wakimoto} into
\begin{align}\label{eq:compose Miura with Wakimoto,k=-n+1}
    \Omega\circ\Xi: \CW^{-n+1}(\mathfrak{gl}_n,\sigma)\hookrightarrow D^{\mathrm{ch}}\left(T^*\mathbb C^{\mathbf v\cdot(\mathbf w-\frac{1}{2}\mathsf C\mathbf v)}\right)\otimes \bigotimes_{i=1}^n\mathcal H^1.
\end{align}
On the other hand, assuming Conjecture \ref{conj:W-alg embedding}, then composing the embedding $\Gamma$ in \eqref{eq:W-alg embedding} with the free-field realization \eqref{eq:free-field realization}, we also get a vertex algebra embedding:
\begin{align}\label{eq:conjectured free-field realization of W-alg}
    \Theta\circ\Gamma:  \CW^{-n+1}(\mathfrak{gl}_n,\sigma)\hookrightarrow D^{\mathrm{ch}}\left(T^*\mathbb C^{\mathbf v\cdot(\mathbf w-\frac{1}{2}\mathsf C\mathbf v)}\right)\otimes \bigotimes_{i=1}^n\mathcal H^{\lambda_i}.
\end{align}

\begin{conjecture}
Assuming Conjecture \ref{conj:W-alg embedding}, then $\Omega\circ\Xi$ agrees with $\Theta\circ\Gamma$, up to an isomorphism between right hand sides of \eqref{eq:compose Miura with Wakimoto,k=-n+1} and \eqref{eq:conjectured free-field realization of W-alg}.
\end{conjecture}

\section{Comparison with 4d \texorpdfstring{$\mathcal{N}=2$}{N=2} VOAs}\label{sec:4d-VOA}

The 3d $\mathcal{N}=4$ balanced linear quiver gauge theories considered in this paper
can be obtained as a dimensional reduction of 4d $\mathcal{N}=2$  quiver gauge theories.
When the quiver is balanced, 
the resulting 4d $\mathcal{N}=2$ quiver gauge theory is superconformal, and 
the work of \cite{Beem:2013sza,Beem:2014kka,Beem:2014rza}
associated a VOA with the 4d theory.
It is therefore a natural question to ask if there are any relations between the two VOAs.
Indeed, the 3d and 4d theories have the same Higgs branches,
and it is natural to conjecture that the two VOAs share the same associated varieties.

There are, however, important differences between the two. 
In 4d, the VOA is defined by 
taking cohomologies with respect to a nilpotent supercharge, which involves a generator of the superconformal charge.
For this reason, we need to impose superconformality of the 4d theory, and this leads to the balanced condition on the quiver. The gauge groups of the quivers are taken to be $\SU$ gauge groups.

In 3d, in contrast, the gauge groups at the quiver vertices are taken to be $\U$ instead of $\SU$ from 4d; this is because the $\U(1)$ part has non-trivial dynamics in 3d. Moreover, the VOA is defined even for non-balanced quivers, as we have seen above. However, 
we have to include extra fermionic matter to cancel the anomalies on the 2d boundary.
As we have seen in the discussion above,
the presence of such fermionic modes for 3d VOAs play a crucial roles in the formulation of our VOAs; for example, some of our currents are not BRST-invariant without such fermionic modes.

In order to make a comparison between 3d and 4d VOAs, let us assume that $Q$ is a balanced quiver, and for simplicity we also assume that the Cartan matrix $\mathsf C$ is invertible. One can then associate three VOAs to $Q$:
\begin{enumerate}
    \item $\CV(Q)$: the BRST reduction of $\beta\gamma$ system $\otimes$ Heisenberg fields $\{h_i,i\in Q_0\}$.
    \item $\CVdel(Q)$: boundary H-twisted VOA of 3d $\CN=4$ gauge theory associated to $Q$. We use the canonical choice of free fermions in \eqref{def:Ff} to cancel the $\SU$ anomalies, and the definition of $\CVdel(Q)$ involves a choice of another set of free fermions that cancel $\U(1)$ anomalies.
    \item $\CV_{\mathrm{4d}}(Q)$: the VOA of the 4d $\CN=2$ gauge theory associated to $Q$ with gauge group $\prod_{i\in Q_0}\SU(\mathbf v_i)$. By definition $\CV_{\mathrm{4d}}(Q)=H^{\infty/2+0}_{\mathrm{BRST}}(\mathfrak{sl}(\mathbf v),\beta\gamma)$\;.
\end{enumerate}
The remaining $(\mathbb C^{\times})^{Q_0}$-symmetry gives a weight decomposition on $\CV_{\mathrm{4d}}(Q)$:
\begin{align*}
    \CV_{\mathrm{4d}}(Q)\cong \bigoplus_{\lambda\in \mathbb Z^{Q_0}}\CV_{\mathrm{4d}}(Q)_{\lambda}\;,
\end{align*}
where $\CV_{\mathrm{4d}}(Q)_{\lambda}$ is the weight $\lambda$ eigenspace. Now consider the following module of $\CV(Q)$:
\begin{align*}
    \CV_\lambda(Q)=H^{\infty/2+0}_{\mathrm{BRST}}(\mathfrak{gl}(\mathbf v),\beta\gamma\otimes \mathcal H_{\lambda})\;,
\end{align*}
$\mathcal H_{\lambda}$ is the Verma module of the Heisenberg algebra $\mathcal H=\{h_i,i\in Q_0\}$ generated from a vector $|\lambda\rangle$ such that $h_{i,(0)}|\lambda\rangle=\lambda_i|\lambda\rangle$. $\CV(Q)$, $\CVdel(Q)$, and $\CV_{\mathrm{4d}}(Q)$ are related by the followings:
\begin{itemize}
    \item[(1)] There is a vertex algebra map $\CV(Q)\to \CV_{\mathrm{4d}}(Q)_{0}$, and there are $\CV(Q)$-module maps $\CV_\lambda(Q)\to \CV_{\mathrm{4d}}(Q)_{-\lambda}$. To see this map, one can consider the vertex algebra extension $\mathcal H\subset \mathcal H\ltimes \mathbb C[{\jet} (\mathbb C^{\times})^{Q_0}]$, where $$\mathbb C[{\jet} (\mathbb C^{\times})^{Q_0}]=\mathbb C[A_{i,(-1)}^{-1},A_{i,(-n)}\:|\:i\in Q_0, n\in \mathbb Z_{\ge 1}]|0\rangle$$ is the function ring on the $\infty$-jet bundle of the central tori $(\mathbb C^{\times})^{Q_0}\subset \GL(\mathbf v)$, regarded as a commutative vertex algebra, and $\mathcal H$ acts on $\mathbb C[{\jet} (\mathbb C^{\times})^{Q_0}]$ via OPEs
    \begin{align}
        h_i(z)(A_j)^n(w)\sim \frac{n\delta_{i,j}(A_j)^n(w)}{z-w}\;, \quad \forall n\in \mathbb Z\;.
    \end{align}
    Then we have a spectral sequence $E^{p,q}$ with second page
    \begin{align}
        E^{p,q}_2=H^{\infty/2+q}_{\mathrm{BRST}}(\mathfrak{sl}(\mathbf v),H^{\infty/2+p}_{\mathrm{BRST}}(\mathfrak{gl}(1)^{Q_0},\beta\gamma\otimes \mathcal H\ltimes \mathbb C[{\jet} (\mathbb C^{\times})^{Q_0}]))\;.
    \end{align}
    Notice that the semidirect product $\widetilde\mu_{\mathrm{ch}}({\mathfrak{gl}(1)^{Q_0}})\ltimes \mathbb C[{\jet} (\mathbb C^{\times})^{Q_0}]))$ between the current of $\mathfrak{gl}(1)^{Q_0}$ and its module $\mathbb C[{\jet} (\mathbb C^{\times})^{Q_0}]$ is a quantization of the $\infty$-jet bundle of $T^*(\mathbb C^{\times})^{Q_0}$; then we have \footnote{Consider $\varepsilon:\beta\gamma\hookrightarrow\beta\gamma\otimes \mathcal H\ltimes \mathbb C[{\jet} (\mathbb C^{\times})^{Q_0}]$, $\varepsilon(x_\lambda)=\exp\left(-\sum_{i\in Q_0}\sum_{n=1}^{\infty}\mu_{\mathrm{ch},i,(n)}A_{i,(-n-1)}A_{i,(-1)}^{-1}\right)(x_\lambda\otimes |-\lambda\rangle)$, where $\mu_{\mathrm{ch}}$ is the non-extended chiral moment map, and $x_\lambda$ belongs to the weight $\lambda$ piece of $\beta\gamma$, and $|-\lambda\rangle=\prod_{i\in Q_0}A_{i,(-1)}^{-\lambda_i}|0\rangle$. $\varepsilon$ is a vertex algebra map, an intuitive way to see this is that $\varepsilon$ can be regarded as the restriction of the vertex algebra ``automorphism'' $\exp\left(-\sum_{i\in Q_0}:\mu_{\mathrm{ch},i}\log A_i:_{(0)}\right)$ of $\beta\gamma\otimes \mathcal H\ltimes \mathbb C[{\jet} (\mathbb C^{\times})^{Q_0}]$ to $\beta\gamma$. Then $\beta\gamma\otimes \mathcal H\ltimes \mathbb C[{\jet} (\mathbb C^{\times})^{Q_0}]=\varepsilon(\beta\gamma)\otimes \widetilde\mu_{\mathrm{ch}}\ltimes \mathbb C[{\jet} (\mathbb C^{\times})^{Q_0}]$, thus $H^{\infty/2+p}_{\mathrm{BRST}}(\mathfrak{gl}(1)^{Q_0},\beta\gamma\otimes \mathcal H\ltimes \mathbb C[{\jet} (\mathbb C^{\times})^{Q_0}])\cong \beta\gamma\otimes H^{\infty/2+p}_{\mathrm{BRST}}(\mathfrak{gl}(1)^{Q_0},\widetilde\mu_{\mathrm{ch}}\ltimes  \mathbb C[{\jet} (\mathbb C^{\times})^{Q_0}])$, which is isomorphic to $\beta\gamma$ when $p=0$ and zero otherwise.}
    \begin{align*}
        H^{\infty/2+p}_{\mathrm{BRST}}(\mathfrak{gl}(1)^{Q_0}, \beta\gamma\otimes \mathcal H\ltimes \mathbb C[{\jet} (\mathbb C^{\times})^{Q_0}])\cong \begin{cases}
            \beta\gamma\;, & q=0\;\\
            0\;, & q\neq 0\;.
        \end{cases}
    \end{align*}
    Therefore $E^{p,q}$ collapse at the second page and converges to $H^{\infty/2+p+q}_{\mathrm{BRST}}(\mathfrak{gl}(\mathbf v),\beta\gamma\otimes \mathcal H\ltimes \mathbb C[{\jet} (\mathbb C^{\times})^{Q_0}])$, thus we get an isomorphism:
    \begin{align*}
        H^{\infty/2+\bullet}_{\mathrm{BRST}}(\mathfrak{gl}(\mathbf v),\beta\gamma\otimes \mathcal H\ltimes \mathbb C[{\jet} (\mathbb C^{\times})^{Q_0}])\cong H^{\infty/2+\bullet}_{\mathrm{BRST}}(\mathfrak{sl}(\mathbf v),\beta\gamma)\;.
    \end{align*}
    Composing with the vertex algebra map $\mathcal H\to \mathcal H\ltimes \mathbb C[{\jet} (\mathbb C^{\times})^{Q_0}]$ the module map $\mathcal H_\lambda\to \mathcal H\ltimes \mathbb C[{\jet} (\mathbb C^{\times})^{Q_0}]_\lambda$, we get the desired maps between BRST reductions:
    \begin{align}
        \CV(Q)\to \CV_{\mathrm{4d}}(Q)_0\;,\quad \CV_\lambda(Q)\to \CV_{\mathrm{4d}}(Q)_{-\lambda}\;.
    \end{align}
    Since we assume that the Cartan matrix $\mathsf C$ is invertible, $\CV(Q)$ becomes a VOA and we can talk about its character. To make an explicit comparison between characters, we add the nonzero cohomology degrees, i.e. with $\mathsf b\mathsf c$ ghosts included, and consider the characters $\chi_{\CV^\bullet}(q):=\mathrm{Tr}_{\CV^\bullet}(-1)^Fq^{L_0-c/24}$. Then the characters of $\CV^\bullet_\lambda(Q)$ and $\CV^\bullet_{\mathrm{4d}}(Q)$ are related by
    \begin{align*}
        \chi_{\CV^\bullet_\lambda(Q)}(q)=\eta(q)^{|Q_0|} \chi_{\CV^\bullet_{\mathrm{4d}}(Q)_{-\lambda}}(q)\;.
    \end{align*}
    \item[(2)] There is an isomorphism
\begin{align*}
    \CVdel(Q)\cong \bigoplus_{\lambda\in \mathbb Z^r}\CV_{\sigma(\lambda)}(Q)\otimes \mathcal H^{\perp}_{\sigma^{\perp}(\lambda)}\;,
\end{align*}
where $r$ is the number of pairs of $\chi\psi$ fermions and $\sigma:\mathbb Z^r\to \mathbb C^{Q_0}$ is the map of charges which is induced from the choice of anomaly cancellation map $\{h_i\}_{i\in Q_0}=\mathcal H\to (\mathcal H^1)^{\otimes r}=\{\chi_j\psi_j\}_{j=1,\cdots,r}$. $\mathcal H^{\perp}$ is the Heisenberg algebra strongly generated by the subspace of $\bigoplus_{j=1}^r\mathbb C\cdot\chi_j\psi_j$ perpendicular to the image of $\mathcal H$, and $\sigma^{\perp}: \mathbb Z^r\to \mathbb C^{r-\Qzero}$ is the induced map for charges.
\end{itemize}

\section{Discussions}\label{sec:discussions}

In this paper we have studied VOAs arising from non-Abelian quiver gauge theories, regarded as ``chiralizations'' of extended quiver varieties. In particular, by localization of quiver vertex algebras, we developed a systematic algorithm for the  reduction of general quiver diagrams which produces a free field realization of the associated VOAs. While some of the mathematical points raised here will be discussed
in our companion paper \cite{math_draft}, there are still many open questions left for future research.
Let us comment here on some of them.

So far, we have concentrated on the H-twisted VOAs of 3d quiver gauge theories. It would nevertheless be interesting to establish a systematic construction of C-twisted VOAs and to then verify the mirror symmetry for the resulting VOAs (which exchanges $\rho$ and $\sigma$ for the $T_{\rho}^{\sigma}[\SU(N)]$ theories).

Another natural direction is to consider 3d $\mathcal{N}=4$ SQFTs with more general Lagrangians. Examples include quiver gauge theories with non-linear quivers, 
quiver gauge theories with more general gauge groups, such as $T[G]$ theories, or 
gauge theories with Chern-Simons terms. 
In relation to this point, 
there have been recent discussions \cite{Gang:2021hrd,Gang:2022kpe,Gang:2023rei,Gang:2023ggt,Ferrari:2023fez} of rational VOAs associated with ``rank $0$'' theories \cite{Terashima:2011qi,Gang:2015wya,Gang:2018huc,Assel:2018vtq,Gang:2021hrd},
whose Coulomb and Higgs branches both have  dimension $0$, and it would be nice to connect our discussion here more directly to the rational VOAs discussed in these references.

While in this paper we defined boundary VOAs from the data of quiver gauge theories, it would be interesting to similarly study 
VOAs with other infinite-dimensional algebras associated with quivers, such as (shifted)
quiver Yangians \cite{Li:2020rij,Galakhov:2020vyb,Galakhov:2021xum}
and quiver W-algebras \cite{Kimura:2015rgi}.

We have seen a close parallel between the geometry of the Higgs branch and the algebra of the VOA, and indeed some information of the latter can be recovered from the former. We expect, however, that some subtle details of the VOA, such as the detailed coefficients of the OPE of the generators, are not readily recovered from the Higgs branch. This suggests that the VOA contains more information than the Higgs branch. It would be interesting to clarify on this point.

\acknowledgments{We would like to thank Toshiro Kuwabara, Hiraku Nakajima, Shigenori Nakatsuka, and Satoshi Nawata for stimulating discussions, and Davide Gaiotto for comments on our manuscript. MS was supported in part by National Research Foundation (NRF) of Korea grant NRF-2022R1A6A3A03053466 and by the Shuimu Scholar program of Tsinghua University.
MY was supported in part by the JSPS Grant-in-Aid for Scientific Research (No.\ 19H00689, 19K03820, 20H05860, 23H01168), and by JST, Japan (PRESTO Grant No. JPMJPR225A, Moonshot R\&D Grant No.\ JPMJMS2061).
}

\appendix

\section{Glossary}\label{sec:glossary}

\begin{itemize}

\item $G$: gauge group

\item $M$: a symplectic representation specifying the matter content

\item $R$: a symplectic representation for Fermi multiplets needed for gauge anomaly cancellations on the boundary

\item $\CVdel=\CVdel[G,M,R]$: H-twisted boundary VOA for 3d $\mathcal{N}=4$ theory 

\item $\CVsdel=\CVsdel[M, R]$: VOA before BRST reduction (so that $\CV=\CVs \Kquotient G$)

\item $\CV=\CV(Q,\mathbf v,\mathbf w)$: vertex algebra associated to quiver data $(Q,\mathbf v,\mathbf w)$

\item $\CVs=\CVs(Q,\mathbf v,\mathbf w)$: vertex algebra associated to quiver data $(Q,\mathbf v,\mathbf w)$ before BRST reduction

\item $D^{\mathrm{ch}}(M)$: $\beta\gamma\mathsf b\mathsf c$ system associated to a symplectic super vector space $M$

\item $\mathcal H^k$: Heisenberg algebra of level $k$, i.e. the OPE is $h(z)h(w)\sim {k}/{(z-w)^2}$. 

\item $k_c$: critical level for a simple Lie algebra $\mathfrak{g}$, which equals to $-\mathsf h^{\vee}_{\mathfrak{g}}$

\item $\QBRST$: BRST charge (BRST differential)

\item $\JBRST$: BRST current 

\item $\widetilde C(\mathfrak g,V)$: BRST complex associated to a vertex algebra $V$ with a $\widehat{\mathfrak g}$-symmetry at the level $k_V=2k_c$

\item $C(\mathfrak g,V)$: relative BRST complex associated to $\widetilde C(\mathfrak g,V)$

\item $H^{\infty/2+\bullet}_{\mathrm{BRST}}(\mathfrak g,V)$: relative BRST cohomology associated to a vertex algebra $V$ with a $\widehat{\mathfrak g}$-symmetry at the level $k_V=2k_c$

\item $\CM=\CM(Q,\mathbf v,\mathbf w)$: Nakajima quiver variety associated to a quiver $Q$ with gauge dimension $\mathbf v$ and framing dimension $\mathbf w$

\item $\CM^0=\CM^0(Q,\mathbf v,\mathbf w)$: affine quiver variety associated to a quiver $Q$ with gauge dimension $\mathbf v$ and framing dimension $\mathbf w$

\item $\widetilde\CM=\widetilde\CM(Q,\mathbf v,\mathbf w)$: extended Nakajima quiver variety associated to a quiver $Q$ with gauge dimension $\mathbf v$ and framing dimension $\mathbf w$

\item $\widetilde\CM^0=\widetilde\CM^0(Q,\mathbf v,\mathbf w)$: extended affine quiver variety associated to a quiver $Q$ with gauge dimension $\mathbf v$ and framing dimension $\mathbf w$

\item $\jet X$: $\infty$-jet scheme of a variety $X$

\item $\mathbb C[X]$: the global ring of functions on a variety $X$

\item $\Zhu(V)$: Zhu's $C_2$ algebra for a vertex algebra $V$

\end{itemize}

\section{BRST reduction}\label{sec:BRST reduction}

In this appendix we recall the mathematical construction of the BRST reduction. 

Let $\mathfrak g$ be a reductive Lie algebra and $\kappa$ be a $\mathfrak g$-invariant symmetric inner form on $\mathfrak g$, then the current algebra $V^\kappa(\mathfrak g)$, as a vector space, is isomorphic to 
$$\mathbb C[J_\infty\mathfrak g]=\mathbb C[E^\alpha_{(-n)}\:|\:n\in \mathbb Z_{\ge 1}, E^\alpha\in\text{a basis of }\mathfrak g]|0\rangle$$ 
endowed with OPE
\begin{align*}
    E^{\alpha}(z)E^\beta(w)\sim \frac{\kappa^{\alpha,\beta}}{(z-w)^2}+\frac{f^{\alpha\beta}_\gamma E^\gamma(w)}{z-w}\;,
\end{align*}
where $\kappa^{\alpha,\beta}=\kappa(E^\alpha,E^\beta)$ is the matrix element of the inner form $\kappa$, and $f^{\alpha\beta}_\gamma E^\gamma=[E^\alpha,E^\beta]$ is the structure constant of $\mathfrak g$.

Denote by $\kappa_{\mathfrak g}$ the killing form on $\mathfrak g$. Suppose that $V$ is a vertex superalgebra such that there exists a vertex superalgebra map
\begin{align*}
    J_V: V^{-\kappa_{\mathfrak g}}(\mathfrak g)\to V \;.
\end{align*}
Then we define vertex superalgebra 
\begin{align}
    \widetilde C(\mathfrak g, V):=V\otimes \CV_{\mathsf b\mathsf c}(\mathfrak g) \;,
\end{align}
where $\CV_{\mathsf b\mathsf c}(\mathfrak g)$ is the $\mathsf b\mathsf c$ system of $\mathfrak g$, i.e.\ chiral differential operators $D^{\mathrm{ch}}(\Pi T^*\mathfrak g)$ of the odd symplectic vector space $\Pi T^*\mathfrak g$. Note that $\CV_{\mathsf b\mathsf c}(\mathfrak g)$ admits a vertex superalgebra map
\begin{align*}
    J_{\mathsf b\mathsf c}: V^{\kappa_{\mathfrak g}}(\mathfrak g)\to \CV_{\mathsf b\mathsf c}(\mathfrak g)\;, \quad J_{\mathsf b\mathsf c}(E^\alpha)=f^{\alpha\beta}{}_{\gamma}\mathsf c_{\beta}\mathsf b^\gamma.
\end{align*}
$\CV_{\mathsf b\mathsf c}(\mathfrak g)$ is graded by ghost numbers, i.e. $\deg \mathsf c=+1,\deg\mathsf b=-1$, so the image of $J$ is contained in the degree zero piece $\CV^0_{\mathsf b\mathsf c}(\mathfrak g)$. Consider the ghost number $1$ element 
\begin{align}
    \JBRST:=\mathsf c_\alpha\left(J_V(E^\alpha)+\frac{1}{2}J_{\mathsf b\mathsf c}(E^\alpha)\right)\in  \widetilde C^1(\mathfrak g, V)\;.
\end{align}
The operator $\mathcal{Q}:=\JBRST_{(0)}$ acts on $\widetilde C(\mathfrak g, V)$, and it squares to zero $\mathcal{Q}^2=0$, therefore $(\widetilde C(\mathfrak g, V), \mathcal{Q})$ is a chain complex of vertex superalgebra, which is graded by ghost numbers. We call it the BRST complex. 

Note that the linear map 
\begin{align}\label{level zero currents in BRST complex}
    E^\alpha\mapsto J(E^\alpha):=\mathcal{Q}(\mathsf b^\alpha)=J_V(E^\alpha)+J_{\mathsf b\mathsf c}(E^\alpha) \;,
\end{align}
is a vertex superalgebra chain complex map $$J: (V^0(\mathfrak g),0)\to (\widetilde C(\mathfrak g, V), \mathcal{Q})$$ where the LHS is endowed with trivial differential. In particular the Lie algebra $\mathfrak g$ acts on the complex $(\widetilde C(\mathfrak g, V), \mathcal{Q})$ by $E^\alpha\mapsto J(E^\alpha)_{(0)}$. 

The relative BRST complex is defined to be the subspace 
\begin{align}
    C(\mathfrak g, V):=\{v\in \widetilde C(\mathfrak g, V)\:|\: J(E^\alpha)_{(0)}v=\mathsf b^\alpha_{(0)}v=0,\forall \alpha\}\;,
\end{align}
which is a sub-complex because $[J(E^\alpha)_{(0)},\mathcal{Q}]=0$ and $[\mathsf b^\alpha_{(0)},\mathcal{Q}]=J(E^\alpha)_{(0)}$. We define the relative BRST cohomology to be
\begin{align}
    H^{\infty/2+\bullet}_{\mathrm{BRST}}(\mathfrak g,V):=H^{\bullet}(C(\mathfrak g, V), \mathcal{Q})\;.
\end{align}
This is a $\mathbb Z$-graded vertex superalgebra.

\subsection{Some BRST exact elements}\label{subsec:Some BRST exact elements}

For the sake of computation in the main contents of this paper, we present some BRST exact elements, i.e. $x\in C(\mathfrak g,V)$ such that $x=\mathcal Q(y)$ for some $y\in C(\mathfrak g,V)$. The first obvious one is $J(E^\alpha)$ in \eqref{level zero currents in BRST complex}.

Next, let us assume that $X\in \mathfrak g$ is an element in the center of $\mathfrak g$, then we take the corresponding ghost field $X_\alpha\mathsf b^\alpha$ and take its BRST variation:
\begin{align}
    \QBRST(X_\alpha\mathsf b^\alpha)=J_V(X).
\end{align}
Thus $J_V(X)$ is BRST exact. For example in the quiver setting, for every gauge node $i\in Q_0$, the operator $\mathrm{Tr}(\widetilde\mu_{\mathrm{ch},i})$ is BRST exact.

Next, let us consider a multiplet $\{\mathcal O_{\alpha}\in V\}_{\alpha=1,\cdots,\dim\mathfrak g}$ living in adjoint representation of the currents of $\mathfrak g$, i.e. 
\begin{align}
    J_V(E^\alpha)(z)\mathcal O_{\beta}(w)\sim\frac{f^{\alpha\gamma}{}_{\beta}\mathcal O_{\gamma}(w)}{z-w}~,
\end{align}
then we have
\begin{align}
    \mathcal Q(\mathcal O_{\alpha}\mathsf b^{\alpha})=:J_V(E^\alpha)\mathcal O_{\alpha}:
\end{align}
Thus $:J_V(E^\alpha)\mathcal O_{\alpha}:$ is BRST exact, or equivalently $:\mathcal O_{\alpha}J_V(E^\alpha):$ is BRST exact.

Furthermore, we can consider a multiplet $\{\mathcal S_{\alpha}\in V\}_{\alpha=1,\cdots,\dim\mathfrak g}$ living in adjoint representation of the currents of $\mathfrak g$ with anomalies, i.e. 
\begin{align}
    J_V(E^\alpha)(z)\mathcal S_{\beta}(w)\sim \frac{\lambda \delta^\alpha_\beta}{(z-w)^2}+\frac{f^{\alpha\gamma}{}_{\beta}\mathcal S_{\gamma}(w)}{z-w}~, \quad \lambda\in \mathbb C~,
\end{align}
then we have
\begin{align}
    \mathcal Q(\mathcal S_{\alpha}\mathsf b^{\alpha})=:J_V(E^\alpha)\mathcal S_{\alpha}:- \lambda :\mathsf b^\alpha\partial \mathsf c_\alpha:
\end{align}
Thus $:J_V(E^\alpha)\mathcal S_{\alpha}:$ equals to $\lambda :\mathsf b^\alpha\partial \mathsf c_\alpha:$ modulo a BRST exact term.

\section{Proof of \texorpdfstring{\eqref{cc agree}}{central charge matching}}
\label{sec:c_W}

To verify \eqref{cc agree},
let us write down the explicit form of each term on the right-hand-side of \eqref{eqn: central charge in Kac-Roan-Wakimoto} in terms of the partition $((N-1)^{w_{N-1}}\cdots i^{w_i}\cdots 1^{w_1})$ as follows (note $\mathsf{h}^{\vee}=-N$ for $\SU(N)$):
\begin{align}\label{eqn: cc_aux_1}
\frac{k\dim\mathfrak{g}}{k+\mathsf{h}^\vee}=-(N-1)(N^2-1)\;,
\end{align}
\begin{align}\label{eqn: cc_aux_2}
(x|x)=\sum_{i=1}^{N-1}w_i\frac{i^3-i}{12}\;,
\end{align}
\begin{equation}\label{eqn: cc_aux_3}
\begin{split}
&\sum_{\alpha\in S_+}(12m_{\alpha}^2-12m_{\alpha}+2)=\sum_{\substack{i>j\\i-j\text{ even}}}\left(w_iw_j\times\sum_{\substack{k\in[i-j,i+j)\\k\equiv 0\mod{2}}}(k^3-2k)\right)\\
&+\sum_{\substack{i>j\\i-j\text{ odd}}}\left(w_iw_j\times\sum_{\substack{k\in[i-j,i+j)\\k\equiv 1\mod{2}}}(k^3-2k-1)\right)+\sum_{i=1}^{N-1}\left(w_i^2 \times\sum_{\substack{k\in[0,2i)\\k\equiv 0\mod{2}}}\frac{k^3-2k}{2}\right)\;,
\end{split}
\end{equation}
\begin{align}\label{eqn: cc_aux_4}
\dim\mathfrak{g}_{1/2}=\sum_{\substack{i>j\\i-j\text{ odd}}}2jw_iw_j\;.
\end{align}
Plug \eqref{eqn: cc_aux_1}-\eqref{eqn: cc_aux_4} into \eqref{eqn: central charge in Kac-Roan-Wakimoto}, and use the identity $N=\sum_{i=1}^{N-1}iw_i$, and then we compute that
\begin{equation}\label{RHS of cc agree}
\begin{split}
\text{RHS of }\eqref{cc agree}=&-N^3+\sum_{i=1}^{N-1}\left(w_i(2i-i^3)+w_i^2(2i^3-i)\right)\\
&+\sum_{i>j}^{N-1}w_iw_j(j^3+3i^2j-2j) \;.
\end{split}
\end{equation}
On the other hand, we can use \eqref{v_i in terms of w_i} to rewrite \eqref{eqn: central charge for linear quivers} in terms of $w_i$. After doing some gymnastics we arrive at the equation
\begin{equation}
\text{LHS of }\eqref{cc agree}=\text{RHS of }\eqref{RHS of cc agree}\;.
\end{equation}
Thus \eqref{cc agree} holds.

\section{Minimal generators in terms of paths}\label{sec:reduction}

We have seen in section \ref{sec:BRST_operators} that we can restrict to operators of the form
\eqref{J_Phi_I} when searching for generators of the VOAs considered in this paper which are associated to paths on linear quivers and which are classically gauge invariant.
For the A-type quiver in Fig.~\ref{fig:quiver for T^[2,1^{n-2}]_[1^n][SU(n)] theory}, this means that there are effectively four types of paths and corresponding operators to consider: 
i) paths which begin and end at the node $w_1$ and loop around the node $v_{\ell+1}$ for all $\ell\in\{0,1,\ldots , k-2\}$, following the notation in \cref{fig:quiver for Trs[SU(N)] theory}, which translate to operators 
\begin{equation}\label{eqdef:typeIpath}
    J_1X_1X_2\cdots X_{\ell}Y_{\ell}\cdots Y_2Y_1I_1 ~,
\end{equation}
ii) operators associated to paths which begin at the node $w_1$, loop around a node $v_{\ell+1}$ with $\ell\geq 1$ and end at $w_2$ 
\begin{equation} \label{eqdef:typeIIpath}
    J_2X_2\cdots X_{\ell}Y_{\ell}\cdots Y_2Y_1I_1 ~,
\end{equation}
iii) operators associated to paths which begin and end at the node $w_2$, looping around $v_{\ell+1}$ for all $\ell\geq 1$
\begin{equation} \label{eqdef:typeIIIpath}
   J_2X_2\cdots X_{\ell}Y_{\ell}\cdots Y_2I_2  ~,
\end{equation}
and iv) operators associated to paths which begin at the node $w_2$, loop around $v_{\ell+1}$ for all $\ell\geq 1$ and end at node $w_1$
\begin{equation} \label{eqdef:typeIVpath}
   J_1X_1X_2\cdots X_{\ell}Y_{\ell}\cdots Y_2I_2  ~.
\end{equation}
We demonstrate now that for the A-type quiver depicted in Fig.~\ref{fig:quiver for T^[2,1^{n-2}]_[1^n][SU(n)] theory}, it is sufficient to consider the operators defined in \eqref{eqdef:genericBRSTbgen} as generators for the corresponding H-twisted VOA which are associated to paths. 
Note that setting $\ell=0$ in \eqref{eqdef:typeIpath} reduces this to $J_1I_1$, which enters the definition of the BRST current $G^{(1,1)}$ in \eqref{eqdef:genericBRSTbgen}, while $\ell=1$ in \eqref{eqdef:typeIIpath} gives precisely $G^{(2,1)}$. 
Using the moment map relations $\hat{\m}^{(i)}{}^{a}_{\ph{a}b} = 0$ from \eqref{eqdef:genericmomentmaprelations}, it is possible to reduce all of the bosonic currents associated to the paths of type  i)-vi) to linear combinations of products of the BRST-invariant currents \eqref{eqdef:genericBRSTbgen} and fermion bilinears of the form $\chi_\ell\psi_\ell$. 
Starting for example with the type i) current \eqref{eqdef:typeIpath} and using the $\ell^{th}$-moment map relation to replace the inner-most product $X_\ell Y_\ell$ with $Y_{\ell-1}X_{\ell-1}$ and the fermion bilinears $\chi_{\ell-1}\psi_{\ell-1}-\chi_\ell\psi_\ell$ reduces this to\footnote{The order of index contractions is shown only for the inner-most products to keep notation light. This carries through to the following equations.} 
\begin{align}\label{eq:typeIpath-Intermediate1}
   & (J_1X_1X_2\cdots (X_{\ell})^a_{\ph{b}b}(Y_{\ell})^b_{\ph{c}c}\cdots Y_2Y_1I_1)^i_{\ph{j}j} = \nonumber\\
   &\qquad\qquad
    (J_1X_1X_2\cdots (X_{\ell-1})^e_{\ph{a}a}(Y_{\ell-1})^a_{\ph{b}b}(X_{\ell-1})^b_{\ph{c}c}(Y_{\ell-1})^c_{\ph{d}d}\cdots Y_2Y_1I_1)^i_{\ph{j}j} \nonumber\\
   &\qquad\qquad -
    (J_1X_1X_2\cdots (X_{\ell-1})^e_{\ph{a}a}(Y_{\ell-1})^a_{\ph{d}d} \cdots Y_2Y_1I_1)^i_{\ph{j}j} (\chi_{\ell-1}\psi_{\ell-1}-\chi_\ell\psi_\ell) ~.
\end{align}
The path corresponding to the second line encircles the node $v_\ell$ and loops once around the segment connecting this to $v_{\ell-1}$, whereas that corresponding to the third line is of type i), but shorter than the initial path. Repeating the procedure for the pairs $X_{\ell-1} Y_{\ell-1}$ and continuing recursively will further reduce equation \eqref{eq:typeIpath-Intermediate1} to 
\begin{align}\label{eq:typeIpath-Intermediate2}
   & (J_1X_1X_2\cdots (X_{\ell})^a_{\ph{b}b}(Y_{\ell})^b_{\ph{c}c}\cdots Y_2Y_1I_1)^i_{\ph{j}j} = 
    (J_1X_1(X_2Y_2)^{\ell-1}Y_1I_1)^i_{\ph{j}j} \nonumber\\
   &\qquad\qquad +
    (J_1X_1(X_2Y_2)^{\ell-2}Y_1I_1)^i_{\ph{j}j} (-(\ell-2)\chi_2\psi_2 + \chi_3\psi_3 + \cdots + \chi_\ell\psi_\ell) \nonumber\\
   &\qquad\qquad + 
   (J_1X_1(X_2Y_2)^{\ell-3}Y_1I_1)^i_{\ph{j}j} ( (\ell-1)(\chi_2\psi_2 - \chi_3\psi_3)^2 + \mathcal{O}(\mathrm{fermion~bilinears})^2) \nonumber\\
   &\qquad\qquad + \cdots + 
   (J_1X_1 X_2Y_2 Y_1I_1)^i_{\ph{j}j} (  \mathcal{O}(\mathrm{fermion~bilinears})^{\ell-1})  ~.
\end{align}
Examining now the generic terms $J_1X_1(X_2Y_2)^n Y_1I_1$, the $2^{nd}$ moment map relation will replace $X_2Y_2$ with $Y_1X_1-I_2J_2$ and the fermion contribution $\chi_1\psi_1-\chi_2\psi_2$, thus leading to linear combinations of products  
\begin{align}\label{expr:genericbosoniccurrent}
J_1X_1(Y_1X_1)^{k_1} (I_2J_2)^{l_1} (Y_1X_1)^{k_2} (I_2J_2)^{l_2} \cdots  (Y_1X_1)^{k_r} (I_2J_2)^{l_r} Y_1 I_1 (\chi_1\psi_1-\chi_2\psi_2)^m
\end{align}
where the $k_\bullet$, $l_\bullet$ and $m$ exponents sum to $n$. The $1^{st}$ moment map relation then replaces all products $X_1Y_1$ with $I_1J_1$ and $\chi_0\psi_0-\chi_1\psi_1$ and thereby generates combinations of bosonic fields of the form 
\begin{align}\label{expr:genericbosoniccurrent2}
    (J_1I_1)^{p_1} J_1X_1I_2 (J_2I_2)^{p_2} J_2Y_1I_1 (J_2I_2)^{p_3} J_2 \cdots  J_2 Y_1 I_1 
\end{align}
or 
\begin{align}\label{expr:genericbosoniccurrent3}
    (J_1I_1)^{p_1} J_1X_1I_2 (J_2I_2)^{p_2} J_2Y_1I_1 (J_2I_2)^{p_3}J_2 \cdots  J_1 X_1 Y_1 I_1  ~.
\end{align}
In the latter case it is necessary to use $1^{st}$ moment map to replace the final remaining product $X_1 Y_1$, but apart from that, the combinations of bosonic fields that appear are precisely 
\begin{equation}\label{expr:simpleBRSTcurrents}
    J_1I_1~, \qquad J_1X_1I_2 ~, \qquad J_2Y_1I_1 ~, 
\end{equation}
which enter the BRST-invariant currents \eqref{eqdef:genericBRSTbgen}, and the $U(1)$-valued current $J_2I_2$. Considering together the traces of the moment map relations at all of the gauge nodes allows to express $J_2I_2$ in terms of $\Tr(J_1I_1)$, showing that it is not an additional independent current. When restricting the quiver depicted in Fig.~\ref{fig:quiver for T^[2,1^{n-2}]_[1^n][SU(n)] theory} to have only three gauge nodes, the most simple example of its type, one can also track the fermion bilinear contributions explicitly
\begin{align}\label{rel:G11squared}
    & (J_1X_1Y_1I_1)^i_{\ph{j}j} = -  (G^{(1,1)}G^{(1,1)})^i_{\ph{j}j} + (G^{(1,1)})^i_{\ph{j}j} (\chi_0\psi_0 + \chi_1\psi_1) - \delta^i_{\ph{j}j} \chi_0\psi_0  \chi_1\psi_1 \\
    & (J_1X_1X_2Y_2Y_1I_1)^i_{\ph{j}j} = -(G^{(1,2)} G^{(2,1)})^i_{\ph{j}j} + ((G^{(1,1)})^3)^i_{\ph{j}j} + ((G^{(1,1)})^2)^i_{\ph{j}j}(1-\chi_0\psi_0 - \chi_1\psi_1) \nonumber \\
    & \qquad \qquad \qquad ~~  + 
    (G^{(1,1)})^i_{\ph{j}j}(- (\chi_0\psi_0)^2 + (\chi_0\psi_0 - \chi_1\psi_1)^2 -4 \chi_0\psi_0  \chi_1\psi_1 - \chi_0\psi_0 + \chi_1\psi_1  )
    \nonumber \\
    & \qquad \qquad \qquad ~~  +  \delta^i_{\ph{j}j} \chi_0\psi_0 (\chi_1\psi_1 - (\chi_1\psi_1)^2) ~.  \label{rel:G11G12G21}
\end{align}
The currents \eqref{eqdef:typeIIpath}-\eqref{eqdef:typeIIIpath} of type ii), iii) and iv) are treated analogously to the type i) examples \eqref{eqdef:typeIpath}. The only variation arises when examining the generic terms $J_2(X_2Y_2)^n Y_1I_1$, $J_2(X_2Y_2)^n I_2$ and $J_1X_1(X_2Y_2)^n I_2$, which replace $J_1X_1(X_2Y_2)^n Y_1I_1$, that led to expressions of the form \eqref{expr:genericbosoniccurrent}. In the type ii) case, the counterpart to \eqref{expr:genericbosoniccurrent2}-\eqref{expr:genericbosoniccurrent3} is 
\begin{align}
    & J_2Y_1I_1 (J_1I_1)^{p_1} J_1X_1I_2 (J_2I_2)^{l_1}J_2 \cdots J_2Y_1I_1 ~\mathrm{or} ~ 
    J_2Y_1I_1 (J_1I_1)^{p_1} J_1X_1I_2 (J_2I_2)^{l_1}J_2 \cdots J_1X_1Y_1I_1 , \nonumber
\end{align}
for type iii), this becomes
\begin{align}
    & J_2Y_1I_1 (J_1I_1)^{p_1} J_1X_1I_2 (J_2I_2)^{l_1}J_2 \cdots J_2I_2 ~\mathrm{or} ~ 
    J_2Y_1I_1 (J_1I_1)^{p_1} J_1X_1I_2 (J_2I_2)^{l_1}J_2 \cdots J_1X_1I_2 ~, \nonumber
\end{align}
and for type iv), 
\begin{align}
    & (J_1I_1)^{p_1} J_1X_1I_2 (J_2I_2)^{l_1}J_2 \cdots J_2I_2 ~\mathrm{or} ~ 
     (J_1I_1)^{p_1} J_1X_1I_2 (J_2I_2)^{l_1}J_2 \cdots J_1X_1I_2 ~. \nonumber
\end{align}
The basic combinations of bosonic fields that appear are still only those listed in \eqref{expr:simpleBRSTcurrents}. In the simple example where the quiver has only three gauge nodes, we thus find 
\begin{equation}\label{rel:G11G21}
    (J_2X_2Y_2Y_1I_1)^1_{\ph{j}j} = - (G^{(2,1)} G^{(1,1)})^1_{\ph{j}j} + (G^{(2,1)})^1_{\ph{j}j} (\chi_2\psi_2 - J_2 I_2) ~.
\end{equation}
It follows that generic operators associated to paths of the type \eqref{J_Phi_I} on the quiver diagram in Fig.~\ref{fig:quiver for T^[2,1^{n-2}]_[1^n][SU(n)] theory}
are generated by those listed in \eqref{expr:simpleBRSTcurrents} using the moment map relations \eqref{eqdef:genericmomentmaprelations} and accounting for fermion bilinear-$\chi\psi$ contributions
\begin{itemize}
    \item $J_1I_1$ is the simplest form of \eqref{eqdef:typeIpath} and enters the definition of $G^{(1,1)}$ \eqref{three_Gs},
    \item $J_2Y_1I_1$ is the simplest form of \eqref{eqdef:typeIIpath} and enters the definition of $G^{(2,1)}$ \eqref{three_Gs}, 
    \item $J_1X_1I_2$ is the simplest form of \eqref{eqdef:typeIVpath} and enters the definition of $G^{(1,2)}$ \eqref{three_Gs}, 
\end{itemize}
with all other operators \eqref{eqdef:typeIpath}-\eqref{eqdef:typeIVpath} generated from these by relations of the form \eqref{rel:G11squared}, \eqref{rel:G11G12G21} and \eqref{rel:G11G21}.

\section{Criteria for two Virasoro operators coincide}
\begin{lem}\label{lem:criterion for two Virasoro coincide_1}
Suppose that $T_1$ and $T_2$ are two Virasoro elements in a vertex algebra $V$, i.e. $T_i, i=1,2$ satisfy the OPE $$T_i(z)T_i(w)\sim \frac{c_i/2}{(z-w)^4}+\frac{2T_i(w)}{(z-w)^2}+\frac{\partial T_i(w)}{z-w},\quad c_i\in \mathbb C.$$ Assume that $T_1-T_2$ is in the center of $V$, then $T_1=T_2$.
\end{lem}

\begin{proof}
Since $T_1-T_2$ is in the center of $V$, we have $(T_1(z)-T_2(z))T_2(w)\sim 0$, which implies that
\begin{align*}
    T_1(z)T_2(w)\sim \frac{c_2/2}{(z-w)^4}+\frac{2T_2(w)}{(z-w)^2}+\frac{\partial T_2(w)}{z-w}.
\end{align*}
In particular $T_{1,(1)}T_2=2T_2$. Then Borcherds identity implies that $T_{2,(1)}T_1=2T_2$. However, if we swap $1\leftrightarrow 2$, we get $T_{2,(1)}T_1=2T_1$, thus we must have $T_1=T_2$.
\end{proof}

\begin{lem}\label{lem:criterion for two Virasoro coincide_2}
Let $V=\bigoplus_{\Delta\ge 0}V_{\Delta}$ be a graded vertex algebra (assume $V_{0}=\mathbb C\cdot|0\rangle$), and assume that the following symmetric pairing $V_1\otimes V_1\to \mathbb C$ is non-degenerate:
\begin{align*}
    a\otimes b\mapsto a_{(1)}b~.
\end{align*}
Suppose that $T_1$ and $T_2$ are two stress-energy operators in $V$ such that both of them are compatible with the grading, i.e. $L_0$ of the $T_1$ and $T_2$ agree and equal to the grading operator on $V$. Assume moreover that $V_1$ is primary with respect to both $T_1$ and $T_2$, then $T_1=T_2$.
\end{lem}

\begin{proof}
Since $T_1$ and $T_2$ are compatible with the grading, we have $T_1,T_2\in V_2$. Since $T_1$ is a stress-energy operator, we have
\begin{align*}
    T_1(z)T_2(w)\sim \frac{\lambda}{(z-w)^4}+\frac{a(w)}{(z-w)^3}+\frac{2T_2(w)}{(z-w)^2}+\frac{\partial T_2(w)}{z-w},\quad \lambda\in \mathbb C,\;a\in V_1~.
\end{align*}
Since $T_2$ is a stress-energy operator, we have
\begin{align*}
    T_2(z)T_1(w)\sim \frac{\lambda'}{(z-w)^4}+\frac{b(w)}{(z-w)^3}+\frac{2T_1(w)}{(z-w)^2}+\frac{\partial T_1(w)}{z-w},\quad \lambda'\in \mathbb C,\;b\in V_1~.
\end{align*}
Using the Borcherds identity we find $\lambda=\lambda'$, and $a=-b$, and $T_1=T_2+\frac{1}{2}\partial a$. Now pick an arbitrary $d\in V_1$, it is primary of weight $1$ with respect to both $T_1$ and $T_2$, thus $(T_1(z)-T_2(z))d(w)\sim 0$. On the other hand,
\begin{align*}
    (T_1(z)-T_2(z))d(w)=\frac{1}{2}\partial a(z)d(w)\sim\frac{-a_{(1)}d}{(z-w)^3}+\text{lower order poles},
\end{align*}
thus we must have $a_{(1)}d=0$ for all $d\in V_1$. It follows from assumption on the nondegeneracy of the pairing $a\otimes d\mapsto a_{(1)}d$ that $a=0$, whence $T_1=T_2$. 
\end{proof}

\bibliographystyle{JHEP}
\bibliography{Bibliography}

\end{document}